%% file: main.tex
\documentclass[11pt]{article}

\usepackage[tt=false]{libertine}

\newif\ifhidecomments
\hidecommentstrue

\input{tex-assets/macros}

\input{components/ToC}

\title{\textbf{
High-Dimensional Robust Change-Point Detection\\ via Angular Kernel Statistics
}}

\input{components/author-info}

\usepackage{setspace}
\begin{document}

\begin{bibunit}

\maketitle
\vspace{1em}

\begingroup
\setstretch{1.15}
\input{abstract}
\endgroup

\makeatletter
\begingroup
\renewcommand{\thefootnote}{}
\renewcommand{\@makefntext}[1]{\noindent #1}
\footnotetext{%
\makebox[1.5em][r]{\(\S\)}\,
\textrm{JRC :} \texttt{\href{mailto:jchoudhury3@gatech.edu}{jchoudhury3@gatech.edu}}\\
\makebox[1.5em][r]{\(\star\)}\,
\textrm{YX :} \texttt{\href{mailto:yao.xie@isye.gatech.edu}{yao.xie@isye.gatech.edu}}%
}
\endgroup
\makeatother

\newpage

\begingroup
\setstretch{1.15}
\setcounter{tocdepth}{3}
\tableofcontents
\endgroup



\section{Introduction}
\label{sec:intro}

Change-point detection is a classical problem in statistics and machine learning, which aims to determine whether, and if so where, the distribution underlying a sequence of observations changes abruptly. Such problems arise naturally in a wide range of applications across domains, including monitoring tasks (e.g., manufacturing, quality control, security systems), signal processing, finance, genomics, and more. There is a rich literature in change-point detection in classical (low-dimensional) settings, with roots in quality control and sequential analysis going back to the works of ~\citet{Shewhart1930}, \citet{page1954continuous}, \citet{Lorden1971optimal}, among others. In many modern applications, however, one often encounters high-dimensional, low-sample-size (HDLSS) data, in which the dimensionality $d$ is large while the sequence length $N$ is small and effectively fixed. This setting is increasingly common in several disciplines, yet poses challenges that are fundamentally different from those encountered in classical asymptotic regimes. For instance, genetic microarray studies routinely measure tens of thousands of features from only a few dozen samples, and in medical imaging, high-resolution scans are sometimes available for only a small group of patients with extremely rare conditions.

The HDLSS regime is statistically delicate for two main reasons. First, when the sequence length is small and effectively fixed, estimation of means, covariances, likelihood scores, and other moment-based quantities becomes inherently unstable. Second, standard Euclidean distance-based discrepancies may behave poorly because of \emph{distance concentration} (\citet{aggarwal2001distconc}, \citet{hall2005hdlss}) in high dimensions. As a result, classical methods developed under low-dimensional or large-sample asymptotics become unreliable in this regime. Early foundational work of \citet{hall2005hdlss} revealed that HDLSS data often exhibit rigid geometric structure, fundamentally altering notions of distance and separation. This insight has motivated a growing body of HDLSS-specific methods for problems such as two-sample testing \citep{BiswasGhosh2014test}, classification (\citet{RoyEtAl2022hdlss}, \citet{ShenErYin2022hdlss}, \citet{raychoudhury2023robust}), and clustering (\citet{LiuEtAl2008ClusteringHDLSS}, \citet{SarkarGhosh2020ClusteringHDLSS}).

In this paper, we primarily study offline single-change-point detection in the HDLSS regime within a nonparametric, moment-free, tuning-free framework. Our theory is developed under a single-change model in which the pre- and post-change distributions differ through their coordinate-wise marginals. The method is designed to detect distributional shifts that are reflected in aggregated one-dimensional marginal behavior across coordinates. Pure dependence, with unchanged one-dimensional marginals, lies outside its target. We later extend the same construction to a sequential (online) setting with streaming high-dimensional data, essentially via a sliding-window scan obtained by repeatedly applying the main offline procedure, but with its own theoretical analysis.

Our starting point is a discrepancy functional built from one-dimensional angular kernels, then averaged over coordinates. The construction is motivated by \emph{generalized energy distance} (\citet{szekely2013energy}, \citet{sejdinovic2013energy}) equipped with an angular kernel, developed in the projection-averaging framework of \citet{kim2020robust}. The specific discrepancy we employ is based on a dimension-averaged angular kernel introduced for HDLSS classification by \citet{raychoudhury2023robust}. Operationally, for each candidate split, we compare the observations before and after the split using a coordinate-aggregated angular kernel statistic, referred to as the \underline{d}imension-averaged \underline{a}ngular \underline{k}ernel (DAK) scan statistic, and then scan over all admissible split locations to identify the maximizing split location. The discrepancy also admits a pair-dependent square-MMD representation, with a distribution-adaptive kernel rather than a single fixed kernel chosen a priori.

The angular kernel has several useful properties. First, the underlying one-dimensional sign kernel is bounded in $[0,1]$, so the resulting discrepancy is well-defined without any moment assumptions and remains valid even for heavy-tailed or mixture distributions whose moments might not even exist. Second, the procedure is invariant under coordinate-wise rescaling, which is helpful when different features operate on different scales. Third, unlike many kernel-based methods, the DAK scan itself does not require hyperparameter tuning. A notable feature of our approach is that it admits a very explicit HDLSS theory. Under a single-change alternative, we show that the population scan curve factorizes into a deterministic (distribution-free) shape function that depends only on the split geometry, and a scalar signal term capturing the aggregated marginal discrepancy. This yields a universal unimodal population form with its unique maximum at the true change-point. Under the null, we derive an exact covariance structure for the scan vector up to a single long-run variance factor, leading to an HDLSS multivariate central limit theorem under suitable mixing conditions, and asymptotically valid calibration.

Our work is related to several strands of the existing change-point literature. A large body of work studies high-dimensional change-point problems for structured parameters such as means and covariances (\citet{wang2018inspect}, \cite{AvanesovBuzun2018}, \citet{EnikeevaHarchaoui2019SparseCPD}), often under sparsity, moment, or dependence assumptions. A separate nonparametric line uses kernel-based change-point methods (\citet{HarchaouiEtAl2008KCPA}, \citet{GarreauArlot2018}, \citet{TruongEtAl2019GreedyKernelCPD}, \citet{ArlotEtAL2019KCPA}, \citet{SongChen2024KernelCPD}) including MMD-based scanners (\citet{LiXieDaiSong2015MStatCPD}, \citet{li2019scanbstat}, \citet{WeiXie2026OnlineKernelCUSUM}), to detect more general distributional shifts. A complementary family uses energy distances and other distance-based two-sample functionals based on the $\ell_2$-norm (\citet{MattesonJames2014}, \citet{BonieceEtAl2025energy}) for nonparametric change-point detection. By contrast, the change-point literature within a truly HDLSS setting remains relatively underexplored. Very recently, there have been some studies, particularly on HDLSS change-point detection using pairwise or interpoint distances (\citet{JunLi2020CPDinterpointCLT}, \citet{Drikvandi2025changepoint}, \citet{ghoshal2025highdimCPD}) and clustering-based (\citet{Dawn2025changepoint}) approaches.

Despite these advancements, the existing literature still leaves a gap for truly HDLSS, nonparametric change-point detection. Many available methods target structured changes in means or covariances, or rely on moment-based separation conditions, multivariate estimation, or hyperparameter tuning choices that become unreliable when $N$ is fixed and $d \to \infty$. This is especially problematic when the relevant moment information is unknown, difficult to estimate from small batches, or does not exist because of heavy tails.
This motivates the search for a robust, tuning-free, moment-agnostic procedure for the HDLSS regime in the offline setting, and high-dimensional streaming data in the online setting. We aim to address the following fundamental question:

\vspace{2mm}
\begin{minipage}{0.95\textwidth}
\begin{center}
\emph{Can we design a robust, nonparametric, hyperparameter-free, moment-agnostic change-point detection procedure for HDLSS data that admits principled theoretical guarantees in high dimensions?}
\end{center}
\end{minipage}
\vspace{2mm}

%

To the best of our knowledge, our work is among the first nonparametric change-point frameworks designed for detecting HDLSS marginal distributional shifts that accommodate distributions without relying on their moments and without hyperparameter tuning, while admitting explicit fixed-\(N\), large-\(d\) calibration theory. The sequential extension further broadens the scope of the method.
The inherent generality of our approach allows it to accommodate a broad class of pre- and post-change distributions, including heavy-tailed distributions and distributions without finite moments. By leveraging a high-dimensional CLT to derive a tractable limiting null distribution, we bypass computationally intensive permutation-based or bootstrap-based threshold calibration. Our main contributions are briefly summarized below.
\begin{itemize}
    \item \textit{Methodology.} We introduce a tuning-free angular kernel scan for detecting marginal distribution shifts in HDLSS data that compares the two sides of a candidate split by averaging one-dimensional angular comparisons across coordinates. The procedure is nonparametric, well-defined without moment assumptions, and invariant under coordinate-wise rescaling.

    \item \textit{Exact (non-asymptotic) structural results.} Under a single-change model with any $d$ and $N$, we derive an exact factorization of the population scan curve into a deterministic (distribution-free) shape term and a scalar distributional signal term. Under the null, we also obtain an exact covariance characterization of the scan vector up to a single long-run variance factor.

    \item \textit{HDLSS asymptotic theory.} Under standard cross-coordinate weak-dependence and mixing conditions, we establish consistency of the scan statistic and the estimated change-point location, prove an HDLSS multivariate CLT under $\alpha$-mixing across coordinates, devise an asymptotically distribution-free test and type-I error control, and derive power and localization guarantees.

    \item \textit{Sequential extension for high-dimensional data.} We extend the offline method to an online change-point monitoring procedure based on sliding-window scans, thereby transferring the proposed discrepancy to a high-dimensional streaming data setting (without a low sample size requirement), with controlled ARL scaling and conditional detection delay bounds.
\end{itemize}

The remainder of the paper is organized as follows. Section~\ref{sec:angular-kernel} introduces the population angular kernel discrepancy together with its generalized energy and pair-dependent MMD representations. Section~\ref{sec:methodology} develops the offline scan procedure. Section~\ref{sec:theory} presents the theoretical analysis for the offline method, including the exact mean and covariance structure, high-dimensional consistency, level-\(\alpha\) calibration, and power analysis. Section~\ref{sec:online} describes the online detection method along with its theoretical guarantees. Section~\ref{sec:experiments} reports simulation studies and real-data experiments. We conclude in Section~\ref{sec:discussion}. All technical proofs and some additional simulations are deferred to the appendix.


\section{Angular kernel discrepancy and its two representations}
\label{sec:angular-kernel}

In this section, we formally introduce the population discrepancy underlying our procedure. We first define a dimension-averaged angular kernel (DAK) discrepancy that compares the one-dimensional marginals of two $d$-dimensional distributions $F_d$ and $G_d$, and then averages across coordinates. The discrepancy is a generalized energy distance generated by a bounded angular kernel, and it also admits an exact \emph{pair-dependent} square-MMD representation. The point is to isolate the population object that is later estimated and scanned over candidate change-point locations.

\subsection{A generalized energy distance via dimension-averaged angular kernel}

For scalars $p,q,r\in\mathbb R$, define the one-dimensional angular kernel:
\begin{equation}
\label{eq:one-dim-angular-kernel}
\rho_0(p,q;r) := \frac{1}{\pi} \cos^{-1} \left\{\frac{(p-r)(q-r)}{|p-r|\cdot|q-r|}\right\} = \mathbf 1\{(p-r)(q-r)<0\},
\end{equation}
with the convention that $\rho_0(p,q;r)=0$ whenever $p=r$ or $q=r$. Thus, $\rho_0(p,q;r)$ is essentially a sign kernel which records $1$ if the anchor $r$ lies strictly between $p$ and $q$; $0$ otherwise.

Fix $\alpha\in[0,1]$. Let $F_d$ and $G_d$ be probability distributions on $\mathbb R^d$, and for each $k \in [d] := \{1,2,\ldots,d\}$, let $F_d^{(k)}$ and $G_d^{(k)}$ denote their $k$-th marginals. For each coordinate $k\in[d]$, let
\(
M_k^{*} \sim \alpha F_d^{(k)} + (1-\alpha)G_d^{(k)} := H_{\alpha,d}^{(k)},
\)
and define the corresponding anchor-averaged coordinate-wise pseudometric
\[
\rho_{F_d,G_d}^{(k)}(u,v)
:=
\mathbb E_{M_k^{*} \sim H_{\alpha,d}^{(k)}}\!\left[\rho_0(u,v;M_k^{*})\right],
\qquad u,v\in\mathbb R.
\]
\indent
We then aggregate across coordinates via averaging over the dimensions:
\begin{equation}
\label{eq:bar-rho-main}
\overline{\rho}_{F_d,G_d}(x,y)
:=
\frac{1}{d}\sum_{k=1}^d \rho_{F_d,G_d}^{(k)}(x_k,y_k),
\qquad x,y\in\mathbb R^d.
\end{equation}
\indent
This quantity is bounded, depends only on the one-dimensional marginals of $(F_d,G_d)$, and is invariant under coordinate-wise rescaling.
Let $X,X'\stackrel{\mathrm{i.i.d.}}{\sim}F_d$ and $Y,Y'\stackrel{\mathrm{i.i.d.}}{\sim}G_d$, all mutually independent. We define the population dimension-averaged angular discrepancy by
\begin{align}
\Delta_{\alpha,d}(F_d,G_d)
&:=
2\,\mathbb{E}\big[\overline{\rho}_{F_d,G_d}(X,Y)\big]
-\mathbb{E}\big[\overline{\rho}_{F_d,G_d}(X,X')\big]
-\mathbb{E}\big[\overline{\rho}_{F_d,G_d}(Y,Y')\big] \label{eq:Delta-main}\\[2mm]
&=
\frac{1}{d}\sum_{k=1}^d
\bigg[
2\,\mathbb{E}\!\left[\rho_{F_d,G_d}^{(k)}(X_k,Y_k)\right]
-\mathbb{E}\!\left[\rho_{F_d,G_d}^{(k)}(X_k,X_k')\right]
-\mathbb{E}\!\left[\rho_{F_d,G_d}^{(k)}(Y_k,Y_k')\right]
\bigg]. \label{eq:Delta-coordinatewise}
\end{align}

Thus, the discrepancy \eqref{eq:Delta-main} is obtained by first comparing the one-dimensional marginals coordinate-by-coordinate, and then averaging the resulting discrepancies \eqref{eq:Delta-coordinatewise} over $k \in [d]$. This marginal aggregation is the key structural feature exploited in the HDLSS regime. 

\begin{proposition}[Properties of the population discrepancy]
\label{lem:signal_factor}
For any \(d\in\bbN\), the following statements hold.
\begin{enumerate}[label=(\alph*)]
\item The function \(\overline\rho_{F_d,G_d}\) is a bounded pseudometric of negative type on \(\bbR^d\); and thus \(\Delta_{\alpha,d}(F_d,G_d)\) is a well-defined generalized energy distance induced by \(\overline\rho_{F_d,G_d}\).

\item Suppose that, for every \(k\in[d]\), the probability measures associated with the coordinate marginals \(F_d^{(k)}\) and \(G_d^{(k)}\) have no atoms; equivalently, the marginal CDFs are continuous. Then \(\Delta_{\alpha,d}(F_d,G_d)\) is independent of \(\alpha\in[0,1]\). That is,
\(
~\Delta_{\alpha,d}(F_d,G_d)=\delta_d^\ast~
\)
for all $\alpha\in[0,1]$, where $\delta_d^{*}$ is given by
\begin{equation}
\label{eq:delta-simplified-cvm}
\delta^{*}_d
=
\frac{2}{d}
\sum_{k=1}^d
\int_{\mathbb R}
\bigl\{F_d^{(k)}(z)-G_d^{(k)}(z)\bigr\}^2 \, dF_d^{(k)}(z)\,.
\end{equation}

\item Under the same atom-free condition as in (b), 
\(
\delta_d^{*} \in \big[0,\nicefrac{2}{3}\big]
\), and the upper bound is sharp.
Moreover,
\[
\delta_d^\ast=0
\quad\Longleftrightarrow\quad
F_d^{(k)}=G_d^{(k)}~~ \text{for every}\,\, k\in[d].
\]
\end{enumerate}
\end{proposition}

An equivalent construction of $\delta^{*}_d$ can be found in \Cref{app-subsec:delta_d_construction} through the expectations of the form $\{\E[\rho_0(U_k,V_k;W_k)]\}_{k=1}^d$ where $U_k,V_k,W_k$ are distributed as $F_d^{(k)}$ or $G_d^{(k)}$ which is key to deriving a closed-form product decomposition of $\E[\mathfrak{W}_d(t)]$ (see \Cref{prop:E[D-star]}). The following two examples illustrate that \(\delta_d^\ast\) captures both classical moment-based changes and genuinely moment-free changes.

\begin{example}[Gaussian mean shift]
\label{ex:ex1-gaussian-mean-shift}
Suppose $F_d = \mcN(0,I_d)~$ and $~G_d = \mcN(\mu,I_d)$, where $\mu=(\mu_1,\dots,\mu_d)^\top$ $\in \bbR^d$. Then,
\begin{equation}
\label{eq:delta-gaussian-shift}
\delta^{*}_d
=
\frac{2}{d}\sum_{k=1}^d
\E_{Z\sim \mcN(0,1)}\,\big(\Phi(Z)-\Phi(Z-\mu_k)\bigr)^2,
\end{equation}
where \(\Phi\) denotes the standard Gaussian CDF.
Moreover, if
\(\|\mu\|_\infty\to0\), then
\(
\delta^{*}_d = \frac{1}{\pi\sqrt3}\frac{\|\mu\|_2^2}{d} + o\!\left(\frac{\|\mu\|_2^2}{d}\right).
\)
\end{example}

\begin{example}[Cauchy scale change]
\label{ex:ex2-cauchy-scale}
Suppose \(F_d=\mathrm{Cauchy}(0,1)^{\otimes d}\) and \(G_d=\mathrm{Cauchy}(0,\lambda)^{\otimes d}\), for some \(\lambda>0\). Then
\begin{equation}
\label{eq:delta-cauchy-scale-integral}
\delta^{*}_d
=
\frac{2}{\pi^3}
\int_{\mathbb R}
\frac{\left[
\tan^{-1}(z)-\tan^{-1}(z/\lambda)
\right]^2}{1+z^2}\,dz 
=
\frac{1}{\pi^2}
\operatorname{Li}_2\!\left[
\left(\frac{\lambda-1}{\lambda+1}\right)^2
\right],
\end{equation}
where \(\operatorname{Li}_2(x)=\sum_{m=1}^\infty x^m/m^2\) denotes the dilogarithm function. Although the Cauchy marginals have no finite moments, \(\delta^{*}_d\) is still well-defined. In particular, for every fixed \(\lambda\neq1\), \(\delta^{*}_d\) is a positive constant.
\end{example}

\begin{remark}
    Since \(\delta_d^{*}\) is obtained by averaging coordinate-wise one-dimensional discrepancies, it is sensitive to changes in the collection of one-dimensional marginals. Under atom-free marginal distributions, $\delta_d^{*} \geq 0$, with equality iff \(F_d^{(k)}=G_d^{(k)}\) for all \(k\in[d]\). Thus, the discrepancy vanishes when all one-dimensional marginals agree, even if the joint dependence structures differ.

    This restriction is deliberate. In HDLSS settings, fully multivariate discrepancies can be difficult to estimate and calibrate when the number of observations is small. The present construction gains stability by averaging the one-dimensional signals across coordinates.
\end{remark}

\subsection{Representation as pair-dependent square-MMD}

The discrepancy in \eqref{eq:Delta-main} admits an exact square-MMD representation, but in a pair-dependent sense. Fix any non-zero reference point $\beta\in\mathbb R^d \setminus \{0\}$ and define the kernel
\begin{equation}
\label{eq:kernel-main}
\mathbf{k}_{F_d,G_d}^{(\beta)}(x,y)
:=
\frac{1}{2}\Big\{
\overline{\rho}_{F_d,G_d}(x,\beta)
+
\overline{\rho}_{F_d,G_d}(y,\beta)
-
\overline{\rho}_{F_d,G_d}(x,y)
\Big\},
\qquad x,y\in\mathbb R^d.
\end{equation}


\begin{proposition}[Exact pair-dependent MMD representation of $\delta_d^{*}$]
\label{prop:pair-dependent-mmd-main}
Recall the dimension-averaged angular kernel discrepancy $\delta_d^{*}$ defined in \eqref{eq:delta-simplified-cvm}. For any pair $(F_d,G_d)$ with atom-free marginal laws, it admits the following square-MMD representation with the pair-dependent kernel $\mathbf{k}_{F_d,G_d}^{(\beta)}$:
\begin{equation}
\label{eq:Delta-mmd-main}
\delta_d^{*}
\,=\,
2\,\mathrm{MMD}^2 (F_d,G_d ; \mathbf{k}_{F_d,G_d}^{(\beta)}),
\end{equation}
where 
\(
\displaystyle
~\mathrm{MMD}^2 (F_d,G_d ; \mathbf{k}_{F_d,G_d}^{(\beta)})
:=
\mathbb{E}\!\left[\mathbf{k}_{F_d,G_d}^{(\beta)}(X,X')\right]
+ \mathbb{E}\!\left[\mathbf{k}_{F_d,G_d}^{(\beta)}(Y,Y')\right]
- 2\,\mathbb{E}\!\left[\mathbf{k}_{F_d,G_d}^{(\beta)}(X,Y)\right].
\)

\noindent
Moreover, the value of the right-hand side does not depend on the choice of the reference point $\beta$.
\end{proposition}

Notably, the representation of $\delta_d^{*}$ in \eqref{eq:Delta-mmd-main} is pair-dependent, i.e., the kernel \(\mathbf{k}_{F_d,G_d}^{(\beta)}\) depends on \((F_d,G_d)\) through the anchor laws \(\alpha F_d^{(k)}+(1-\alpha)G_d^{(k)}\) for all \(k\in[d]\). 
This viewpoint is complementary to the generalized energy representation \eqref{eq:Delta-main} in the previous subsection.



\section{Methodology: Offline change-point detection}
\label{sec:methodology}

We work under a \emph{nonparametric, offline, single change-point model with focus on marginal distribution shifts}, i.e., in which the pre-change and post-change distributions differ in at least one one-dimensional marginal. Formally, for fixed dimension $d$ we consider the following change-point detection or testing problem:
\begin{align}
\label{eq:single-cp}
\mathbf{H}_{0,d} &: \quad Z_1, Z_2, \dots, Z_N \overset{\text{i.i.d.}}{\sim} F_d, \quad \text{vs.} \notag \\
\mathbf{H}_{1,d} &: \quad \exists\, \tau \in \mcT := \{2,3,\dots,N-2\} \ \text{such that}\ 
\begin{cases}
Z_1, Z_2, \dots, Z_\tau \overset{\text{i.i.d.}}{\sim} F_d, \\
Z_{\tau+1}, Z_{\tau+2}, \dots, Z_N \overset{\text{i.i.d.}}{\sim} G_d,
\end{cases}
\end{align}
where $F_d$ and $G_d$ are probability distributions on $\mathbb{R}^d$ such that $F_d^{(k)} \neq G_d^{(k)}$ for at least one $k \in [d]$. No parametric or moment assumptions are imposed on \(F_d\) or \(G_d\) beyond those stated explicitly later for theoretical analysis. The distributions may differ in arbitrary ways, including in higher-order structure or tail behavior. Our asymptotic results study this sequence of testing problems as \(d\to\infty\), with \(N\) fixed. 

\paragraph{Empirical DAK discrepancy.}
We first define an empirical analogue of the population discrepancy, following \citet{raychoudhury2023robust}. Let
$\mcX = \{X_1, \ldots, X_{n_1}\}$ and $\mcY = \{Y_1, \ldots, Y_{n_2}\}$. Write $\widehat{\alpha} = n_1 / n$.
For any pair of vectors \( \mathbf{u} = (u_k)_{k \in [d]}\) and \(\mathbf{v} = (v_k)_{k \in [d]}\), we define the pooled-anchor empirical kernel:

\begingroup
\setlength\abovedisplayskip{-2pt}
\setlength\belowdisplayskip{10pt}
\begin{equation*}
\widehat{\rho}_{\mcX,\mcY}^{(k)}(u_k, v_k)
:=
\frac{1}{n_1+n_2}
\left[
\sum_{i=1}^{n_1}\rho_0(u_k,v_k;X_{i,k})
+
\sum_{j=1}^{n_2}\rho_0(u_k,v_k;Y_{j,k})
\right],
\quad\text{for all}~~ k\in[d],
\end{equation*}
\endgroup
where \(\rho_0\) is the one-dimensional angular kernel introduced in \eqref{eq:one-dim-angular-kernel}. Averaging across coordinates gives
\begin{equation}
\label{eq:empirical-bar-rho}
\widehat{\overline\rho}_{\mcX,\mcY}(\mathbf{u}, \mathbf{v})
:=
\frac{1}{d} \sum_{k=1}^d \widehat{\rho}_{\mcX,\mcY}^{(k)}(u_k,v_k).
\end{equation}
\indent
This is the natural empirical counterpart of \(\overline{\rho}_{F_d,G_d}\), obtained by replacing the mixture anchor law \(\alpha F_d^{(k)}+(1-\alpha)G_d^{(k)}\) with the pooled empirical anchor distribution on the \(k\)-th coordinate.
%
We use the shorthand $\widehat{\overline{\rho}}(\cdot, \cdot)$ for $\widehat{\overline\rho}_{\mcX,\mcY}(\cdot,\cdot)$, and define the following empirical quantities:
\[
\widehat{T}_{\mcX\mcX} = \frac{1}{n_1(n_1 - 1)} \sum_{\substack{X, X' \in \mcX \\ [X \neq X']}} \widehat{\overline{\rho}}(X, X'), \quad
\widehat{T}_{\mcY\mcY} = \frac{1}{n_2(n_2 - 1)} \sum_{\substack{Y, Y' \in \mcY \\ [Y \neq Y']}} \widehat{\overline{\rho}}(Y, Y'), \quad
\widehat{T}_{\mcX\mcY} = \frac{1}{n_1 n_2} \sum_{\substack{X \in \mcX, \\ Y \in \mcY}} \widehat{\overline{\rho}}(X, Y).
\]
\indent
Finally, define the empirical dimension-averaged angular kernel discrepancy (between $\mcX$ and $\mcY$) by
\begin{equation}
\label{eq:Delta_alpha}
\widehat{\Delta}_{\widehat{\alpha},d}(\mcX,\mcY) := 2\, \widehat{T}_{\mcX\mcY} - \widehat{T}_{\mcX\mcX} - \widehat{T}_{\mcY\mcY}\,,
\end{equation}
which serves as a sample analogue of the population quantity \(\Delta_{\alpha,d}(F_d,G_d)\), with \(\alpha\) replaced by the empirical proportion \(\widehat\alpha=n_1/n\), and measures the separation between the samples in \( \mcX \) and \( \mcY \).

\paragraph{Detection statistic and decision rule.}
Consider the following detection statistic.   
For each candidate change-point location \( t \in \mcT \), we split the data into two parts: $\mcX_t = \{Z_1, \ldots, Z_t\}$ and $\mcY_t = \{Z_{t+1}, \ldots, Z_N\}$, and compute the statistic with $n_1 := t$, $n_2 := N-t$, and $\widehat{\alpha}_t = t/N$: 
\begin{equation}
\label{eq:scan-stat}
\mathfrak{W}_d(t) 
\, := \, \widehat{\Delta}_{\widehat{\alpha}_t,d}(\mcX_t,\mcY_t) 
\, = \, \widehat{\Delta}_{\widehat{\alpha}_t,d}\Big( \{Z_1, \ldots, Z_t\} ,  \{Z_{t+1}, \ldots, Z_N\} \Big).
\end{equation}
\indent
For testing at level $\alpha$, we use a variance-calibrated version of the scan. Let \(\widehat\sigma_{\mathrm{long}}(N)\) denote the plug-in estimator of a long-run variance factor, constructed in \Cref{subsec:level-alpha-calibration} below. Define:
\begin{equation}
\label{eq:studentized-max}
S_d
\;:=\;
\max_{t\in\mcT}\frac{\sqrt d\,\mathfrak W_d(t)}{\widehat{\sigma}_{\mathrm{long}}(N)}.
\end{equation}
\indent
We declare a change if \(\,S_d > c_\alpha\), where \(c_\alpha\) is a fixed threshold that can be computed independently of the observed data. Its definition, based on the asymptotic null distribution of $S_d$, is given in \Cref{subsubsec:adf-test}.



Finally, once a change-point has been detected, we can estimate the actual change-point location by the value of $t \in \mcT$ where $\{\mathfrak{W}_d(t)\}_{t \in \mcT}$ is maximized, i.e.,
\begin{equation}
\label{eq:cp-estimate}
\widehat{\tau}_{d} 
\,\in\, \underset{t \in \mcT}{\argmax} \,\, \mathfrak{W}_{d}(t).
\end{equation}

\begin{remark}[Discussion on an alternative approach]
A similar estimator can be constructed based on the multivariate angular kernel in \citet{kim2020robust}, which avoids the dimension-averaging step in \eqref{eq:empirical-bar-rho} and tests whether $F_d \neq G_d$ in a fully multivariate way. However, its HDLSS behavior would still be governed by certain first- and second-moment-type quantities; see, e.g., \citet[Lemmas A.1 and A.2]{raychoudhury2023robust}, and thus would still require asymptotic separation in those moments of $F_d$ and $G_d$. 

By contrast, our DAK statistic \(~\mathfrak W_d\), due to the dimension-averaging, makes our method applicable to heavy-tailed or contaminated distributions, and to settings in which low-order moments of pre- and post-change distributions may be identical, infinite, or even fail to exist. This robustness comes with a clear trade-off: our method is sensitive to changes in the collection of one-dimensional marginals, but not to changes in pure dependence with completely unchanged marginals. 
This is a deliberate compromise that makes possible a tuning-free, moment-agnostic method with explicit finite-$N$, large-$d$ theory.
%

\end{remark}

\vspace{-2mm}
\paragraph{Sequential extension.} 
While the offline DAK scan is designed for HDLSS data, it also provides the basic ingredient for a simple, sequential monitoring procedure for streaming high-dimensional data. By repeatedly applying the offline statistic over small rolling windows, we exploit the low-sample regime in which the offline method is most effective while extending the approach to a high-dimensional online setting. We defer the methodology and theoretical analysis of this sequential extension to Section~\ref{sec:online}.


\section{Theoretical properties of the offline scan}
\label{sec:theory}


\subsection{Exact structural decomposition of the expectation of $\mathfrak{W}_d(t)$}

\begin{proposition}[Exact closed-form expression of mean discrepancy]
\label{prop:E[D-star]}
    Suppose $Z_1, Z_2, \dots, Z_N$ are such that $~\exists\, \tau \in \mcT$ such that $Z_1, \dots, Z_\tau \overset{\text{i.i.d.}}{\sim} F_d$, and $Z_{\tau+1}, \dots, Z_N \overset{\text{i.i.d.}}{\sim} G_d$ for some $d$-dimensional distributions $F_d$ and $G_d$ with atom-free marginal laws. Define the sample size-adjusted signal factor 
    \begin{equation}
        \delta_d := \frac{N-1}{N} \cdot \delta^{*}_d~.
    \end{equation}
    Then, the (pointwise) expected value of $\mathfrak{W}_{d}(t)$ admits a closed-form expression, given by 
    \begin{align}\label{eq:E[D-star]}
        \mu_d(t) := \E\big[\mathfrak{W}_{d}(t)\big] =
        \begin{cases}
            \, f_d(t) = \dfrac{(N - \tau)(N - \tau - 1)}{(N - t)(N - t - 1)} \cdot \delta_d & \text{~for~} t \in [1,\tau] \cap \bbN \,, \\[1.2em]
            \, g_d(t) = \dfrac{\tau(\tau - 1)}{t(t - 1)} \cdot \delta_d & \text{~for~} t \in [\tau,N] \cap \bbN \,.
        \end{cases}
    \end{align}
    %
    Moreover, $\mu_d(t) \in \left[0,\nicefrac{2}{3}\right)$ for all $t \in [N]$.
\end{proposition}

\begin{remark}[$\mu_d(t)$ is maximized at the true change-point $\tau$]
\label{rem:E[D]-inc-dec}
Under $\mathbf{H}_{0,d}$, we always have \(\mu_d\equiv0\). Under $\mathbf{H}_{1,d}$, $\exists\,k\in[d]\ \text{such that}\ F_d^{(k)}\neq G_d^{(k)}$; thus \Cref{lem:signal_factor}(c) ensures that $\delta_d > 0$, provided the coordinate marginals are atom-free. Hence, under $\mathbf{H}_{1,d}$, the mean discrepancy function $\mu_d(\cdot)$, defined in \eqref{eq:E[D-star]}, strictly increases on $[1,\tau) \cap \bbN$, and strictly decreases on $(\tau,N] \cap \bbN$. At~ $t = \tau$, we have $\mu_d(\tau) = f_d(\tau) = g_d(\tau) = \delta_d$. 

Thus, \(\mu_d(t)\) attains its unique maximum at \(t=\tau\), the true change-point location, whenever \(F_d\) and \(G_d\) differ in at least one coordinate marginal. The magnitude of the peak is governed by \(\delta_d>0\).
\end{remark}

\begin{figure}[t]
    \centering
    \includegraphics[width=\linewidth]{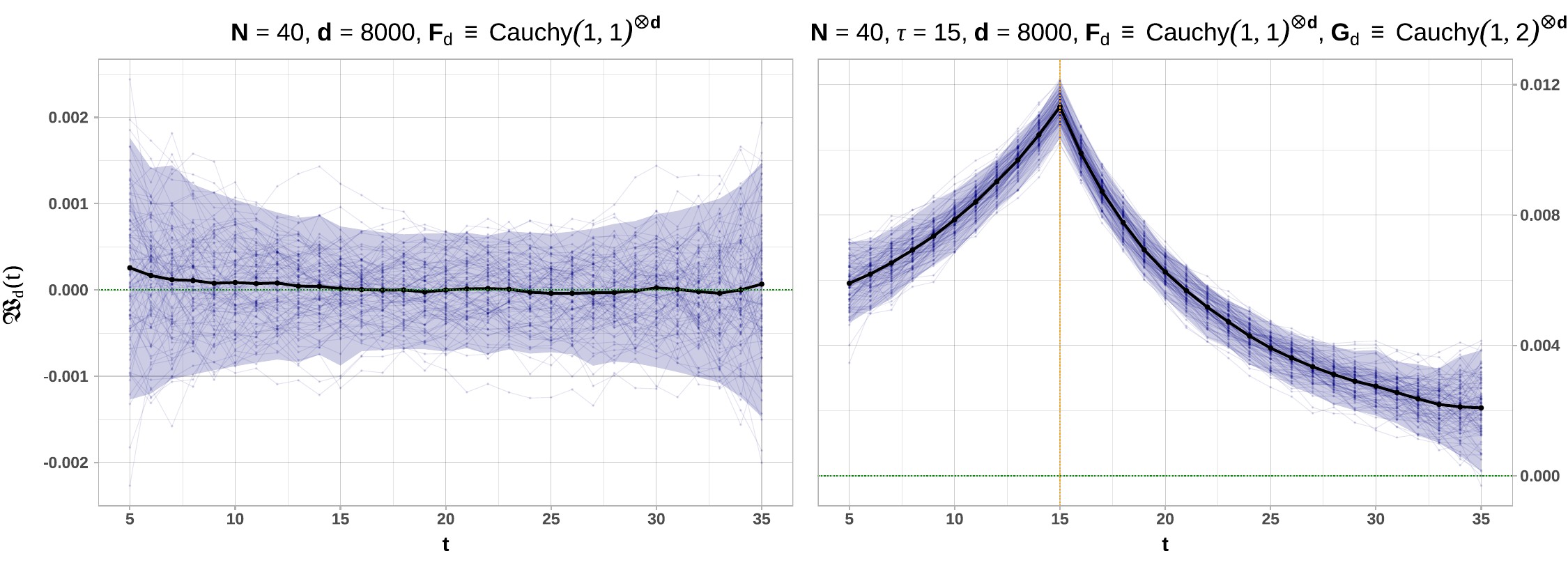}
    \captionsetup{font=footnotesize,justification=centering}
    \caption{
    Empirical trajectories and pointwise $95\%$ quantile envelopes of $\mathfrak{W}_d(t)$ with $N=40$ and $d=8000$, based on $100$ independent replications. The black curve denotes the pointwise empirical mean across replications.\\[0.8mm]
    \textbf{Left ($\mathbf{H}_{0,d}$):} $Z_1, \ldots, Z_N \sim F_d \equiv \mathrm{Cauchy}(1,1)^{\otimes d}$. 
    The process fluctuates around zero without a systematic structure.\\[0.6mm]
    \textbf{Right ($\mathbf{H}_{1,d}$):} $Z_1, \ldots, Z_\tau \sim F_d \equiv \mathrm{Cauchy}(1,1)^{\otimes d}$ and $Z_{\tau+1}, \ldots, Z_N \sim G_d \equiv \mathrm{Cauchy}(1,2)^{\otimes d}$. 
    The vertical dashed line marks the true change-point, at $\tau = 15$. 
    Under the alternative, the DAK discrepancies concentrate around a strictly unimodal signal with the mode at $\tau$, exhibiting clear distinction from the null fluctuation envelope.
    }
    \label{fig:E[D]-inc-dec}
\end{figure}

\begin{remark}[Product decomposition of the expected DAK statistic]
\label{rem:geometric_factorization}
Proposition~\ref{prop:E[D-star]}, paired with \Cref{lem:signal_factor} and \Cref{rem:E[D]-inc-dec}, reveals a remarkable structural property of the DAK discrepancy.
The population mean
$\mu_d(t)=\mathbb{E}[\mathfrak{W}_d(t)]$ factorizes as
\begin{equation}
\label{eq:structural-decomp}
\mu_d(t)
= \Lambda_{\tau,N}(t)\cdot \delta_d \, ,
\end{equation}
where the \emph{deterministic} (increasing till $\tau$, decreasing thereafter) shape function
\begin{equation}
\label{eq:det-shape-func}
\Lambda_{\tau,N}(t)
=
\begin{cases}
\, \dfrac{(N-\tau)(N-\tau-1)}{(N-t)(N-t-1)}~, &  \text{for } t \in [1,\tau] \cap \bbN \,, \\[1em]
\, \dfrac{\tau(\tau-1)}{t(t-1)}~, &  \text{for } t \in [\tau,N] \cap \bbN \,,
\end{cases}
\end{equation}
depends only on the combinatorial geometry of the split $(t,\tau,N)$, i.e., independent of $F_d$ and $G_d$. The \emph{signal factor} $\delta_d$ is a finite-$N$ rescaling of the dimension-averaged, Cramér--von Mises--type pseudometric $\delta^{*}_d$ between $F_d$ and $G_d$, that gives a measure of separation between the marginals of $(F_d,G_d)$ without any dependence on $(t,\tau)$. By \Cref{lem:signal_factor}(c), $\delta_d > 0$ under $\mathbf{H}_{1,d}$, provided the coordinate marginals are atom-free.

Therefore, under $\mathbf{H}_{1,d}$, the population DAK discrepancy curve $\mu_d(t)$ takes a universal, deterministic unimodal shape entirely determined by $(t,\tau,N)$, while its magnitude is solely governed by the scalar distributional discrepancy factor $\delta_d$ independent of $(t,\tau)$. Thus, \eqref{eq:structural-decomp} can be viewed as a product decomposition of the DAK discrepancy into \emph{shape} and \emph{signal}. The above results hold for all fixed $N$ and $d$. 
\end{remark}



\subsection{Exact covariance structure of $\{\mathfrak{W}_d(t)\}_{t \in \mcT}$, and high-dimensional consistency}

Recall the definition of $\mathfrak{W}_d(t)$ in \eqref{eq:scan-stat} which is a sum of terms of the form $\rho_0(\cdot, \cdot; Z_k)$ for all $k \in [N]$, with the anchors being $\{Z_1, \ldots, Z_N\}$. Consequently, each $\mathfrak{W}_d(t)$ is a function of the entire sample $\{Z_k\}_{k \in [N]}$, and hence it is natural to expect that they are heavily correlated. Interestingly, we show that under $\mathbf{H}_{0,d}$, for all $t, t' \in \mcT$, $\cov\!\big(\mathfrak{W}_d(t),\mathfrak{W}_d(t')\big)$ has a closed form product decomposition into a time component and a distributional component; similar to the decomposition \eqref{eq:structural-decomp} of $\mu_d(\cdot)$ as in \Cref{rem:geometric_factorization}.

\begin{proposition}[Covariance structure of $\mathfrak{W}_d(t)$ under $\mathbf{H}_{0,d}$]
\label{prop:var-cov-null}
Let $N$ be fixed and consider the offline discrepancy statistic $\mathfrak{W}_d(t)$ computed from $N$ observations in ambient dimension $d$. Under $\mathbf{H}_{0,d}$, the covariance of $\mathfrak{W}_d$ admits the following closed-form expression:
\begin{align}
\cov\!\big(\mathfrak{W}_d(t),\mathfrak{W}_d(t')\big)
&= \frac{2(N-1)(N-2)}{\,t'(t'-1)(N-t)(N-t-1)} \, \mathcal{V}_{d,N}~,
\qquad \text{for~} t' \geq t~;
\label{eq:covDt}
\end{align}
where
\(
\mathcal{V}_{d,N} := \sigma_d^2 - 2\kappa_d + \eta_d,
\)
with
$\sigma_d^2 := \var\big(\widehat{\overline{\rho}}(Z_i,Z_j)\big)$, $\kappa_d := \cov\!\big(\widehat{\overline{\rho}}(Z_i,Z_j), \widehat{\overline{\rho}}(Z_i,Z_k)\big)$, and $\eta_d := \cov\!\big(\widehat{\overline{\rho}}(Z_i,Z_j), \widehat{\overline{\rho}}(Z_k,Z_l)\big)$,
for distinct indices \(i,j,k,l \in [N]\).
Also, \(\mathcal{V}_{d,N} \in [0,1]\) is a deterministic scalar depending on the sample size $N$ and the null law $F_d$; it is independent of $(t,t')$.
\end{proposition}

\begin{assumption}
\label{ass:low-cross-cov}    
For any four $d$-dimensional random vectors $U, V, Q, Q^{*}$ with distribution $F_d$ or $G_d$, such that they are mutually independent,
\begin{itemize}
    \item[i.] $\sum_{1\leq k_1<k_2 \leq d}~\cov(\rho_0(U_{k_1},V_{k_1};Q_{k_1}), \rho_0(U_{k_2},V_{k_2};Q^{}_{k_2}))=o(d^2);$
    \item[ii.] $ \sum_{1 \leq k_1<k_2\leq d}~\cov(\rho_0(U_{k_1},V_{k_1};Q_{k_1}), \rho_0(U_{k_2},V_{k_2};Q^{*}_{k_2}))=o(d^2).$
\end{itemize}
\end{assumption} 

\Cref{ass:low-cross-cov} is the same condition that \citet{raychoudhury2023robust} worked with; and in fact it is a standard ``weak covariance structure across coordinates'' assumption in the HDLSS classification, clustering, and change-point literature (see, e.g., \citet{hall2005hdlss}, \citet{ghoshal2025highdimCPD}, \citet{SarkarGhosh2020ClusteringHDLSS}). It requires that the aggregate cross-coordinate covariance among the angular indicators vanish under a \(d^{-2}\) normalization, irrespective of whether $U$, $V$, $Q$, $Q^{*}$ arise from $F_d$ or $G_d$.
The assumption is trivially satisfied if the component variables of the $F_d$ and $G_d$ are independently distributed, and \citet{hall2005hdlss} showed that it continues to hold when the components have $\rho$-mixing property.

\begin{proposition}[Generic variance bound under $\mathbf{H}_{0,d}$ and $\mathbf{H}_{1,d}$]
\label{prop:var(D)}
    Under both $\mathbf{H}_{0,d}$ and $\mathbf{H}_{1,d}$, if \Cref{ass:low-cross-cov} is satisfied, then for each $t \in \mcT$,
    \[
    \var\left(\mathfrak{W}_{d}(t)\right) 
    \leq \frac{16}{d} + o(1) \, .
    \]
\end{proposition}


\begin{theorem}[Pointwise consistency of $\mathfrak{W}_{d}(t)$ under $\mathbf{H}_{0,d}$ and $\mathbf{H}_{1,d}$]
\label{thm:D-consistency}
Under both $\mathbf{H}_{0,d}$ and $\mathbf{H}_{1,d}$, if \Cref{ass:low-cross-cov} is satisfied, then for every $t \in \mcT$, $\mathfrak{W}_d(t)$ consistently estimates $\mu_d(t)$. That is,
\begin{equation*}
\mathfrak{W}_{d}(t) - \mu_d(t) \stackrel{\mathbb{P}}{\longrightarrow} 0 ~,\quad \text{as~} d \to\infty.
\end{equation*}
Moreover, since $\mcT$ is finite, we have $~\max_{t\in\mcT}\left|\mathfrak W_d(t)-\mu_d(t)\right| \stackrel{\mathbb{P}}{\longrightarrow} 0$.
\end{theorem}

\Cref{thm:D-consistency} establishes uniform consistency of the scan curve. Exact inference in high dimensions, however, requires a distributional limit for the centered null process. We next derive such a limit under cross-coordinate strong mixing.


\subsection[A multivariate CLT for $\{\mathfrak{W}_d(t)\}_{t \in \mcT}$ under high-dimensional strong mixing]{A multivariate CLT for $\{\mathfrak{W}_d(t)\}_{t \in \mcT}$ under high-dimensional strong mixing}
\label{subsec:clt}

In the HDLSS regime, since the number of observations $N$ is fixed while the dimension $d \to \infty$, the CLT is radically different from its classical ($N \to \infty$, fixed-$d$) counterpart. Throughout the rest of the paper, for random vectors \(\{A_d\}\) and \(A\), the notation \(A_d \stackrel{\mathscr L}{\longrightarrow} A\) denotes convergence in distribution, equivalently weak convergence, as \(d\to\infty\).

Note that the vector of scan statistics $\,\widetilde{\mathfrak{W}}_d := \big(\mathfrak{W}_d(t)\big)_{t\in\mcT}$ can be viewed as a sum of random vectors indexed by coordinate dimension. Recall from \eqref{eq:Delta_alpha} and \eqref{eq:scan-stat} that the statistic $\mathfrak{W}_d(t)$ is a linear combination of empirical angular distances between univariate distributions. We can thus decompose it as
\begin{equation}
\label{eq:coord-decomposition}
\mathfrak{W}_d(t) = \frac{1}{d} \sum_{k=1}^d \xi_k(t),
\end{equation}
where $\xi_k(t)$ denotes the univariate discrepancy statistic computed using solely the $k$-th coordinate of the samples $\{Z_1, \dots, Z_N\}$, denoted by $\mathbf{Z}^{(k)}:=(Z_{1,k},\dots,Z_{N,k})^\top$. Explicitly,
\begin{align}
\label{eq:xi_k}
\begin{aligned}
\xi_k(t) &:= \frac{2}{t(N-t)} \sum_{i=1}^t \sum_{j=t+1}^N \widehat{\rho}^{(k)}(Z_{i,k}, Z_{j,k})
- \frac{1}{t(t-1)} \sum_{1 \le i \neq j \le t} \widehat{\rho}^{(k)}(Z_{i,k}, Z_{j,k})\\
&\hspace{5cm}- \frac{1}{(N-t)(N-t-1)} \sum_{t+1 \le i \neq j \le N} \widehat{\rho}^{(k)}(Z_{i,k}, Z_{j,k}) \,,
\end{aligned}
\end{align}
where $\widehat{\rho}^{(k)}(u, v) = \frac{1}{N} \sum_{r=1}^N \rho_0(u, v; Z_{r,k})$ is the empirical one-dimensional angular kernel (or, the sign kernel) defined in \eqref{eq:one-dim-angular-kernel}.
Since $\rho_0\in \{0,1\}$ and $N$ is fixed, each univariate statistic is uniformly bounded; in particular, we have
$\sup_{k\ge 1}\sup_{t\in\mcT}|\xi_k(t)| \leq 2$.
Let $\boldsymbol{\xi}_k := (\xi_k(t))_{t\in\mcT}\in\R^{|\mcT|}$, so that \eqref{eq:coord-decomposition} implies
\begin{equation*}
\widetilde{\mathfrak{W}}_d = \frac{1}{d}\sum_{k=1}^d \boldsymbol{\xi}_k.
\end{equation*}
\indent
Under $\mathbf{H}_{0,d}$, $\E[\boldsymbol{\xi}_k]=\mathbf{0}$ for each $k$; thus $\E[\widetilde{\mathfrak{W}}_d] = \mathbf{0}$ as well, but the random vectors $\boldsymbol{\xi}_k := (\xi_k(t))_{t \in \mcT}$ need not be independent across $k$. We make this dependence concrete using the notion of stationarity and strong mixing, in particular, $\alpha$-mixing across coordinates. We view the \(d\)-dimensional observations as the first \(d\) coordinates of an infinite coordinate process. Thus, under \(\mathbf{H}_{0,d}\), the coordinate trajectories \(\mathbf{Z}^{(k)}\) form a strictly stationary sequence across \(k\). Equivalently, the sequence of \(d\)-dimensional null laws is assumed to be projectively generated by this infinite coordinate process.

\begin{assumption}[Stationarity and $\alpha$-mixing across coordinates]
\label{ass:mixing}
Let $\mathbf{Z}^{(k)} := (Z_{1,k},\dots,\allowbreak Z_{N,k})^\top$ denote the $k$-th coordinate trajectory.
\begin{itemize}
\item[(i)]
The sequence $\{ \mathbf{Z}^{(k)} \}_{k \ge 1}$ is strictly stationary and $\alpha$-mixing with
mixing coefficients
\[
\alpha(r) := \sup_{k \ge 1}~
\sup_{\substack{A \in \mathcal{A}_k,\, B \in \mathcal{B}_{k+r}}}
\big| \P(A \cap B) - \P(A)\P(B) \big|~,
\]
satisfying $\sum_{r=1}^\infty \alpha(r) < \infty$, where $\mathcal{A}_k = \sigma(\mathbf{Z}^{(i)} : i \le k)$ and $\mathcal{B}_\ell = \sigma(\mathbf{Z}^{(i)} : i \ge \ell)$ denote the $\sigma$-fields generated by $\{\mathbf{Z}^{(i)}\}_{i \leq k}$ and $\{\mathbf{Z}^{(i)}\}_{i \geq \ell}$, respectively.

\item[(ii)]
The mixing coefficients $\{\alpha(r)\}_{r \in \bbN}$ satisfy the polynomial decay bound
\[
\alpha(r)\le C_\alpha\, r^{-\lambda}\qquad \text{for all } r \in \bbN,
\]
for some constant $C_\alpha>0$ and an exponent $\lambda > 4$.
\end{itemize}
\end{assumption}

The use of $\alpha$-mixing for cross-coordinate dependence follows existing work such as \citet{JunLi2020CPDinterpointCLT}. 
\Cref{ass:mixing} is a sufficient condition for the CLT, not a structural requirement of the DAK discrepancy itself. It is most natural for ordered features (e.g., spatial, genomic, etc.); otherwise, it can be viewed as a regularity condition on a chosen ordering. The scan statistic itself remains order-agnostic.

Under $\mathbf{H}_{0,d}$, Proposition~\ref{prop:var-cov-null} shows that the covariance structure of the scan vector is completely determined by the scalar quantity $\mathcal V_{d,N}$. Thus, to obtain a non-degenerate HDLSS limit, it suffices to understand the asymptotic behavior of $d\,\mathcal V_{d,N}$ under cross-coordinate dependence. The next proposition shows that this limit exists under strong mixing.

\begin{proposition}[Existence of the long-run variance factor under $\mathbf{H}_{0,d}$]
\label{prop:longrun-exists}
Suppose Assumption~\ref{ass:mixing}(i) holds under $\mathbf{H}_{0,d}$.
Then there exists a finite constant $\sigma_{\mathrm{long}}^2(N)\in[0,\infty)$ such that
\begin{equation}
\label{eq:dV-limit}
d\,\mathcal V_{d,N}\longrightarrow \sigma_{\mathrm{long}}^2(N),
\qquad\text{as~}~ d\to\infty.
\end{equation}
\end{proposition}

The proof identifies $\sigma_{\mathrm{long}}^2(N)$ through the absolutely summable lag-covariance series of the stationary coordinate-wise process $\{\boldsymbol{\xi}_k\}_{k\ge1}$; we defer the explicit representation to the appendix. To avoid degeneracy in the CLT, we assume the mild condition that the limit $\sigma_{\mathrm{long}}^2(N)$ is non-zero.


\begin{assumption}[Positive long-run variance]
\label{ass:nondeg-longrun}
The long run variance factor $\sigma_{\mathrm{long}}^2(N)$ satisfies $\sigma_{\mathrm{long}}^2(N)>0$.
\end{assumption}

\begin{theorem}[Multivariate CLT in the HDLSS regime under $\{\mathbf{H}_{0,d}\}$]
\label{thm:multivariate-clt}
Suppose \Cref{ass:mixing,ass:nondeg-longrun} hold.
Let $\mathbf{K} := \mathbf{K}(N) = (\mathbf{K}_{t,t'})_{t,t' \in \mcT}$ denote the $|\mcT| \times |\mcT|$ deterministic matrix of time components from \Cref{prop:var-cov-null}, defined as
\begin{equation}
\label{eq:K-mat-original}
\mathbf{K}_{t,t'} = \mathbf{K}_{t',t} = \frac{2(N-1)(N-2)}{\,t'(t'-1)(N-t)(N-t-1)} \,,
\qquad~\text{for}~\, t \leq t'.
\end{equation}
Then, under the sequence of null hypotheses $\{\mathbf{H}_{0,d}\}$, as $d \to \infty$,
\begin{equation}
\label{eq:clt-final-result}
\sqrt{d}\;\widetilde{\mathfrak{W}}_d
\ \stackrel{\mathscr{L}}{\longrightarrow} \
\mathcal{N}\!\left(\mathbf{0},\, \sigma^2_{\mathrm{long}}(N)\cdot\mathbf{K}\right).
\end{equation}
\end{theorem}

Theorem~\ref{thm:multivariate-clt} shows that after $\sqrt d \,$ scaling, the scan process weakly converges to a finite-dimensional Gaussian vector whose covariance structure is completely determined by the deterministic matrix $\mathbf K$ and the (unknown) long-run variance factor $\sigma^2_{\mathrm{long}}(N)$.
In particular, it enables principled threshold calibration (under $\mathbf{H}_{0,d}$) and power analysis (under $\mathbf{H}_{1,d}$) in high dimensions. Nevertheless, in its current form, \eqref{eq:clt-final-result} is not immediately useful unless the long-run variance factor $\sigma^2_{\mathrm{long}}(N)$ is consistently estimated. We do that in the next subsection which leads to an asymptotically distribution-free test.


\subsection{An asymptotically distribution-free test via plug-in estimation}
\label{subsec:level-alpha-calibration}

\subsubsection[Plug-in estimation of the variance factor under $\mathbf{H}_{0}$]{Plug-in estimation of $\sigma^2_{\mathrm{long}}(N)$ under $\mathbf{H}_{0,d}$}
\label{subsec:plugin-sigmalong}

Fix $t\in\mcT$. Recall the coordinate decomposition \eqref{eq:coord-decomposition} of $\mathfrak{W}_d(t)$, as the sample mean of $\{\xi_k(t)\}_{k=1}^d$. By definition, each $\xi_k(t)$ is a function of the full sample. Also, under $\mathbf{H}_{0,d}$, the scalar process $\{\xi_k(t)\}_{k\ge1}$ is strictly stationary with mean zero. Define its lag-$r$ covariance by
\[
\gamma_t(r):=\cov\bigl(\xi_1(t),\xi_{1+r}(t)\bigr),
\qquad r\ge0.
\]
Its long-run variance is therefore given by
\begin{equation}
\label{eq:lrv-t}
\mathrm{LRV}(N;t)
:=
\gamma_t(0)+2\sum_{r=1}^\infty \gamma_t(r)
\,\stackrel{(\star)}{=}\,
\mathbf K_{tt}\,\sigma^2_{\mathrm{long}}(N).
\end{equation}
where the equality $(\star)$ follows from Proposition~\ref{prop:longrun-exists} together with the exact covariance factorization in Proposition~\ref{prop:var-cov-null}.
Thus, for each fixed $t$, the scalar quantity $\sigma^2_{\mathrm{long}}(N)$ can be recovered from the coordinate process $\{\xi_k(t)\}_{k=1}^d$ by dividing its long-run variance by the known factor $\mathbf K_{tt}$.
Toward estimating $\sigma^2_{\mathrm{long}}(N)$, define the sample autocovariance at lag $r\ge0$ by
\begin{equation}
\label{eq:gammahat-full}
\widehat\gamma_{r,d}(t)
:=
\frac{1}{d}\sum_{k=1}^{d-r}
\bigl(\xi_k(t)-\mathfrak{W}_d(t)\bigr)
\bigl(\xi_{k+r}(t)-\mathfrak{W}_d(t)\bigr).
\end{equation}
Let $L=L(d)\in\mathbb N$ be a bandwidth and use Bartlett weights
$w_r:=1-\frac{r}{L+1}$ for $1\le r\le L$.
We define the heteroskedasticity and autocorrelation consistent (HAC) estimator \citep{NeweyWest1987HAC}:
\begin{equation}
\label{eq:HAC-LRV}
\widehat{\mathrm{LRV}}_d(N;t)
\;:=\;
\widehat\gamma_{0,d}(t)
\;+\;
2\sum_{r=1}^{L} w_r\,\widehat\gamma_{r,d}(t).
\end{equation}
This naturally yields the pointwise plug-in estimator of $\sigma^2_{\mathrm{long}}(N)$:
\begin{equation*}
\label{eq:sigmalonghat-single-t}
\widehat\sigma^2_{\mathrm{long}}(N;t)
\;:=\;
\frac{\widehat{\mathrm{LRV}}_d(N;t)}{\mathbf K_{tt}},
\qquad t\in\mcT.
\end{equation*}

Since $|\mcT|=N-3$ is fixed, the final plug-in estimator of $\sigma^2_{\mathrm{long}}$ can be taken to be any of these $\widehat\sigma^2_{\mathrm{long}}(N;t)$, or can be defined by aggregating $\{\widehat\sigma^2_{\mathrm{long}}(N;t)\}_{t \in \mcT}$ through an appropriate centrality measure such as weighted mean, median, etc. We choose the median due to its robustness property, viz.,
\begin{equation}
\label{eq:sigmalonghat-aggregate}
\widehat\sigma^2_{\mathrm{long}}(N)
:=
\mathrm{median}\bigl\{
\widehat\sigma^2_{\mathrm{long}}(N;t):\, t\in\mcT
\bigr\}.
\end{equation}

\begin{lemma}[Consistency of the plug-in estimator of $\sigma^2_{\mathrm{long}}(N)$]
\label{lem:sigmalong-consistency}
Assume that under $\{\mathbf{H}_{0,d}\}$, \Cref{ass:mixing,ass:nondeg-longrun} hold, and that the bandwidth satisfies
\begin{equation}
\label{eq:bandwidth-cond}
L(d)\to\infty
~\quad\text{and}\quad~
L(d)=o\bigl(d^{1/2}\bigr),
\qquad d\to\infty.
\end{equation}
Then, under the sequence of null hypotheses $\{\mathbf{H}_{0,d}\}$, for each fixed $t\in\mcT$,
\begin{equation}
\label{eq:lrv-consistency-pointwise}
\widehat{\mathrm{LRV}}_d(N;t)
\stackrel{\mathbb P}{\longrightarrow}
\mathbf K_{tt}\,\sigma^2_{\mathrm{long}}(N),
\qquad d\to\infty.
\end{equation}
Consequently, the plug-in estimator $\widehat\sigma^2_{\mathrm{long}}(N)$ is consistent for $\sigma^2_{\mathrm{long}}(N)$ as $d \to \infty$:
\[
\widehat\sigma^2_{\mathrm{long}}(N)
\stackrel{\mathbb P}{\longrightarrow}
\sigma^2_{\mathrm{long}}(N).
\]
\end{lemma}


\begin{remark}[Hyperparameter-free scan vs. variance calibration]
The DAK scan itself is hyperparameter-free: unlike Gaussian kernel, projection-based, or regularized procedures, it involves no bandwidth, scale parameter, number of projections, or penalty parameter. The plug-in calibrated test only requires estimating the scalar long-run variance factor \(\sigma^2_{\mathrm{long}}(N)\). The HAC bandwidth \(L(d)\) enters only through this nuisance-scale estimation step. Lemma~\ref{lem:sigmalong-consistency} shows that consistency holds for any bandwidth sequence satisfying \eqref{eq:bandwidth-cond}. Thus, this bandwidth does not affect the scan statistic \(\mathfrak{W}_d(\cdot)\), or the localization rule \(\widehat\tau_d\).
\end{remark}

\subsubsection{An asymptotically distribution-free test}
\label{subsubsec:adf-test} 

Let \(\mcZ=(\mcZ_t)_{t\in\mcT}\sim \mcN(0,\mathbf{K}(N))\), where \(\mathbf{K}(N)\) is the deterministic matrix from \eqref{eq:K-mat-original}. For a fixed level of significance \(\alpha\in(0,1)\), let \(c_\alpha\) denote the \((1-\alpha)\)-quantile of $\max_{t\in\mathcal T} Z_t$. 

\medskip
We define the plug-in variance-calibrated test by
\begin{equation}
\label{eq:level-alpha-test}
\widehat\phi_d(\alpha)
\,:=\,
\mathbbm{1}\{S_d>c_\alpha\}
\,=\,
\mathbbm{1}\left\{
\max_{t\in\mcT}\frac{\sqrt d\,\mathfrak W_d(t)}{\widehat{\sigma}_{\mathrm{long}}(N)}
>c_\alpha\right\}
\end{equation}
where \(S_d\) is the studentized scan statistic defined in \eqref{eq:studentized-max}, with $\widehat\sigma_{\mathrm{long}}(N) := \sqrt{|\widehat\sigma^2_{\mathrm{long}}(N)|}$. $S_d$ is a fully computable function of the data; and since $\mathbf K$ is known once $N$ is fixed, the critical value $c_\alpha$ can be computed offline by Monte Carlo simulation from $\mathcal N(\mathbf 0,\mathbf K)$.


The limiting null law of $S_d$ depends only on the deterministic matrix $\mathbf{K}(N)$, and not on the unknown $F_d$. Thus, $\widehat\phi_d(\alpha)$ is an asymptotically distribution-free and size-$\alpha$ test under appropriate conditions.

\begin{theorem}[Asymptotic level-$\alpha$ validity of the test]
\label{thm:level-alpha}
Assume that under $\mathbf{H}_{0,d}$, \Cref{ass:mixing,ass:nondeg-longrun} hold, and that the bandwidth condition \eqref{eq:bandwidth-cond} is satisfied. Then,
\[
S_d
\ \stackrel{\mathscr{L}}{\longrightarrow}\
\max_{t\in\mcT}\mcZ_t,
\qquad d\to\infty.
\]
Consequently, the test $\widehat\phi_d(\alpha)$ in \eqref{eq:level-alpha-test} is an asymptotically distribution-free and size-$\alpha$ test; i.e.,
\begin{equation}
\label{eq:level-alpha-claim}
\lim_{d\to\infty} \mathbb{E}_{\mathbf{H}_{0,d}}\!\bigl[\widehat\phi_d(\alpha)\bigr]
= \lim_{d\to\infty} \mathbb{P}_{\mathbf{H}_{0,d}}\!\bigl(\widehat\phi_d(\alpha)=1\bigr)
= \alpha.
\end{equation}
\end{theorem}

\subsection{Power analysis under the alternative $\mathbf{H}_{1}$}
\label{subsec:power-analysis}

The multivariate CLT in \Cref{thm:multivariate-clt} shows that under $\mathbf{H}_{0,d}$, the scan process fluctuates on the scale $d^{-1/2}$. A natural question is how the proposed test behaves under the sequence of alternatives $\{\mathbf{H_{1,d}}\}$, and in particular, what signal strength is required for reliable detection and accurate change-point localization.

We first establish a general power and localization result under high-level alternative-side fluctuation conditions, relying only on the deterministic structure of the population mean and the variance scaling of the statistic. We then provide a refined asymptotic characterization of power under additional weak limit approximation assumptions, which recovers a phase-transition-type behavior under local alternatives.

\subsubsection{Lower bounds on detection power and localization accuracy}
Recall the deterministic shape function $\Lambda_{\tau,N}(t)$ defined in \eqref{eq:det-shape-func} which was shown to be uniquely maximized at $t=\tau$. Since $N$ is fixed and $\mcT$ is finite, this implies a strictly positive separation between the true change-point and all other candidates. Define the separation gap
\begin{equation}
\label{eq:gap}
\Omega_{\tau,N}
:= \min_{t\in\mcT\setminus\{\tau\}} \bigl[\Lambda_{\tau,N}(\tau)-\Lambda_{\tau,N}(t)\bigr]
:= 1 - \max_{t\in\mcT\setminus\{\tau\}} \Lambda_{\tau,N}(t) \,\in\, (0,1).
\end{equation}

On the stochastic side, the finite-\(d\) lower bounds below require the scan fluctuations under \(\mathbf{H}_{1,d}\) to remain on the \(d^{-1/2}\) scale. We state this as part of the following alternative-side assumption.


\begin{assumption}[Alternative-side scale and fluctuation conditions]
\label{ass:sigma-long-ALT}
Under \(\{\mathbf{H}_{1,d}\}\), the following hold.
\begin{itemize}
\item[(i)] There exists \(\sigma_*^2(N)\in(0,\infty)\) such that
$~\widehat\sigma_{\mathrm{long}}^2(N) \stackrel{\mathbb P}{\longrightarrow} \sigma_*^2(N)$, as $d \to \infty$.

\item[(ii)] There exists a constant \(C_{\mathrm{alt}} \in [0,\infty)\) such that, for all sufficiently large \(d\),
\[
\max_{t\in\mcT}
\;\left\{d\cdot\var_{\mathbf{H}_{1,d}}\!\bigl(\mathfrak W_d(t)\bigr)\right\}
\le C_{\mathrm{alt}}.
\]
\end{itemize}
\end{assumption}

The following theorem shows that these two ingredients, viz., deterministic separation and $d^{-1/2}$-scale fluctuations, are sufficient for high-probability recovery of the true change-point location and high power.

\begin{theorem}[Change-point localization and power under the alternative]
\label{thm:power-main}
Suppose \Cref{ass:sigma-long-ALT} holds, and $N$ is fixed.
Let $\widehat\phi_d(\alpha)$ be the test defined in \eqref{eq:level-alpha-test}. Then the following holds.

\begin{enumerate}
\item[(i)] (Localization) For all sufficiently large $d$,
\[
\mathbb{P}_{\mathbf{H}_{1,d}}\!\left(\widehat\tau_d=\tau\right)
\ge
1-\frac{4\,C_{\mathrm{alt}}}{\Omega_{\tau,N}^2}\cdot \frac{N}{d\,\delta_d^2}.
\]

\item[(ii)] (Power) Fix any $\eta\in(0,1)$, define the event
\(
A_{d,\eta}
:=
\left\{
\left|\widehat\sigma_{\mathrm{long}}(N)-\sigma_*(N)\right|
\le \eta\,\sigma_*(N)
\right\}
\).
Then, for all sufficiently large $d$ satisfying $\sqrt d\,\delta_d > (1+\eta)\sigma_*(N)\,c_\alpha $, provided \(c_\alpha>0\), we have
\[
\mathbb P_{\mathbf{H}_{1,d}}\!\left(\widehat\phi_d(\alpha)=1\right)
\ge
1 - \mathbb P_{\mathbf{H}_{1,d}}(A_{d,\eta}^c) - \frac{C_{\mathrm{alt}}}{\bigl(\sqrt d\,\delta_d-(1+\eta)\sigma_*(N)c_\alpha\bigr)^2}.
\]
\end{enumerate}

\noindent
In particular, if $\sqrt d\,\delta_d\to\infty$, then
\(
\displaystyle\lim_{d \to \infty} \mathbb P_{\mathbf{H}_{1,d}}\!\left(\widehat\tau_d=\tau\right) = 1
\)
and
\(
\displaystyle\;\lim_{d \to \infty} \mathbb P_{\mathbf{H}_{1,d}}\!\left(\widehat\phi_d(\alpha)=1\right) = 1.
\)
\end{theorem}

The theorem shows that once the signal dominates the intrinsic fluctuation scale $d^{-1/2}$, the scan statistic reliably identifies both the presence and the location of the change-point. The condition $\sqrt d\,\delta_d \to \infty$ thus ensures exact localization and consistency of the test in high dimensions.

\subsubsection{Power under local alternatives}

To understand the behavior of the procedure at the boundary of detectability, we now turn to a local asymptotic analysis.

\noindent\textbf{Local alternatives and weak limit of the centered scan process.}
Fix a change-point location $\tau\in\mcT$. We consider a sequence of local alternatives $\{\mathbf H_{1,d}\}$ such that the signal factor $\delta_d$ satisfies
\begin{equation}
\label{eq:local-signal-scaling}
\delta_d = \frac{h}{\sqrt d},
\end{equation}
for some fixed $h\in(0,\infty)$. Under this scaling, the exact factorization of the mean yields
\begin{equation}
\label{eq:local-mean-current-notation}
\E_{\mathbf H_{1,d}}[\mathfrak W_d(t)]
=
\Lambda_{\tau,N}(t) \cdot \delta_d
=
\frac{h}{\sqrt d}\,\Lambda_{\tau,N}(t),
\qquad t\in\mcT.
\end{equation}
\indent
Thus, after multiplication by $\sqrt d$, the deterministic component of the scan process converges to the non-trivial signal profile $h\,\bigl(\Lambda_{\tau,N}(t)\bigr)_{t\in\mcT}$.

To obtain a limiting power curve, we impose the following additional assumption on the centered scan vector under the local alternatives.

\begin{assumption}[Existence of weak limit of the centered scan vector under local alternatives]
\label{ass:local-weak-limit}
Suppose $\{\mathbf H_{1,d}\}$ is a sequence of local alternatives satisfying \eqref{eq:local-signal-scaling} for some fixed $h\in(0,\infty)$. Then, as $d\to\infty$,
\[
\sqrt d\left(\widetilde{\mathfrak W}_d-\E_{\mathbf H_{1,d}}[\widetilde{\mathfrak W}_d]\right)
\ \stackrel{\mathscr{L}}{\longrightarrow}\
\mathfrak{B}^{(h)},
\]
for some $\R^{|\mcT|}$-valued random vector $\mathfrak{B}^{(h)} = (\mathfrak{B}^{(h)}_t)_{t \in \mcT}$.
\end{assumption}

Assumption~\ref{ass:local-weak-limit} is a local asymptotic regularity condition. It merely assumes the existence of a non-degenerate weak limit for the centered scan vector at the critical $d^{-1/2}$ scale. The limiting law is allowed to depend on the local signal intensity $h$.
Our next result shows that under this assumption, the variance-calibrated scan admits a nontrivial asymptotic power curve.

\begin{theorem}[Local asymptotic power curve]
\label{thm:local-power}
Fix $h\in(0,\infty)$, and consider a sequence of local alternatives $\{\mathbf H_{1,d}\}$ satisfying \eqref{eq:local-signal-scaling}. Suppose \Cref{ass:sigma-long-ALT}(i) and \Cref{ass:local-weak-limit} hold. Then, as $d \to \infty$,
\[
\mathbb P_{\mathbf H_{1,d}}\!\left(\widehat\phi_d(\alpha)=1\right)
\;\longrightarrow\;
\mathbb{P}\!\left(M_h>c_\alpha\right)
=
\mathbb{P}\!\left(
\max_{t\in\mcT}
\left[\mathfrak{B}^{(h)}_t+h\,\Lambda_{\tau,N}(t)\right]
>c_\alpha \, \sigma_{*}(N)
\right),
\]
if $c_\alpha$ is a continuity point of the distribution function of $~M_h = \displaystyle\max_{t\in\mcT} \,\dfrac{\mathfrak{B}^{(h)}_t+h\,\Lambda_{\tau,N}(t)}{\sigma_{*}(N)}$.
\end{theorem}

\smallskip
\begin{remark}[Interpretation of the detection scale]
Collectively, \Cref{thm:power-main,thm:local-power} show that the scale
$\delta_d \asymp d^{-1/2}$ is the critical detection scale for the proposed DAK-based scan statistic in the HDLSS regime. When $\sqrt d\,\delta_d \to h\in(0,\infty)$, the test exhibits a non-degenerate limiting power characterized by \Cref{thm:local-power}; when $\sqrt d\,\delta_d \to \infty$, the power tends to $1$ by \Cref{thm:power-main}. 
\end{remark}

\section{Sequential extension for streaming high-dimensional data}
\label{sec:online}

\subsection{Sequential monitoring via small windows: Motivation and setup}

The offline DAK scan is developed for the HDLSS regime, where the sample size $N$ is fixed and the dimension $d$ diverges. This makes it particularly well-suited to small batches of high-dimensional observations. The same feature also suggests a natural sequential extension. Rather than relying on large windows (see, e.g., \citet{wang2020multiscale, XieEtAl2023WindowCPD, WeiXie2026OnlineKernelCUSUM}) or temporal accumulation such as recursive likelihood-ratio procedures (see, e.g., \citet{Lai1995SeqCPD, XieSiegmund2013MultiSensorCPD, CaoEtAl2018SeqCPDviaOCO}), we repeatedly apply the offline scan to small, moving windows of fixed size. In this way, the offline procedure serves as a building block for online detection in high-dimensional streaming data.

The online method thus has a different scope than the offline method. While the offline theory is tailored to a single HDLSS batch, the sequential procedure can be run over arbitrarily long streams by continually reusing the same fixed-window statistic. The underlying discrepancy, however, remains the same DAK construction; hence, the online method continues to target changes reflected in the aggregated one-dimensional marginals across coordinates.
Formally, suppose we observe a streaming sequence of $d$-dimensional random vectors $\{Z_t\}_{t\in\mathbb N}\subset\mathbb R^d$. We consider the following online single-change-point model with a change at $\nu$.
\begin{align}
\label{eq:single-cp-online}
\mathbf{H}_{0,d} : ~~ Z_t ~{\sim}~ F_d~\text{~for all~} t \in \bbN, \quad \text{vs.} \quad
\mathbf{H}_{1,d} : ~~ \exists\, \nu \in \bbN \ \text{such that}\ 
\begin{cases}
Z_1, Z_2, \dots, Z_\nu \overset{\text{i.i.d.}}{\sim} F_d, \\
Z_{\nu+1}, Z_{\nu+2}, \dots\overset{\text{i.i.d.}}{\sim} G_d,
\end{cases}
\end{align}
where $F_d$ and $G_d$ are probability distributions on $\mathbb R^d$ with at least one unequal coordinate marginal, i.e., $F_d^{(k)} \neq G_d^{(k)}$ for at least one $k \in [d]$.

For notational convenience, let \(\mathbb P_{d,\infty}\) denote the law under \(\mathbf{H}_{0,d}\), and let \(\mathbb P_{d,\nu}\) denote the law under $\mathbf{H}_{1,d}$ when $\nu$ is the true change-point.
We write \(\mathbb E_{d,\infty}\) and \(\mathbb E_{d,\nu}\) for the corresponding expectations. When the dimension \(d\) is clear from context, we suppress it and write simply \(\mathbb P_\infty,\mathbb P_\nu\) and \(\mathbb E_\infty,\mathbb E_\nu\).

\subsection{Sliding fixed-window DAK scan and stopping rule}

Fix a window length \(N_0\ge4\), treated as a constant independent of the dimension \(d\). Let $\mcT_0:=\{2,\dots,N_0-2\}$ denote the set of admissible split points inside a window of length \(N_0\). 
For each time \(s\ge N_0\), define the rolling window $\mcW_s:=\{Z_{s-N_0+1},\dots,Z_s\}$.

For each candidate split \(k\in\mcT_0\), partition \(\mcW_s\) into $\mcX_{s,k}:=\{Z_{s-N_0+1},\dots,Z_{s-N_0+k}\}$ and $\mcY_{s,k}:=\{Z_{s-N_0+k+1},\dots,Z_s\}$, and compute the offline DAK statistic
\[
\mathfrak W_{d,s}(k)
:=
\widehat\Delta_{\widehat\alpha_k,d}\bigl(\mcX_{s,k},\mcY_{s,k}\bigr),
\quad~\text{where}~~
\widehat\alpha_k:=\frac{k}{N_0}.
\]

We assume the availability of a pre-change reference sample $\mathcal C_{N_0} := \{Z_{1}^{\mathrm{(cal)}},\dots,Z_{N_0}^{\mathrm{(cal)}}\}$, consisting of \(N_0\) independent observations from \(F_d\), independent of the monitoring stream. This reference sample is used only once, before monitoring begins, to estimate the long-run variance factor $\sigma^2_{\mathrm{long}}(N_0)$ corresponding to window length \(N_0\), exactly as in \Cref{subsec:plugin-sigmalong}, with \(N\) replaced by \(N_0\). Denote the resulting variance and scale estimators by
$\widehat\sigma_{\mathrm{long}}^2(N_0;\mcC_{N_0})$ and $\widehat\sigma_{\mathrm{long}}(N_0;\mcC_{N_0}) :=
\sqrt{\bigl|\widehat\sigma_{\mathrm{long}}^2(N_0;\mcC_{N_0})\bigr|}$.
These are computed once at the calibration stage and then kept fixed throughout monitoring.

We define the calibrated, window-level DAK scan statistic for the $s$-th window, for $s \geq N_0$, as
\begin{equation}
\label{eq:window-studentized}
\mcM_d(s)
:=
\max_{k\in\mcT_0}~
\frac{\sqrt d\,\mathfrak W_{d,s}(k)}
{\widehat\sigma_{\mathrm{long}}(N_0;\mcC_{N_0})}.
\end{equation}

Let \(\mcZ=(\mcZ_k)_{k\in\mcT_0}\sim\mcN(\mathbf 0,\mathbf K(N_0))\), where \(\mathbf K(N_0)\) is the deterministic matrix with entries
\begin{equation}
\label{eq:K-mat-N0}
\mathbf K_{t,t'}(N_0)=\mathbf K_{t',t}(N_0)
=
\frac{2(N_0-1)(N_0-2)}
{t'(t'-1)(N_0-t)(N_0-t-1)},
\qquad t\le t',
\end{equation}
and let \(c_{\alpha,N_0}\) denote the \((1-\alpha)\)-quantile of \(\max_{k\in\mcT_0}\mcZ_k\). 
The stopping time is then defined by
\begin{equation}
\label{eq:online-stop}
\widehat\nu
:=
\inf\{s\ge N_0:\mcM_d(s) > c_{\alpha,N_0}\}.
\end{equation}

If \(\widehat\nu<\infty\), we estimate the corresponding change-point location by
\begin{equation}
\label{eq:online-cp-estimate}
\widehat\tau_{\mathrm{on}}
:= (\widehat\nu-N_0) + \underset{k\in\mcT_0}{\argmax}~\mathfrak W_{d,\widehat\nu}(k).
\end{equation}
where \(\argmax\) denotes the smallest maximizer in case of ties.
Thus, the online procedure is obtained by repeatedly applying the offline DAK scan over moving windows of fixed size, while using a single pre-change reference block for variance calibration.

\subsection{Behavior under $\mathbf{H}_{0}$ and ARL calibration}
\label{subsec:online-arl}

We first characterize the behavior of the sequential procedure under $\mathbf{H}_{0,d}$. 
For each fixed $s\ge N_0$, the window-level statistic $\mathfrak W_{d,s}(k)$ is computed from $N_0$ consecutive observations. Since the calibration sample $\mcC_{N_0}$ is independent of the data stream, the fixed-sample multivariate HDLSS CLT from \Cref{thm:multivariate-clt} applies to each window with $N$ replaced by $N_0$:
\[
\sqrt d\,(\mathfrak W_{d,s}(k))_{k\in\mcT_0}
\ \stackrel{\mathscr{L}}{\longrightarrow}\
\mcN\!\left(\mathbf 0,\,
\sigma_{\mathrm{long}}^2(N_0)\,\mathbf K(N_0)\right),
\qquad d\to\infty.
\]

Fix a threshold $c\in\mathbb R$. Since the online statistic is calibrated using the random reference block $\mcC_{N_0}$, it is natural to work conditionally on $\mcC_{N_0}$. Define the conditional stopping time and the corresponding conditional average run length by
\begin{equation}
\label{eq:conditional-arl}
\widehat\nu_d(c)
:=\inf\{s\ge N_0:\mcM_d(s)>c\},
\qquad
\mathrm{ARL}_{d}(c\,|\,\mcC_{N_0})
:=
\E_{d,\infty}\!\big[\widehat\nu_d(c)\mid \mcC_{N_0}\big].
\end{equation}
Also, define the conditional one-step exceedance probability
\begin{equation}
\label{eq:qd-conditional}
q_d(c\,|\,\mcC_{N_0})
:=
\P_{d,\infty}\!\big(\mcM_d(N_0)>c \,\big|\, \mcC_{N_0}\big).
\end{equation}

Conditional on $\mcC_{N_0}$, the denominator $\widehat\sigma_{\mathrm{long}}(N_0;\mcC_{N_0})$ is fixed. Hence, under $\mathbf{H}_{0,d}$, the sequence $\{\mcM_d(s)\}_{s\ge N_0}$ is stationary and $(N_0-1)$-dependent in time: if $|s-t|\ge N_0$, then the windows $\mcW_s$ and $\mcW_t$ are disjoint and therefore $\mcM_d(s)\indep \mcM_d(t)$, conditionally on $\mcC_{N_0}$.

\begin{lemma}[Conditional ARL bounds under window overlap]
\label{lem:online-arl-bounds}
Assume that $\mathbf{H}_{0,d}$ holds. Fix a threshold $c\in\mathbb R$. Then, almost surely with respect to the calibration sample $\mcC_{N_0}$,
\begin{equation}
\label{eq:arl-sandwich-conditional}
\frac{1}{4\,q_d(c\,|\,\mcC_{N_0})}-\frac{1}{2}
\ \le\
\mathrm{ARL}_{d}(c\,|\,\mcC_{N_0})
\ \le\
\frac{N_0}{q_d(c\,|\,\mcC_{N_0})}.
\end{equation}
In particular, since $N_0$ is fixed,
\(
~\mathrm{ARL}_{d}(c\,|\,\mcC_{N_0})
\asymp
q_d(c\,|\,\mcC_{N_0})^{-1}.
\)
\end{lemma}

We specialize to the Gaussian calibration threshold \(c=c_{\alpha,N_0}\), defined as the \((1-\alpha)\)-quantile of
\(
\mcZ_{\max}(\mcT_0):=\max_{k\in\mcT_0}\mcZ_k
\), where
\(
\mcZ=(\mcZ_k)_{k\in\mcT_0}\sim\mcN(\mathbf 0,\mathbf K(N_0)).
\)

Since the calibration sample \(\mcC_{N_0}\) is independent of the data stream, and since \(\widehat\sigma_{\mathrm{long}}(N_0;\mcC_{N_0})\) is consistent for \(\sigma_{\mathrm{long}}(N_0)\) under \(\mathbf{H}_{0,d}\), the studentized window statistic satisfies
\(
\mcM_d(s)
\ \stackrel{\mathscr{L}}{\longrightarrow}\
\mcZ_{\max}(\mcT_0)
\)
as $d\to\infty$, for each fixed \(s\ge N_0\). Thus,
\begin{equation}
\label{eq:qd-alpha-conditional}
q_d(c_{\alpha,N_0} \mid \mcC_{N_0})
\ \stackrel{\bbP}{\longrightarrow}\
\alpha,
\qquad\text{as~}\, d\to\infty.
\end{equation}

\begin{remark}[Role of the calibration level \(\alpha\)]
\label{rem:alpha-arl}
By \Cref{lem:online-arl-bounds}, the conditional ARL is controlled by the conditional one-step exceedance probability \(q_d(c_{\alpha,N_0}\mid \mcC_{N_0})\). Therefore, \eqref{eq:qd-alpha-conditional} implies that
\[
\mathrm{ARL}_{d}(c_{\alpha,N_0} \mid \mcC_{N_0})
=
\Theta_{\bbP}\!\left(\frac{1}{\alpha}\right),
\qquad d\to\infty.
\]
More explicitly, as $d \to \infty$, for every fixed \(\varepsilon\in(0,\alpha)\),
\[
\P\!\left(
\frac{1}{4(\alpha+\varepsilon)}-\frac12
\le
\mathrm{ARL}_{d}(c_{\alpha,N_0} \mid \mcC_{N_0})
\le
\frac{N_0}{\alpha-\varepsilon}
\right)
\longrightarrow 1.
\]
Thus, to target a desired ARL level \(\gamma\), it suffices to choose \(\alpha\) of order \(1/\gamma\). In this sense, although \(\alpha\) is introduced through a fixed-window Gaussian limit, it serves as a design parameter roughly controlling the long-run false-alarm frequency of the sequential procedure. An explicit renewal-type analysis (see, e.g., \citet{SiegmundYakir2000tailprob}) of the stopping time distribution is left as future work.
\end{remark}

\subsection{Expected detection delay under $\mathbf{H}_{1}$}
\label{subsec:online-edd}

We now study the delay of the sequential procedure following the change. Because alarms at or before the change-point are false alarms, it is natural to measure detection delay conditionally on no false alarm before the change. Moreover, since the same calibration estimate \(\widehat\sigma_{\mathrm{long}}(N_0;\mcC_{N_0})\) is reused throughout monitoring, all probabilities and expectations in this subsection are taken conditionally on the fixed calibration sample \(\mcC_{N_0}\); to avoid cumbersome notation, we suppress this conditioning.

Let \(\widehat\nu_d:=\widehat\nu_d(c_{\alpha,N_0})\) denote the stopping time in \eqref{eq:online-stop}, and let \(\nu\) be the deterministic change-point in \eqref{eq:single-cp-online}. We define the conditional expected detection delay by
\begin{equation}
\label{eq:cedd-def}
\operatorname{CEDD}_{d,\nu}
:=
\E_{d,\nu}\!\big[\,\widehat\nu_d-\nu \,\big|\, \widehat\nu_d>\nu\big].
\end{equation}
We also define the Pollak-type worst-case conditional EDD by
\begin{equation}
\label{eq:pcedd-def}
\operatorname{CEDD}_{d}^{\mathbf{(P)}}
:=
\sup_{\nu\in\mathbb N} \; \operatorname{CEDD}_{d,\nu} = \sup_{\nu\in\mathbb N} \; \E_{d,\nu}\!\big[\,\widehat\nu_d-\nu \,\big|\, \widehat\nu_d>\nu\big].
\end{equation}

To bound these quantities, we separate the mixed-window phase from the fully post-change phase. Define $\mcJ_\nu:=\{\nu+1,\nu+2,\dots,\nu+N_0-1\}$. These are the possible mixed-window indices relative to the change-point; among monitored times, the event \(\{\widehat\nu_d\in\mcJ_\nu\}\) corresponds to stopping during the mixed-window phase. Windows with index \(s\ge \nu+N_0\) are entirely post-change. We therefore introduce the following two quantities:
\begin{align}
\label{eq:pi-qG-def}
\pi_{d,\nu}
:=
\bbP_{d,\nu}\!\left(\widehat\nu_d\in\mcJ_\nu \,\middle|\, \widehat\nu_d>\nu\right),\quad~\text{and}\quad~
q_d^{(G)}
:=
\bbP_{d,\nu}\!\left(\mcM_d(\nu+2N_0-1)>c_{\alpha,N_0}\right). \notag
\end{align}

The quantity \(\pi_{d,\nu}\) is the conditional probability of stopping during the mixed-window phase. The quantity \(q_d^{(G)}\) is the one-window exceedance probability after the mixed-window phase has been cleared by one full window length. By stationarity of the post-change regime, \(q_d^{(G)}\) does not depend on \(\nu\).

\begin{lemma}[Conditional EDD and worst-case Pollak EDD bounds]
\label{lem:online-cedd}
Assume the online procedure is calibrated as in Section~\ref{sec:online}, and
suppose \(q_d^{(G)}>0\). Then, for every \(\nu\in\mathbb N\),
\begin{equation}
\label{eq:cedd-bound}
\operatorname{CEDD}_{d,\nu}
\le
(N_0-1)
+
\bigl(1-\pi_{d,\nu}\bigr)\frac{N_0}{q_d^{(G)}}.
\end{equation}
Consequently,
\begin{equation}
\label{eq:pcedd-bound-pi}
\operatorname{CEDD}_{d}^{\mathbf{(P)}}
\le
(N_0-1)
+
\left(1-\inf_{\nu\in\mathbb N}\pi_{d,\nu}\right)
\frac{N_0}{q_d^{(G)}}.
\end{equation}
\end{lemma}

\begin{remark}[Interpretation of the bound]
\label{rem:cedd-interpretation}
The bound in \eqref{eq:cedd-bound} separates the delay into two parts. The term \(N_0-1\) is the largest possible delay if the procedure stops during the mixed-window phase. If this phase is missed, which occurs with conditional probability \(1-\pi_{d,\nu}\), the remaining delay is controlled by the post-change exceedance probability \(q_d^{(G)}\).
\end{remark}

We next specialize the bound using the HDLSS limit for a fully post-change window. Suppose that, for a fully post-change window, the analogue of \Cref{ass:mixing,ass:nondeg-longrun} holds under \(G_d\), with \(N\) replaced by \(N_0\), and with positive long-run variance factor \(\sigma_*^2(N_0)\). Define
\begin{equation}
\label{eq:qG-limit-def}
q_\infty^{(G)}
:=
\bbP\!\left(
\max_{k\in\mcT_0}\mcZ_k
>
\frac{\sigma_{\mathrm{long}}(N_0)}{\sigma_*(N_0)}
\,c_{\alpha,N_0}
\right),
\quad~\text{where~}~
\mcZ=(\mcZ_k)_{k\in\mcT_0}\sim
\mcN(\mathbf 0,\mathbf K(N_0)).
\end{equation}

\begin{corollary}[Pollak EDD bounds under a post-change HDLSS CLT]
\label{cor:online-pollak-cedd}
Assume the conditions of \Cref{lem:online-cedd}. Suppose that the reference calibration block satisfies the conditions of \Cref{lem:sigmalong-consistency} with \(N\) replaced by \(N_0\), and that the fully post-change HDLSS conditions stated above hold.
%
Then,
\(
q_d^{(G)}
\stackrel{\bbP}{\longrightarrow}
q_\infty^{(G)}\in(0,1)\,
\)
as $d \to \infty$,
and consequently,
\begin{equation}
\label{eq:pcedd-sandwich}
N_0-1
\;\le\;
\operatorname{CEDD}_{d}^{\mathbf{(P)}}
\;\le\;
N_0\left(1+\frac{1}{q_\infty^{(G)}}\right)-1
+
o_{\bbP}(1).
\end{equation}
\end{corollary}

\begin{remark}[Interpretation of the Pollak bound]
\label{rem:pcedd-interpretation}
The lower bound in \eqref{eq:pcedd-sandwich} is deterministic: since \(\widehat\nu_d\ge N_0\), taking the early change-point \(\nu=1\)
gives
\[
\operatorname{CEDD}_{d}^{\mathbf{(P)}}
\ge
\E_{d,1}\!\left[\widehat\nu_d-1 \,\middle|\, \widehat\nu_d>1\right]
\ge
N_0-1.
\]
The upper bound shows that the worst-case conditional delay remains \(O_{\bbP}(1)\) in the stream length, with the residual waiting time controlled by \(q_\infty^{(G)}\). The bound, however, is conservative; sharpening it would require uniform control of the mixed-window phase after conditioning on no earlier alarm.
\end{remark}


\section{Numerical experiments}
\label{sec:experiments}

We assess the empirical performance of the proposed offline and online methods in both synthetic and real-data settings. We focus on marginal distributional changes, which are the target alternatives of the proposed discrepancy, and include heavy-tailed examples for which moment-based methods may be unreliable. To this end, we organize the numerical results into two parts: offline and online experiments.

\subsection{Simulation experiments on offline detection}
\label{subsec:simulation}

For the offline simulations, we compare our method against several existing approaches representing different methodological paradigms. These include \emph{E-Divisive} and \emph{E-CP3O} \citep{MattesonJames2014}, which are energy-statistic-based nonparametric changepoint procedures; \emph{Kernel CPA}, the kernel changepoint method of \citep{ArlotEtAL2019KCPA}; \emph{MMD-$\mcN$} and \emph{MMD-$\mcE$}, which use MMD with a Gaussian kernel and the energy kernel with $\ell_2$ distance \citep{sejdinovic2013energy}, respectively; \emph{HDD-DM}, the high-dimensional mean and covariance shift detection procedure of \citet{Drikvandi2025changepoint}; and a \emph{sliced Wasserstein} distance based on sliced optimal transport \citep{NadjahiEtAl2020SlicedWass}. Collectively, we cover a broad range of approaches, including energy-based, kernel-based, and optimal-transport-based methods. In particular, we include HDD-DM as a recent state-of-the-art method designed specifically for HDLSS data.

We first study the localization accuracy of the proposed offline DAK scan under single-change alternatives. Throughout this subsection, the sequence length is fixed at \(N=40\), the true change-point is \(\tau=15\), and the dimension is varied over \(d\in\{200,1000,5000\}\). For each simulated sequence, every method returns an estimated change-point \(\widehat\tau_d\), and performance is summarized by the empirical distribution of the absolute localization error $|\widehat\tau_d-\tau|$.
Specifically, Table~\ref{tab:cp_accuracy_main} reports the proportions of replications for which this deviation equals \(0,1,2\), or is at least \(3\), computed based on $1000$ independent replications.

\input{tables/cp-loc-main}

We report five representative alternatives: two heavy-tailed product examples involving Cauchy location-scale changes, a Dirichlet example on the simplex, a sparse Gaussian mean-shift example, and a mixed Cauchy--Gaussian example. These settings cover dense, heavy-tailed marginal shifts; scale changes without finite moments; constrained non-Gaussian data; sparse marginal signals; and heterogeneous coordinate types. 

The five representative examples reported are as follows. Fix \(N=40\) and \(\tau=15\) in the original testing problem \eqref{eq:single-cp}. We take \(Z_1,\ldots,Z_\tau\sim F_d\) and \(Z_{\tau+1},\ldots,Z_N\sim G_d\), for $d \in \{200, 1000, 5000\}$.

\begin{enumerate}
    \item \textit{Cauchy location change:} \(F_d=\mathrm{Cauchy}(0,1)^{\otimes d}\) and \(G_d=\mathrm{Cauchy}(1,1)^{\otimes d}\).

    \item \textit{Cauchy scale change:} \(F_d=\mathrm{Cauchy}(1,1)^{\otimes d}\) and \(G_d=\mathrm{Cauchy}(1,2)^{\otimes d}\).

    \item \textit{Dirichlet change:} \(F_d=\mathrm{Dirichlet}(\mathbf 1_d)\) and \(G_d=\mathrm{Dirichlet}(0.1\cdot\mathbf 1_d)\).

    \item \textit{Sparse Gaussian mean change:} \(F_d=\mathcal N(0,I_d)\) and \(G_d=\mathcal N(\mu_d,I_d)\), where \(\mu_{d,k}=1\) for \(1\le k\le s_d\), with \(s_d=\lfloor 0.05d\rfloor\); \(\mu_{d,k}=0\) otherwise.

    \item \textit{Cauchy--Gaussian mixture change:} Let \(d_2=\lfloor d/3\rfloor\) and \(d_1 = \lceil 2d/3\rceil\), so that $d_1 +d_2 = d$. Then, \(F_d=\mathcal N(0,1)^{\otimes d_1}\otimes \mathrm{Cauchy}(0,1)^{\otimes d_2}\) and \(G_d=\mathcal N(1,1)^{\otimes d_1}\otimes \mathrm{Cauchy}(1,1)^{\otimes d_2}\).
\end{enumerate}

The results in Table~\ref{tab:cp_accuracy_main} show that DAK Scan accurately localizes the change-point across the representative marginal-shift settings considered above. In the \emph{Cauchy Location}, \emph{Dirichlet}, and \emph{Cauchy--Gaussian Mix} examples, DAK achieves essentially exact localization across all three choices of the dimension $d$. These examples highlight the main robustness feature of the proposed statistic: localization remains reliable even when some or all coordinates are heavy-tailed or their moments may fail to exist.

The \emph{Cauchy Scale} and \emph{Gaussian Sparse Mean} examples illustrate the HDLSS accumulation effect more clearly. For the Cauchy scale change, the exact localization probability of DAK increases from \(0.13\) at \(d=200\), to \(0.68\) at \(d=1000\), and to \(0.99\) at \(d=5000\). For the sparse Gaussian mean change, where only \(5\%\) of the coordinates carry signal, the corresponding probabilities increase from \(0.03\), to \(0.33\), and then to \(0.90\). Thus, even when the per-coordinate signal is weak or sparse, the proposed scan becomes substantially more accurate as the dimension grows, consistent with the fixed-\(N\), large-\(d\) theory.

The strong separation between DAK and the competing procedures in several examples should be interpreted in light of the HDLSS regime considered here. The heavy-tailed settings, particularly the \emph{Cauchy location}, \emph{Cauchy scale}, and \emph{Cauchy--Gaussian mixture} examples, deliberately stress regimes in which $\ell_2$-distance-based, moment-based, or tuning-sensitive procedures can become unstable. HDD-DM, which is designed for HDLSS data, performs well in some settings but is less uniformly accurate across the heavy-tailed and sparse alternatives. In contrast, DAK is based on bounded coordinate-wise angular comparisons and therefore remains well-defined without requiring finite marginal moments. 

Its performance in these examples reflects the main design principle of the method: in fixed-\(N\), large-\(d\) problems, stable bounded marginal signals can be accumulated across coordinates even when classical moment-based summaries are unreliable. To check that the strong performance above is not solely an artifact of heavy-tailed alternatives, \Cref{tab:cp_accuracy_appendix} in Appendix~\ref{app-sec:additional-simulations} reports additional finite-moment examples that are more favorable to standard distance- and kernel-based procedures. Among the competing methods, HDD-DM is particularly relevant, as it performs strongly on several moment-friendly examples, but our proposed DAK scan remains broadly competitive across most settings.

\subsection{Experiments on sequential (online) monitoring}

\subsubsection{Simulation studies under a single change-point}

We next evaluate the online DAK Scan procedure on synthetic streams generated from the same five alternatives used in the offline experiments. Throughout, the sliding-window length is \(N_0=10\), the change occurs at \(\nu=50\), and the calibration level is \(\alpha=0.002\). Thus, each scan uses only ten observations at a time, with admissible splits \(k\in\mcT_0 := \{2,\ldots,N_0-2\}\).

We report only the online results for the proposed method. Unlike in the offline setting, where several nonparametric procedures can be directly compared, few sequential methods are tailored to the high-dimensional regime considered here. Sliding-window adaptations of offline methods are possible, but require separate calibration choices. We therefore focus on the empirical behavior of the proposed stopping rule.

Under \(\mathbf{H}_{0,d}\), all observations are generated from \(F_d\), and we report the empirical monitoring average run length (ARL). Under \(\mathbf{H}_{1,d}\), \(Z_1,\ldots,Z_\nu\sim F_d\) and \(Z_{\nu+1},Z_{\nu+2},\ldots\sim G_d\). We report the false-alarm proportion before \(\nu\), the empirical conditional expected detection delay (CEDD), and the non-detection proportion by the terminal horizon.

\input{tables/online-ARL-EDD-proposed}

Table~\ref{tab:online_arl_edd} shows that the online calibration produces the intended null behavior. The choice \(\alpha=0.002\) is motivated by the \(1/\alpha\) scaling of the null ARL, corresponding to a nominal target near \(500\); empirically, the ARLs across all examples and dimensions remain close to this target, lying roughly between \(490\) and \(550\). The false-alarm proportions by time \(\nu=50\) range from \(0.07\) to \(0.12\), reflecting repeated monitoring over overlapping windows.

Detection is rapid for strong marginal changes. For the Dirichlet and Cauchy--Gaussian mixture examples, the CEDD is essentially \(2\) across all dimensions. The Cauchy location example exhibits a clear HDLSS effect, with CEDD decreasing from \(42.73\) at \(d=200\) to approximately \(2\) for \(d=1000\) and \(d=5000\). The Cauchy scale and sparse Gaussian mean examples are harder under such short windows, but still show improvement with dimension. Overall, the results support the main operational message: the proposed angular-kernel scan can be run with very short windows while maintaining long null run lengths and delivering a small delay when the aggregated marginal signal is sufficiently strong.

\subsubsection{Multiple change-point detection in solar flare imaging: An illustrative case study}
\label{sec:solar-flare}

We next apply our online method to a solar-flare image sequence\footnote{The Solar Object Locator for the original data is \texttt{SOL2011-04-30T21-45-49L061C108}.} from the Solar Dynamics Observatory. Each frame is vectorized into a high-dimensional observation, with \(d=67744\) pixels, while the temporal length is only a few hundred frames. 
The example is challenging because the background solar activity is itself nonstationary, with bright regions whose shape changes gradually over time, making simple background-subtraction approaches unreliable.

We use a sliding window of length \(N_0 = 15\). For each window beginning at time \(s\), we compute the scan statistic
\(
   \overline{W}[s]
   =
   \max_{k\in\{2,\ldots,N_0-2\}}
   \mathfrak W_{d,s}(k),
\)
where \(\mathfrak W_{d,s}(k)\) compares the first \(k\) observations in the window with the remaining \(N_0-k\) observations. The scale parameter \(\sigma_{\mathrm{long}}(N_0)\) is estimated from a clean calibration window of length \(N_0\). Specifically, we repeatedly permute the calibration samples, compute the corresponding raw scan vector \(\widetilde{\mathfrak{W}}^{(\pi)}\), and form the whitened vector
\[
   \mcU^{(\pi)}
   =
   \sqrt d \; \mathbf{K}(N_0)^{-1/2}\; \widetilde{\mathfrak{W}}^{(\pi)} ,
\]
where \(\mathbf{K}(N_0)\) is the deterministic null covariance matrix of the scan vector.
Under the null CLT, \(\mcU^{(\pi)}\) approximately has covariance \(\sigma_{\mathrm{long}}^2(N_0) I\). We therefore estimate \(\sigma_{\mathrm{long}}(N_0)\) by the empirical standard deviation of the entries of ~\(\mcU^{(\pi)}\) over the generated permutations. The monitoring threshold is then
\[
   \widehat c_{N_0,\alpha}
   =
   c_{\alpha,N_0}
   \frac{\widehat\sigma_{\mathrm{long}}}{\sqrt d},
\]
where \(c_{\alpha,N_0}\) is the \((1-\alpha)\)-quantile of \(\max_{k\in\mcT_0} \mcZ_k\), with \(\mcZ\sim \mcN(0,\mathbf{K}(N_0))\). Notably, the implementation requires no hyperparameter tuning; $N_0$ and $\alpha$ are fixed design choices rather than data-tuned parameters.

\begin{figure}
    \centering
    \includegraphics[width=\linewidth]{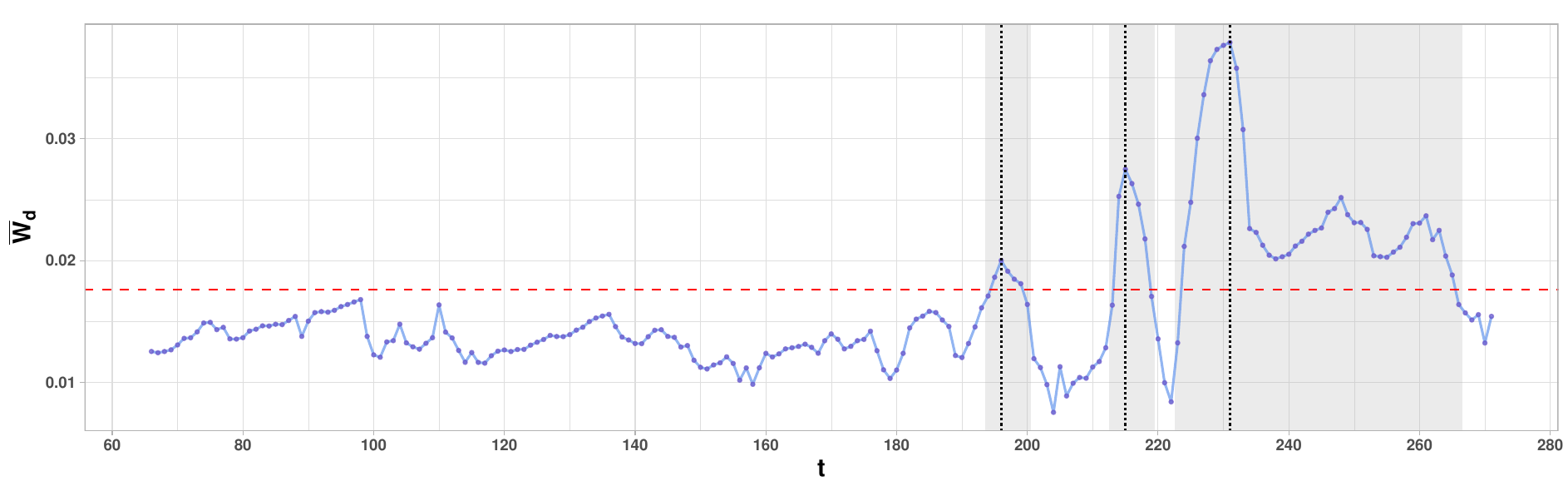}
    \captionsetup{font=footnotesize,justification=centering}
    \caption{
    Solar-flare image sequence. The curve shows the online scan statistic \( \overline{W}_d[s] = \max_{2\le k\le N_0-2} \mathfrak{W}_{d,s}(k)\) with window length \(N_0=15\). The dashed red line is the CLT-calibrated threshold using the permutation-whitened estimate of \(\sigma_{\mathrm{long}}(N_0)\). Shaded regions denote suprathreshold excursion bands, and dotted vertical lines mark the maximizing window within each band. The method identifies three event regions around \(t\approx 198,217,\) and \(227\).
    }
\label{fig:solar-flare}
\end{figure}

Figure~\ref{fig:solar-flare} shows the resulting scan path for $N_0 = 15$ and $\alpha = 2 \times 10^{-4}$. The dashed red line is the CLT-calibrated threshold, the shaded regions are excursion bands, and the dotted vertical lines mark the maximizing window within each detected excursion. The method identifies three distinct event regions, aligned with the visually apparent changes around times \(198\), \(217\), and \(227\). The largest excursion corresponds to the main flare burst, while the two smaller excursions capture weaker transient changes.
This example illustrates that the proposed angular kernel-based scan can detect multiple distributional changes in a very high-dimensional image stream without any bandwidth selection, moment estimation, or distributional modeling.


\section{Discussion}
\label{sec:discussion}

This paper introduced a dimension-averaged angular kernel (DAK) scan for detecting marginal distributional changes in high-dimensional data. The proposed statistic aggregates bounded one-dimensional discrepancies across coordinates, yielding a tuning-free and moment-agnostic procedure for detecting marginal distributional changes. For the HDLSS regime, we established an exact population mean factorization of our discrepancy statistic, characterized the null covariance structure, derived a high-dimensional CLT under cross-coordinate mixing, and used these results to obtain calibration, power, and localization guarantees. We also developed a fixed-window sequential version for high-dimensional streaming data, with ARL calibration and conditional detection delay bounds.

The scope of the method is deliberately targeted. Since the discrepancy is constructed from coordinate-wise marginal information, it is not intended to detect pure dependence changes when all one-dimensional marginals remain unchanged. This is a conscious trade-off: by focusing on marginal distributional shifts via bounded angular comparisons, the method applies to a broad class of heavy-tailed distributions, including settings where some of their marginal moments may not exist, which are beyond the scope of most existing change-point methods.

Several extensions are natural. One direction is to incorporate dependence-aware features while preserving the fixed-\(N\), large-\(d\) asymptotic structure. Another is to develop adaptive or multiscale online windows, which may improve delay in settings where the change duration or signal strength is unknown. Finally, sharper non-asymptotic delay bounds and distributionally robust variants of the angular discrepancy would further broaden the applicability of the framework.


\section*{Acknowledgment}

This work is partially supported by NSF DMS-2220495, CMMI-2112533, and the Coca-Cola Foundation. Jyotishka Ray Choudhury is also partially supported by the Ronald J. and Carol T. Beerman Presidential Fellowship at Georgia Tech.

\begingroup
\setstretch{1.05}
\putbib[references]
\endgroup

\end{bibunit}

\newpage


\appendix
\phantomsection
\addcontentsline{toc}{section}{Appendix}

\BeginAppendixToC

\input{supp}

\EndAppendixToC

\end{document}

%% file: tex-assets/macros.tex
\usepackage[margin=1in]{geometry}
\makeatletter
\def\thm@space@setup{%
  \thm@preskip=0.3\baselineskip
  \thm@postskip=0.3\baselineskip
}
\makeatother

\usepackage[max2,noblocks]{authblk}
\usepackage[symbol]{footmisc}

\usepackage[nottoc]{tocbibind}
\usepackage[toc,page,header]{appendix}
\usepackage{minitoc}

\usepackage[utf8]{inputenc} 
\usepackage[T1]{fontenc}    
\usepackage{amsmath, amsthm, amssymb, amsfonts}
\usepackage{float}
\usepackage{upgreek}
\usepackage{bm} 
\usepackage{scalerel}
\newlength\bshft
\def\fakebold#1{\ThisStyle{\ooalign{$\SavedStyle#1$\cr%
  \kern-\bshft$\SavedStyle#1$\cr%
  \kern\bshft$\SavedStyle#1$}}}

\usepackage{tikz}

\newcommand{\finedotul}[1]{%
  \tikz[baseline=(txt.base)]{
    \node[inner sep=0pt, outer sep=0pt] (txt) {#1};
    \draw[
      line width=0.45pt,
      line cap=round,
      dash pattern=on 0pt off 0.95pt
    ]
    ([yshift=-0.55ex]txt.south west) -- ([yshift=-0.55ex]txt.south east);
  }%
}

\usepackage[skip=2.5mm, indent=0.5cm]{parskip}

\makeatletter
\makeatletter

\renewcommand\section{\@startsection{section}{1}{\z@}%
  {-3.5ex plus -1ex minus -.2ex}
  {2ex plus .3ex}
  {\normalfont\Large\bfseries}}

\renewcommand\subsection{\@startsection{subsection}{2}{\z@}%
  {-3.0ex plus -1ex minus -.2ex}
  {1.5ex plus .2ex}
  {\normalfont\large\bfseries}}

\renewcommand\subsubsection{\@startsection{subsubsection}{3}{\z@}%
  {-2.0ex plus -1ex minus -.2ex}%
  {0.6ex plus .2ex}%
  {\normalfont\normalsize\bfseries}}
\makeatother

\usepackage{needspace}

\usepackage{color}
\usepackage{xcolor}
\PassOptionsToPackage{dvipsnames,svgnames,table}{xcolor}
\definecolor{darkblue}{rgb}{0,0,1}

\usepackage[colorlinks=true,
            urlcolor=purple,
            linkcolor=blue,
            citecolor=blue]{hyperref}
\usepackage{url}
\usepackage{colortbl}
\usepackage{orcidlink}

\makeatletter\def\Hy@Warning#1{}\makeatother 

\usepackage[round,authoryear]{natbib}
\usepackage{aliascnt}
\usepackage{cleveref}
\usepackage{etoolbox}
\usepackage{bibunits}

\defaultbibliographystyle{apalike}
\defaultbibliography{references}

\usepackage{nicefrac}
\usepackage{mathtools}
\usepackage{mathrsfs}
\usepackage{bm, bbm}

\usepackage{graphicx}
\usepackage{caption}
\usepackage{subcaption}
\usepackage{wrapfig}
\usepackage[section]{placeins}

\usepackage{enumitem}
\setlist[itemize]{itemsep=0.25\baselineskip, topsep=0.25\baselineskip, parsep=0pt, partopsep=0pt}
\setlist[enumerate]{itemsep=0.25\baselineskip, topsep=0.25\baselineskip, parsep=0pt, partopsep=0pt}

\usepackage{lipsum}
\usepackage{tcolorbox}
\usepackage{comment}
\usepackage{float}
\usepackage{caption}
\usepackage{array}

\allowdisplaybreaks

\usepackage{booktabs}
\usepackage{longtable}
\usepackage{makecell}
\usepackage{multirow}
\usepackage{tabu}
\usepackage{threeparttable}
\usepackage{threeparttablex}
\usepackage{xltabular}
\usepackage{adjustbox}

\newcolumntype{L}[1]{>{\raggedright\arraybackslash}p{#1}}
\newcolumntype{C}[1]{>{\centering\arraybackslash}p{#1}}

\usepackage{algorithm}
\usepackage{algpseudocode}



\makeatletter
\renewcommand{\fnum@algorithm}{\fname@algorithm}
\makeatother

\newcommand{\mcC}{\mathcal{C}}

\newcommand{\mcE}{\mathcal{E}}
\newcommand{\mcF}{\mathcal{F}}

\newcommand{\mcJ}{\mathcal{J}}

\newcommand{\mcM}{\mathcal{M}}
\newcommand{\mcN}{\mathcal{N}}
\newcommand{\mcO}{\mathcal{O}}

\newcommand{\mcT}{\mathcal{T}}
\newcommand{\mcU}{\mathcal{U}}

\newcommand{\mcW}{\mathcal{W}}
\newcommand{\mcX}{\mathcal{X}}
\newcommand{\mcY}{\mathcal{Y}}
\newcommand{\mcZ}{\mathcal{Z}}

\newcommand{\E}{\mathbb{E}}

\newcommand{\bbN}{\mathbb{N}}

\newcommand{\bbP}{\mathbb{P}}
\renewcommand{\P}{\mathbb{P}}

\newcommand{\bbR}{\mathbb{R}}
\newcommand{\R}{\mathbb{R}}

\newcommand{\bbZ}{\mathbb{Z}}

\newcommand{\var}{\operatorname{var}}
\newcommand{\cov}{\operatorname{cov}}

\newtheorem{example}{Example}[section]

\newtheorem{theorem}{Theorem}[section]

\newtheorem{assumption}{Assumption}
\newtheorem{remark}{\textbf{Remark}}

\newaliascnt{lemma}{theorem}
\newtheorem{lemma}[lemma]{Lemma}
\aliascntresetthe{lemma}
\crefname{lemma}{lemma}{lemmas}
\Crefname{lemma}{Lemma}{Lemmas}

\newaliascnt{proposition}{theorem}
\newtheorem{proposition}[proposition]{Proposition}
\aliascntresetthe{proposition}
\crefname{proposition}{proposition}{propositions}
\Crefname{proposition}{Proposition}{Propositions}

\newaliascnt{corollary}{theorem}
\newtheorem{corollary}[corollary]{Corollary}
\aliascntresetthe{corollary}
\crefname{corollary}{corollary}{corollaries}
\Crefname{corollary}{Corollary}{Corollaries}

\crefname{assumption}{assumption}{assumptions}
\Crefname{assumption}{Assumption}{Assumptions}

\newenvironment{fminipage}%
  {\begin{Sbox}\begin{minipage}}%
  {\end{minipage}\end{Sbox}\fbox{\TheSbox}}

\makeatletter
\newcommand{\ostar}{\mathbin{\mathpalette\make@circled\star}}
\newcommand{\make@circled}[2]{%
	\ooalign{$\m@th#1\smallbigcirc{#1}$\cr\hidewidth$\m@th#1#2$\hidewidth\cr}%
}
\newcommand{\smallbigcirc}[1]{%
	\vcenter{\hbox{\scalebox{0.77778}{$\m@th#1\bigcirc$}}}%
}
\makeatother

\NewDocumentCommand{\tens}{e{_^}}{%
  \mathbin{\mathop{\otimes}\displaylimits
    \IfValueT{#1}{_{#1}}
    \IfValueT{#2}{^{#2}}
  }%
}

\def\argmax{\mathop{\rm argmax}\nolimits}

\newcommand{\indep}{\perp\hspace{-2.3mm}\perp} 

\ifhidecomments
  \newcommand{\yx}[1]{}
  \newcommand{\jrc}[1]{}
  \newcommand{\jrcred}[1]{}
\else
  \newcommand{\red}[1]{\textcolor{red}{#1}}
  \newcommand{\yx}[1]{\medskip\noindent\textbf{{\red{[{YX -- #1}]}}}}

  \newcommand{\blue}[1]{\textcolor{blue}{#1}}
  \newcommand{\jrc}[1]{\medskip\noindent\textbf{{\blue{[{JRC -- #1}]}}}}

  \newcommand{\jrcred}[1]{\medskip\noindent\textbf{{\red{[{JRC -- #1}]}}}}
\fi


%% file: components/ToC.tex
\usepackage{etoolbox}
\usepackage{xparse}

\makeatletter

\newcommand{\appendixtableofcontents}{%
  \section*{Appendix Contents}%
  \@starttoc{apx}%
}

\newcommand{\BeginAppendixToC}{%
  \begingroup
  \let\apx@origaddcontentsline\addcontentsline
  \RenewDocumentCommand{\addcontentsline}{mmm}{%
    \ifstrequal{##1}{toc}{%
      \ifstrequal{##3}{Appendix}{%
        \apx@origaddcontentsline{toc}{##2}{##3}%
      }{%
        \apx@origaddcontentsline{apx}{##2}{##3}%
        \ifstrequal{##2}{section}{%
          \apx@origaddcontentsline{toc}{subsection}{##3}%
        }{%
        }%
      }%
    }{%
      \apx@origaddcontentsline{##1}{##2}{##3}%
    }%
  }%
}

\newcommand{\EndAppendixToC}{\endgroup}

\makeatother

\makeatletter
\newif\if@toc@aftersection
\@toc@aftersectionfalse

\@ifundefined{l@chapter}{}{%
  \pretocmd{\l@chapter}{\global\@toc@aftersectiontrue}{}{}
}

\pretocmd{\l@section}{\global\@toc@aftersectiontrue}{}{}

\pretocmd{\l@subsection}{%
  \if@toc@aftersection
    \addvspace{0.55ex}
    \global\@toc@aftersectionfalse
  \fi
}{}{}

\@ifundefined{l@subsubsection}{}{%
  \pretocmd{\l@subsubsection}{\global\@toc@aftersectionfalse}{}{}
}
\makeatother

%% file: components/author-info.tex
\renewcommand{\footnoterule}{%
  \kern -3pt
  \hrule width 155pt height 0.4pt
  \kern 2.6pt
}

\newcommand{\authdag}{\hspace{0.4mm}\textsuperscript{\ensuremath{\S}}\hspace{0.25mm}}
\newcommand{\authstar}{\hspace{0.4mm}\textsuperscript{\ensuremath{\star}}\hspace{0.25mm}}

\newcommand{\orcid}[1]{\raisebox{0.5ex}{\scalebox{1}{\orcidlink{#1}}}}

\author{Jyotishka {Ray Choudhury}\authdag \orcid{0000-0002-7070-683X} $\And$ 
Yao Xie\authstar\orcid{0000-0001-6777-2951}}

\DeclareRobustCommand{\GTAffiliation}{%
\begin{tabular}[t]{c}
\vspace{2mm}
H. Milton Stewart School of Industrial and Systems Engineering,\\
\textsc{Georgia Institute of Technology}\\[1.5mm]
755 Ferst Drive N.W., Atlanta, GA 30332
\end{tabular}%
}

\affil{\GTAffiliation}

%% file: abstract.tex
\begin{abstract}
We study nonparametric change-point detection for high-dimensional data in regimes where inference must be performed from small batches of observations. Our primary focus is the high-dimensional, low sample size (HDLSS) regime, where the sequence length is fixed while the ambient dimension diverges. We propose a dimension-averaged angular kernel scan framework for detecting marginal distributional shifts. The statistic aggregates bounded one-dimensional angular discrepancies across coordinates, yielding a fully nonparametric, hyperparameter-free, and moment-agnostic estimator that remains well-defined without specifying, estimating, or assuming finite marginal moments; for example, under heavy-tailed or contaminated distributions. For the offline single-change problem, we derive an exact population mean factorization into a universal deterministic shape function and a scalar signal factor, and characterize the exact null covariance structure up to a scalar variance factor, both valid for any fixed sample size and dimension. We also establish an HDLSS multivariate central limit theorem under cross-coordinate strong mixing which leads to a variance-calibrated asymptotically distribution-free test, asymptotic type-I error control, and lower bounds on power and localization accuracy. We further extend the offline procedure to a fixed-window sequential monitoring procedure for high-dimensional streaming data, and obtain ARL calibration and worst-case Pollak EDD bounds. Simulation studies demonstrate that the proposed method can accurately detect and localize changes in many challenging HDLSS and streaming high-dimensional settings where moment-based or hyperparameter-sensitive procedures may be extremely unstable or inaccurate.
\end{abstract}

\medskip
\smallskip
\begin{center}
\begin{minipage}[t]{0.9\textwidth}
\noindent\textbf{Keywords:}
\textit{change-point detection; high-dimensional statistics; nonparametric statistics; HDLSS; angular kernel; energy distance; MMD; sequential analysis.}
\end{minipage}
\end{center}

%% file: tables/cp-loc-main.tex
\begingroup
\footnotesize
\setlength{\tabcolsep}{2.7pt}
\renewcommand{\arraystretch}{1.05}
\setlength{\LTleft}{0pt plus 1fill}
\setlength{\LTright}{0pt plus 1fill}

\begin{longtable}[c]{
>{\centering\arraybackslash}m{1.9cm}
>{\raggedright\arraybackslash}m{2.4cm}
ccccc
@{}p{0.40cm}@{}
ccccc
@{}p{0.40cm}@{}
ccccc
}

\caption{\footnotesize
Empirical localization accuracy (offline problem) for \(N=40\), true change-point \(\tau=15\), and \(d\in\{200,1000,5000\}\).
Columns \(0,1,2,\geqslant\! 3\) report empirical proportions for \(|\widehat{\tau}_d-\tau|=0,1,2\), at least 3, respectively; and \(\widehat{\mathbb E}\) reports the empirical mean of \(|\widehat{\tau}_d-\tau|\).
Boldface marks the largest entry in column \(0\) within each example--dimension block; underlining marks the smallest empirical mean in the corresponding block.
}
\label{tab:cp_accuracy_main}
\\[-1mm]

\toprule
\multirow{2}{1.55cm}{\centering\arraybackslash\textbf{Example}}
&
\multirow{2}{2.25cm}{\centering\arraybackslash\textbf{Method~~~~}}
& \multicolumn{5}{c}{$d = 200$}
& \multicolumn{1}{c}{}
& \multicolumn{5}{c}{$d = 1000$}
& \multicolumn{1}{c}{}
& \multicolumn{5}{c}{$d = 5000$} \\
\cmidrule(l{4pt}r{4pt}){3-7}
\cmidrule(l{4pt}r{4pt}){9-13}
\cmidrule(l{4pt}r{4pt}){15-19}
&
& \textbf{$0$}
& \textbf{$1$}
& \textbf{$2$}
& \textbf{$\geqslant\!3$}
& \(\widehat{\mathbb E}\)
&
& \textbf{$0$}
& \textbf{$1$}
& \textbf{$2$}
& \textbf{$\geqslant\!3$}
& \(\widehat{\mathbb E}\)
&
& \textbf{$0$}
& \textbf{$1$}
& \textbf{$2$}
& \textbf{$\geqslant\!3$}
& \(\widehat{\mathbb E}\) \\
\midrule
\endfirsthead

\caption[]{\footnotesize
Empirical localization accuracy (offline problem) results \textit{(continued)}.
}
\\[-1mm]

\toprule
\multirow{2}{1.55cm}{\centering\arraybackslash\textbf{Example}}
&
\multirow{2}{2.25cm}{\centering\arraybackslash\textbf{Method~~~~}}
& \multicolumn{5}{c}{$d = 200$}
& \multicolumn{1}{c}{}
& \multicolumn{5}{c}{$d = 1000$}
& \multicolumn{1}{c}{}
& \multicolumn{5}{c}{$d = 5000$} \\
\cmidrule(l{4pt}r{4pt}){3-7}
\cmidrule(l{4pt}r{4pt}){9-13}
\cmidrule(l{4pt}r{4pt}){15-19}
&
& \textbf{$0$}
& \textbf{$1$}
& \textbf{$2$}
& \textbf{$\geqslant\!3$}
& \(\widehat{\mathbb E}\)
&
& \textbf{$0$}
& \textbf{$1$}
& \textbf{$2$}
& \textbf{$\geqslant\!3$}
& \(\widehat{\mathbb E}\)
&
& \textbf{$0$}
& \textbf{$1$}
& \textbf{$2$}
& \textbf{$\geqslant\!3$}
& \(\widehat{\mathbb E}\) \\
\midrule
\endhead

\endfoot

\bottomrule
\endlastfoot

 & DAK Scan (Ours) & \textbf{0.98} & 0.01 & 0.00 & 0.01 & \finedotul{0.12} &  & \textbf{1.00} & 0.00 & 0.00 & 0.00 & \finedotul{0.00} &  & \textbf{1.00} & 0.00 & 0.00 & 0.00 & \finedotul{0.00}\\
\nopagebreak
 & E-Divisive & 0.00 & 0.00 & 0.00 & 1.00 & 26.00 &  & 0.00 & 0.00 & 0.00 & 1.00 & 26.00 &  & 0.00 & 0.00 & 0.00 & 1.00 & 26.00\\
\nopagebreak
 & E-CP3O & 0.00 & 0.00 & 0.00 & 1.00 & 26.00 &  & 0.00 & 0.00 & 0.00 & 1.00 & 26.00 &  & 0.00 & 0.00 & 0.00 & 1.00 & 26.00\\
\nopagebreak
 & KCPA & 0.00 & 0.00 & 0.00 & 1.00 & 26.00 &  & 0.00 & 0.00 & 0.00 & 1.00 & 26.00 &  & 0.00 & 0.00 & 0.00 & 1.00 & 26.00\\
\nopagebreak
 & MMD-$\mcN$ & 0.00 & 0.00 & 0.00 & 1.00 & 17.93 &  & 0.00 & 0.00 & 0.00 & 1.00 & 18.15 &  & 0.00 & 0.00 & 0.00 & 1.00 & 18.03\\
\nopagebreak
 & MMD-$\mcE$ & 0.04 & 0.07 & 0.09 & 0.80 & 7.50 &  & 0.04 & 0.08 & 0.09 & 0.79 & 7.32 &  & 0.04 & 0.08 & 0.07 & 0.81 & 7.06\\
\nopagebreak
 & HDD-DM & 0.04 & 0.03 & 0.05 & 0.88 & 11.55 &  & 0.04 & 0.06 & 0.07 & 0.83 & 11.43 &  & 0.03 & 0.04 & 0.07 & 0.87 & 11.72\\
\nopagebreak
\multirow{-8}{1.7cm}{\centering\arraybackslash \textit{Cauchy Location}} & Sliced-Wass & 0.00 & 0.00 & 0.00 & 1.00 & 14.24 &  & 0.00 & 0.00 & 0.00 & 1.00 & 14.25 &  & 0.00 & 0.00 & 0.00 & 1.00 & 14.32\\
\cmidrule{1-19}\pagebreak[0]

 & DAK Scan (Ours) & \textbf{0.13} & 0.05 & 0.01 & 0.81 & 14.68 &  & \textbf{0.68} & 0.03 & 0.00 & 0.30 & \finedotul{5.39} &  & \textbf{0.99} & 0.00 & 0.00 & 0.01 & \finedotul{0.10}\\
\nopagebreak
 & E-Divisive & 0.00 & 0.00 & 0.00 & 1.00 & 26.00 &  & 0.00 & 0.00 & 0.00 & 1.00 & 26.00 &  & 0.00 & 0.00 & 0.00 & 1.00 & 26.00\\
\nopagebreak
 & E-CP3O & 0.00 & 0.00 & 0.00 & 1.00 & 26.00 &  & 0.00 & 0.00 & 0.00 & 1.00 & 26.00 &  & 0.00 & 0.00 & 0.00 & 1.00 & 26.00\\
\nopagebreak
 & KCPA & 0.00 & 0.00 & 0.00 & 1.00 & 26.00 &  & 0.00 & 0.00 & 0.00 & 1.00 & 26.00 &  & 0.00 & 0.00 & 0.00 & 1.00 & 26.00\\
\nopagebreak
 & MMD-$\mcN$ & 0.00 & 0.00 & 0.00 & 0.99 & 20.62 &  & 0.00 & 0.00 & 0.00 & 1.00 & 20.29 &  & 0.00 & 0.00 & 0.00 & 1.00 & 19.99\\
\nopagebreak
 & MMD-$\mcE$ & 0.02 & 0.04 & 0.05 & 0.88 & \finedotul{7.90} &  & 0.01 & 0.04 & 0.05 & 0.89 & 7.73 &  & 0.01 & 0.03 & 0.06 & 0.89 & 7.58\\
\nopagebreak
 & HDD-DM & 0.03 & 0.04 & 0.04 & 0.89 & 11.82 &  & 0.04 & 0.05 & 0.06 & 0.84 & 11.74 &  & 0.02 & 0.04 & 0.05 & 0.89 & 12.05\\
\nopagebreak
\multirow{-8}{1.7cm}{\centering\arraybackslash \textit{Cauchy Scale\phantom{h}}} & Sliced-Wass & 0.00 & 0.00 & 0.00 & 1.00 & 15.22 &  & 0.00 & 0.00 & 0.00 & 1.00 & 15.09 &  & 0.00 & 0.00 & 0.00 & 0.99 & 15.12\\
\cmidrule{1-19}\pagebreak[0]

 & DAK Scan (Ours) & \textbf{1.00} & 0.00 & 0.00 & 0.00 & \finedotul{0.00} &  & \textbf{1.00} & 0.00 & 0.00 & 0.00 & \finedotul{0.00} &  & \textbf{1.00} & 0.00 & 0.00 & 0.00 & \finedotul{0.00}\\
\nopagebreak
 & E-Divisive & 0.00 & 0.00 & 0.00 & 1.00 & 26.00 &  & 0.00 & 0.00 & 0.00 & 1.00 & 26.00 &  & 0.00 & 0.00 & 0.00 & 1.00 & 26.00\\
\nopagebreak
 & E-CP3O & 0.00 & 0.00 & 0.00 & 1.00 & 26.00 &  & 0.00 & 0.00 & 0.00 & 1.00 & 26.00 &  & 0.00 & 0.00 & 0.00 & 1.00 & 26.00\\
\nopagebreak
 & KCPA & 0.00 & 0.00 & 0.00 & 1.00 & 26.00 &  & 0.00 & 0.00 & 0.00 & 1.00 & 26.00 &  & 0.00 & 0.00 & 0.00 & 1.00 & 26.00\\
\nopagebreak
 & MMD-$\mcN$ & 0.00 & 0.00 & 0.00 & 1.00 & 23.00 &  & 0.00 & 0.00 & 0.00 & 1.00 & 23.00 &  & 0.00 & 0.00 & 0.00 & 1.00 & 23.00\\
\nopagebreak
 & MMD-$\mcE$ & 0.00 & 0.00 & 0.00 & 1.00 & 10.37 &  & 0.00 & 0.00 & 0.00 & 1.00 & 10.61 &  & 0.00 & 0.00 & 0.00 & 1.00 & 10.59\\
\nopagebreak
 & HDD-DM & 0.99 & 0.00 & 0.00 & 0.01 & 0.06 &  & \textbf{1.00} & 0.00 & 0.00 & 0.00 & \finedotul{0.00} &  & \textbf{1.00} & 0.00 & 0.00 & 0.00 & \finedotul{0.00}\\
\nopagebreak
\multirow{-8}{1.7cm}{\centering\arraybackslash \textit{Dirichlet}} & Sliced-Wass & 0.00 & 0.00 & 0.00 & 1.00 & 23.00 &  & 0.00 & 0.00 & 0.00 & 1.00 & 23.00 &  & 0.00 & 0.00 & 0.00 & 1.00 & 23.00\\
\cmidrule{1-19}\pagebreak[0]

 & DAK Scan (Ours) & \textbf{0.03} & 0.01 & 0.00 & 0.95 & 16.33 &  & \textbf{0.33} & 0.03 & 0.01 & 0.62 & 10.19 &  & \textbf{0.90} & 0.00 & 0.00 & 0.09 & \finedotul{1.26}\\
\nopagebreak
 & E-Divisive & 0.00 & 0.00 & 0.00 & 1.00 & 26.00 &  & 0.00 & 0.00 & 0.00 & 1.00 & 26.00 &  & 0.00 & 0.00 & 0.00 & 1.00 & 26.00\\
\nopagebreak
 & E-CP3O & 0.00 & 0.00 & 0.00 & 1.00 & 26.00 &  & 0.00 & 0.00 & 0.00 & 1.00 & 26.00 &  & 0.00 & 0.00 & 0.00 & 1.00 & 26.00\\
\nopagebreak
 & KCPA & 0.00 & 0.00 & 0.00 & 1.00 & 26.00 &  & 0.00 & 0.00 & 0.00 & 1.00 & 26.00 &  & 0.00 & 0.00 & 0.00 & 1.00 & 26.00\\
\nopagebreak
 & MMD-$\mcN$ & 0.00 & 0.00 & 0.00 & 1.00 & 17.14 &  & 0.00 & 0.00 & 0.00 & 1.00 & 16.54 &  & 0.00 & 0.00 & 0.00 & 1.00 & 15.06\\
\nopagebreak
 & MMD-$\mcE$ & 0.00 & 0.00 & 0.00 & 1.00 & \finedotul{7.20} &  & 0.00 & 0.00 & 0.00 & 1.00 & \finedotul{7.03} &  & 0.00 & 0.00 & 0.00 & 1.00 & 7.11\\
\nopagebreak
 & HDD-DM & 0.00 & 0.01 & 0.01 & 0.98 & 23.96 &  & 0.06 & 0.02 & 0.05 & 0.87 & 14.39 &  & 0.78 & 0.02 & 0.02 & 0.18 & 2.12\\
\nopagebreak
\multirow{-8}{1.7cm}{\centering\arraybackslash \textit{Gaussian Sparse Mean}} & Sliced-Wass & 0.00 & 0.00 & 0.00 & 1.00 & 17.36 &  & 0.00 & 0.00 & 0.00 & 1.00 & 17.88 &  & 0.00 & 0.00 & 0.00 & 1.00 & 17.54\\
\cmidrule{1-19}\pagebreak[0]

 & DAK Scan (Ours) & \textbf{1.00} & 0.00 & 0.00 & 0.00 & \finedotul{0.00} &  & \textbf{1.00} & 0.00 & 0.00 & 0.00 & \finedotul{0.00} &  & \textbf{1.00} & 0.00 & 0.00 & 0.00 & \finedotul{0.00}\\
\nopagebreak
 & E-Divisive & 0.00 & 0.00 & 0.00 & 1.00 & 26.00 &  & 0.00 & 0.00 & 0.00 & 1.00 & 26.00 &  & 0.00 & 0.00 & 0.00 & 1.00 & 26.00\\
\nopagebreak
 & E-CP3O & 0.00 & 0.00 & 0.00 & 1.00 & 26.00 &  & 0.00 & 0.00 & 0.00 & 1.00 & 26.00 &  & 0.00 & 0.00 & 0.00 & 1.00 & 26.00\\
\nopagebreak
 & KCPA & 0.00 & 0.00 & 0.00 & 1.00 & 26.00 &  & 0.00 & 0.00 & 0.00 & 1.00 & 26.00 &  & 0.00 & 0.00 & 0.00 & 1.00 & 26.00\\
\nopagebreak
 & MMD-$\mcN$ & 0.00 & 0.00 & 0.00 & 1.00 & 17.52 &  & 0.00 & 0.00 & 0.00 & 1.00 & 17.88 &  & 0.00 & 0.00 & 0.00 & 1.00 & 17.74\\
\nopagebreak
 & MMD-$\mcE$ & 0.03 & 0.06 & 0.05 & 0.86 & 8.02 &  & 0.04 & 0.08 & 0.08 & 0.80 & 7.28 &  & 0.03 & 0.07 & 0.07 & 0.83 & 7.72\\
\nopagebreak
 & HDD-DM & 0.02 & 0.04 & 0.05 & 0.88 & 12.17 &  & 0.03 & 0.05 & 0.03 & 0.89 & 12.52 &  & 0.02 & 0.05 & 0.06 & 0.87 & 11.67\\
\nopagebreak
\multirow{-8}{1.7cm}{\centering\arraybackslash \textit{Cauchy-Gaussian Mix}} & Sliced-Wass & 0.00 & 0.00 & 0.00 & 1.00 & 13.86 &  & 0.00 & 0.00 & 0.00 & 1.00 & 14.55 &  & 0.00 & 0.00 & 0.00 & 1.00 & 14.14\\
\bottomrule

\end{longtable}
\endgroup

%% file: tables/online-ARL-EDD-proposed.tex
\begingroup
\footnotesize
\setlength{\tabcolsep}{2.25pt}
\renewcommand{\arraystretch}{1.25}

\begin{table}[!ht]
\centering
\caption{
\footnotesize
Empirical online performance for \(N_0=10\), change point \(\nu=50\), and nominal calibration level \(\alpha=0.002\).
ARL denotes the empirical monitoring ARL under \(\mathbf H_{0,d}\); FA is the empirical false-alarm proportion by time \(\nu\) under \(\mathbf H_{1,d}\); CEDD is the empirical conditional expected detection delay \eqref{eq:cedd-def}; and ND is the empirical non-detection proportion by the horizon.
}
\label{tab:online_arl_edd}

\begin{adjustbox}{max width=0.9\textwidth}
\begin{tabular}{
>{\raggedright\arraybackslash}m{4cm}
cccc
@{}p{0.3cm}@{}
cccc
@{}p{0.3cm}@{}
cccc
}
\specialrule{1.1pt}{0pt}{0pt}
\multirow{2}{3cm}{\centering\arraybackslash\textbf{Example}}
& \multicolumn{4}{c}{$d = 200$}
& \multicolumn{1}{c}{}
& \multicolumn{4}{c}{$d = 1000$}
& \multicolumn{1}{c}{}
& \multicolumn{4}{c}{$d = 5000$} \\
\cmidrule(l{12pt}r{0pt}){2-5}
\cmidrule(l{10pt}r{0pt}){7-10}
\cmidrule(l{10pt}r{0pt}){12-15}
& \textbf{~~ARL~~}
& \textbf{FA}
& \textbf{~CEDD~}
& \textbf{ND}
&
& \textbf{~~ARL~~}
& \textbf{FA}
& \textbf{~CEDD~}
& \textbf{ND}
&
& \textbf{~~ARL~~}
& \textbf{FA}
& \textbf{~CEDD~}
& \textbf{ND} \\
\midrule

\textit{Cauchy Location}
& 502.2 & 0.12 & 42.73 & 0.01
&  & 494.6 & 0.09 & 1.99 & 0.00
&  & 518.1 & 0.07 & 2.00 & 0.00\\
\nopagebreak

\textit{Cauchy Scale}
& 502.2 & 0.12 & 393.65 & 0.05
&  & 494.6 & 0.09 & 386.99 & 0.03
&  & 518.1 & 0.07 & 62.82 & 0.00\\
\nopagebreak

\textit{Dirichlet}
& 485.6 & 0.12 & 1.99 & 0.00
&  & 537.6 & 0.08 & 2.00 & 0.00
&  & 531.0 & 0.08 & 2.00 & 0.00\\
\nopagebreak

\textit{Gaussian Sparse Mean}
& 521.9 & 0.12 & 402.93 & 0.05
&  & 518.1 & 0.09 & 432.24 & 0.04
&  & 529.6 & 0.07 & 292.91 & 0.02\\
\nopagebreak

\textit{Cauchy-Gaussian Mix}
& 490.9 & 0.10 & 1.99 & 0.00
&  & 542.5 & 0.09 & 2.00 & 0.00
&  & 532.3 & 0.07 & 2.00 & 0.00\\[3pt]
\specialrule{1.1pt}{0pt}{0pt}
\end{tabular}
\end{adjustbox}
\end{table}

\endgroup

%% file: supp.tex

\begin{center}
    \LARGE\textbf{\textsc{Appendix:}\\
    High-dimensional Change-point Detection\\ 
    via Angular Kernel Statistics}    
\end{center}

\vspace{1cm}
\begingroup
\setstretch{1.15}
\appendixtableofcontents
\endgroup

\begin{bibunit}


\section{Some useful lemmas and their proofs}

\subsection{One-dimensional Stieltjes identities for the signal factor \texorpdfstring{\(\delta^{*}_d\)}{delta-d}}


\begin{lemma}[One-dimensional Stieltjes identities for atom-free marginals]
\label{lem:onedim-stieltjes-identities}
Let \(F\) and \(G\) be atom-free probability distributions on \(\mathbb R\), and write \(F\) and \(G\) also for their distribution functions. Let \(h:=F-G\), and write \(dF,dG,dh\) for the associated Lebesgue--Stieltjes measures, with
\(dh=dF-dG\). Define
\[
\mathfrak{M}(F,G):=\int_{\mathbb R} h(z)^2\,dF(z).
\]
Then
\[
\mathfrak{M}(F,G) = \int_{\mathbb R} h(z)^2\,dF(z)
=
\int_{\mathbb R} h(z)^2\,dG(z).
\]
Moreover, if
\[
a^{(1)}:=\int (F+G-2FG)\,dF,\qquad
a^{(2)}:=\int (F+G-2FG)\,dG,
\]
\[
b^{(1)}:=\int (2F-2F^2)\,dF,\qquad
b^{(2)}:=\int (2F-2F^2)\,dG,
\]
\[
c^{(1)}:=\int (2G-2G^2)\,dF,\qquad
c^{(2)}:=\int (2G-2G^2)\,dG,
\]
then
\begin{align}
2a^{(1)}-b^{(1)}-c^{(1)} &= 2\,\mathfrak{M}(F,G), \label{eq:onedim-id-F}\\
2a^{(2)}-b^{(2)}-c^{(2)} &= 2\,\mathfrak{M}(F,G), \label{eq:onedim-id-G}\\
a^{(1)}+a^{(2)}-b^{(1)}-c^{(2)} &= \mathfrak{M}(F,G). \label{eq:onedim-id-mixed}
\end{align}
\end{lemma}

\begin{proof}
Since \(F\) and \(G\) are atom-free, their distribution functions are continuous. As distribution functions they are non-decreasing and bounded; hence \(F\) and \(G\), and therefore \(h=F-G\), are continuous functions of bounded variation. The finite signed Lebesgue--Stieltjes measure induced by \(h\) satisfies \(dh=dF-dG\), because the two signed measures agree on all half-open intervals \((a,b]\).

All Stieltjes integrals over \(\mathbb R\) are understood as limits over \((-m,m]\) as \(m\to\infty\); the limits exist because the integrators have finite total variation and the integrands are bounded.
We first show that the two \(h^2\)-integrals agree. Since \(h\) is continuous of bounded variation and \(h(-\infty)=h(\infty)=0\), the Stieltjes chain rule gives
\[
\int h^2\,dF-\int h^2\,dG
=
\int h^2\,dh
=
\frac{1}{3} h^3\Big|_{-\infty}^{\infty}
=0.
\]
Thus \(\mathfrak{M}(F,G)\) may equivalently be written as \(\int h^2\,dG\).

The identities \eqref{eq:onedim-id-F} and \eqref{eq:onedim-id-G} follow by direct subtraction:
\[
2(F+G-2FG)-(2F-2F^2)-(2G-2G^2)=2(F-G)^2.
\]
It remains to prove \eqref{eq:onedim-id-mixed}. From the definitions,
\[
a^{(1)}-b^{(1)}=\int h(2F-1)\,dF,
\qquad
a^{(2)}-c^{(2)}=\int h(1-2G)\,dG.
\]
Therefore,
\[
a^{(1)}+a^{(2)}-b^{(1)}-c^{(2)}
=
\int h(2F-1)\,dF+\int h(1-2G)\,dG.
\]
Using \(dF=dG+dh\), the right-hand side equals
\[
2\int h^2\,dG+\int h(2F-1)\,dh.
\]
Since \(F=G+h\),
\[
\int h(2F-1)\,dh
=
2\int hG\,dh+2\int h^2\,dh-\int h\,dh.
\]
The last two terms vanish by the Stieltjes chain rule:
\[
\int h^2\,dh= \frac{1}{3} h^3\Big|_{-\infty}^{\infty}=0,
\qquad
\int h\,dh= \frac12 h^2\Big|_{-\infty}^{\infty}=0.
\]
For the remaining term, the product rule applied to \(Gh^2\) gives
\[
Gh^2\Big|_{-\infty}^{\infty}
=
\int h^2\,dG+2\int Gh\,dh.
\]
The boundary term is zero because \(0\le G\le1\) and \(h(\pm\infty)=0\). Hence
\(2\int Gh\,dh=-\int h^2\,dG\). Consequently,
\[
a^{(1)}+a^{(2)}-b^{(1)}-c^{(2)}
=
2\int h^2\,dG-\int h^2\,dG
=
\mathfrak{M}(F,G),
\]
which proves \eqref{eq:onedim-id-mixed}.
\end{proof}

\subsection{A direct construction of $\delta^{*}_d$ from scratch}
\label{app-subsec:delta_d_construction}

Let $U,V,Z$ be mutually independent random vectors. For every $k \in [d]$ and $i = 1,2$, define $a_k^{(i)}$, $b_k^{(i)}$, $c_k^{(i)}$ in the following way, where the expectation is taken with respect to the laws of $U_k$, $V_k$, and $Z_k$.
\begin{equation}
\label{eq:E[rho_0]}
\E[\rho_0(U_k,V_k;Z_k)] := 
\begin{cases}
    ~a_k^{(1)} & \text{~if~~} U_k \sim F_d^{(k)},\, V_k \sim G_d^{(k)},\, \text{and~} Z_k \sim F_d^{(k)},\\
    ~b_k^{(1)} &  \text{~if~~} U_k \sim F_d^{(k)},\, V_k \sim F_d^{(k)},\, \text{and~} Z_k \sim F_d^{(k)},\\
    ~c_k^{(1)} & \text{~if~~} U_k \sim G_d^{(k)},\, V_k \sim G_d^{(k)},\, \text{and~} Z_k \sim F_d^{(k)}.\\
\end{cases}
\end{equation}
Analogously, $a_k^{(2)}$, $b_k^{(2)}$, and $c_k^{(2)}$ are defined as the values taken by $\E[\rho_0(U_k,V_k;Z_k)]$ when $Z_k \sim G_d^{(k)}$.

Recall the one-dimensional angular kernel $\rho_0(\cdot,\cdot;\cdot)$ defined in \eqref{eq:one-dim-angular-kernel}. For $i = 1,2$, define the dimension-averaged quantities:
\begin{equation*}
\bar a^{(i)} := \frac{1}{d}\sum_{k=1}^d a_k^{(i)}~,\qquad
\bar b^{(i)} := \frac{1}{d}\sum_{k=1}^d b_k^{(i)}~,\qquad
\bar c^{(i)} := \frac{1}{d}\sum_{k=1}^d c_k^{(i)}.
\end{equation*}
Given a change-point $\tau \in \mcT$, define
\begin{equation*}
\bar a_d := \frac{\tau}{N} \, \bar a^{(1)} + \left(1-\frac{\tau}{N}\right) \, \bar a^{(2)},~\quad
\bar b_d := \frac{\tau}{N} \, \bar b^{(1)} + \left(1-\frac{\tau}{N}\right) \, \bar b^{(2)},~\quad
\bar c_d := \frac{\tau}{N} \, \bar c^{(1)} + \left(1-\frac{\tau}{N}\right) \, \bar c^{(2)}.
\end{equation*}

\begin{lemma}[Direct construction of $\delta^{*}_d$]
\label{lem:signal_factor_construction}
\phantom{abcd}

\begin{itemize}
\item[(a)] (Existence and boundedness) For arbitrary probability distributions
\(F_d,G_d\) on \(\mathbb R^d\), each of the six quantities~
\(a_k^{(1)}\), \(a_k^{(2)}\), \(b_k^{(1)}\), \(b_k^{(2)}\), \(c_k^{(1)}\), and
\(c_k^{(2)}\) exists and lies in \([0,1]\) for every \(k\in[d]\). Consequently,
\[
\bar a^{(1)},\bar a^{(2)},\bar b^{(1)},\bar b^{(2)},\bar c^{(1)},\bar c^{(2)}
\in[0,1],
\quad\text{and}\quad
\bar a_d,\bar b_d,\bar c_d\in[0,1].
\]

\item[(b)] (Non-negativity of the signal factor under atom-free marginals)
Assume that, for every \(k\in[d]\), the one-dimensional marginal laws \(F_d^{(k)}\) and \(G_d^{(k)}\) are atom-free; equivalently, their distribution functions are continuous. Then,
\begin{equation}
\label{eq:signal_nonneg_bar_general}
\delta^{*}_d = 2\bar a_d-\bar b_d-\bar c_d\bigr.
\end{equation}

\end{itemize}

\begin{proof}
\phantom{abcd}

\begin{itemize}
\item[(a)]
Since $\rho_0(p,q;r)\in\{0,1\}$ for all $p,q,r \in \R$, each random variable $\rho_0(U_k,V_k;Z_k)$ is bounded and measurable. Hence, each expectation defining $a_k^{(1)}$, $a_k^{(2)}$, $b_k^{(1)}$, $b_k^{(2)}$, $c_k^{(1)}$, $c_k^{(2)}$ exists and lies in $[0,1]$. Averaging over $k$ preserves the interval $[0,1]$, proving that $$\bar a^{(1)}, \bar a^{(2)}, \bar b^{(1)}, \bar b^{(2)}, \bar c^{(1)}, \bar c^{(2)} \in [0,1].$$

Finally, since \(\bar a_d,\bar b_d,\bar c_d\) are convex combinations with weights \(\tau/N\) and \(1-\tau/N\), it follows that \(\bar a_d,\bar b_d,\bar c_d\in[0,1]\).
\qed

\item[(b)] 
Fix $k\in[d]$ and write $F_k:=F_d^{(k)}$ and $G_k:=G_d^{(k)}$ for the $k$-th marginals, with distribution functions (also denoted by) $F_k(\cdot)$ and $G_k(\cdot)$. 
Let \(U,V,Z\) be mutually independent real-valued random variables with the laws specified in the definitions of \(a_k^{(i)},b_k^{(i)},c_k^{(i)}\).

\medskip
Let \(h_k:=F_k-G_k\), and set
\begin{equation}
\label{eq:def-M}
\mathfrak{M}_k:=\int h_k(z)^2\,dF_k(z).
\end{equation}
By \Cref{lem:onedim-stieltjes-identities}, this also equals
\(\int h_k(z)^2\,dG_k(z)\).

\medskip
Since the marginals are continuous, strict and non-strict inequalities have the
same probability. Conditioning on \(Z=z\), if \(U\) has distribution function
\(A\) and \(V\) has distribution function \(B\), then
\[
\rho_0(U,V;z)
=
\mathbbm{1}\{U<z<V\}
+
\mathbbm{1}\{V<z<U\},
\]
and hence
\[
\E[\rho_0(U,V;Z)\mid Z=z]
=
A(z)\{1-B(z)\}+B(z)\{1-A(z)\}
=
A(z)+B(z)-2A(z)B(z).
\]
Using the above identity and by \Cref{lem:onedim-stieltjes-identities}, we have
\begin{align}
2a_k^{(1)}-b_k^{(1)}-c_k^{(1)}
&=
2\,\E_{Z\sim F_k}
\!\left[
F_k(Z)^2+G_k(Z)^2-2F_k(Z)G_k(Z)
\right] \notag\\
&=
2\,\E_{Z\sim F_k}
\!\left[
\bigl(F_k(Z)-G_k(Z)\bigr)^2
\right] \notag\\
&= 2 \int h_k^2 dF_k \notag\\
&= 2 \, \mathfrak{M}_k \,,
\label{eq:key_identity_Fk}
\end{align}
and similarly,
\begin{align}
2a_k^{(2)}-b_k^{(2)}-c_k^{(2)}
&=
2\,\E_{Z\sim G_k}
\!\left[
\bigl(F_k(Z)-G_k(Z)\bigr)^2
\right] = 2 \int h_k^2 dG_k = 2 \, \mathfrak{M}_k \,.
\label{eq:key_identity_Gk}
\end{align}

Starting from the definition of $\delta_d$, and using the definitions of \(\bar a_d,\bar b_d,\bar c_d\), we obtain
\begin{align}
&2\bar a_d-\bar b_d-\bar c_d \nonumber \\[2mm]
&= \frac{\tau}{N} \left[2 \bar a^{(1)} - \bar b^{(1)} - \bar c^{(1)}\right] + \left(1-\frac{\tau}{N}\right) \left[2 \bar a^{(2)} - \bar b^{(2)} - \bar c^{(2)}\right] \nonumber \\
&= \frac{1}{d}\frac{\tau}{N} \left[2 \sum_{k=1}^d a_k^{(1)} - \sum_{k=1}^d b_k^{(1)} - \sum_{k=1}^d c_k^{(1)}\right] + \frac{1}{d}\left(1-\frac{\tau}{N}\right) \left[2 \sum_{k=1}^d a_k^{(2)} - \sum_{k=1}^d b_k^{(2)} - \sum_{k=1}^d c_k^{(2)}\right] \nonumber \\
&= \frac{1}{d} \sum_{k=1}^d \left[ \frac{\tau}{N} \left(2a_k^{(1)}-b_k^{(1)}-c_k^{(1)}\right) + \left(1-\frac{\tau}{N}\right) \left(2a_k^{(2)}-b_k^{(2)}-c_k^{(2)}\right)\right] \nonumber \\
&= \frac{1}{d} \sum_{k=1}^d \left[ 2\frac{\tau}{N} \mathfrak{M}_k + 2\left(1-\frac{\tau}{N}\right) \mathfrak{M}_k \right] \label{eq:M-equals-2a-b-c}\\
&= \frac{2}{d} \sum_{k=1}^d ~\mathfrak{M}_k \nonumber \\
&=
\frac{2}{d} \sum_{k=1}^d \int_{\mathbb R}
\bigl\{F_d^{(k)}(z)-G_d^{(k)}(z)\bigr\}^2\,dF_d^{(k)}(z) \label{eq:use-def-M}\\
&= \delta^{*}_d \nonumber
\end{align}
which proves \eqref{eq:signal_nonneg_bar_general}. We use \eqref{eq:key_identity_Fk} and \eqref{eq:key_identity_Gk} in \eqref{eq:M-equals-2a-b-c}, and the definition of $\mathfrak{M}_k$ from \eqref{eq:def-M} in \eqref{eq:use-def-M}.

\end{itemize}

\end{proof}
\end{lemma}

\subsection{Counting lemmas pertaining to \Cref{prop:var-cov-null}}

We state and prove $10$ counting lemmas that we crucially use in the proof of \Cref{prop:var-cov-null}. The proofs are purely combinatorial, and rely on delicate counting arguments.

\input{components/computational-lemmas}


\subsection{A simple variance lemma pertaining to \Cref{prop:var(D)}}

\begin{lemma}
\label{lem:sum-of-var-upper}
Let \( X_1, X_2, \dots, X_N \) be real-valued random variables with finite second moments. Then
\[
\var\left(\frac{1}{N} \sum_{i=1}^N X_i \right)
\leq \frac{1}{N} \sum_{i=1}^N \var(X_i).
\]
\end{lemma}

\begin{proof}
We first observe that for any square-integrable real-valued random variables \( X_1, X_2, \allowbreak \dots,\allowbreak X_N \), the variance of their sum can be expressed as
\[
\var\left( \sum_{i=1}^N X_i \right)
= \sum_{i=1}^N \var(X_i) + \sum_{\substack{i,j = 1 \\ i \neq j}}^N \cov(X_i, X_j).
\]
To bound the off-diagonal terms, we apply the Cauchy--Schwarz inequality:
\[
\left| \cov(X_i, X_j) \right|
\leq \sqrt{ \var(X_i) \var(X_j) }
\leq \frac{1}{2} \Big[ \var(X_i) + \var(X_j) \Big].
\]
Using this, we obtain
\[
\sum_{\substack{i,j = 1 \\ i \neq j}}^N \left| \cov(X_i, X_j) \right|
\leq \frac{1}{2} \sum_{i \neq j} \left( \var(X_i) + \var(X_j) \right)
= (N - 1) \sum_{i=1}^N \var(X_i),
\]
where the last equality uses the fact that for each fixed \( i \), the term \( \var(X_i) \) appears exactly \( N - 1 \) times in the sum over \( i \neq j \). Therefore,
\[
\var\left( \sum_{i=1}^N X_i \right)
\leq \sum_{i=1}^N \var(X_i) + (N - 1) \sum_{i=1}^N \var(X_i)
= N \sum_{i=1}^N \var(X_i).
\]
\end{proof}

\begin{remark}
The upper bound stated above is crude for most practical cases, especially when the variables are not highly correlated (e.g., independent or weakly correlated). Nevertheless, since we are in a finite sample size regime, $N$ is fixed. We exploit this fact to show that the variance of $~\mathfrak{W}_{d}(t)$ decays with increasing dimensions.
\end{remark}


\subsection{Lemmas pertaining to \Cref{thm:multivariate-clt}}

\begin{lemma}[Monotonicity of strong mixing coefficients]
\label{lem:alpha-monotone}
Let $\{X_k\}_{k\in\bbZ}$ be any (not necessarily stationary) sequence, and define the $\alpha$-mixing coefficients by
\[
\alpha(r)
:=
\sup_{k\in\bbZ}\;
\sup_{\substack{A\in\sigma(X_i:i\le k)\\ B\in\sigma(X_i:i\ge k+r)}}
\big|\P(A\cap B)-\P(A)\P(B)\big|,
\qquad r\ge 1.
\]
Then $r\mapsto \alpha(r)$ is non-increasing, i.e., $\alpha(r_2)\le \alpha(r_1)$ for all $r_1 < r_2$.
\end{lemma}

\begin{proof}
Fix any $r_1,r_2\in\bbN$ with $r_1<r_2$. For each $k\in\bbZ$, define the past and future $\sigma$-fields
\[
\mcF_{-\infty}^k := \sigma(X_i:i\le k),
\qquad
\mcF_{k+r}^\infty := \sigma(X_i:i\ge k+r).
\]
Since $r_2>r_1$, we have the set inclusion $\mcF_{k+r_2}^\infty \subseteq \mcF_{k+r_1}^\infty$. Consequently, for every fixed $k$ and every pair of events $A\in \mcF_{-\infty}^k$ and $B\in \mcF_{k+r_2}^\infty$, the pair $(A,B)$ is also admissible in the definition of $\alpha(r_1)$ (because $B\in \mcF_{k+r_1}^\infty$ as well). Hence,
\[
\big|\P(A\cap B)-\P(A)\P(B)\big|
\le
\sup_{\substack{A'\in \mcF_{-\infty}^k\\ B'\in \mcF_{k+r_1}^\infty}}
\big|\P(A'\cap B')-\P(A')\P(B')\big|
\le \alpha(r_1).
\]
Taking the supremum over all $k\in\bbZ$ and all admissible $(A,B)$ with $A\in \mcF_{-\infty}^k$ and $B\in \mcF_{k+r_2}^\infty$ yields $\alpha(r_2)\le \alpha(r_1)$. Since $r_1<r_2$ were arbitrary, $\alpha(\cdot)$ is non-increasing.
\end{proof}


\begin{lemma}[A subsequence summability bound for non-increasing $\alpha$]
\label{lem:alpha-subsequence-bound}
Let $\alpha:\bbN\to[0,\infty)$ be non-increasing and assume $\sum_{m=1}^\infty \alpha(m)<\infty$.
Then for every integer $p\ge 1$,
\[
\sum_{r=1}^\infty \alpha(rp)
\;\le\;
\frac{1}{p}\sum_{m=1}^\infty \alpha(m).
\]
\end{lemma}

\begin{proof}
Fix any integer $p\ge 1$. For each $r\ge 1$ and each $m\in\{(r-1)p+1,\dots,rp\}$, monotonicity due to \Cref{lem:alpha-monotone} of $\alpha(\cdot)$ implies $\alpha(m)\ge \alpha(rp)$.
Therefore,
\[
\sum_{m=(r-1)p+1}^{rp} \alpha(m)
\;\ge\;
\sum_{m=(r-1)p+1}^{rp} \alpha(rp)
\;=\;
p\,\alpha(rp).
\]
Summing over $r\ge 1$ gives
\[
\sum_{m=1}^\infty \alpha(m)
=
\sum_{r=1}^\infty \sum_{m=(r-1)p+1}^{rp} \alpha(m)
\;\ge\;
p\sum_{r=1}^\infty \alpha(rp),
\]
which rearranges to the claimed inequality.
\end{proof}


\begin{lemma}[Covariance bound via strong mixing]
\label{lem:cov-alpha}
Let $U$ be $\mathcal A$-measurable and $V$ be $\mathcal B$-measurable, and assume that both random variables are bounded. Then
\[
\big|\cov(U,V)\big|
\;\le\;
4\,\|U\|_\infty\,\|V\|_\infty\,\alpha(\mathcal A,\mathcal B),
\]
where
\[
\alpha(\mathcal A,\mathcal B)
:= \sup_{A\in\mathcal A,\,B\in\mathcal B}
\big|\P(A\cap B)-\P(A)\P(B)\big|.
\]
\end{lemma}

\begin{proof}
We proceed in three steps.

\noindent\textbf{Step 1: Indicator functions.}
For any $A\in\mathcal A$ and $B\in\mathcal B$,
\[
\cov(\mathbf 1_A,\mathbf 1_B)
= \P(A\cap B)-\P(A)\P(B).
\]
By definition of $\alpha(\mathcal A,\mathcal B)$,
\[
\big|\cov(\mathbf 1_A,\mathbf 1_B)\big|
\le \alpha(\mathcal A,\mathcal B).
\]

\noindent\textbf{Step 2: Non-negative bounded random variables.}
Let $U,V\ge 0$ be bounded and measurable with respect to $\mathcal A$ and $\mathcal B$, respectively. Using the tail-integral representation (valid pointwise), we can write
\[
U = \int_0^{\|U\|_\infty} \mathbf 1_{\{U>s\}}\,ds~,
\qquad
V = \int_0^{\|V\|_\infty} \mathbf 1_{\{V>t\}}\,dt~.
\]
Next, by the bilinearity of covariance,
\[
\cov(U,V)
=
\int_0^{\|U\|_\infty}\!\!\int_0^{\|V\|_\infty}
\cov\big(\mathbf 1_{\{U>s\}},\,\mathbf 1_{\{V>t\}}\big)
\,dt\,ds.
\]
The integrand satisfies
\[
\big|\cov(\mathbf 1_{\{U>s\}},\mathbf 1_{\{V>t\}})\big|
\le 1,
\]
and hence is absolutely integrable over $[0,\|U\|_\infty]\times[0,\|V\|_\infty]$; therefore Fubini’s theorem applies. Moreover, since $\{U>s\}\in\mathcal A$ and $\{V>t\}\in\mathcal B$ for all $s,t$, Step~1 yields
\[
\big|\cov(\mathbf 1_{\{U>s\}},\mathbf 1_{\{V>t\}})\big|
\le \alpha(\mathcal A,\mathcal B).
\]
Consequently,
\[
\big|\cov(U,V)\big|
\le
\|U\|_\infty\,\|V\|_\infty\,\alpha(\mathcal A,\mathcal B).
\]

\noindent\textbf{Step 3: General bounded random variables.}
For general bounded $U,V$, write $U=U^+-U^-$ and $V=V^+-V^-$, where $U^\pm,V^\pm\ge 0$.
By bilinearity of covariance,
\begin{align*}
\cov(U,V)
&= \cov(U^+-U^-, V^+-V^-)\\
&= \cov(U^+,V^+)
-\cov(U^+,V^-)
-\cov(U^-,V^+)
+\cov(U^-,V^-).
\end{align*}
Each term is the covariance of non-negative bounded random variables, and $\|U^\pm\|_\infty\le \|U\|_\infty$, $\|V^\pm\|_\infty\le \|V\|_\infty$.
Applying Step~2 to each term yields
\[
\big|\cov(U,V)\big|
\le
4\,\|U\|_\infty\,\|V\|_\infty\,\alpha(\mathcal A,\mathcal B),
\]
which completes the proof.
\end{proof}


\begin{lemma}[Explicit block lengths]
\label{lem:explicit-blocks}
Under \Cref{ass:mixing}(ii), define
\[
p_d:=\lfloor d^{1/3}\rfloor,\qquad q_d:=\lfloor d^{1/6}\rfloor,\qquad
k_d:=\Big\lfloor\frac{d}{p_d+q_d}\Big\rfloor.
\]
Then, as $d\to\infty$,
\[
p_d\to\infty,\qquad q_d\to\infty,\qquad \frac{q_d}{p_d}\to 0,\qquad \frac{p_d}{\sqrt d}\to 0,
\qquad \text{and}\qquad k_d\,\alpha(q_d)\to 0.
\]
\end{lemma}

\begin{proof}
The first four limits are immediate:
$p_d\asymp d^{1/3}\to\infty$, $q_d\asymp d^{1/6}\to\infty$,
$q_d/p_d\asymp d^{-1/6}\to0$, and $p_d/\sqrt d\asymp d^{-1/6}\to0$.

\medskip\noindent
Next, since $p_d+q_d\asymp d^{1/3}$, we have
\[
k_d=\left\lfloor\frac{d}{p_d+q_d}\right\rfloor \asymp d^{2/3}.
\]
By \Cref{ass:mixing}(ii), $\alpha(q_d)\le C_\alpha q_d^{-\lambda}\lesssim d^{-\lambda/6}$. Hence, as $d \to \infty$,
\[
k_d\,\alpha(q_d)\ \lesssim\ d^{2/3}\cdot d^{-\lambda/6}
= d^{\,2/3-\lambda/6}\ \longrightarrow\ 0
\]
since $\lambda>4$ implies $2/3-\lambda/6<0$.
\end{proof}


\section{Deferred proofs from \Cref{sec:angular-kernel,sec:theory,sec:online}}

\subsection[Proof of \Cref{lem:signal_factor}: Properties of the population discrepancy]{Proof of \Cref{lem:signal_factor}}

\begin{enumerate}[label=(\alph*)]
\item 
By \citet[Lemma D.2]{kim2020robust}, for each \(k\in[d]\), the coordinate-wise
angular pseudometric \(\rho_{F_d,G_d}^{(k)}\) is of negative type. That is, for
any \(m\ge2\), any \(z_1,\ldots,z_m\in\bbR^d\), and any
\(a_1,\ldots,a_m\in\bbR\) satisfying \(\sum_{i=1}^m a_i=0\),
\[
\sum_{i=1}^m\sum_{j=1}^m
a_i a_j\,\rho_{F_d,G_d}^{(k)}(z_i,z_j)
\le 0.
\]
Averaging over \(k\), we obtain
\[
\sum_{i=1}^m\sum_{j=1}^m
a_i a_j\,\overline\rho_{F_d,G_d}(z_i,z_j)
=
\frac1d\sum_{k=1}^d
\sum_{i=1}^m\sum_{j=1}^m
a_i a_j\,\rho_{F_d,G_d}^{(k)}(z_{i,k},z_{j,k})
\le 0.
\]
Thus \(\overline\rho_{F_d,G_d}\) is of negative type. Since each \(\rho_{F_d,G_d}^{(k)}\) is a bounded pseudometric, their average \(\overline\rho_{F_d,G_d}\) is also a bounded pseudometric.

\medskip
Since \(\overline\rho_{F_d,G_d}\) is bounded, all expectations appearing in \(\Delta_{\alpha,d}(F_d,G_d)\) are finite. Therefore,
\[
\Delta_{\alpha,d}(F_d,G_d)
=
2\mathbb E \big[\overline\rho_{F_d,G_d}(X,Y)\big]
-\mathbb E \big[\overline\rho_{F_d,G_d}(X,X')\big]
-\mathbb E \big[\overline\rho_{F_d,G_d}(Y,Y')\big]
\]
is a well-defined generalized energy distance generated by \(\overline\rho_{F_d,G_d}\), where \(X,X'\stackrel{\mathrm{i.i.d.}}{\sim}F_d\), \(Y,Y'\stackrel{\mathrm{i.i.d.}}{\sim}G_d\), and all variables are independent.
\qed

\item
Fix any $d \in \bbN$. The existence of $\delta_d^{*}$ is guaranteed by \Cref{lem:signal_factor_construction}.

Fix \(k\in[d]\), and write
\[
F=F_d^{(k)},\qquad G=G_d^{(k)},\qquad
H_\alpha=\alpha F+(1-\alpha)G .
\]
Since \(F\) and \(G\) are atom-free, \(H_\alpha\) is atom-free. Hence, for
\(u,v\in\bbR\),
\[
\mathbb E_{R\sim H_\alpha}\{\rho_0(u,v;R)\}
=
\mathbb P_{R\sim H_\alpha}\{R \text{ lies strictly between } u \text{ and } v\}
=
|H_\alpha(u)-H_\alpha(v)|.
\]
Therefore, the \(k\)-th coordinate contribution to
\(\Delta_{\alpha,d}(F_d,G_d)\) is
\[
\Delta_{\alpha}^{(k)}
=
2\mathbb E|H_\alpha(X)-H_\alpha(Y)|
-\mathbb E|H_\alpha(X)-H_\alpha(X')|
-\mathbb E|H_\alpha(Y)-H_\alpha(Y')|,
\]
where \(X,X'\stackrel{\mathrm{i.i.d.}}{\sim}F\), \(Y,Y'\stackrel{\mathrm{i.i.d.}}{\sim}G\), and all variables are independent.

\medskip
By the one-dimensional energy identity applied to the transformed metric
\((u,v)\mapsto |H_\alpha(u)-H_\alpha(v)|\),
\[
\Delta_{\alpha}^{(k)}
=
2\int_{\bbR} \{F(z)-G(z)\}^2\,dH_\alpha(z).
\]
Since \(H_\alpha=\alpha F+(1-\alpha)G\), this becomes
\begin{align*}
\Delta_{\alpha}^{(k)}
&= 2\alpha\int_{\bbR}\{F(z)-G(z)\}^2\,dF(z)
+
2(1-\alpha)\int_{\bbR}\{F(z)-G(z)\}^2\,dG(z)\\
&= 2\alpha\int_{\bbR} h(z)^2 \, dF(z)
+
2(1-\alpha)\int_{\bbR} h(z)^2 \, dG(z)\\
&= 2\alpha\; \mathfrak{M}(F,G) + 2(1-\alpha)\; \mathfrak{M}(F,G)\\
&= 2\;\mathfrak{M}(F,G)\\
&= 2 \int_{\bbR} \{F(z)-G(z)\}^2\,dF(z)
\end{align*}
where the last three equations follow from \Cref{lem:onedim-stieltjes-identities}. Thus, $\Delta_{\alpha}^{(k)}$ is independent of \(\alpha\) for each $k \in [d]$. Averaging over \(k\in[d]\), we obtain
\[
\Delta_{\alpha,d}(F_d,G_d)
=
\frac1d\sum_{k=1}^d \Delta_{\alpha}^{(k)}
=
\frac{2}{d}\sum_{k=1}^d
\int_{\bbR}
\{F_d^{(k)}(z)-G_d^{(k)}(z)\}^2\,dF_d^{(k)}(z)
= \delta_d^{*}
\]
for every \(\alpha\in[0,1]\). \qed

\item 
Next, we prove that $0 \leq \delta_d^{*} \leq \frac{2}{3}$. We know that
\[ 
\delta_d^{*} = \frac{2}{d}\sum_{k=1}^d \int_{\bbR} \bigl\{F_d^{(k)}(z)-G_d^{(k)}(z)\bigr\}^2 \,dF_d^{(k)}(z). 
\] 
The lower bound is immediate. We prove the upper bound coordinate by coordinate. Fix \(k\in[d]\), and write \( F=F_d^{(k)}\) and \(G=G_d^{(k)}.\) Let \(X\sim F\), and set 
\[ 
U=F(X),\qquad V=G(X). 
\] 
Since \(F\) is atom-free, \(U\sim \operatorname{Unif}(0,1)\). Also, \(0\le V\le1\). Hence \[ \int_{\bbR} \{F(z)-G(z)\}^2\,dF(z) = \mathbb E\{F(X)-G(X)\}^2 = \mathbb E(U-V)^2 = \mathbb E U^2+\mathbb E V^2-2\mathbb E(UV). 
\] 

Since \(0\le V\le1\), we have \(V^2\le V\), and therefore \(\mathbb E V^2\le \mathbb E [V]\). 

\smallskip
Next, we lower bound \(\mathbb E(UV)\). Let \(X'\) be an independent copy of \(X\), and set \(U'=F(X')\), \(V'=G(X')\). Since both \(F\) and \(G\) are non-decreasing, $(U-U')(V-V') \geq 0$ almost surely. Taking expectations gives 
\[ 
0 \le \mathbb E\{(U-U')(V-V')\} = 2\{\mathbb E(UV)-\mathbb E [U]\,\mathbb E [V]\}. 
\] 
Thus, 
\[ 
\mathbb E(UV)\ge \mathbb E [U] \cdot \mathbb E [V] = \frac12\,\mathbb E [V]. 
\] 
Combining the preceding inequalities, 
\[ 
\mathbb E V^2-2\mathbb E(UV) \le \mathbb E [V] - 2\cdot \frac12\mathbb E [V] = 0. 
\] 
Therefore, 
\[ 
\int_{\bbR} \{F(z)-G(z)\}^2\,dF(z) 
= \mathbb E(U-V)^2 \le \mathbb E U^2 
= \int_0^1 u^2\,du = \frac{1}{3}. 
\] 
Since this holds for every coordinate \(k\), 
\[ 
\delta_d^{*} ~\leq~ \frac{2}{d}\sum_{k=1}^d \frac{1}{3} ~=~ \frac{2}{3}. 
\] 
Finally, the bound is sharp. Take, for example, 
\[ F_d^{(k)}=\operatorname{Unif}(0,1), \qquad G_d^{(k)}=\operatorname{Unif}(2,3), \qquad k=1,\ldots,d. 
\] 
Then, for \(X_k\sim F_d^{(k)}\), we have \(G_d^{(k)}(X_k)=0\) and \(F_d^{(k)}(X_k)\sim \operatorname{Unif}(0,1)\). Hence 
\[ 
\int_{\bbR} \bigl\{F_d^{(k)}(z)-G_d^{(k)}(z)\bigr\}^2 \,dF_d^{(k)}(z) = \int_0^1 u^2\,du = \frac{1}{3} 
\] 
for every \(k\), and therefore \(\delta_d^{*}=2/3\). \qed

\medskip
It remains to prove the $\delta_d^{*} = 0$ characterization result.
\(k\in[d]\),
\[
\E_{Z\sim F_k}
\!\left[
\bigl(F_k(Z)-G_k(Z)\bigr)^2
\right]
=
\E_{Z\sim G_k}
\!\left[
\bigl(F_k(Z)-G_k(Z)\bigr)^2
\right]
=
0.
\]
Equivalently, if
\[
h_k(x):=F_k(x)-G_k(x),
\]
then \(h_k(Z)=0\) almost surely under both \(Z\sim F_k\) and
\(Z\sim G_k\). Thus
\[
\int h_k(x)^2\,d(F_k+G_k)(x)=0.
\]
Let \(\mu_k:=F_k+G_k\), viewed as a finite measure on \(\mathbb R\). The above equation implies that \(h_k=0\) \(\mu_k\)-almost everywhere.

\medskip
We now use the atom-free assumption. Since \(F_k\) and \(G_k\) are atom-free, their distribution functions are continuous, and hence \(h_k\) is continuous. We claim that \(h_k(x)=0\) for every \(x\in\mathbb R\). 
\begin{itemize}
\item[\textbullet] First, suppose \(x\in\operatorname{supp}(\mu_k)\) and \(h_k(x)\neq0\). Then by continuity there exists an open neighborhood \(I\) of \(x\) and a constant \(\varepsilon>0\) such that \(|h_k(y)|\ge \varepsilon\) for all \(y\in I\). Since \(x\in\operatorname{supp}(\mu_k)\), we have \(\mu_k(I)>0\), contradicting \(\int h_k^2\,d\mu_k=0\). Hence $h_k(x)=0$ for all $x\in\operatorname{supp}(\mu_k)$.

\item[\textbullet] Next, suppose \(x\notin\operatorname{supp}(\mu_k)\). Let \(I\) be the connected component of \(\mathbb R\setminus\operatorname{supp}(\mu_k)\) containing \(x\). Since \(\mu_k(I)=0\), both \(F_k\) and \(G_k\) assign zero mass to \(I\); therefore both distribution functions are constant on \(I\), and so \(h_k\) is constant on \(I\). If \(I=(a,b)\) with finite endpoints, then \(a,b\in\operatorname{supp}(\mu_k)\), and by the previous paragraph \(h_k(a)=h_k(b)=0\). By continuity, the constant value of \(h_k\) on \(I\) must also be $0$. If \(I\) is an unbounded component, then either both distribution functions are identically \(0\) on \(I\), or both are identically \(1\) on \(I\), so again \(h_k=0\) on \(I\).
\end{itemize}

Thus \(h_k(x)=0\) for all \(x\in\mathbb R\). Consequently, \(F_k(x)=G_k(x)\) for every \(x\in\mathbb R\), and hence \(F_k=G_k\) as distributions. We have shown that \(\delta_d^{*}=0\) implies \(F_k=G_k\) for every \(k\in[d]\). The converse is immediate: if \(F_k=G_k\) for every \(k\), then every squared term in the representation of \(\delta_d^{*}\) is identically zero, and hence \(\delta_d^{*}=0\). 

\medskip
Therefore,
\[
\delta_d^{*}>0
\quad\Longleftrightarrow\quad
\exists\,k\in[d]\ \text{such that}\ F_d^{(k)}\neq G_d^{(k)}.
\]
In other words,
\[
\delta_d^{*}=0
\quad\Longleftrightarrow\quad
F_d^{(k)} \equiv G_d^{(k)}
\ \text{for every }k = 1,2, \ldots,d.
\]
\qed
\end{enumerate}


\subsection[Proof of \Cref{prop:pair-dependent-mmd-main}: Exact pair-dependent MMD representation of $\delta_d^{*}$]{Proof of \Cref{prop:pair-dependent-mmd-main}}

Fix a pair of distributions $(F_d,G_d)$ and a reference point $\beta\in\mathbb R^d\setminus\{0\}$. For brevity, write
\[
\overline{\rho}(x,y):=\overline{\rho}_{F_d,G_d}(x,y),
\qquad
\mathbf k(x,y):=\mathbf k_{F_d,G_d}^{(\beta)}(x,y).
\]
By the discussion preceding the proposition, $\overline{\rho}$ is a pseudometric of negative type. Hence the centered distance kernel
\[
\mathbf k(x,y)
=
\frac12\{\overline{\rho}(x,\beta)+\overline{\rho}(y,\beta)-\overline{\rho}(x,y)\}
\]
is positive semidefinite. Since $\overline{\rho}$ is bounded, all expectations below are finite.

Let $X,X'\stackrel{\mathrm{i.i.d.}}{\sim}F_d$ and
$Y,Y'\stackrel{\mathrm{i.i.d.}}{\sim}G_d$, all mutually independent. Then
\begin{align*}
\mathbb E[\mathbf k(X,X')] =
\frac12\mathbb E\!\left[
\overline{\rho}(X,\beta)+\overline{\rho}(X',\beta)-\overline{\rho}(X,X')
\right] =
\mathbb E[\overline{\rho}(X,\beta)]
-\frac12\mathbb E[\overline{\rho}(X,X')],
\end{align*}
and similarly,
\[
\mathbb E[\mathbf k(Y,Y')]
=
\mathbb E[\overline{\rho}(Y,\beta)]
-\frac12\mathbb E[\overline{\rho}(Y,Y')].
\]
Moreover,
\[
\mathbb E[\mathbf k(X,Y)]
=
\frac12\left\{
\mathbb E[\overline{\rho}(X,\beta)]
+
\mathbb E[\overline{\rho}(Y,\beta)]
-
\mathbb E[\overline{\rho}(X,Y)]
\right\}.
\]
Therefore,
\begin{align*}
\mathrm{MMD}^2(F_d,G_d;{\mathbf k})
&=
\mathbb E[\mathbf k(X,X')]
+
\mathbb E[\mathbf k(Y,Y')]
-
2\mathbb E[\mathbf k(X,Y)] \\
&\begin{aligned}
&= \mathbb E[\overline{\rho}(X,\beta)]
-\frac12\mathbb E[\overline{\rho}(X,X')] + \mathbb E[\overline{\rho}(Y,\beta)]
-\frac12\mathbb E[\overline{\rho}(Y,Y')]\\ 
&\hspace{5cm} -\mathbb E[\overline{\rho}(X,\beta)] - \mathbb E[\overline{\rho}(Y,\beta)] + \mathbb E[\overline{\rho}(X,Y)]
\end{aligned}
\\
&=
\mathbb E[\overline{\rho}(X,Y)]
-\frac12\mathbb E[\overline{\rho}(X,X')]
-\frac12\mathbb E[\overline{\rho}(Y,Y')].
\end{align*}
Multiplying both sides by $2$ gives
\[
2\,\mathrm{MMD}^2(F_d,G_d;{\mathbf k})
=
2\,\mathbb E\big[\overline{\rho}(X,Y)\big]
-
\mathbb E\big[\overline{\rho}(X,X')\big]
-
\mathbb E\big[\overline{\rho}(Y,Y')\big]
=
\Delta_{\alpha,d}(F_d,G_d) = \delta_d^{*},
\]
which proves \eqref{eq:Delta-mmd-main}.

Finally, all terms involving the reference point $\beta$ cancel in the above expansion. Hence the value of \(\mathrm{MMD}^2(F_d,G_d;\mathbf k_{F_d,G_d}^{(\beta)})\), and therefore the right-hand side of \eqref{eq:Delta-mmd-main}, does not depend on the choice of $\beta$.
\qed


\subsection{Derivation of \(\delta_d^{*}\) in the examples}
\label{subsec:derive-delta-examples}

We derive the expressions stated in Examples~\ref{ex:ex1-gaussian-mean-shift} and \ref{ex:ex2-cauchy-scale}. Throughout, we use the representation from \Cref{lem:signal_factor}:
\[
\delta_d^{*}
=
\frac{2}{d}
\sum_{k=1}^d
\int_{\mathbb R}
\{F_d^{(k)}(z)-G_d^{(k)}(z)\}^2\,dF_d^{(k)}(z)\,.
\]

\subsubsection{\Cref{ex:ex1-gaussian-mean-shift}: Gaussian mean shift}

Suppose \(F_d=\mcN(0,I_d)\) and \(G_d=\mcN(\mu,I_d)\), where \(\mu=(\mu_1,\dots,\mu_d)^\top\in\bbR^d\). Then, for each \(k\in[d]\),
\[
F_d^{(k)}(z)=\Phi(z),
\qquad
G_d^{(k)}(z)=\Phi(z-\mu_k).
\]
Substituting these marginal CDFs into the general representation immediately gives
\[
\delta^{*}_d
=
\frac{2}{d}
\sum_{k=1}^d
\E_{Z\sim\mcN(0,1)}
\bigl(\Phi(Z)-\Phi(Z-\mu_k)\bigr)^2,
\]
which proves \eqref{eq:delta-gaussian-shift}.

It remains to derive the small-shift expansion. For \(m\in\bbR\), define
\[
I(m):=
\E_{Z\sim\mcN(0,1)}
\bigl(\Phi(Z)-\Phi(Z-m)\bigr)^2.
\]
Using
\[
\Phi(z)-\Phi(z-m)
=
m\int_0^1 \phi(z-um)\,du,
\]
where \(\phi\) is the standard Gaussian density, we obtain
\[
\frac{I(m)}{m^2}
=
\int_{\bbR}
\left\{\int_0^1 \phi(z-um)\,du\right\}^2
\phi(z)\,dz.
\]
For \(|m|\le1\), the integrand is bounded by
\[
\sup_{|a|\le1}\phi(z-a)^2\,\phi(z),
\]
and this envelope is integrable because it is bounded on compact sets and has Gaussian tails. Hence dominated convergence applies. Hence, by the dominated convergence theorem,
\[
\frac{I(m)}{m^2}
\longrightarrow
\int_{\bbR}\phi(z)^3\,dz.
\]
Since
\[
\int_{\bbR}\phi(z)^3\,dz
=
(2\pi)^{-3/2}\int_{\bbR}e^{-3z^2/2}\,dz
=
\frac{1}{2\pi\sqrt3},
\]
we have
\[
I(m)=\frac{m^2}{2\pi\sqrt3}+o(m^2),
\qquad m\to0.
\]
Therefore, if \(\|\mu\|_\infty\to0\), then
\[
\delta^{*}_d
=
\frac{2}{d}
\sum_{k=1}^d
\left\{
\frac{\mu_k^2}{2\pi\sqrt3}+o(\mu_k^2)
\right\}
=
\frac{1}{\pi\sqrt3}\frac{\|\mu\|_2^2}{d}
+
o\!\left(\frac{\|\mu\|_2^2}{d}\right),
\]
where the remainder is uniform over \(k\) because \(\max_{1\le k\le d}|\mu_k|\le\|\mu\|_\infty\to0\).
\qed

\subsubsection{\Cref{ex:ex2-cauchy-scale}: Cauchy scale change}

Let \(F\) and \(G\) denote the univariate CDFs of \(\mathrm{Cauchy}(0,1)\) and \(\mathrm{Cauchy}(0,\lambda)\), respectively. Since \(F_d=F^{\otimes d}\) and \(G_d=G^{\otimes d}\), all coordinate marginals are identical. Hence, from the definition of the population signal,
\[
\delta^{*}_d
=
\frac{2}{d}
\sum_{k=1}^d
\int_{\mathbb R}
\{F(z)-G(z)\}^2\,dF(z)
=
2
\int_{\mathbb R}
\{F(z)-G(z)\}^2\,dF(z).
\]
For \(F=\mathrm{Cauchy}(0,1)\) and \(G=\mathrm{Cauchy}(0,\lambda)\), we know that
\[
F(z)=\frac12+\frac1\pi \tan^{-1}(z),
\qquad
G(z)=\frac12+\frac1\pi \tan^{-1}(z/\lambda),
\]
and
\[
dF(z)=\frac{1}{\pi(1+z^2)}\,dz.
\]
Therefore
\[
\begin{aligned}
\delta^{*}_d
&=
2
\int_{\mathbb R}
\left[
\frac{1}{\pi}
\{\tan^{-1}(z)-\tan^{-1}(z/\lambda)\}
\right]^2
\frac{1}{\pi(1+z^2)}\,dz \\[2mm]
&=
\frac{2}{\pi^3}
\int_{\mathbb R}
\frac{\{\tan^{-1}(z)-\tan^{-1}(z/\lambda)\}^2}{1+z^2}\,dz.
\end{aligned}
\]
This proves the first equality.

\medskip
\noindent
It remains to evaluate the integral
\[
I_\lambda
:=
\int_{\mathbb R}
\frac{
\left[
\tan^{-1}(z)-\tan^{-1}(z/\lambda)
\right]^2
}
{1+z^2}\,dz .
\]
Set \(z=\tan(\theta)\), where \(\theta\in(-\pi/2,\pi/2)\). Since \(dz/(1+z^2)=d\theta\), we get
\[
I_\lambda
=
\int_{-\pi/2}^{\pi/2}
\left[
\theta-\tan^{-1}\!\left(\frac{\tan(\theta)}{\lambda}\right)
\right]^2d\theta.
\]
Let $\rho:=\frac{\lambda-1}{\lambda+1}$. Since \(\lambda>0\), we have \(|\rho|<1\). We now derive a Fourier expansion for
\[
A_\lambda(\theta)
:=
\theta-\tan^{-1}\!\left(\frac{\tan(\theta)}{\lambda}\right).
\]
First,
\[
\tan A_\lambda(\theta)
=
\frac{\tan(\theta)-\tan(\theta)/\lambda}
{1+\tan^2(\theta)/\lambda}
=
\frac{(\lambda-1)\tan(\theta)}{\lambda+\tan^2(\theta)}.
\]
On the other hand,
\[
1+\rho e^{2i\theta}
=
1+\rho\cos(2\theta)+i\rho\sin(2\theta).
\]
Since \(|\rho|<1\), its real part satisfies $1+\rho\cos(2\theta)\ge 1-|\rho|>0$, so the principal argument lies in \((- \pi/2,\pi/2)\). 
Moreover,
\begin{align*}
\tan\!\left\{\operatorname{Arg}(1+\rho e^{2i\theta})\right\}
=
\frac{\rho\sin(2\theta)}{1+\rho\cos(2\theta)} 
&= \frac{2\rho\tan(\theta)}{(1+\tan^2(\theta))+\rho(1-\tan^2(\theta))} \\[2mm]
&= \frac{2\frac{\lambda-1}{\lambda+1}\tan(\theta)}{(1+\tan^2(\theta))+\frac{\lambda-1}{\lambda+1}(1-\tan^2(\theta))} \\[2mm]
&= \frac{2(\lambda-1)\tan(\theta)}{(\lambda+1)(1+\tan^2(\theta))+(\lambda-1)(1-\tan^2(\theta))} \\[2mm]
&= \frac{2(\lambda-1)\tan(\theta)}{2\lambda+2\tan^2(\theta)} \\[2mm]
&=
\frac{(\lambda-1)\tan(\theta)}{\lambda+\tan^2(\theta)},
\end{align*}
Thus \(A_\lambda(\theta)\) and \(\operatorname{Arg}(1+\rho e^{2i\theta})\) have the same tangent. Since both belong to \((- \pi/2,\pi/2)\), they are equal:
\[
A_\lambda(\theta)
=
\operatorname{Arg}(1+\rho e^{2i\theta}).
\]
Next, because \(|\rho|<1\), the logarithmic expansion is absolutely and uniformly convergent:
\[
\log(1+\rho e^{2i\theta})
=
\sum_{m=1}^{\infty}
(-1)^{m+1}\frac{\rho^m}{m}e^{2im\theta}.
\]
Taking imaginary parts gives
\[
A_\lambda(\theta)
=
\sum_{m=1}^{\infty}
(-1)^{m+1}\frac{\rho^m}{m}\sin(2m\theta).
\]
By orthogonality on \((-\pi/2,\pi/2)\),
\[
\int_{-\pi/2}^{\pi/2}
\sin(2m\theta)\sin(2n\theta)\,d\theta
=
\begin{cases}
\pi/2, & m=n,\\
0, & m\ne n.
\end{cases}
\]
Therefore, by Parseval's identity for the sine expansion on \((-\pi/2,\pi/2)\), we get
\[
I_\lambda
=
\int_{-\pi/2}^{\pi/2} A_\lambda(\theta)^2\,d\theta
=
\frac{\pi}{2}
\sum_{m=1}^\infty \frac{\rho^{2m}}{m^2}
=
\frac{\pi}{2}
\operatorname{Li}_2(\rho^2).
\]
where $\operatorname{Li}_2$ denotes the dilogarithm function.
Since \(\rho=\frac{\lambda-1}{\lambda+1}\), this gives
\[
I_\lambda
=
\frac{\pi}{2}
\operatorname{Li}_2\!\left[
\left(\frac{\lambda-1}{\lambda+1}\right)^2
\right].
\]
Substituting this into the main expression yields
\[
\delta^{*}_d
=
\frac{2}{\pi^3}
\int_{\mathbb R}
\frac{
\left[
\tan^{-1}(z)-\tan^{-1}(z/\lambda)
\right]^2
}
{1+z^2}\,dz
=
\frac{1}{\pi^2}
\operatorname{Li}_2\!\left[
\left(\frac{\lambda-1}{\lambda+1}\right)^2
\right].
\]
\qed

\subsection[Proof of \Cref{prop:E[D-star]}: Exact closed-form expression of mean discrepancy]{Proof of \Cref{prop:E[D-star]}}

For distinct sample observations \(Z_i,Z_j\), write
\[
\widehat{\overline{\rho}}(Z_i,Z_j)
=
\frac1d\sum_{k=1}^d
\frac1N\sum_{r=1}^N \rho_0(Z_{i,k},Z_{j,k};Z_{r,k}).
\]
When the anchor \(Z_r\) coincides with either endpoint \(Z_i\) or \(Z_j\), the corresponding \(\rho_0\)-term is zero by the convention in \eqref{eq:one-dim-angular-kernel}. Thus endpoint anchors must be counted as zero terms in the pooled-anchor average.

Define the endpoint-corrected expected pairwise quantities
\[
A := \frac{1}{dN}\sum_{k=1}^d
\left[(\tau-1)a_k^{(1)}+(N-\tau-1)a_k^{(2)}\right],
\]
\[
B := \frac{1}{dN}\sum_{k=1}^d
\left[(\tau-2)b_k^{(1)}+(N-\tau)b_k^{(2)}\right],
\qquad
C := \frac{1}{dN}\sum_{k=1}^d
\left[\tau c_k^{(1)}+(N-\tau-2)c_k^{(2)}\right].
\]

Then \(A,B,C\) are respectively the expectations of~ \(\widehat{\overline\rho}(F,G)\), \(\widehat{\overline\rho}(F,F)\), and \(\widehat{\overline\rho}(G,G)\), where \(F\) and \(G\) indicate the laws of the two sample endpoints. For instance, an \(F\)-\(G\) pair leaves \(\tau-1\) available \(F\)-anchors and \(N-\tau-1\) available \(G\)-anchors; the two endpoint anchors are included in the factor \(1/N\) but contribute zero.
We first show that
\begin{equation}
\label{eq:ABC-delta-proof}
2A-B-C=\delta_d.
\end{equation}

\noindent
For each coordinate \(k\), let 
\[\mathfrak{M}_k:=\int\{F_d^{(k)}-G_d^{(k)}\}^2\,dF_d^{(k)}.\]
By \Cref{lem:onedim-stieltjes-identities},
\[
2a_k^{(1)}-b_k^{(1)}-c_k^{(1)}=2\,\mathfrak{M}_k,\qquad
2a_k^{(2)}-b_k^{(2)}-c_k^{(2)}=2\,\mathfrak{M}_k,\qquad
a_k^{(1)}+a_k^{(2)}-b_k^{(1)}-c_k^{(2)}=\mathfrak{M}_k.
\]
Therefore,
\begin{align*}
2A-B-C
&=
\frac1{dN}\sum_{k=1}^d
\bigg[
\tau\{2a_k^{(1)}-b_k^{(1)}-c_k^{(1)}\}
+
(N-\tau)\{2a_k^{(2)}-b_k^{(2)}-c_k^{(2)}\} \\
&\hspace{3.5cm}
-
2\{a_k^{(1)}+a_k^{(2)}-b_k^{(1)}-c_k^{(2)}\}
\bigg] \\
&=
\frac1{dN}\sum_{k=1}^d
\{\tau(2\,\mathfrak{M}_k)+(N-\tau)(2\,\mathfrak{M}_k)-2\,\mathfrak{M}_k\} \\
&=
\frac{2(N-1)}{Nd}\sum_{k=1}^d \mathfrak{M}_k \\
&=
\delta_d,
\end{align*}
where the last equality follows from \Cref{lem:signal_factor}. This proves
\eqref{eq:ABC-delta-proof}.

Fix any \(t \in [\tau]\). For compactness, write \(\widehat\rho=\widehat{\overline{\rho}}\) throughout this proof. Recall the definition of \(\mathfrak{W}_{d}(t)\):
\begin{equation}
\mathfrak{W}_{d}(t)
=
\frac{2}{t(N-t)}
\sum_{x \in \mathcal{X}_t}
\sum_{y \in \mathcal{Y}_t}
\widehat{\rho}(x, y)
-
\frac{1}{t(t-1)}
\sum_{\substack{x,x' \in \mathcal{X}_t\\ x\neq x'}}
\widehat{\rho}(x, x')
-
\frac{1}{(N-t)(N-t-1)}
\sum_{\substack{y,y' \in \mathcal{Y}_t\\ y\neq y'}}
\widehat{\rho}(y, y').
\end{equation}
We simply take expectation on both sides, and compute:
\begin{align*}
\E [\mathfrak{W}_{d}(t)]
&=
\begin{aligned}[t]
&\frac{2}{t(N-t)}
\sum_{x \in \mathcal{X}_t}
\sum_{y \in \mathcal{Y}_t}
\E[\widehat{\rho}(x, y)]
-
\frac{1}{t(t-1)}
\sum_{\substack{x,x' \in \mathcal{X}_t\\ x\neq x'}}
\E [\widehat{\rho}(x, x')]
\\
&\qquad
-
\frac{1}{(N-t)(N-t-1)}
\sum_{\substack{y,y' \in \mathcal{Y}_t\\ y\neq y'}}
\E [\widehat{\rho}(y, y')]
\end{aligned}
\\
&=:
2P^{(1)}_t - Q^{(1)}_t - R^{(1)}_t .
\end{align*}

We simplify \(P^{(1)}_t\), \(Q^{(1)}_t\), and \(R^{(1)}_t\) individually for better readability. Since \(t\le \tau\), all observations in \(\mathcal X_t\) are from \(F_d\), whereas \(\mathcal Y_t\) contains \(\tau-t\) observations from \(F_d\) and \(N-\tau\) observations from \(G_d\).
Note that
\begin{align*}
P^{(1)}_t
&:=
\frac{1}{t(N-t)}
\sum_{x \in \mathcal{X}_t}
\sum_{y \in \mathcal{Y}_t}
\E[\widehat{\rho}(x,y)]
\\
&=
\frac{1}{t(N-t)}
\left\{
\sum_{x \in \mathcal{X}_t}
\sum_{\substack{y \in \mathcal{Y}_t\\ y\sim F_d}}
\E[\widehat{\rho}(x,y)]
+
\sum_{x \in \mathcal{X}_t}
\sum_{\substack{y \in \mathcal{Y}_t\\ y\sim G_d}}
\E[\widehat{\rho}(x,y)]
\right\}
\\
&=
\frac{1}{t(N-t)}
\left\{
t(\tau-t)B+t(N-\tau)A
\right\}
\\
&=
\frac{1}{N-t}
\left\{
(\tau-t)B+(N-\tau)A
\right\}.
\end{align*}
Similarly,
\begin{align*}
Q^{(1)}_t
&:=
\frac{1}{t(t-1)}
\sum_{\substack{x,x' \in \mathcal{X}_t\\ x\neq x'}}
\E[\widehat{\rho}(x,x')]
= \frac{1}{t(t-1)} \sum_{\substack{x,x' \in \mathcal{X}_t\\ x\neq x'}} B = B.
\end{align*}
For the third term,
\begin{align*}
R^{(1)}_t
&:=
\frac{1}{(N-t)(N-t-1)}
\sum_{\substack{y,y' \in \mathcal{Y}_t\\ y\neq y'}}
\E[\widehat{\rho}(y,y')]
\\
&=
\frac{1}{(N-t)(N-t-1)}
\Bigg\{
\sum_{\substack{y,y' \in \mathcal{Y}_t\\ y\sim F_d,\ y'\sim G_d}} \!\!\!\!\!\!\E[\widehat{\rho}(y,y')]
\quad+~
\sum_{\substack{y,y' \in \mathcal{Y}_t\\ y\sim G_d,\ y'\sim F_d}} \!\!\!\!\!\!\E[\widehat{\rho}(y,y')]
\\
&\hspace{5cm}
+ \sum_{\substack{y,y' \in \mathcal{Y}_t\\ y\sim F_d,\ y'\sim F_d\\ y\neq y'}} \!\!\!\!\!\!\E[\widehat{\rho}(y,y')]
\quad+~
\sum_{\substack{y,y' \in \mathcal{Y}_t\\ y\sim G_d,\ y'\sim G_d\\ y\neq y'}}
\!\!\!\!\!\!\E[\widehat{\rho}(y,y')]
\Bigg\}
\\
&=
\frac{1}{(N-t)(N-t-1)}
\Big\{
2(\tau-t)(N-\tau)A + (\tau-t)(\tau-t-1)B + (N-\tau)(N-\tau-1)C
\Big\}.
\end{align*}

\noindent
Denote
\(\tau_{\ell}:=\tau-t\) and \(\tau_r:=N-\tau\). Then \(N-t=\tau_\ell+\tau_r\). We combine these to obtain the following.
\begin{align*}
\E\left[\mathfrak{W}_{d}(t)\right]
&=
2P^{(1)}_t - Q^{(1)}_t - R^{(1)}_t
\\
&=
\begin{aligned}[t]
&\bigg[
\frac{2}{N-t}
\Big\{(\tau-t)B+(N-\tau)A\Big\}
-B
\\
&\qquad
-
\frac{1}{(N-t)(N-t-1)}
\Big\{
2(\tau-t)(N-\tau)A
+
(\tau-t)(\tau-t-1)B\\
&\hspace{9cm}+
(N-\tau)(N-\tau-1)C
\Big\}
\bigg]
\end{aligned}
\\
&=
\begin{aligned}[t]
&\frac{1}{(N-t)(N-t-1)}
\bigg[
2(\tau_\ell+\tau_r-1)\tau_\ell B
+
2\tau_r(\tau_\ell+\tau_r-1)A
\\
&\hspace{3cm}
-
(\tau_\ell+\tau_r)(\tau_\ell+\tau_r-1)B
-
2\tau_\ell\tau_r A
-
\tau_\ell(\tau_\ell-1)B
-
\tau_r(\tau_r-1)C
\bigg]
\end{aligned}
\\
&=
\begin{aligned}[t]
&\frac{1}{(N-t)(N-t-1)}
\bigg[
A\Big\{2\tau_r(\tau_\ell+\tau_r-1)-2\tau_\ell\tau_r\Big\}
\\
&\hspace{2.6cm}
+
B\Big\{
2\tau_\ell(\tau_\ell+\tau_r-1)
-
(\tau_\ell+\tau_r)(\tau_\ell+\tau_r-1)
-
\tau_\ell(\tau_\ell-1)
\Big\}
\\
&\hspace{2.6cm}
-
C\Big\{\tau_r(\tau_r-1)\Big\}
\bigg]
\end{aligned}
\\
&=
\begin{aligned}[t]
&\frac{1}{(N-t)(N-t-1)}
\bigg[
2A\Big\{\tau_\ell\tau_r+\tau_r^2-\tau_r-\tau_\ell\tau_r\Big\}
-
C\Big\{\tau_r(\tau_r-1)\Big\}
\\
&\hspace{2.6cm}
+
B\Big\{
2\tau_\ell^2+2\tau_\ell\tau_r-2\tau_\ell
-\tau_\ell^2-2\tau_\ell\tau_r-\tau_r^2+\tau_\ell+\tau_r
-\tau_\ell^2+\tau_\ell
\Big\}
\bigg]
\end{aligned}
\\
&=
\frac{\tau_r(\tau_r-1)}
{(N-t)(N-t-1)}
\Big[2A-B-C\Big]
\\[2mm]
&=
\frac{(N-\tau)(N-\tau-1)}
{(N-t)(N-t-1)}
\Big[2A-B-C\Big]
\\[2mm]
&=
\frac{(N-\tau)(N-\tau-1)}
{(N-t)(N-t-1)}
\delta_d\\[2mm]
&=: f_d(t).
\end{align*}
This proves the first part of \eqref{eq:E[D-star]}. 

The second part can be proved in an analogous manner. Suppose \(t\ge\tau\). Then \(\mcY_t\) contains only \(G_d\)-observations, whereas \(\mcX_t\) contains \(\tau\) observations from \(F_d\) and \(t-\tau\) observations from \(G_d\). Again write \(\E[\mathfrak W_d(t)]=2P_t^{(2)}-Q_t^{(2)}-R_t^{(2)}\). Here,
\begin{align*}
P_t^{(2)} &= \frac{\tau A+(t-\tau)C}{t},\\
Q_t^{(2)} &=
\frac{\tau(\tau-1)B + 2\tau(t-\tau)A + (t-\tau)(t-\tau-1)C}{t(t-1)},\\
R_t^{(2)} &= C.
\end{align*}
Substituting and simplifying gives
\[
\E[\mathfrak W_d(t)]
=
\frac{\tau(\tau-1)}{t(t-1)}(2A-B-C)
=
\frac{\tau(\tau-1)}{t(t-1)}\,\delta_d \,.
\]
Combining the two cases proves \eqref{eq:E[D-star]}.

Finally, the shape factor \(\Lambda_{\tau,N}(t)\) lies in \([0,1]\) for all admissible \(t\), and \(\delta_d\in[0,\nicefrac{2}{3})\) by \eqref{eq:delta-simplified-cvm}. Hence \(0\le\mu_d(t)\le\nicefrac{2}{3}\), completing the proof.
\qed

\subsection[Proof of \Cref{prop:var-cov-null}: Covariance structure of $\mathfrak{W}_d(t)$ under $\mathbf{H}_{0,d}$]{Proof of \Cref{prop:var-cov-null} 
}

We aim to show the following.
\begin{itemize}
    \item[(a)] 
    $\mathcal{V}_{d,N} := \var(\widehat{\overline{\rho}}(Z_i,Z_j)) - 2 \cov(\widehat{\overline{\rho}}(Z_i,Z_j), \widehat{\overline{\rho}}(Z_i,Z_k)) + \cov(\widehat{\overline{\rho}}(Z_i,Z_j), \widehat{\overline{\rho}}(Z_k,Z_l)) =: \sigma_d^2 - 2\kappa_d + \eta_d \in [0,1]$ (with distinct $i,j,k,l$) depends only on $(d,N,F_d)$, not on $t$.

    \item[(b)]
    For all $t,t' \in \mcT$ such that $t \leq t'$,
    \begin{equation}
    \cov\!\big(\mathfrak{W}_d(t),\mathfrak{W}_d(t')\big) = \frac{2(N-1)(N-2)}{\,t'(t'-1)(N-t)(N-t-1)} \, \mathcal{V}_{d,N}~.
    \end{equation}
\end{itemize}

\noindent\textbf{Proof of (a).} Lastly, we want to show that $\mathcal{V}_{d,N} \in [0,1]$. Fix four distinct indices $i,j,k,l$ and write
\[
G_{ab}:=\widehat{\overline{\rho}}(Z_a,Z_b)\in[0,1]\qquad(a\neq b),
\]
where the bound follows since $\rho_0(\cdot,\cdot;\cdot)\in[0,1]$ and $G_{ab}$ is an average of such terms.
Under $\mathbf{H}_{0,d}$, we have $Z_1,\dots,Z_N \stackrel{\mathrm{i.i.d.}}{\sim} F_d$, hence the array $\{G_{ab}\}$ is jointly exchangeable; in particular, for distinct indices the covariances depend only on the overlap pattern. Therefore the quantities
\[
\sigma_d^2:=\var(G_{ij}),\qquad 
\kappa_d:=\cov(G_{ij},G_{ik}),\qquad 
\eta_d:=\cov(G_{ij},G_{kl})
\]
are well-defined (do not depend on the specific choice of distinct indices), and
\[
\mathcal{V}_{d,N}=\sigma_d^2-2\kappa_d+\eta_d.
\]
Now define the quantity
\[
\varphi := (G_{ij}-G_{ik})-(G_{lj}-G_{lk}).
\]
A direct expansion gives
\begin{align*}
\var(\varphi)
&= \var(G_{ij}-G_{ik})+\var(G_{lj}-G_{lk})
   -2\cov(G_{ij}-G_{ik},\,G_{lj}-G_{lk}) \\
&= (2\sigma_d^2-2\kappa_d) + (2\sigma_d^2-2\kappa_d)
   -2\Big(\kappa_d-\eta_d-\eta_d+\kappa_d\Big) \\
&= 4(\sigma_d^2-2\kappa_d+\eta_d) \\
&= 4\mathcal{V}_{d,N}, 
\end{align*}
where we used that $\cov(G_{ij},G_{lj})=\kappa_d$ and $\cov(G_{ij},G_{lk})=\eta_d$ for distinct $i,j,k,l$ (one shared index vs.\ disjoint pairs), and similarly for the other terms. Hence
\[
\mathcal{V}_{d,N}=\frac14\,\var(\varphi)\ge 0.
\]
Moreover, since each $G_{ab}\in[0,1]$, we have $G_{ij}-G_{ik}\in[-1,1]$ and $G_{lj}-G_{lk}\in[-1,1]$, so $\Delta \in [-2,2]$ almost surely. Thus, by Popoviciu's variance inequality,
\[
\var(\varphi)\le \frac{(2-(-2))^2}{4}=4
\quad\Longrightarrow\quad
\mathcal{V}_{d,N}=\frac14\var(\varphi)\le 1.
\]
Combining the two bounds yields $\mathcal{V}_{d,N}\in[0,1]$.
\qed

\noindent\textbf{Proof of (b).} Recall that for any $t\in\mcT$, we write $\mathfrak{W}_d(t) = 2A_t - B_t - C_t$, 
where
\[
A_t
:= \frac{1}{t(N-t)}
\sum_{i\in \mathcal{X}_t}\,
\sum_{j\in \mathcal{Y}_t}
\widehat{\rho}(Z_i,Z_j) \;,
\]
\[
B_t
:= \frac{1}{t(t-1)}
\sum_{\substack{i,i'\in \mathcal{X}_t\\ i \neq i'}}
\widehat{\rho}(Z_i,Z_{i'}) \;,\qquad
C_t
:= 
\frac{1}{(N-t)(N-t-1)}
\sum_{\substack{j,j'\in \mathcal{Y}_t\\ j \neq j'}}
\widehat{\rho}(Z_j,Z_{j'})\,,
\]
where $\mathcal{X}_t = \{1,\dots,t\}$ and $\mathcal{Y}_t=\{t+1,\dots,N\}$.
Thus,
\begin{align}
\cov(\mathfrak{W}_{d}(t), \mathfrak{W}_{d}(t')) 
&= \cov(2 A_t - B_t - C_t , 2 A_{t'} - B_{t'} - C_{t'}) \notag\\
&\begin{aligned}
\notag
&=4\,\cov(A_t,A_{t'})
-2\,\cov(A_t,B_{t'})
-2\,\cov(A_t,C_{t'})\\
        &\qquad-2\,\cov(B_t,A_{t'})
+ \cov(B_t,B_{t'})
+ \cov(B_t,C_{t'})\\ 
        &\qquad\qquad-2\,\cov(C_t,A_{t'})
+ \cov(C_t,B_{t'})
+ \cov(C_t,C_{t'}).
    \end{aligned}\\
    &\begin{aligned}
        &=4\,\cov(A_t,A_{t'})
        + \cov(B_t,B_{t'})
        + \cov(C_t,C_{t'})\\
        &\qquad
        -2\,\cov(A_t,B_{t'})
        + \cov(B_t,C_{t'})
        -2\,\cov(C_t,A_{t'})\\ 
        &\qquad\qquad
        -2\,\cov(A_{t'},B_t)
        + \cov(B_{t'},C_t)
        -2\,\cov(C_{t'},A_t).\\
    \end{aligned}
\label{eq:cov_W_expansion_initial}
\end{align}

Each covariance term can be written explicitly in terms of the individual kernel evaluations $\widehat{\rho}(Z_a,Z_b)$.  
For example,
\[
\cov(A_t,A_{t'})
=
\frac{1}{t(N-t)\,t'(N-t')}
\sum_{i\in\mathcal{X}_t}\sum_{j\in\mathcal{Y}_t}
\sum_{k\in\mathcal{X}_{t'}}\sum_{\ell\in\mathcal{Y}_{t'}}
\cov\bigl(\widehat{\rho}(Z_i,Z_j),\widehat{\rho}(Z_k,Z_\ell)\bigr).
\]
Analogous expansions hold for all the remaining eight covariance blocks, each differing only in the admissible index ranges and the corresponding normalization factor.  
For instance,
\[
\cov(A_t,B_{t'})
=
\frac{1}{t(N-t)\,t'(t'-1)}
\sum_{i\in\mathcal{X}_t}\sum_{j\in\mathcal{Y}_t}
\sum_{\substack{k,k'\in\mathcal{X}_{t'}\\ k \neq k'}}
\cov\bigl(\widehat{\rho}(Z_i,Z_j),\widehat{\rho}(Z_k,Z_{k'})\bigr),
\]
\[
\cov(B_t,C_{t'})
=
\frac{1}{t(t-1)\,(N-t')(N-t'-1)}
\sum_{\substack{i,i'\in\mathcal{X}_t\\ i \neq i'}}
\sum_{\substack{\ell,\ell'\in\mathcal{Y}_{t'}\\ \ell \neq \ell'}}
\cov\bigl(\widehat{\rho}(Z_i,Z_{i'}),\widehat{\rho}(Z_\ell,Z_{\ell'})\bigr),
\]
and similarly for the remaining six terms.

\medskip

Throughout this analysis, we work under the null hypothesis
\(
\mathbf{H}_{0,d} : Z_1,\dots,Z_N \stackrel{\mathrm{i.i.d.}}{\sim} F.
\)
Hence the random vectors $(Z_a,Z_b)$ and $(Z_c,Z_d)$ have the same joint distribution whenever the index pairs $(a,b)$ and $(c,d)$ share the same \emph{overlap pattern}. In particular, for $a\neq b$ and $c\neq d$,
\(
\cov\bigl(\widehat{\rho}(Z_a,Z_b),\widehat{\rho}(Z_c,Z_d)\bigr)
\)
depends only on whether the unordered pairs $\{a,b\}$ and $\{c,d\}$ are:
\begin{itemize}
    \item[(i)] identical;
    \item[(ii)] overlapping in exactly one index;
    \item[(iii)] disjoint.
\end{itemize}
This follows from exchangeability of $(Z_1,\dots,Z_N)$ and the fact that $\widehat{\rho}$ is defined symmetrically in its arguments. Accordingly, we introduce the dimension--dependent constants
\[
\sigma_d^2
:= 
\var\bigl(\widehat{\rho}(Z_i,Z_j)\bigr),
\qquad (i\neq j),
\]
\[
\kappa_d
:=
\cov\bigl(\widehat{\rho}(Z_i,Z_j),\,\widehat{\rho}(Z_i,Z_k)\bigr),
\qquad (i,j,k \text{ all distinct}),
\]
\[
\eta_d
:=
\cov\bigl(\widehat{\rho}(Z_i,Z_j),\,\widehat{\rho}(Z_k,Z_\ell)\bigr),
\qquad (i,j,k,\ell \text{ all distinct}).
\]
These quantities are well defined under $H_0$ because all index tuples of the same overlap type have identical joint distributions.

\medskip

Each of the nine covariance blocks appearing in
\eqref{eq:cov_W_expansion_initial}
therefore admits a representation of the form
\[
\cov(\cdot,\cdot)
=
\frac{1}{N_0^{(\cdot)} + N_1^{(\cdot)} + N_2^{(\cdot)}}
\Bigl[
N_0^{(\cdot)}\,\sigma_d^2
\;+\;
N_1^{(\cdot)}\,\kappa_d
\;+\;
N_2^{(\cdot)}\,\eta_d
\Bigr],
\]
where $N_0^{(\cdot)}$, $N_1^{(\cdot)}$, $N_2^{(\cdot)}$ denote, respectively, the numbers of quadruples of indices contributing to that block for which the ordered pairs share.
\begin{itemize}
    \item $N_0~:$ Two overlapping indices (``identical pairs''),
    \item $N_1~:$ Exactly one overlapping index,
    \item $N_2~:$ Zero overlapping indices (``disjoint pairs'').
\end{itemize}
The superscript indicates the covariance block, e.g.\ $(A_t,B_{t'})$, $(C_t,C_{t'})$, and so forth.  
The corresponding normalization factors are the prefactors determined earlier, such as $[t(N-t)\,t'(N-t')]^{-1}$ for $(A_t,A_{t'})$.
Since for each block the total number of quadruples is known explicitly, once $N_0^{(\cdot)}$ and $N_2^{(\cdot)}$ are computed, the quantity $N_1^{(\cdot)}$ is obtained by subtraction. For example, denoting $N^{(A_t ,B_{t'})} = N_0^{(A_t ,B_{t'})} + N_1^{(A_t ,B_{t'})} + N_2^{(A_t ,B_{t'})}$, note that
\begin{align*}
    \cov(A_t ,B_{t'}) &=
\frac{1}{N^{(A_t ,B_{t'})}}
\left[ N_0^{(A_t ,B_{t'})}\,\sigma_d^2 \;+\; N_1^{(A_t ,B_{t'})}\,\kappa_d \;+\; N_2^{(A_t ,B_{t'})}\,\eta_d \right]\\
&= \frac{1}{N^{(A_t ,B_{t'})}}
\left[ N_0^{(A_t ,B_{t'})}\,\sigma_d^2 \;+\; \left( N^{(A_t ,B_{t'})} - N_0^{(A_t ,B_{t'})} - N_2^{(A_t ,B_{t'})} \right) \,\kappa_d \;+\; N_2^{(A_t ,B_{t'})}\,\eta_d \right]\\
&= \frac{1}{N^{(A_t ,B_{t'})}}
\left[ N^{(A_t ,B_{t'})} \kappa_d + N_0^{(A_t ,B_{t'})}\,(\sigma_d^2 - \kappa_d) \;+\; N_2^{(A_t ,B_{t'})}\, (\eta_d - \kappa_d) \right]\\
&= \kappa_d + \frac{1}{N^{(A_t ,B_{t'})}}
\left[ N_0^{(A_t ,B_{t'})}\,(\sigma_d^2 - \kappa_d) \;+\; N_2^{(A_t ,B_{t'})}\, (\eta_d - \kappa_d) \right]
\end{align*}
Analogous expressions can be found for the other $8$ pairs too. Let us denote the ordered collection of $9$ pairs by 
\[
\mathscr{P} = \Big\{ (A_t ,A_{t'}), (B_t ,B_{t'}), (C_t ,C_{t'}), (A_t ,B_{t'}), (B_t ,C_{t'}), (C_t ,A_{t'}), (A_{t'} ,B_t), (B_{t'} ,C_t), (C_{t'} ,A_t) \Big\}~.
\]
Since the weights $w = \{w^{(P,Q)}\}_{(P,Q) \in \mathscr{P}} = (4,1,1,-2,1,-2,-2,1,-2)$ in \eqref{eq:cov_W_expansion_initial} sums up to 0, thus we have
\begin{align}
&\cov(\mathfrak{W}_{d}(t), \mathfrak{W}_{d}(t')) \notag \\
&= \kappa_d\sum_{(P,Q) \in \mathscr{P}} w^{(P,Q)} + \sum_{(P,Q) \in \mathscr{P}}~ \frac{w^{(P,Q)}}{N^{(P,Q)}}
\left[ N_0^{(P,Q)}\,(\sigma_d^2 - \kappa_d) \;+\; N_2^{(P,Q)}\, (\eta_d - \kappa_d) \right] \notag \\
&\begin{aligned}
    \notag
        &=(\sigma_d^2 - \kappa_d) 
        \Bigg[4\, \frac{N_0^{(A_t ,A_{t'})}}{N^{(A_t ,A_{t'})}} 
        + \frac{N_0^{(B_t ,B_{t'})}}{N^{(B_t ,B_{t'})}} 
        + \frac{N_0^{(C_t ,C_{t'})}}{N^{(C_t ,C_{t'})}} 
        - 2\, \frac{N_0^{(A_t ,B_{t'})}}{N^{(A_t ,B_{t'})}} 
        + \frac{N_0^{(B_t ,C_{t'})}}{N^{(B_t ,C_{t'})}} \\
        &\hspace{14em}
        - 2\, \frac{N_0^{(C_t ,A_{t'})}}{N^{(C_t ,A_{t'})}} 
        - 2\, \frac{N_0^{(A_{t'} ,B_t)}}{N^{(A_{t'} ,B_t)}} 
        + \frac{N_0^{(B_{t'} ,C_t)}}{N^{(B_{t'} ,C_t)}} 
        - 2\, \frac{N_0^{(C_{t'} ,A_t)}}{N^{(C_{t'} ,A_t)}}\Bigg]
    \end{aligned}\\
&\begin{aligned}
    \notag
        &\qquad + (\eta_d - \kappa_d) 
        \Bigg[4\, \frac{N_2^{(A_t ,A_{t'})}}{N^{(A_t ,A_{t'})}} 
        + \frac{N_2^{(B_t ,B_{t'})}}{N^{(B_t ,B_{t'})}} 
        + \frac{N_2^{(C_t ,C_{t'})}}{N^{(C_t ,C_{t'})}} 
        - 2\, \frac{N_2^{(A_t ,B_{t'})}}{N^{(A_t ,B_{t'})}} 
        + \frac{N_2^{(B_t ,C_{t'})}}{N^{(B_t ,C_{t'})}} \\
        &\qquad \hspace{12em}
        - 2\, \frac{N_2^{(C_t ,A_{t'})}}{N^{(C_t ,A_{t'})}} 
        - 2\, \frac{N_2^{(A_{t'} ,B_t)}}{N^{(A_{t'} ,B_t)}} 
        + \frac{N_2^{(B_{t'} ,C_t)}}{N^{(B_{t'} ,C_t)}} 
        - 2\, \frac{N_2^{(C_{t'} ,A_t)}}{N^{(C_{t'} ,A_t)}}\Bigg]
    \end{aligned}\\[1mm]
&=: (\sigma_d^2 - \kappa_d) \mathscr{J}_0 + (\eta_d - \kappa_d) \mathscr{J}_2
\label{eq:cov-Wk-Wl-main}
\end{align}
Next, we claim that 
\begin{equation}
\label{eq:J0-J2-equal}
\mathscr{J}_0 = \mathscr{J}_2 = \frac{2(N-1)(N-2)}{t'(t'-1)(N-t)(N-t-1)}
\end{equation} 
We prove that by computing the explicit expressions for $\mathscr{J}_0$ and $\mathscr{J}_2$. In \Cref{lem:N0_all_pairs,lem:N2_At_Atprime,lem:N2_Bt_Btprime,lem:N2_Ct_Ctprime,lem:N2_At_Btprime,lem:N2_Bt_Ctprime,lem:N2_Ct_Atprime,lem:N2_Atprime_Bt,lem:N2_Btprime_Ct,lem:N2_Ctprime_At}, we compute $N_0^{(P,Q)}$ and $N_2^{(P,Q)}$ for all $(P,Q) \in \mathscr{P}$, which we use extensively in this computation.


\input{components/J0-J2-computation}

\noindent
Comparing \eqref{eq:J0-expr-final} and \eqref{eq:J2-expr-final}, we see that
\begin{align*}
\mathscr{J}_0 = \mathscr{J}_2 = \frac{2(N-1)(N-2)}{t'(t'-1)(N-t)(N-t-1)}\,,
\end{align*}
proving our claim in \eqref{eq:J0-J2-equal}. Plugging this back into \eqref{eq:cov-Wk-Wl-main}, we obtain that 
\begin{align*}
\cov(\mathfrak{W}_{d}(t), \mathfrak{W}_{d}(t'))
&= (\sigma_d^2 - \kappa_d) \mathscr{J}_0 + (\eta_d - \kappa_d) \mathscr{J}_2\\[0.2\baselineskip]
&= \frac{2(N-1)(N-2)}{t'(t'-1)(N-t)(N-t-1)} \, (\sigma_d^2 - 2\kappa_d + \eta_d)\\[0.2\baselineskip]
&= \frac{2(N-1)(N-2)}{t'(t'-1)(N-t)(N-t-1)} \, \mathcal{V}_{d,N} \,,
\end{align*}
yielding exactly what we claimed.
\qed

\noindent
Setting $t'=t$ in (b) gives
\[
\var\big(\mathfrak W_d(t)\big)
=
\cov\!\big(\mathfrak W_d(t),\mathfrak W_d(t)\big)
=
\frac{2(N-1)(N-2)}{t(t-1)(N-t)(N-t-1)}\,\mathcal{V}_{d,N}.
\]

\subsection[Proof of \Cref{prop:var(D)}: Generic variance bound under $\mathbf{H}_{0,d}$ and $\mathbf{H}_{1,d}$]{Proof of \Cref{prop:var(D)}}

    We start off by directly employing \Cref{lem:sum-of-var-upper} on the variance of $\mathfrak{W}_{d}(t)$.
    \begin{align}
    \var\left( \mathfrak{W}_{d}(t) \right) 
    &= \begin{aligned}[t]
    &\var\Bigg[ \bigg\{\frac{2}{t(N-t)} \sum_{x \in \mathcal{X}_t} \sum_{y \in \mathcal{Y}_t} \widehat{\rho}(x, y) \bigg\} - \bigg\{\frac{1}{t(t-1)} \sum_{x, x' \in \mathcal{X}_t} \widehat{\rho}(x, x') \bigg\}\\
    &\qquad\qquad- \bigg\{\frac{1}{(N-t)(N-t-1)} \sum_{y, y' \in \mathcal{Y}_t} \widehat{\rho}(y, y') \bigg\}\Bigg]\\ \nonumber
    \end{aligned}\\
    &\leqslant \begin{aligned}[t]
    &3 \Bigg[\var\bigg\{\frac{2}{t(N-t)} \sum_{x \in \mathcal{X}_t} \sum_{y \in \mathcal{Y}_t} \widehat{\rho}(x, y) \bigg\} + \var\bigg\{\frac{1}{t(t-1)} \sum_{x, x' \in \mathcal{X}_t} \widehat{\rho}(x, x') \bigg\}\\
    &\qquad\qquad +\var\bigg\{\frac{1}{(N-t)(N-t-1)} \sum_{y, y' \in \mathcal{Y}_t} \widehat{\rho}(y, y') \bigg\}\Bigg]\\ \nonumber
    \end{aligned}\\
    &= \begin{aligned}[t]
    &\frac{12}{[t(N-t)]^2} \var\Bigg\{ \sum_{x \in \mathcal{X}_t} \sum_{y \in \mathcal{Y}_t} \widehat{\rho}(x, y) \Bigg\} + \frac{4}{[t(t-1)]^2} \var\Bigg\{ \sum_{x, x' \in \mathcal{X}_t} \widehat{\rho}(x, x') \Bigg\}\\
    &\qquad\qquad+ \frac{4}{[(N-t)(N-t-1)]^2} \var\Bigg\{\sum_{y, y' \in \mathcal{Y}_t} \widehat{\rho}(y, y') \Bigg\}\Bigg]\\ \nonumber
    \end{aligned}\\
    &=: \frac{12}{[t(N-t)]^2} \cdot K_t + \frac{4}{[t(t-1)]^2} \cdot L_t + \frac{4}{[(N-t)(N-t-1)]^2} \cdot M_t
    \label{var-upper-inter}
    \end{align}


    Before proceeding further along this line of argument, we claim that for any $\mathbf{u}$, $\mathbf{v}$ arising from $F_d$ or $G_d$, we must have $\var\left(\widehat{\rho}(\mathbf{u},\mathbf{v})\right) = o(1)$ as $d \to \infty$. We prove the claim below first. Note that based on the sample $\mcX_t \cup \mcY_t = \{Z_1, Z_2, \ldots, Z_N\}$,
    \[
    \widehat{\rho}(\mathbf{u}, \mathbf{v}) = \frac{1}{N} \sum_{t=1}^{N} \rho_0(\mathbf{u}, \mathbf{v}; Z_t) \in [0,1],
    \]
    since each summand lies in $[0,1]$. Thus, for each coordinate $i \in [d]$, $\var({\widehat{\rho}}\left(u_i, v_i\right)) \leqslant \mathbb{E}\left[{\widehat{\rho}^2}\left(u_i, v_i\right)\right] \leqslant \mathbb{E}\left[{\widehat{\rho}}\left(u_i, v_i\right)\right] \in [0,1]$, owing to the fact that $\widehat{\rho}\left(u_i, v_i\right) \in [0,1]$. Therefore, we have
    \begin{align*}
        \var\left(\widehat{\rho}(\mathbf{u},\mathbf{v})\right)
        &= \var\left(\frac{1}{d}\sum_{i=1}^d{\widehat{\rho}}\left({u_i},{v_i}\right)\right)\\
        &=\frac{1}{d^2}\left[\sum_{i=1}^d \var({\widehat{\rho}}\left(u_i, v_i\right)) ~+~ 2\!\!\sum_{1\leq i<j\leq d}\cov(\widehat{\rho}(u_i, v_i), \widehat{\rho}(u_j, v_j))\right]\\
        &\leqslant \frac{1}{d^2}\left[\sum_{i=1}^d \mathbb{E}({\widehat{\rho}^2}\left(u_i, v_i\right)) ~+~ 2\!\!\sum_{1\leq i<j\leq d}\cov(\widehat{\rho}(u_i, v_i), \widehat{\rho}(u_j, v_j ))\right]\\
        &\leqslant \frac{1}{d^2}\left[d + o(d^2)\right]\\[0.2\baselineskip]
        &= \frac{1}{d} + o(1) ~\longrightarrow~ 0 ~\text{ ,  as } d \to\infty,
    \end{align*}
    where the fact that $\sum_{1\leq i<j\leq d}\cov(\widehat{\rho}(u_i, v_i), \widehat{\rho}(u_j, v_j )) = o(d^2)$ follows directly from \citet[Lemma A.6]{raychoudhury2023robust}. Coming back to the original proof, observe that
    \begin{align*}
        K_t &= \var\Bigg\{ \sum_{x \in \mathcal{X}_t} \sum_{y \in \mathcal{Y}_t} \widehat{\rho}(x, y) \Bigg\}\\
        &\leqslant t(N-t)\sum_{x \in \mathcal{X}_t} \sum_{y \in \mathcal{Y}_t} \var\left(\widehat{\rho}(x, y)\right)\\ 
        &\leqslant t(N-t)\sum_{x \in \mathcal{X}_t} \sum_{y \in \mathcal{Y}_t} \left[\frac{1}{d} + o(1)\right] = \left[t(N-t)\right]^2 \cdot \left[\frac{1}{d} + o(1)\right]
    \end{align*}
    where we use above result on convergence of the variance of $\widehat{\rho}(x, y)$. By an exactly analogous argument, one can show that as $d \to \infty$,
    \[
    L_t \leqslant \left[t(t-1)\right]^2 \cdot \left[\frac{1}{d} + o(1)\right]\quad \text{and} \quad M_t \leqslant \left[(N-t)(N-t-1)\right]^2 \cdot \left[\frac{1}{d} + o(1)\right].
    \]

    \noindent
    Combining these three bounds and plugging them back into \eqref{var-upper-inter}, we obtain
    \[
    \var\{\mathfrak W_d(t)\}
    \le C_{N,t}\left[d^{-1}+o(1)\right],
    \]
    where \(C_{N,t}<\infty\) is a constant depending only on \(N\) and \(t\).
    Thus, for each fixed \(t\in\mathcal T\), we have $\var\{\mathfrak W_d(t)\}\to0$ as $d\to\infty$, proving the claimed result.
    \qed

\subsection[Proof of \Cref{thm:D-consistency}: Pointwise consistency of $\mathfrak{W}_{d}(t)$ under $\mathbf{H}_{0,d}$ and $\mathbf{H}_{1,d}$]{Proof of \Cref{thm:D-consistency}}

    Let \( t \in \mcT \) be fixed. For any \( \varepsilon > 0 \), we apply Chebyshev's inequality:
    \[
    \mathbb{P}\Big( \left| \mathfrak{W}_{d}(t) - \mathbb{E}[\mathfrak{W}_{d}(t)] \right| > \varepsilon \Big) \leq \frac{\var(\mathfrak{W}_{d}(t))}{\varepsilon^2}.
    \]
    By \Cref{prop:var(D)}, \( \var(\mathfrak{W}_{d}(t)) \longrightarrow 0 \) as \( d \to \infty \). Therefore,
    \[
    \lim_{d \to \infty} \mathbb{P}\left( \left| \mathfrak{W}_{d}(t) - \mathbb{E}[\mathfrak{W}_{d}(t)] \right| > \varepsilon \right) = 0.
    \]
    Since this holds for every \( \varepsilon > 0 \), we conclude that
    \[
    \mathfrak{W}_{d}(t) \stackrel{\mathbb{P}}{\longrightarrow} \E[\mathfrak{W}_{d}(t)] \text{ , as } d \to\infty. \qedhere
    \]
    Under \(\mathbf{H}_{1,d}\), \(\E[\mathfrak W_d(t)]=\mu_d(t)\) by \Cref{prop:E[D-star]}; under \(\mathbf{H}_{0,d}\), the same identity holds with \(\mu_d(t)=0\). Therefore,
    \[
    \mathfrak W_d(t)-\mu_d(t) \stackrel{\mathbb{P}}{\longrightarrow} 0.
    \]
    Finally, since \(\mcT\) is finite,
    \[
    \mathbb P\left(
    \max_{t\in\mcT}
    \left|\mathfrak W_d(t)-\mu_d(t)\right|>\varepsilon
    \right)
    \le
    \sum_{t\in\mcT}
    \mathbb P\left(
    \left|\mathfrak W_d(t)-\mu_d(t)\right|>\varepsilon
    \right)
    \longrightarrow 0 \,.
    \]
    This proves the uniform statement as well.
    \qed

\subsection[Proof of \Cref{prop:longrun-exists}: Existence of the long-run variance factor under $\mathbf{H}_{0,d}$]{Proof of \Cref{prop:longrun-exists}}

Throughout the proof, expectations, variances, and covariances are taken under
\(\mathbf{H}_{0,d}\). By \eqref{eq:coord-decomposition},
\[
\mathfrak W_d(u)=\frac1d\sum_{k=1}^d \xi_k(u),
\qquad u\in\mcT,
\]
where, by \eqref{eq:xi_k}, each \(\xi_k(u)\) is a measurable function only of the \(k\)-th coordinate trajectory
\[
\mathbf Z^{(k)}=(Z_{1,k},\dots,Z_{N,k})^\top\in\bbR^N.
\]
Thus, for each \(u\in\mcT\), there exists a measurable map \(g_u:\bbR^N\to\bbR\) such that $\xi_k(u)=g_u(\mathbf Z^{(k)})$ for all $k \geq 1$.
Equivalently, with $\boldsymbol{\xi}_k:=(\xi_k(u))_{u\in\mcT}\in\bbR^{|\mcT|}$, there exists a measurable map \(g:\bbR^N\to\bbR^{|\mcT|}\) such that
\[
\boldsymbol{\xi}_k=g(\mathbf Z^{(k)}),\qquad k\ge1.
\]
Since \(\rho_0\in\{0,1\}\), it follows from \eqref{eq:xi_k} that
$|\xi_k(u)| \leq 2$ for all $k\ge1$ and $u\in\mcT$. Hence,
\[
\sup_{k\ge1}\sup_{u\in\mcT}\|\xi_k(u)\|_{L^\infty}\le2.
\]

We claim that \(\{\boldsymbol{\xi}_k\}_{k\ge1}\) is strictly stationary and \(\alpha\)-mixing, with mixing coefficients bounded above by those in \Cref{ass:mixing}. Strict stationarity follows immediately from the strict stationarity of \(\{\mathbf Z^{(k)}\}_{k\ge1}\), since \(\boldsymbol{\xi}_k\) is obtained from \(\mathbf Z^{(k)}\) by the same measurable map \(g\) for every \(k\).

For the mixing property, define
\[
\mathcal G_m^-:=\sigma(\boldsymbol{\xi}_k:k\le m),
\qquad
\mathcal G_{m+r}^+:=\sigma(\boldsymbol{\xi}_k:k\ge m+r).
\]
Since \(\boldsymbol{\xi}_k\) is measurable with respect to
\(\sigma(\mathbf Z^{(k)})\),
\[
\mathcal G_m^-\subseteq \sigma(\mathbf Z^{(i)}:i\le m)=\mathcal A_m,
\qquad
\mathcal G_{m+r}^+\subseteq \sigma(\mathbf Z^{(i)}:i\ge m+r)=\mathcal B_{m+r}.
\]
Therefore, for every \(m\ge1\),
\[
\sup_{\substack{A\in\mathcal G_m^-\\ B\in\mathcal G_{m+r}^+}}
\left|\P(A\cap B)-\P(A)\P(B)\right|
\le
\sup_{\substack{A\in\mathcal A_m\\ B\in\mathcal B_{m+r}}}
\left|\P(A\cap B)-\P(A)\P(B)\right|
\le \alpha(r).
\]
Taking the supremum over \(m\), we obtain $\alpha_{\boldsymbol\xi}(r)\le \alpha(r)$.
Thus \(\{\boldsymbol{\xi}_k\}_{k\ge1}\), and hence each scalar sequence \(\{\xi_k(u)\}_{k\ge1}\), is \(\alpha\)-mixing with coefficients bounded by \(\alpha(r)\).

Now define the centered variables
\[
\widetilde \xi_k(u):=\xi_k(u)-\E[\xi_k(u)],
\qquad u\in\mcT,\ k\ge1.
\]
Since \(|\xi_k(u)|\le2\), we have
\[
|\widetilde \xi_k(u)|\le4
\qquad\text{for all }k\ge1,\ u\in\mcT.
\]
For \(r\ge0\), define
\[
\gamma_{u,v}(r)
:=
\cov\!\left(\widetilde\xi_1(u),\widetilde\xi_{1+r}(v)\right),
\qquad u,v\in\mcT.
\]

We first show that the corresponding long-run covariance series is absolutely summable. For \(r\ge1\), the variables \(\widetilde\xi_1(u)\) and \(\widetilde\xi_{1+r}(v)\) are measurable with respect to sigma-fields separated by lag \(r\). Hence, by \Cref{lem:cov-alpha},
\[
|\gamma_{u,v}(r)|
\le
4\|\widetilde\xi_1(u)\|_{L^\infty}
  \|\widetilde\xi_{1+r}(v)\|_{L^\infty}
  \alpha(r)
\le
64\,\alpha(r).
\]
Since \(\sum_{r=1}^\infty\alpha(r)<\infty\), we have
\[
\sum_{r=1}^\infty |\gamma_{u,v}(r)|<\infty
\qquad\text{for every }u,v\in\mcT.
\]
Consequently, for every \(u,v\in\mcT\), the series
\[
\Gamma_N(u,v)
:=
\gamma_{u,v}(0)
+
\sum_{r=1}^\infty \gamma_{u,v}(r)
+
\sum_{r=1}^\infty \gamma_{v,u}(r)
\]
is absolutely convergent and finite.

We now compute the covariance limit. By \eqref{eq:coord-decomposition},
\[
\mathfrak W_d(u)-\E[\mathfrak W_d(u)]
=
\frac1d\sum_{k=1}^d \widetilde\xi_k(u),
\qquad u\in\mcT.
\]
Hence,
\[
\cov\!\left(\mathfrak W_d(u),\mathfrak W_d(v)\right)
=
\frac1{d^2}
\sum_{k=1}^d\sum_{\ell=1}^d
\cov\!\left(\widetilde\xi_k(u),\widetilde\xi_\ell(v)\right).
\]
By stationarity of \(\{\boldsymbol{\xi}_k\}_{k\ge1}\),
\[
\cov\!\left(\widetilde\xi_k(u),\widetilde\xi_k(v)\right)
=
\gamma_{u,v}(0),
\]
and, for \(r\ge1\),
\[
\cov\!\left(\widetilde\xi_k(u),\widetilde\xi_{k+r}(v)\right)
=
\gamma_{u,v}(r),
\qquad
\cov\!\left(\widetilde\xi_{k+r}(u),\widetilde\xi_k(v)\right)
=
\gamma_{v,u}(r).
\]
Therefore,
\[
\cov\!\left(\mathfrak W_d(u),\mathfrak W_d(v)\right)
=
\frac1{d^2}
\left[
d\,\gamma_{u,v}(0)
+
\sum_{r=1}^{d-1}(d-r)\gamma_{u,v}(r)
+
\sum_{r=1}^{d-1}(d-r)\gamma_{v,u}(r)
\right].
\]
Multiplying by \(d\), we get
\begin{equation}
\label{eq:d-cov-finite-proof}
d\,\cov\!\left(\mathfrak W_d(u),\mathfrak W_d(v)\right)
=
\gamma_{u,v}(0)
+
\sum_{r=1}^{d-1}\left(1-\frac rd\right)\gamma_{u,v}(r)
+
\sum_{r=1}^{d-1}\left(1-\frac rd\right)\gamma_{v,u}(r).
\end{equation}

To pass to the limit in \eqref{eq:d-cov-finite-proof}, we use dominated convergence for series (see, e.g.\ \citet[Theorem~10.28]{Apostol1974analysis}). For each fixed \(r\ge1\),
\[
\mathbf 1\{r\le d-1\}\left(1-\frac rd\right)\gamma_{u,v}(r)
\longrightarrow \gamma_{u,v}(r),
\qquad
\mathbf 1\{r\le d-1\}\left(1-\frac rd\right)\gamma_{v,u}(r)
\longrightarrow \gamma_{v,u}(r).
\]
Moreover,
\[
\left|
\mathbf 1\{r\le d-1\}\left(1-\frac rd\right)\gamma_{u,v}(r)
\right|
\le |\gamma_{u,v}(r)|,
\qquad
\left|
\mathbf 1\{r\le d-1\}\left(1-\frac rd\right)\gamma_{v,u}(r)
\right|
\le |\gamma_{v,u}(r)|,
\]
and both dominating sequences are summable. Hence, by DCT,
\[
\sum_{r=1}^{d-1}\left(1-\frac rd\right)\gamma_{u,v}(r)
\longrightarrow
\sum_{r=1}^\infty \gamma_{u,v}(r),
\qquad
\sum_{r=1}^{d-1}\left(1-\frac rd\right)\gamma_{v,u}(r)
\longrightarrow
\sum_{r=1}^\infty \gamma_{v,u}(r).
\]
Thus, from \eqref{eq:d-cov-finite-proof},
\[
\lim_{d \to \infty}\, d\,\cov\!\left(\mathfrak W_d(u),\mathfrak W_d(v)\right)
= \Gamma_N(u,v)\,.
\]

We now specialize to the variance to identify the limit of
\(d\,\mathcal V_{d,N}\). Fix any \(t_0\in\mcT\). By
\Cref{prop:var-cov-null}, with \(t=t'=t_0\),
\[
\var\left(\mathfrak W_d(t_0)\right)
=
\frac{2(N-1)(N-2)}
{t_0(t_0-1)(N-t_0)(N-t_0-1)}
\,\mathcal V_{d,N}.
\]
Therefore,
\[
d\,\mathcal V_{d,N}
=
\frac{
t_0(t_0-1)(N-t_0)(N-t_0-1)
}
{2(N-1)(N-2)}
\,d\,\var\left(\mathfrak W_d(t_0)\right).
\]
From the covariance limit above with \(u=v=t_0\),
\[
d\,\var\left(\mathfrak W_d(t_0)\right)
\longrightarrow
\Gamma_N(t_0,t_0).
\]
Consequently,
\[
d\,\mathcal V_{d,N}
\longrightarrow
\frac{
t_0(t_0-1)(N-t_0)(N-t_0-1)
}
{2(N-1)(N-2)}
\,\Gamma_N(t_0,t_0)
=:\sigma_{\mathrm{long}}^2(N).
\]
This limit is finite and non-negative. Moreover, the right-hand side does not depend on the particular choice of \(t_0\), since it is the limit of the same sequence \(d\,\mathcal V_{d,N}\). Hence there exists \(\sigma_{\mathrm{long}}^2(N)\in[0,\infty)\) such that
\[
d\,\mathcal V_{d,N}\longrightarrow \sigma_{\mathrm{long}}^2(N).
\]
Combining this convergence with \Cref{prop:var-cov-null} also yields the limiting covariance factorization
\begin{equation}
\label{eq:longrun-cov-factorization}
\Gamma_N(u,v)
=
\frac{2(N-1)(N-2)}
{v(v-1)(N-u)(N-u-1)}
\,\sigma_{\mathrm{long}}^2(N),
\qquad u\le v,\ \ u,v\in\mcT.
\end{equation}
This completes the proof.
\qed

\subsection[Proof of \Cref{thm:multivariate-clt}: Multivariate CLT in the HDLSS regime under $\{\mathbf{H}_{0,d}\}$]{Proof of \Cref{thm:multivariate-clt}}

We employ the Cram\'er--Wold device (see, e.g., \citet{Kallenberg2002book}). Fix any non-zero $\mathbf{a}\in\R^{|\mcT|}$ and define
\[
Y_k:=\mathbf{a}^\top \boldsymbol{\xi}_k~,\qquad \text{for~} k\ge 1.
\]

Note that since $\boldsymbol{\xi}_k$ is a measurable function of $\mathbf{Z}^{(k)}$, we have
$\sigma(\boldsymbol{\xi}_k)\subseteq \sigma(\mathbf{Z}^{(k)})$. Hence, the induced sequence
$\{\boldsymbol{\xi}_k\}_{k\ge 1}$ is strictly stationary and $\alpha$-mixing with coefficients
$\alpha_{\xi}(r)$ satisfying
\begin{equation}
\label{eq:alpha-monotone}
\alpha_{\xi}(r)\le \alpha(r)\qquad \text{for all } r\ge 1.
\end{equation}

Therefore, $\{Y_k\}$ being a function of $\{\boldsymbol{\xi}_k\}$, is strictly stationary, $\alpha$-mixing with coefficients upper bounded by $\alpha(\cdot)$
via \eqref{eq:alpha-monotone}, and is uniformly bounded:
\begin{equation}
\label{eq:Y-bounded}
|Y_k|\le 2\|\mathbf{a}\|_1=:B
\qquad \text{a.s. for all }k.
\end{equation}
Under $\mathbf{H}_{0,d}$, $\E[Y_k]= \E[\mathbf{a}^\top \boldsymbol{\xi}_k] = \mathbf{a}^\top \E[\boldsymbol{\xi}_k] = 0$. It suffices to show that
\begin{equation}
\label{eq:target-scalar}
S_d:=\frac{1}{\sqrt{d}}\sum_{k=1}^d Y_k
\ \stackrel{\mathscr{L}}{\longrightarrow}\ \mathcal{N}(0,\nu^2),
\qquad
\nu^2:=\sigma^2_{\mathrm{long}}(N)\cdot \mathbf{a}^\top \mathbf{K} \mathbf{a}.
\end{equation}

\noindent\textbf{Step 0 (Sub-linear variance growth).}
First, we put an upper bound on $\var\left(\sum_{k=1}^m Y_k\right)$. By expanding the variance into a double summation over all index pairs $(i, j)$, we separate the diagonal variance terms from the off-diagonal covariance terms. Stationarity implies that $\cov(Y_i, Y_j)$ depends only on the lag $h = |j-i|$, and for each fixed lag $h \in \{1, \dots, m-1\}$, there are exactly $m-h$ such pairs.
\begin{align}
\var\left(\sum_{k=1}^m Y_k\right)
&= \sum_{i=1}^m\sum_{j=1}^m \cov(Y_i,Y_j) \notag \\
&= m\var(Y_1)+2\sum_{h=1}^{m-1}(m-h)\cov(Y_1,Y_{1+h}) \notag \\
&\le m\var(Y_1)+2m \sum_{h=1}^{m-1}\big|\cov(Y_1,Y_{1+h})\big| \notag \\
&\le m\left[\var(Y_1)+8B^2\sum_{h=1}^\infty \alpha(h)\right] \label{eq:appl-lem-cov}\\
&=: C_0\,m,\label{eq:var-linear}
\end{align}
where we apply \Cref{lem:cov-alpha} with $U=Y_1$ and $V=Y_{1+h}$ to obtain
$|\cov(Y_1,Y_{1+h})|\le 4B^2\,\alpha_Y(h)$, and use $\alpha_Y(h)\le \alpha(h)$ together with $\sum_{h=1}^\infty \alpha(h)<\infty$ from \Cref{ass:mixing}(i).

\noindent\textbf{Step 1 (Blocking).}
Let us define the following sequences:
\[
p_d:=\lfloor d^{1/3}\rfloor,\qquad q_d:=\lfloor d^{1/6}\rfloor,\qquad
k_d:=\Big\lfloor\frac{d}{p_d+q_d}\Big\rfloor.
\]
Define big blocks $I_j$ and small blocks $J_j$ by
\[
I_j=\{(j-1)(p_d+q_d)+1,\dots,(j-1)(p_d+q_d)+p_d\}~,~~
J_j=\{(j-1)(p_d+q_d)+p_d+1,\dots,j(p_d+q_d)\},
\]
and set
\[
U_j:=\sum_{i\in I_j} Y_i~,\qquad
V_j:=\sum_{i\in J_j} Y_i~,\qquad
R_d:=\sum_{i=k_d(p_d+q_d)+1}^{d} Y_i~.
\]
Then,
\begin{equation}
\sum_{k=1}^d Y_k=\sum_{j=1}^{k_d}U_j+\sum_{j=1}^{k_d}V_j+R_d~.
\end{equation}

\noindent\textbf{Step 2 (Remainder is negligible).}
The remainder has at most $(p_d+q_d)$ terms. Hence, by \eqref{eq:Y-bounded},
\[
\Big|\frac{R_d}{\sqrt d}\Big|
\le \frac{(p_d+q_d)B}{\sqrt d}
\ \longrightarrow\ 0~,
\]
since by \Cref{lem:explicit-blocks}, $p_d + q_d \asymp d^\frac{1}{3}$. Therefore $R_d/\sqrt d \to 0$ almost surely (and hence, in probability too).

\noindent\textbf{Step 3 (Small blocks are negligible).}
Write $S_{V,d}:=\sum_{j=1}^{k_d}V_j$. First, by \eqref{eq:var-linear} applied to a block of length $q_d$,
\begin{equation}
\label{eq:VarVj}
\var(V_j)\le C_0\,q_d\qquad \text{for all }j.
\end{equation}
Second, for $i<j$, the $\sigma$-fields $\sigma(V_i)$ and $\sigma(V_j)$, generated by $V_i$ and $V_j$, respectively, are separated by at least
$(j-i-1)p_d$ coordinates, and by at least $p_d$ when $j=i+1$. Thus, by \Cref{lem:cov-alpha},
\begin{equation}
\label{eq:CovViVj}
|\cov(V_i,V_j)|
\le 4\,\|V_i\|_\infty \|V_j\|_\infty\,\alpha\big((j-i-1)p_d \vee p_d\big)
\le 4\,(q_d B)^2\,\alpha\big((j-i-1)p_d \vee p_d\big).
\end{equation}
Hence, we have the following:
\begin{align}
\var(S_{V,d})
&= \sum_{j=1}^{k_d}\var(V_j) +2\sum_{1\le i<j\le k_d}\cov(V_i,V_j) \notag\\
&\le C_0 \, k_d q_d
+ 8\,q_d^2 B^2 \sum_{r=1}^{k_d-1} (k_d-r)\,\alpha((r-1) p_d \vee p_d) \notag\\
&\le C_0 \, k_d q_d + 16\,q_d^2 B^2\,k_d \sum_{s=1}^{\infty} \alpha(s p_d) \notag \\
&\le C_0 \, k_d q_d + 16\,B^2\,\frac{q_d^2 k_d}{p_d} \sum_{s=1}^{\infty} \alpha(s) \label{eq:lem-alpha-subsequence-bound-application}\\
&\le C_0 \, k_d q_d + 16\,C'_0\,B^2\,\frac{q_d^2 k_d}{p_d}~, \label{eq:VarSumV}
\end{align}
where \eqref{eq:lem-alpha-subsequence-bound-application} is a direct application of \Cref{lem:alpha-subsequence-bound}, and \eqref{eq:VarSumV} follows from \Cref{ass:mixing}(i) which states that $\sum_{r=1}^\infty \alpha(r)< C'_0$ for some $C'_0 < \infty$.
Therefore, \eqref{eq:VarSumV} implies 
\begin{align}
\var\left(\frac{1}{\sqrt d} \sum_{j=1}^{k_d}V_j \right)
=
\frac{1}{d}\var(S_{V,d})
&= \mcO\left(\frac{k_d q_d}{d} + \frac{q_d^2 k_d}{dp_d}\right)\\
&= \mcO\left(\frac{d^{\frac{2}{3} + \frac{1}{6}}}{d} + \frac{d^{\frac{1}{3} + \frac{2}{3}}}{d^{1 + \frac{1}{3}}}\right)
= \mcO\left(d^{-\frac{1}{6}} + d^{-\frac{1}{3}}\right)
\ \longrightarrow\ 0
\end{align}
by \Cref{lem:explicit-blocks}. Hence $d^{-1/2}\sum_{j=1}^{k_d}V_j\to 0$ in $L^2$, and therefore in probability.

\noindent\textbf{Step 4 (Coupling big blocks to independent copies).}
Let $\widetilde U_1,\dots,\widetilde U_{k_d}$ be independent random variables such that
$\widetilde U_j\stackrel{d}{=}U_j$ for each $j$.
Let
\[
\psi_d(t):=\E\left[\exp\Big(it\sum_{j=1}^{k_d}U_j\Big)\right] =\E\left[\prod_{j=1}^{k_d}e^{itU_j}\right] ,\quad
\widetilde\psi_d(t):=\E\left[\exp\Big(it\sum_{j=1}^{k_d}\widetilde U_j\Big)\right]=\prod_{j=1}^{k_d}\E\left[ e^{it\widetilde U_j}\right].
\]
By the Volkonski\u{\i}--Rozanov inequality (\citet{VolkonskiiRozanov1959}; and also, \citet[Lemma 2.4]{FanMasry1992RegressionMixing}), we have the uniform bound
\begin{equation}
\label{eq:VR}
\sup_{t\in\R}|\psi_d(t)-\widetilde\psi_d(t)|
\le 16\,(k_d-1)\,\alpha(q_d).
\end{equation}
Applying \eqref{eq:VR} at \(t/\sqrt d\) shows that the characteristic functions of \(d^{-1/2}\sum_{j=1}^{k_d}U_j\) and \(d^{-1/2}\sum_{j=1}^{k_d}\widetilde U_j\) differ by at most \(16(k_d-1)\alpha(q_d)\), uniformly in \(t\in\mathbb R\).
By \Cref{lem:explicit-blocks}, the right-hand side tends to $0$, as $d \to \infty$. Therefore, $\frac{1}{\sqrt{d}}\sum_{j=1}^{k_d}U_j$ and $\frac{1}{\sqrt{d}}\sum_{j=1}^{k_d}\widetilde U_j$ have the same asymptotic law.

\noindent\textbf{Step 5a (Existence of the long-run variance for $\{Y_k\}$).}
Let $\gamma(h):=\cov(Y_1,Y_{1+h})$ for $h\ge 0$. We first show that
$\sum_{h\ge 1}|\gamma(h)|<\infty$ and that
\begin{equation}
\label{eq:lrv-Y-basic}
\lim_{m\to\infty}\frac{1}{m}\var\left(\;\sum_{k=1}^m Y_k\right)
=
\gamma(0)+2\sum_{h=1}^\infty \gamma(h)
=
\var(Y_1)+2\sum_{h=1}^\infty \cov(Y_1,Y_{1+h}).
\end{equation}
For each $h\ge 1$, by bilinearity of covariance,
\[
\gamma(h)
=\cov(Y_1,Y_{1+h})
=\cov(\mathbf a^\top\boldsymbol{\xi}_1,\mathbf a^\top\boldsymbol{\xi}_{1+h})
=\mathbf a^\top\cov(\boldsymbol{\xi}_1,\boldsymbol{\xi}_{1+h})\,\mathbf a.
\]
Applying \Cref{lem:cov-alpha} to the bounded variables $Y_1$ and $Y_{1+h}$ yields
\[
|\gamma(h)|
\le
4\|Y_1\|_\infty^2\,\alpha_\xi(h)
\le
C\,\alpha_\xi(h),
\]
for some constant $C<\infty$ depending only on $\mathbf a$ and
$\|\boldsymbol{\xi}_1\|_\infty$;
and hence $\sum_{h\ge 1}|\gamma(h)|<\infty$ since $\sum_{h\ge 1}\alpha_\xi(h)<\infty$.

\noindent
Next, let $S_m:=\sum_{k=1}^m Y_k$. By stationarity,
\[
\var(S_m)
=
\sum_{i=1}^m\sum_{j=1}^m \cov(Y_i,Y_j)
=
m\,\gamma(0)+2\sum_{h=1}^{m-1}(m-h)\gamma(h).
\]
Therefore,
\[
\frac{1}{m}\var(S_m)
=
\gamma(0)+2\sum_{h=1}^{m-1}\Big(1-\frac{h}{m}\Big)\gamma(h).
\]
Since $\sum_{h\ge1}|\gamma(h)|<\infty$, dominated convergence theorem (DCT) for series (e.g.\ \citet[Theorem~10.28]{Apostol1974analysis}) implies
\[
\lim_{m\to\infty}\frac{1}{m}\var(S_m)
=
\gamma(0)+2\sum_{h=1}^\infty \gamma(h),
\]
which proves \eqref{eq:lrv-Y-basic}.

\noindent\textbf{Step 5b (Identification and appearance of \(\mathbf K\)).}
We now identify the scalar long-run variance in \eqref{eq:lrv-Y-basic}. Define the long-run covariance matrix
\[
\Sigma_\xi
:=
\left(\Gamma_N(u,v)\right)_{u,v\in\mcT},
\]
where \(\Gamma_N(u,v)\) is the absolutely convergent long-run covariance identified in \Cref{prop:longrun-exists}; equivalently,
\[
\Gamma_N(u,v)
=
\gamma_{u,v}(0)
+
\sum_{h=1}^\infty \gamma_{u,v}(h)
+
\sum_{h=1}^\infty \gamma_{v,u}(h),
\]
with
\[
\gamma_{u,v}(h)
:=
\cov\!\left(\widetilde\xi_1(u),\widetilde\xi_{1+h}(v)\right).
\]
By \Cref{prop:longrun-exists}, or equivalently by
\eqref{eq:longrun-cov-factorization},
$\Sigma_\xi
=
\sigma_{\mathrm{long}}^2(N)\,\mathbf K$.
Now recall that \(Y_k=\mathbf a^\top\boldsymbol\xi_k\). Since
\(\E[Y_k]=0\) under \(\mathbf{H}_{0,d}\), we have
\[
\gamma(h):=\cov(Y_1,Y_{1+h})
=
\mathbf a^\top
\cov(\boldsymbol\xi_1,\boldsymbol\xi_{1+h})
\mathbf a.
\]
Therefore,
\begin{align*}
\gamma(0)+2\sum_{h=1}^\infty \gamma(h)
&=
\mathbf a^\top
\left[
\cov(\boldsymbol\xi_1,\boldsymbol\xi_1)
+
2\sum_{h=1}^\infty
\cov(\boldsymbol\xi_1,\boldsymbol\xi_{1+h})
\right]
\mathbf a.
\end{align*}
For scalar quadratic forms, the one-sided and symmetrized versions agree:
\[
\mathbf a^\top
\cov(\boldsymbol\xi_1,\boldsymbol\xi_{1+h})
\mathbf a
=
\mathbf a^\top
\cov(\boldsymbol\xi_{1+h},\boldsymbol\xi_1)
\mathbf a.
\]
Hence
\[
\gamma(0)+2\sum_{h=1}^\infty \gamma(h)
=
\mathbf a^\top \Sigma_\xi \mathbf a
=
\sigma_{\mathrm{long}}^2(N)\,\mathbf a^\top \mathbf K\mathbf a.
\]
Combining this with \eqref{eq:lrv-Y-basic}, we obtain
\begin{equation}
\label{eq:lrv-Y}
\lim_{m\to\infty}\frac{1}{m}\var\left(\sum_{k=1}^m Y_k\right)
=
\sigma_{\mathrm{long}}^2(N)\,\mathbf a^\top \mathbf K\mathbf a
=: \nu^2.
\end{equation}

\noindent\textbf{Step 6 (CLT for the independent triangular array).}
Define
\[
Z_d:=\frac{1}{\sqrt d}\sum_{j=1}^{k_d}\widetilde U_j.
\]
We verify Lindeberg--Feller.

\begin{itemize}
\item[(i)] \emph{Lindeberg condition.}
Since $|\widetilde U_j|\le p_d B$ almost surely, we have by \Cref{lem:explicit-blocks},
\[
\max_{1\le j\le k_d}\frac{|\widetilde U_j|}{\sqrt d}
\le \frac{p_d B}{\sqrt d} \asymp d^{-\frac{1}{6}} B \longrightarrow 0~.
\]
Hence, for any $\varepsilon>0$ and all sufficiently large $d$, $\mathbbm{1}\{|\widetilde U_j|/\sqrt d>\varepsilon\}\equiv 0$, so the Lindeberg condition holds.

\item[(ii)] \emph{Variance convergence.}
Let $s_d^2:=\var(Z_d)=\frac{1}{d}\sum_{j=1}^{k_d}\var(U_j)$.
By stationarity of $\{Y_k\}$, $\var(U_j) = \allowbreak \var(U_1)$ for all $j$, hence
\begin{equation}
\label{eq:sd2}
s_d^2=\frac{k_d}{d}\var(U_1).
\end{equation}
Moreover, $\var(U_1)=\var(\sum_{k=1}^{p_d}Y_k)$ and, by \eqref{eq:lrv-Y}, the limit
\[
\lim_{m\to\infty}\frac{1}{m}\var\left(\;\sum_{k=1}^{m}Y_k\right) = \sigma^2_{\mathrm{long}}(N)\;\mathbf a^\top \mathbf{K}\,\mathbf a
\]
exists. Denote this limit by $\nu^2$. Then, $\var(U_1)=p_d\nu^2+o(p_d)$.
Combining with \eqref{eq:sd2} and $k_d p_d/d\to 1$ (from \Cref{lem:explicit-blocks}) yields
\begin{equation*}
s_d^2 =\frac{k_d}{d}\var(U_1) = \frac{k_d p_d}{d} \nu^2 + o\left(\frac{k_d p_d}{d}\right) \longrightarrow 1 \cdot \nu^2 + o(1) = \nu^2.
\end{equation*}
\end{itemize}
Thus, by Lindeberg--Feller CLT, we have 
\[\frac{1}{\sqrt d}\sum_{j=1}^{k_d}\widetilde U_j ~=~ Z_d ~\stackrel{\mathscr{L}}{\longrightarrow}~ \mathcal{N}(0,\nu^2)~\equiv~\mathcal{N}(0, \sigma^2_{\mathrm{long}}(N)\;\mathbf a^\top \mathbf{K}\,\mathbf a)~.\]

\noindent\textbf{Step 7 (conclusion).}
Steps 2--3 show that
\[
\frac{1}{\sqrt d}\sum_{k=1}^d Y_k
=
\frac{1}{\sqrt d}\sum_{j=1}^{k_d}U_j
+ o_{\P}(1).
\]
Step 4 transfers the limit law from $\sum_{j=1}^{k_d} U_j$ to the independent sum $\sum_{j=1}^{k_d} \widetilde U_j$; Steps 5a and 5b compute the long-run variance of $\sum_{k=1}^m Y_k$ under a $m^{-1}$ scaling; and Step 6 shows that $\sum_{j=1}^{k_d} \widetilde U_j$ converges weakly to $\mathcal{N}(0,\sigma^2_{\mathrm{long}}(N)\;\mathbf a^\top \mathbf{K} \,\mathbf a)$. This proves
\begin{equation}
\mathbf{a}^\top \sqrt{d} \ \widetilde{\mathfrak{W}}_d
= S_d 
= \mathbf{a}^\top \left( \frac{1}{\sqrt{d}}\sum_{k=1}^d \boldsymbol{\xi}_k \right) 
= \frac{1}{\sqrt{d}}\sum_{k=1}^d Y_k
\ \stackrel{\mathscr{L}}{\longrightarrow}\ \mathcal{N}(0,\sigma^2_{\mathrm{long}}(N)\; \mathbf{a}^\top \mathbf{K}\mathbf{a})~,
\end{equation}
which is precisely what we claimed in \eqref{eq:target-scalar}.
Finally, the Cram\'er--Wold theorem yields the desired multivariate limit.
\hfill $\square$

\begin{remark}[On strict stationarity]
Strict stationarity in \Cref{ass:mixing} can be weakened to the corresponding second-order stationarity conditions needed to identify $\var(U_j)$ and the long-run variance; we keep the standard strict formulation for simplicity.
\end{remark}


\subsection[Proof of \Cref{lem:sigmalong-consistency}: Consistency of the plug-in estimator of $\sigma^2_{\mathrm{long}}(N)$]{Proof of \Cref{lem:sigmalong-consistency}}

Fix \(t\in\mcT\). Throughout the proof, we suppress the explicit dependence on \(t\) whenever no confusion may arise. For each \(d\ge1\), write
\[
\xi_{k,d}:=\xi_k(t),\qquad k\ge1,
\]
and let
\[
\bar\xi_d:=\frac1d\sum_{k=1}^d \xi_{k,d}.
\]
By the coordinate decomposition, \(\bar\xi_d=\mathfrak W_d(t)\). Under \(\mathbf{H}_{0,d}\), \(\E[\xi_{k,d}]=0\) for every \(k\).

\paragraph{Step 1: Basic properties of the coordinate process.}
Since \(\xi_{k,d}\) is a measurable function of the coordinate trajectory \(\mathbf Z^{(k)}\), and since \(\{\mathbf Z^{(k)}\}_{k\ge1}\) is strictly stationary and \(\alpha\)-mixing under \Cref{ass:mixing}, the scalar sequence \(\{\xi_{k,d}\}_{k\ge1}\) is also strictly stationary and \(\alpha\)-mixing, with mixing coefficients bounded by the same sequence \(\alpha(r)\). Moreover, because \(\rho_0\in\{0,1\}\) and \(N\) is fixed,
\[
|\xi_{k,d}|\le 2
\qquad\text{a.s. for all }k,d.
\]
Thus \(\xi_{k,d}\in L^p\) for every \(p\ge1\), uniformly in \(k\) and \(d\).

\paragraph{Step 2: Population long-run variance and absolute summability.}
For \(r\ge0\), define
\[
\gamma_{r,d}(t):=\cov(\xi_{1,d},\xi_{1+r,d}).
\]
By Davydov's covariance inequality, for every \(p>2\), there exists a constant
\(C_p<\infty\), independent of \(r\) and \(d\), such that
\[
|\gamma_{r,d}(t)|
\le
C_p\,\alpha(r)^{1-2/p},
\qquad r\ge1.
\]
Since, by \Cref{ass:mixing}(ii), \(\alpha(r)\le C_\alpha r^{-\lambda}\) with \(\lambda>4\), choose \(p>2\) large enough so that \(\eta:=\lambda(1-2/p)>1\). Then, for some constant \(C<\infty\),
\[
\sup_{d\ge1}|\gamma_{r,d}(t)|
\le C r^{-\eta},
\qquad r\ge1.
\]
Hence
\begin{equation}
\label{eq:uniform-abs-summability}
\sum_{r=1}^\infty \sup_{d\ge1}|\gamma_{r,d}(t)|<\infty.
\end{equation}
It follows that the population long-run variance
\begin{equation}
\label{eq:Gamma-d-t}
\Gamma_d(t)
:=
\gamma_{0,d}(t)+2\sum_{r=1}^\infty \gamma_{r,d}(t)
\end{equation}
is well-defined and finite for every \(d\).

\paragraph{Step 3: HAC consistency for the target \(\Gamma_d(t)\).}
Recall that
\[
\widehat\gamma_{r,d}(t)
:=
\frac1d\sum_{k=1}^{d-r}
\bigl(\xi_{k,d}-\mathfrak W_d(t)\bigr)
\bigl(\xi_{k+r,d}-\mathfrak W_d(t)\bigr),
\qquad 0\le r\le d-1.
\]
With Bartlett weights \(w_r=1-r/(L(d)+1)\), the HAC estimator is
\[
\widehat{\mathrm{LRV}}_d(N;t)
=
\widehat\gamma_{0,d}(t)
+
2\sum_{r=1}^{L(d)}
\left(1-\frac{r}{L(d)+1}\right)\widehat\gamma_{r,d}(t).
\]
Because \(\{\xi_{k,d}\}_{k\ge1}\) is uniformly bounded, strictly stationary, and \(\alpha\)-mixing with the uniform summability bound \eqref{eq:uniform-abs-summability}, and because \(L(d)\to\infty\) and \(L(d)=o(d^{1/2})\), standard Bartlett HAC consistency results (see, e.g., \citet{Andrews1991HAC}) for bounded strongly mixing triangular arrays imply
\begin{equation}
\label{eq:HAC-to-Gamma}
\widehat{\mathrm{LRV}}_d(N;t)-\Gamma_d(t)
\stackrel{\bbP}{\longrightarrow}0.
\end{equation}

\paragraph{Step 4: Exact variance identity for \(\sqrt d\,\mathfrak W_d(t)\).}
Since \(\mathfrak W_d(t)=d^{-1}\sum_{k=1}^d\xi_{k,d}\), the variance of the partial-sum average satisfies the exact identity
\begin{equation}
\label{eq:exact-var-identity}
\var\bigl(\sqrt d\,\mathfrak W_d(t)\bigr)
=
\gamma_{0,d}(t) + 2\sum_{r=1}^{d-1}\Bigl(1-\frac{r}{d}\Bigr)\gamma_{r,d}(t).
\end{equation}
Therefore, subtracting \eqref{eq:exact-var-identity} from \eqref{eq:Gamma-d-t} gives
\begin{align}
\label{eq:Gamma-minus-var}
\left|
\Gamma_d(t)-\var\bigl(\sqrt d\,\mathfrak W_d(t)\bigr)
\right|
&\le
2\sum_{r=1}^{d-1}\frac{r}{d}\,|\gamma_{r,d}(t)|
+
2\sum_{r=d}^{\infty}|\gamma_{r,d}(t)|.
\end{align}

We now show that the right-hand side of \eqref{eq:Gamma-minus-var} converges to zero. Fix \(\varepsilon>0\). By \eqref{eq:uniform-abs-summability}, there exists \(M\in\bbN\) such that
\[
\sum_{r=M+1}^\infty \sup_{d'\ge1}|\gamma_{r,d'}(t)|<\frac{\varepsilon}{8}.
\]
Then, for every $d>M$,
\begin{align*}
\kappa_1 = \sum_{r=1}^{d-1}\frac{r}{d}\,|\gamma_{r,d}(t)|
&=
\sum_{r=1}^{M}\frac{r}{d}\,|\gamma_{r,d}(t)|
+
\sum_{r=M+1}^{d-1}\frac{r}{d}\,|\gamma_{r,d}(t)| \\
&\le
\underbrace{\frac{M}{d}\sum_{r=1}^{M}\sup_{d'\ge1}|\gamma_{r,d'}(t)|}_{:= \kappa_{1,1}}
+
\underbrace{\sum_{r=M+1}^{\infty}\sup_{d'\ge1}|\gamma_{r,d'}(t)|}_{:= \kappa_{1,2}} = \kappa_{1,1} + \kappa_{1,2}.
\end{align*}
By the choice of $M$, we have $\kappa_{1,2} \leq \varepsilon/8$. Since $M$ is now fixed, the numerator of $\kappa_{1,1}$ is bounded; thus, there exists $d_1=d_1(\varepsilon,M)$ such that for all $d\ge d_1$,
\[
\kappa_{1,1} = \frac{M}{d}\sum_{r=1}^{M}\sup_{d'\ge1}|\gamma_{r,d'}(t)|
< \frac{\varepsilon}{8}.
\]
Hence, for all $d\ge \max\{M+1,d_1\}$,
\[
\kappa_1 = \sum_{r=1}^{d-1}\frac{r}{d}\,|\gamma_{r,d}(t)| \leq \kappa_{1,1} + \kappa_{1,2} < \frac{\varepsilon}{8} + \frac{\varepsilon}{8} = \frac{\varepsilon}{4}.
\]
Next, again by \eqref{eq:uniform-abs-summability}, there exists
$d_2=d_2(\varepsilon)$ such that for all $d\ge d_2$,
\[
\sum_{r=d}^{\infty}\sup_{d'\ge1} |\gamma_{r,d'}(t)| < \frac{\varepsilon}{4}.
\]
Therefore, for all $d\ge d_2$,
\[
\kappa_2 = \sum_{r=d}^{\infty} |\gamma_{r,d}(t)|
\le
\sum_{r=d}^{\infty}\sup_{d'\ge1} |\gamma_{r,d'}(t)|
< \frac{\varepsilon}{4}.
\]
Combining the last two displays with \eqref{eq:Gamma-minus-var}, we conclude that for all $d\ge \max\{M+1,d_1,d_2\}$,
\[
\Bigl|
\Gamma_d(t)-\var\bigl(\sqrt d\,\mathfrak W_d(t)\bigr)
\Bigr|
<
2\cdot\frac{\varepsilon}{4}
+
2\cdot\frac{\varepsilon}{4}
=
\varepsilon.
\]
Since our choice of $\varepsilon > 0$ was arbitrary, we conclude that
\begin{equation}
\label{eq:Gamma-var-gap}
\left|\,\Gamma_d(t)-\var\bigl(\sqrt d\,\mathfrak W_d(t)\bigr)\, \right|\longrightarrow 0.
\end{equation}

\paragraph{Step 5: Identification of the limit and pointwise consistency.}
By \Cref{prop:var-cov-null}, under \(\mathbf{H}_{0,d}\),
\[
\var\bigl(\mathfrak W_d(t)\bigr)
=
\mathbf K_{tt}\,\mathcal V_{d,N}.
\]
Thus
\[
\var\bigl(\sqrt d\,\mathfrak W_d(t)\bigr)
=
d\,\mathbf K_{tt}\,\mathcal V_{d,N}
\longrightarrow
\mathbf K_{tt}\,\sigma_{\mathrm{long}}^2(N),
\]
where the convergence follows from \Cref{prop:longrun-exists}. Combining this
with \eqref{eq:Gamma-var-gap}, we obtain
\[
\Gamma_d(t)\longrightarrow \mathbf K_{tt}\,\sigma_{\mathrm{long}}^2(N).
\]
Together with \eqref{eq:HAC-to-Gamma}, this gives
\[
\widehat{\mathrm{LRV}}_d(N;t)
\stackrel{\bbP}{\longrightarrow}
\mathbf K_{tt}\,\sigma_{\mathrm{long}}^2(N).
\]
Since \(\mathbf K_{tt}>0\) for every \(t\in\mcT\),
\[
\widehat\sigma_{\mathrm{long}}^2(N;t)
=
\frac{\widehat{\mathrm{LRV}}_d(N;t)}{\mathbf K_{tt}}
\stackrel{\bbP}{\longrightarrow}
\sigma_{\mathrm{long}}^2(N).
\]

\paragraph{Step 6: Consistency of the aggregated estimator.}
Because \(\mcT\) is fixed and finite, the pointwise convergence above implies
\begin{equation}
\label{eq:max-sigmasq-diff}
\max_{t\in\mcT}
\left|
\widehat\sigma_{\mathrm{long}}^2(N;t)
-
\sigma_{\mathrm{long}}^2(N)
\right|
\stackrel{\bbP}{\longrightarrow}0.
\end{equation}
Now let
\[
\widehat\sigma_{\mathrm{long}}^2(N)
=
\mathrm{median}\left\{
\widehat\sigma_{\mathrm{long}}^2(N;t):t\in\mcT
\right\}.
\]
Since the median of finitely many real numbers lies between their minimum and maximum,
\[
\left|
\widehat\sigma_{\mathrm{long}}^2(N)
-
\sigma_{\mathrm{long}}^2(N)
\right|
\le
\max_{t\in\mcT}
\left|
\widehat\sigma_{\mathrm{long}}^2(N;t)
-
\sigma_{\mathrm{long}}^2(N)
\right|.
\]
Therefore, \eqref{eq:max-sigmasq-diff} gives
\[
\widehat\sigma_{\mathrm{long}}^2(N)
\stackrel{\bbP}{\longrightarrow}
\sigma_{\mathrm{long}}^2(N).
\]
This completes the proof.
\qed


\subsection[Proof of \Cref{thm:level-alpha}: Asymptotic level-$\alpha$ validity of the plug-in calibrated scan]{Proof of \Cref{thm:level-alpha}}

By \Cref{thm:multivariate-clt}, under $\mathbf{H}_{0,d}$ we have
\[
\sqrt d\,\widetilde{\mathfrak W}_d
\ \stackrel{\mathscr{L}}{\longrightarrow} \
\mathcal N\!\bigl(\mathbf 0,\,\sigma^2_{\mathrm{long}}(N)\mathbf K\bigr).
\]
By \Cref{lem:sigmalong-consistency}, $\widehat{\sigma}_{\mathrm{long}}\xrightarrow{\mathbb{P}}\sigma_{\mathrm{long}}$.
By Slutsky's theorem,
\[
\frac{\sqrt d\,\widetilde{\mathfrak W}_d}{\widehat{\sigma}_{\mathrm{long}}}
\ \stackrel{\mathscr{L}}{\longrightarrow} \
\mathcal N(\mathbf 0,\mathbf K).
\]
The mapping $x\mapsto \max_{t\in\mcT} x_t$ is continuous on $\mathbb R^{|\mcT|}$,
so the continuous mapping theorem yields
\[
S_d
=
\max_{t\in\mcT}\frac{\sqrt d\,\mathfrak W_d(t)}{\widehat{\sigma}_{\mathrm{long}}}
\ \stackrel{\mathscr{L}}{\longrightarrow} \
\max_{t\in\mcT}\mcZ_t,
\qquad \mcZ\sim\mathcal N(\mathbf 0,\mathbf K).
\]
The random variable \(\max_{t\in\mcT}\mcZ_t\) has a continuous distribution because each marginal Gaussian coordinate has positive variance \(\mathbf K_{tt}>0\). Therefore \(c_\alpha\) is a continuity point of the limit distribution.
Since the limit distribution is continuous, convergence in distribution implies convergence of tail probabilities at the quantile point $c_\alpha$:
\[
\lim_{d\to\infty}\mathbb P_{\mathbf{H}_{0,d}}(S_d>c_\alpha)
=
\mathbb{P}\Bigl(\max_{t\in\mcT}\mcZ_t>c_\alpha\Bigr)
=
\alpha,
\]
by the definition of $c_\alpha$ in \Cref{subsubsec:adf-test}. This proves \eqref{eq:level-alpha-claim}.
\qed


\subsection[Proof of \Cref{thm:power-main}: Change-point localization and power under the alternative]{Proof of \Cref{thm:power-main}}


\begin{proof}
We prove the two claims separately.

\paragraph{Proof of (i): Localization bound.}

\noindent
Recall that \(\widehat\tau_d\in\argmax_{t\in\mcT}\mathfrak W_d(t)\). Hence
\[
\{\widehat\tau_d\neq\tau\}
\subseteq
\bigcup_{t\in\mcT\setminus\{\tau\}}
\{\mathfrak W_d(t)\ge \mathfrak W_d(\tau)\}.
\]
Therefore, by the union bound,
\begin{equation}
\label{eq:loc-union-updated}
\mathbb P_{\mathbf{H}_{1,d}}(\widehat\tau_d\neq\tau)
\le
\sum_{t\in\mcT\setminus\{\tau\}}
\mathbb P_{\mathbf{H}_{1,d}}\!\left(\mathfrak W_d(t)\ge\mathfrak W_d(\tau)\right).
\end{equation}
Fix \(t\in\mcT\setminus\{\tau\}\), and define $D_{d,t}:=\mathfrak W_d(\tau)-\mathfrak W_d(t)$. Then,
\[
\mathbb P_{\mathbf{H}_{1,d}}\!\left(\mathfrak W_d(t)\ge\mathfrak W_d(\tau)\right)
=
\mathbb P_{\mathbf{H}_{1,d}}(D_{d,t}\le0).
\]
By the mean factorization in \Cref{prop:E[D-star]} and \Cref{rem:geometric_factorization},
\[
\mu_d(s)=\Lambda_{\tau,N}(s) \; \delta_d,\qquad s\in\mcT.
\]
Since \(\Lambda_{\tau,N}(\tau)=1\) and \(\Lambda_{\tau,N}\) is uniquely maximized
at \(\tau\),
\[
\E_{\mathbf{H}_{1,d}}[D_{d,t}]
=
\mu_d(\tau)-\mu_d(t)
=
\{\Lambda_{\tau,N}(\tau)-\Lambda_{\tau,N}(t)\}\delta_d
\ge
\Omega_{\tau,N}\delta_d.
\]
Hence,
\begin{align}
\mathbb P_{\mathbf{H}_{1,d}}(D_{d,t}\le 0)
&= \mathbb P_{\mathbf{H}_{1,d}}\!\left(D_{d,t}-\E_{\mathbf{H}_{1,d}}[D_{d,t}] \leq -\E_{\mathbf{H}_{1,d}}[D_{d,t}]\right) \notag\\
&\leq \mathbb P_{\mathbf{H}_{1,d}}\!\left(|D_{d,t}-\E_{\mathbf{H}_{1,d}}[D_{d,t}]|\geq \Omega_{\tau,N}\,\delta_d \right) \notag\\
&\leq \frac{\var_{\mathbf{H}_{1,d}}(D_{d,t})}{\Omega_{\tau,N}^2\delta_d^2}.
\label{eq:cheb-loc-updated}
\end{align}
where we use Chebyshev's inequality in the last step.

\noindent
Using \(\var(A-B)\le2\var(A)+2\var(B)\) and
\Cref{ass:sigma-long-ALT}(ii),
\[
\var_{\mathbf{H}_{1,d}}(D_{d,t})
\le
2\var_{\mathbf{H}_{1,d}}\{\mathfrak W_d(\tau)\}
+
2\var_{\mathbf{H}_{1,d}}\{\mathfrak W_d(t)\}
\le
\frac{4C_{\mathrm{alt}}}{d}
\]
for all sufficiently large \(d\). Substituting this into \eqref{eq:cheb-loc-updated} gives
\[
\mathbb P_{\mathbf{H}_{1,d}}(D_{d,t}\le0)
\le
\frac{4C_{\mathrm{alt}}}{\Omega_{\tau,N}^2\,d\,\delta_d^2}.
\]
Returning to \eqref{eq:loc-union-updated} and using
\(|\mcT\setminus\{\tau\}|=N-4\), we obtain
\[
\mathbb P_{\mathbf{H}_{1,d}}(\widehat\tau_d\neq\tau)
\le
\frac{4C_{\mathrm{alt}}(N-4)}
{\Omega_{\tau,N}^2\,d\,\delta_d^2}.
\]
Thus, we finally obtain
\[
\mathbb P_{\mathbf{H}_{1,d}}(\widehat\tau_d=\tau)
\ge
1- \frac{N-4}{\Omega_{\tau,N}^2} \cdot \frac{4 C_{\mathrm{alt}}}{d\,\delta_d^2}.
\]
Using $N-4 \leq N$ proves (i).

\smallskip
\paragraph{Proof of (ii): Lower bound on power.}

\noindent
Recall that
\[
S_d
=
\max_{t\in\mcT}
\frac{\sqrt d\,\mathfrak W_d(t)}{\widehat\sigma_{\mathrm{long}}(N)},
\qquad
\widehat\phi_d(\alpha)=\mathbbm 1\{S_d>c_\alpha\}.
\]
Since \(S_d\ge \sqrt d\,\mathfrak W_d(\tau)/\widehat\sigma_{\mathrm{long}}(N)\),
\[
\mathbb P_{\mathbf{H}_{1,d}}\{\widehat\phi_d(\alpha)=1\}
\ge
\mathbb P_{\mathbf{H}_{1,d}}\!\left(
\frac{\sqrt d\,\mathfrak W_d(\tau)}
{\widehat\sigma_{\mathrm{long}}(N)}
>c_\alpha
\right).
\]
Fix \(\eta\in(0,1)\), and let
\[
A_{d,\eta}
:=
\left\{
\left|
\widehat\sigma_{\mathrm{long}}(N)-\sigma_*(N)
\right|
\le
\eta \, \sigma_*(N)
\right\}.
\]
On \(A_{d,\eta}\), we have
\(\widehat\sigma_{\mathrm{long}}(N)\le(1+\eta)\sigma_*(N)\). Therefore,
\[
\left\{
\sqrt d\,\mathfrak W_d(\tau)>
(1+\eta) \, \sigma_*(N) \, c_\alpha
\right\}
\cap A_{d,\eta}
\subseteq
\left\{
\frac{\sqrt d\,\mathfrak W_d(\tau)}
{\widehat\sigma_{\mathrm{long}}(N)}
>c_\alpha
\right\}.
\]
Consequently,
\begin{align}
\mathbb P_{\mathbf{H}_{1,d}}\{\widehat\phi_d(\alpha)=1\}
&\ge
\mathbb P_{\mathbf{H}_{1,d}}\!\left(
\sqrt d\,\mathfrak W_d(\tau)>
(1+\eta) \, \sigma_*(N) \, c_\alpha ~,~ A_{d,\eta}
\right) \notag\\
&\ge
\mathbb P_{\mathbf{H}_{1,d}}\!\left(
\sqrt d\,\mathfrak W_d(\tau)>
(1+\eta) \, \sigma_*(N) \, c_\alpha
\right) - \mathbb P_{\mathbf{H}_{1,d}}(A_{d,\eta}^c).
\label{eq:power-sigma-reduction-updated}
\end{align}
Since \(\mu_d(\tau)=\delta_d\), we have
\begin{align*}
&\mathbb P_{\mathbf{H}_{1,d}}\!\left(
\sqrt d\,\mathfrak W_d(\tau)>
(1+\eta) \, \sigma_*(N) \, c_\alpha
\right)\\
&=
1-
\mathbb P_{\mathbf{H}_{1,d}}\!\left(
\sqrt d\{\mathfrak W_d(\tau)-\mu_d(\tau)\}
\le
-\left[
\sqrt d\,\delta_d-(1+\eta) \, \sigma_*(N) \, c_\alpha
\right]
\right).
\end{align*}
If $\sqrt d\,\delta_d > (1+\eta) \, \sigma_*(N) \, c_\alpha$, then Chebyshev's inequality gives
\[
\mathbb P_{\mathbf{H}_{1,d}}\!\left(
\sqrt d\,\mathfrak W_d(\tau)>
(1+\eta) \, \sigma_*(N) \, c_\alpha
\right)
\ge
1-
\frac{
d\,\var_{\mathbf{H}_{1,d}}\{\mathfrak W_d(\tau)\}
}{
\left[
\sqrt d\,\delta_d-(1+\eta) \, \sigma_*(N) \, c_\alpha
\right]^2
}.
\]
By \Cref{ass:sigma-long-ALT}(ii),
\(d\,\var_{\mathbf{H}_{1,d}}\{\mathfrak W_d(\tau)\}\le C_{\mathrm{alt}}\) for all sufficiently large \(d\). Hence, combining with
\eqref{eq:power-sigma-reduction-updated} yields
\[
\mathbb P_{\mathbf{H}_{1,d}}\{\widehat\phi_d(\alpha)=1\}
\ge
1-\mathbb P_{\mathbf{H}_{1,d}}(A_{d,\eta}^c)
-
\frac{C_{\mathrm{alt}}}{
\left[
\sqrt d\,\delta_d-(1+\eta)\sigma_*(N)c_\alpha
\right]^2
}.
\]
This proves (ii).

Finally, if \(\sqrt d\,\delta_d\to\infty\), then
\(d\delta_d^2\to\infty\), so the localization error bound in (i) tends to zero.
Also, by \Cref{ass:sigma-long-ALT}(i), for every fixed \(\eta\in(0,1)\),
\[
\mathbb P_{\mathbf{H}_{1,d}}(A_{d,\eta}^c)\to0,
\]
and the denominator in the power bound diverges. Therefore,
\[
\mathbb P_{\mathbf{H}_{1,d}}(\widehat\tau_d=\tau)\to1,
\qquad
\mathbb P_{\mathbf{H}_{1,d}}\{\widehat\phi_d(\alpha)=1\}\to1.
\]
This completes the proof.
\end{proof}

\subsection[Proof of \Cref{thm:local-power}: Local asymptotic power curve]{Proof of \Cref{thm:local-power}}

Under the local alternatives \eqref{eq:local-signal-scaling}, \eqref{eq:local-mean-current-notation} implies
\[
\sqrt d\,\E_{\mathbf H_{1,d}}[\widetilde{\mathfrak W}_d]
=
h\,\bigl(\Lambda_{\tau,N}(t)\bigr)_{t\in\mcT}.
\]
By \Cref{ass:local-weak-limit},
\[
\sqrt d\left(\widetilde{\mathfrak W}_d-\E_{\mathbf H_{1,d}}[\widetilde{\mathfrak W}_d]\right)
\ \stackrel{\mathscr{L}}{\longrightarrow}\
\mathfrak{B}^{(h)}.
\]
Therefore, by Slutsky's theorem,
\[
\sqrt d\,\widetilde{\mathfrak W}_d
=
\sqrt d\left(\widetilde{\mathfrak W}_d-\E_{\mathbf H_{1,d}}[\widetilde{\mathfrak W}_d]\right)
+
\sqrt d\,\E_{\mathbf H_{1,d}}[\widetilde{\mathfrak W}_d]
\ \stackrel{\mathscr{L}}{\longrightarrow}\
\mathfrak{B}^{(h)} + h\,\bigl(\Lambda_{\tau,N}(t)\bigr)_{t\in\mcT}.
\]
Since \(\widehat\sigma_{\mathrm{long}}(N)\stackrel{\mathbb P}{\longrightarrow}\sigma_*(N)\)
by \Cref{ass:sigma-long-ALT}(i), another application of Slutsky's theorem yields
\[
\frac{\sqrt d\,\widetilde{\mathfrak W}_d}{\widehat\sigma_{\mathrm{long}}(N)}
\ \stackrel{\mathscr{L}}{\longrightarrow}\
\frac{\mathfrak{B}^{(h)} + h\,\bigl(\Lambda_{\tau,N}(t)\bigr)_{t\in\mcT}}{\sigma_*(N)}.
\]
Now the map $x\mapsto \max_{t\in\mcT}x_t$ is continuous on $\R^{|\mcT|}$, so by continuous mapping theorem,
\[
S_d
=
\max_{t\in\mcT}\frac{\sqrt d\,\mathfrak W_d(t)}{\widehat\sigma_{\mathrm{long}}(N)}
\ \stackrel{\mathscr{L}}{\longrightarrow}\
M_h.
\]
Finally, since $\widehat\phi_d(\alpha)=\mathbbm 1\{S_d>c_\alpha\}$ and $c_\alpha$ is assumed to be a continuity point of the distribution function of $M_h$, weak convergence implies
\[
\mathbb P_{\mathbf H_{1,d}}\!\left(\widehat\phi_d(\alpha)=1\right)
=
\mathbb P_{\mathbf H_{1,d}}(S_d>c_\alpha)
\longrightarrow
\mathbb{P}(M_h>c_\alpha),
\]
which proves the claim.
\qed

\subsection[Proof of \Cref{lem:online-arl-bounds}: Conditional ARL bounds under window overlap]{Proof of \Cref{lem:online-arl-bounds}}

Throughout this proof, probabilities and expectations are taken conditionally on \(\mcC_{N_0}\). We suppress this conditioning from the notation. Write \(q:=q_d(c\mid\mcC_{N_0})\).

If \(q=0\), then \(\widehat\nu_d(c)=\infty\) almost surely and the bound is understood with \(1/q=\infty\), so the claim is trivial. Hence assume \(q>0\). We proceed in a few steps.

\paragraph{Step 1 (Finite-range dependence).}
For each $s\ge N_0$, $\mcM_d(s)$ is a measurable function of the window $\mcW_s=\{Z_{s-N_0+1},\dots,Z_s\}$.
If $r\ge N_0$, then $\mcW_s$ and $\mcW_{s+r}$ are disjoint, hence independent under i.i.d.\ $\{Z_t\}$.
Therefore $\mcM_d(s)$ and $\mcM_d(s+r)$ are independent, and $\{\mcM_d(s)\}$ is $(N_0-1)$-dependent.

\paragraph{Step 2 (Upper bound).}
Define the thinned sequence $Y_j:=\mcM_d(jN_0)$, $j\ge1$, and the stopping time $$\widehat\nu^\star:=\inf\{j\ge1:Y_j>c\}.$$
By Step~1, $(Y_j)_{j\ge1}$ are i.i.d.\ and $\P(Y_1>c)=q$ by stationarity. Hence \(\widehat\nu^\star\) is geometric with success probability \(q\), so $\E\left[\widehat\nu^\star\right]=1/q$.
Moreover, $\widehat\nu_d(c)\le jN_0$ whenever $\widehat\nu^\star=j$, hence
$\widehat\nu_d(c) \leq N_0\,\widehat\nu^\star$ a.s.\ and thus we have
$$\E\left[\widehat\nu_d(c)\right]\leq N_0\,\E\left[\widehat\nu^\star\right]=\frac{N_0}{q}.$$

\paragraph{Step 3 (Lower bound).}
For any $n\ge0$, by the union bound and stationarity,
\[
\P\Big(\widehat\nu_d(c) \leq N_0+n\Big)
=\P\!\left(\bigcup_{j=0}^{n}\{\mcM_d(N_0+j)>c\}\right)
\le \sum_{j=0}^{n}\P(\mcM_d(N_0+j)>c)
=(n+1)q.
\]
Hence,
\begin{equation}
\label{eq:tail-lb}
\P\Big(\widehat\nu_d(c) > N_0+n\Big)\ \ge\ 1-(n+1)q.
\end{equation}
Let $L:=\big\lfloor (2q)^{-1}\big\rfloor$. Then $(n+1)q\le Lq\le \tfrac12$ for all $0\le n\le L-1$,
so \eqref{eq:tail-lb} gives $\P(\widehat\nu_d(c)> N_0+n)\ge \tfrac12$ for $0\le n\le L-1$.
Using $\E[T]=\sum_{m\ge0}\P(T>m)$ for integer-valued $T\ge0$,
\[
\E\left[\widehat\nu_d(c)\right]
\ge \sum_{n=0}^{L-1}\P(\widehat\nu_d(c)> N_0+n)
\ge \sum_{n=0}^{L-1}\frac12
=\frac{L}{2}.
\]
Finally, $L=\lfloor (2q)^{-1}\rfloor\ge (2q)^{-1}-1$, hence
\[
\E\left[\widehat\nu_d(c)\right]\ \ge\ \frac{1}{2}\left(\frac{1}{2q}-1\right)
=\frac{1}{4q}-\frac{1}{2}.
\]
Combining Steps 2 and 3 yields \eqref{eq:arl-sandwich-conditional}.
\qed


\subsection[Proof of \Cref{lem:online-cedd}: Conditional EDD and Pollak EDD bounds]{Proof of \Cref{lem:online-cedd}}

Throughout the proof, all probabilities and expectations are understood conditionally on the fixed calibration sample \(\mcC_{N_0}\), and we suppress this conditioning from the notation.

Let \(\widehat\nu_d:=\widehat\nu_d(c_{\alpha,N_0})\). Next, let us define $S_{d,\nu}:=\{\widehat\nu_d>\nu\}$, and $A_{d,\nu}:=\{\widehat\nu_d\ge \nu+N_0\}$.
Since $\mcJ_\nu=\{\nu+1,\dots,\nu+N_0-1\}$, we have the disjoint decomposition
\[
S_{d,\nu}
=
\{\widehat\nu_d\in\mcJ_\nu\}\,\sqcup\,A_{d,\nu},
\]
where $A \sqcup B$ denotes disjoint union, i.e., $A \sqcup B = A \cup B$ when $A \cap B = \varnothing$. Therefore,
\[
\bbP_{\mathbf{H}_{1,d}}(A_{d,\nu}\mid S_{d,\nu})
=
1-\pi_{d,\nu}.
\]
Hence,
\begin{align*}
\operatorname{CEDD}_{d,\nu}
&=
\E_{\mathbf{H}_{1,d}}\!\left[\widehat\nu_d-\nu \,\middle|\, S_{d,\nu},\,\widehat\nu_d\in\mcJ_\nu\right]
\pi_{d,\nu} 
+
\E_{\mathbf{H}_{1,d}}\!\left[\widehat\nu_d-\nu \,\middle|\, A_{d,\nu}\right]
(1-\pi_{d,\nu}).
\end{align*}
On the event \(\{\widehat\nu_d\in\mcJ_\nu\}\), we have
\[
1\le \widehat\nu_d-\nu\le N_0-1,
\]
and therefore
\[
\E_{\mathbf{H}_{1,d}}\!\left[\widehat\nu_d-\nu \,\middle|\, S_{d,\nu},\,\widehat\nu_d\in\mcJ_\nu\right]
\le N_0-1.
\]
It remains to bound the second term. Define
\[
s_{0,\nu}:=\nu+2N_0-1,
\qquad
\tau_{d,\nu}^{(G)}
:=
\inf\{s\ge s_{0,\nu}:\mcM_d(s)>c_{\alpha,N_0}\}.
\]

By construction, every window \(\mcW_s\) with \(s\ge s_{0,\nu}\) is fully post-change and is disjoint from every window with index at most \(\nu+N_0-1\). Consequently, under \(\mathbf{H}_{1,d}\), the event \(A_{d,\nu}\) depends only on \(\{Z_1,\dots,Z_{\nu+N_0-1}\}\), whereas \(\tau_{d,\nu}^{(G)}\) depends only on \(\{Z_{\nu+N_0},Z_{\nu+N_0+1},\dots\}\). Since the change-point \(\nu\) is deterministic and the observations are independent across time under \(\mathbf{H}_{1,d}\), it follows that \(A_{d,\nu}\) and \(\tau_{d,\nu}^{(G)}\) are conditionally independent given \(\mcC_{N_0}\), i.e., $A_{d,\nu} \indep \tau_{d,\nu}^{(G)}~|~\mcC_{N_0}$.

Moreover, on the event \(A_{d,\nu}\), either the procedure stops before
\(s_{0,\nu}\), in which case \(\widehat\nu_d<s_{0,\nu}\le \tau_{d,\nu}^{(G)}\),
or it does not, in which case \(\widehat\nu_d=\tau_{d,\nu}^{(G)}\). Thus,
$\widehat\nu_d\le \tau_{d,\nu}^{(G)}$ on $A_{d,\nu}$,
and hence
\[
\E_{\mathbf{H}_{1,d}}\!\left[\widehat\nu_d-\nu \,\middle|\, A_{d,\nu}\right]
\le
\E_{\mathbf{H}_{1,d}}\!\left[\tau_{d,\nu}^{(G)}-\nu\right].
\]
Now define the shifted post-change process
\[
\widetilde{\mcM}_d(u)
:=
\mcM_d(s_{0,\nu}-N_0+u),
\qquad u\ge N_0,
\]
and the associated stopping time
\[
T_d^{(G)}
:=
\inf\{u\ge N_0:\widetilde{\mcM}_d(u)>c_{\alpha,N_0}\}.
\]
Then we have $\tau_{d,\nu}^{(G)} = s_{0,\nu}-N_0+T_d^{(G)}$.  Since every window entering \(\widetilde{\mcM}_d(u)\) is fully post-change, the process \(\{\widetilde{\mcM}_d(u)\}_{u\ge N_0}\) is stationary and \((N_0-1)\)-dependent under the post-change law \(G_d\). Also,
\[
\bbP_{\mathbf{H}_{1,d}}\!\left(\widetilde{\mcM}_d(N_0)>c_{\alpha,N_0}\right)
=
q_d^{(G)}.
\]
Applying the same \(m\)-dependence argument used in \Cref{lem:online-arl-bounds} to this shifted process yields
\[
\E_{\mathbf{H}_{1,d}}[T_d^{(G)}]\le \frac{N_0}{q_d^{(G)}}.
\]
Therefore,
\begin{align*}
\E_{\mathbf{H}_{1,d}}\!\left[\tau_{d,\nu}^{(G)}-\nu\right]
=
(s_{0,\nu}-N_0-\nu)+\E_{\mathbf{H}_{1,d}}[T_d^{(G)}]
\le
(s_{0,\nu}-N_0-\nu)+\frac{N_0}{q_d^{(G)}}
=
N_0-1+\frac{N_0}{q_d^{(G)}}.
\end{align*}
Thus,
\[
\E_{\mathbf{H}_{1,d}}\!\left[\widehat\nu_d-\nu \,\middle|\, A_{d,\nu}\right]
\le
N_0-1+\frac{N_0}{q_d^{(G)}}.
\]
Combining the two parts, we obtain
\begin{align*}
\operatorname{CEDD}_{d,\nu}
&\le
(N_0-1)\pi_{d,\nu}
+
\left(N_0-1+\frac{N_0}{q_d^{(G)}}\right)(1-\pi_{d,\nu})
\\
&=
(N_0-1)
+
\bigl(1-\pi_{d,\nu}\bigr)\frac{N_0}{q_d^{(G)}},
\end{align*}
which proves \eqref{eq:cedd-bound}.
\qed

\medskip\noindent
Taking the supremum over \(\nu\in\bbN\) and using \(q_d^{(G)}>0\), we obtain
\[
\operatorname{CEDD}_{d}^{\mathbf{(P)}}
=
\sup_{\nu\in\bbN} \; \operatorname{CEDD}_{d,\nu}
\le
(N_0-1)
+
\left(1-\inf_{\nu\in\bbN}\pi_{d,\nu}\right)
\frac{N_0}{q_d^{(G)}},
\]
which proves \eqref{eq:pcedd-bound-pi}.
\qed

\subsection[Proof of \Cref{cor:online-pollak-cedd}: Pollak EDD bounds under a post-change HDLSS CLT]{Proof of \Cref{cor:online-pollak-cedd}}

Throughout the proof, all probabilities and expectations are understood conditionally on the fixed calibration sample \(\mcC_{N_0}\), and we suppress this conditioning from the notation.

\paragraph{Proof of $q_d^{(G)} \stackrel{\bbP}{\to} q_\infty^{(G)}$.}

By assumption, for any fixed fully post-change window \(s\), the
coordinate-trajectory sequence generated under \(G_d\) satisfies the analogue
of \Cref{ass:mixing,ass:nondeg-longrun} with \(N\) replaced by \(N_0\). Hence,
by the same argument as in \Cref{thm:multivariate-clt},
\[
\sqrt d\,(\mathfrak W_{d,s}(k))_{k\in\mcT_0}
\ \stackrel{\mathscr L}{\longrightarrow}\
\mcN\!\left(\mathbf 0,\sigma_*^2(N_0)\mathbf K(N_0)\right).
\]
Since the calibration sample is independent of the monitoring stream and
\[
\widehat\sigma_{\mathrm{long}}(N_0;\mcC_{N_0})
\stackrel{\bbP}{\longrightarrow}
\sigma_{\mathrm{long}}(N_0),
\]
Slutsky's theorem gives
\[
\mcM_d(s)
=
\max_{k\in\mcT_0}
\frac{\sqrt d\,\mathfrak W_{d,s}(k)}
{\widehat\sigma_{\mathrm{long}}(N_0;\mcC_{N_0})}
\ \stackrel{\mathscr L}{\longrightarrow}\
\frac{\sigma_*(N_0)}{\sigma_{\mathrm{long}}(N_0)}
\max_{k\in\mcT_0}\mcZ_k,
\]
where \(\mcZ=(\mcZ_k)_{k\in\mcT_0}\sim
\mcN(\mathbf 0,\mathbf K(N_0))\).

By stationarity of the fully post-change regime,
\[
q_d^{(G)} = \P_{d,\nu}\!\left(\mcM_d(s)>c_{\alpha,N_0}\right)
\]
for any fully post-change window \(s\). Since \(\max_{k\in\mcT_0}\mcZ_k\) has a continuous distribution, the threshold \(c_{\alpha,N_0}\) is a continuity point of the limiting law. Therefore,
\[
q_d^{(G)}
\stackrel{\bbP}{\longrightarrow}
\P\!\left(
\frac{\sigma_*(N_0)}{\sigma_{\mathrm{long}}(N_0)}
\max_{k\in\mcT_0}\mcZ_k
>
c_{\alpha,N_0}
\right)
=
q_\infty^{(G)}.
\]
Because \(\sigma_*(N_0)\in(0,\infty)\), \(\sigma_{\mathrm{long}}(N_0)\in(0,\infty)\), and \(\max_{k\in\mcT_0}\mcZ_k\) has a non-degenerate continuous Gaussian law, it follows that $q_\infty^{(G)}\in(0,1)$.

\paragraph{Proof of the Pollak EDD lower bound.}
The lower bound is deterministic. Since the stopping rule only begins once a full window is available, \(\widehat\nu_d\ge N_0\). Taking \(\nu=1\) in the Pollak supremum gives
\[
\operatorname{CEDD}_{d}^{\mathbf{(P)}}
\ge
\operatorname{CEDD}_{d,1}
=
\E_{d,1}\!\left[
\widehat\nu_d-1
\,\middle|\,
\widehat\nu_d>1
\right]
\ge
N_0-1.
\]

\paragraph{Proof of the Pollak EDD upper bound.}
For the upper bound, \eqref{eq:pcedd-bound-pi} and
\(0\le\pi_{d,\nu}\le1\) imply
\[
\operatorname{CEDD}_{d}^{\mathbf{(P)}}
\le
N_0-1+\frac{N_0}{q_d^{(G)}}.
\]
Since \(q_d^{(G)} \stackrel{\bbP}{\longrightarrow} q_\infty^{(G)}\in(0,1)\), the continuous mapping theorem gives
\[
\frac{1}{q_d^{(G)}}
\stackrel{\bbP}{\longrightarrow}
\frac{1}{q_\infty^{(G)}}.
\]
Therefore,
\[
\operatorname{CEDD}_{d}^{\mathbf{(P)}}
\le
N_0\left(1+\frac{1}{q_\infty^{(G)}}\right)-1
+
o_{\bbP}(1),
\]
which proves \eqref{eq:pcedd-sandwich}.
\qed

\section{Operational details for the numerical experiments}
\label{app-sec:operational-details-exp}

All offline simulations were conducted on a MacBook Pro equipped with an Apple M3 processor, \(8\) cores, and \(16\) GB of memory. The online simulations were conducted on a Windows workstation equipped with AMD EPYC-Genoa processors (2.25 GHz), \(64\) vCPUs, and \(256\) GB of memory. All experiments were implemented in the \texttt{R} programming language \citep{R-lang}, version \(4.6.0\).


We summarize the implementation details for the competing methods used in the offline simulation study. All methods were applied to the same simulated sequences with \(N=40\), true change-point \(\tau=15\), candidate split set \(\mcT := \{2,\ldots,N-2\}\), and \(1000\) Monte Carlo replications for each distributional setting and dimension. For methods that returned no valid interior change-point, or returned a boundary point, we recorded the estimate as \(N+1\), so that it falls in the largest localization error category.

\begin{itemize}
    \item \textit{DAK Scan :} We computed the proposed scan statistic over all candidate split points \(t\in\mcT\) using the bit-level implementation of the angular-kernel scan. The estimated change-point was the maximizer of the scan statistic. The computationally intensive components of the proposed DAK scan were implemented in \texttt{C++} and called from \texttt{R} through \texttt{Rcpp} \citep[version \(1.1.1.1.1\)]{Rcpp-package} and \texttt{RcppArmadillo} \citep[version \(15.2.6.1\)]{RcppArmadillo-package}, substantially reducing the computation time.

    \item \textit{E-Divisive \citep{MattesonJames2014}} : We used the function \texttt{e.divisive()} from the \texttt{ecp} package\footnote{\texttt{ecp} package: \url{https://cran.r-project.org/web/packages/ecp/index.html}} \citep{ecp-package} in \texttt{R}, with significance level \(0.05\), number of permutations \(\texttt{R}=200\), and maximum number of changes \(\texttt{k}=1\). When the output contained exactly one interior estimated change-point, we used that estimate; otherwise the method was treated as returning no valid change-point.

    \item \textit{E-CP3O \citep{MattesonJames2014}} : We used the function \texttt{e.cp3o()} from the \texttt{ecp} package with \(\texttt{K}=1\). If the returned estimate was an interior point, it was used as the estimated change-point; boundary estimates were treated as invalid.

    \item \textit{KCPA \citep{ArlotEtAL2019KCPA}} : We used the function \texttt{kcpa()} from the \texttt{ecp} package, which is an implementation of kernel change-point analysis with \(\texttt{L}=1\) and \(\texttt{C}=1\). If the method returned no estimate, or returned a boundary point, it was treated as returning no valid interior change-point.

    
    \item \textit{MMD-\(\mathcal N\) :} We implemented a scan based on the empirical Gaussian kernel MMD statistic (see, e.g., \citet{GrettonEtAl2012Kernel}). For a candidate split \(t\), the two samples are \(\{Z_1,\ldots,Z_t\}\) and \(\{Z_{t+1},\ldots,Z_N\}\), and the statistic is computed using the V-statistic form
    \[
    \mathrm{MMD}^{(\mcN)}_{\sigma}(t)
    =
    \frac{1}{t^2}\sum_{i,i'\le t} k_\sigma(Z_i,Z_{i'})
    +
    \frac{1}{(N-t)^2}\sum_{j,j'>t} k_\sigma(Z_j,Z_{j'})
    -
    \frac{2}{t(N-t)}\sum_{i\le t<j} k_\sigma(Z_i,Z_j),
    \]
    where \(k_\sigma(x,y)=\exp\{-\|x-y\|_2^2/(2\sigma^2)\}\). The bandwidth \(\sigma\) was chosen once for each simulated sequence using the median heuristic, \(\sigma=\{ \operatorname{median}_{i<j}\|Z_i-Z_j\|_2^2/2\}^{1/2}\). The estimated change-point was the maximizer over \(t\in\{2,\ldots,N-2\}\).


    \item \textit{MMD-\(\mathcal E\) :} We implemented an analogous MMD scan using the energy kernel based on the Euclidean distance (see, e.g., \citet{sejdinovic2013energy}). For a candidate split \(t\), the statistic was
    \[
    \mathrm{MMD}^{(\mcE)}_{\ell_2}(t)
    =
    \frac{2}{t(N-t)}\sum_{i\le t<j}\|Z_i-Z_j\|_2
    -
    \frac{1}{t^2}\sum_{i,i'\le t}\|Z_i-Z_{i'}\|_2
    -
    \frac{1}{(N-t)^2}\sum_{j,j'>t}\|Z_j-Z_{j'}\|_2.
    \]
    The estimated change-point was the maximizer of this statistic over \(t\in\mcT\).

    \item \textit{HDD-DM \citep{Drikvandi2025changepoint}} : HDD-DM is a recent, distribution-free, state-of-the-art method designed specifically for change-point detection in the HDLSS regime. We use the function \texttt{test\_single\_changepoint()} from the \texttt{HDDchangepoint} package \footnote{\texttt{HDDchangepoint} package: \url{https://github.com/rezadrikvandi/HDDchangepoint}} in \texttt{R}, with \(\texttt{npermut}=100\). The returned change-point index was shifted by one to match our convention for split locations. If no valid interior estimate was returned, the method was treated as returning no valid change-point.

    \item \textit{Sliced-Wass :} We implemented a Sliced Wasserstein scan over \(t\in\mcT\) using \(100\) random projections and Wasserstein order \(p=1\). The estimated change-point was the maximizer of the resulting scan statistic.
    \vspace{0.5cm}
\end{itemize}

\section{Additional simulations under finite-moment alternatives}
\label{app-sec:additional-simulations}

\Cref{tab:cp_accuracy_appendix} reports additional localization accuracy results for the offline simulation study. The setup is the same as in Section~\ref{subsec:simulation}: \(N=40\), true change-point \(\tau=15\), dimensions \(d\in\{200,1000,5000\}\), and \(1000\) independent replications. The table includes additional light-tailed, contaminated, sparse, and covariance/dependence-dominated examples that were omitted from the main text for space. As before, performance is summarized through the empirical distribution of \(|\widehat\tau_d-\tau|\), together with its empirical mean.

\input{tables/cp-loc-supp-new}

The additional examples reported in Table~\ref{tab:cp_accuracy_appendix} are as follows. As in the main text, we fix \(N=40\) and \(\tau=15\), and take \(Z_1,\ldots,Z_\tau\sim F_d\) and \(Z_{\tau+1},\ldots,Z_N\sim G_d\), for \(d\in\{200,1000,5000\}\).

\begin{enumerate}
    \item \textit{Gaussian location change:} \(F_d=\mathcal N(0,1)^{\otimes d}\) and \(G_d=\mathcal N(1,1)^{\otimes d}\).

    \item \textit{Gaussian scale change:} \(F_d=\mathcal N(0,1)^{\otimes d}\) and \(G_d=\mathcal N(0,4)^{\otimes d}\).

    \item \textit{Gaussian spiked covariance change:} \(F_d=\mathcal N(0,I_d)\) and \(G_d=\mathcal N(0,I_d+bvv^\top)\), where \(b=5\) and \(v\in\mathbb R^d\) is a unit vector.

    \item \textit{Gaussians with same marginals:} \(F_d=\mathcal N(0,I_d)\), while \(G_d=\mathcal N(0,\Sigma_d)\) with \((\Sigma_d)_{kk}=1\) for all \(k\) and \((\Sigma_d)_{k\ell}=\rho\) for \(k\neq \ell\), where \(\rho=0.3\). Thus, \(F_d \neq G_d\), but $F_d^{(k)} \equiv G_d^{(k)}$ for all $k \in [d]$.

    \item \textit{Laplace location change:} \(F_d=\mathrm{Laplace}(0,1)^{\otimes d}\) and \(G_d=\mathrm{Laplace}(1,1)^{\otimes d}\).

    \item \textit{Bernoulli--Gaussian change:} \(F_d\) is the product distribution of \(B_{0.1}X\), and \(G_d\) is the product distribution of \(B_{0.9}X\), where \(B_p\sim\mathrm{Bernoulli}(p)\), \(X\sim\mathcal N(0,1)\), and \(B_p\) and \(X\) are independent.

    \item \textit{Gaussian mixture:} \(F_d=\mathcal N(0,1)^{\otimes d}\), while \(G_d=\{(1-\epsilon)\mathcal N(0,1)+\epsilon\mathcal N(10,1)\}^{\otimes d}\), with \(\epsilon=0.1\).
%
\end{enumerate}

The additional results are broadly consistent with the main simulation findings. In standard marginal-shift settings such as Gaussian location, Gaussian scale, Laplace location, Bernoulli--Gaussian, and Gaussian mixture examples, DAK localizes accurately, with performance improving as \(d\) increases in the harder scale and mixture settings. HDD-DM is also highly competitive in several of these examples, especially in the Gaussian and Bernoulli--Gaussian settings.

The covariance- and dependence-dominated examples clarify the scope of the proposed method. In the Gaussian same-marginals example, the one-dimensional marginals are unchanged and the difference is only in the dependence structure; DAK is therefore not expected to localize the change. Similarly, the Gaussian spiked covariance example is primarily a covariance change and is not favorable to a marginal scan. The sparse variance example also illustrates a limitation of dimension averaging: only a fixed number of coordinates carry the variance change, so the fraction of affected coordinates decreases with \(d\). These cases reinforce that the proposed procedure is designed for changes reflected in aggregated coordinate-wise marginals, rather than arbitrary joint-distribution changes.

\begingroup
\setstretch{1.05}
\putbib[references]
\endgroup

\end{bibunit}

%% file: components/computational-lemmas.tex

\begin{lemma}
\label{lem:N0_all_pairs}
Fix integers $N \in \bbN$ and $1 \leq t \leq t' \leq N$.
For each ordered configuration $(P,Q) \in \mathscr{P}$, let $N_0^{(P,Q)}(t,t',N)$ denote the number of ordered couples $\bigl((i,j),(i',j')\bigr)$ in configuration $(P,Q)$ for which the two ordered pairs exhibit two complete overlaps, i.e.\ $(i',j') = (i,j)$ or $(i',j') = (j,i)$, whenever these equalities are compatible with the index-range constraints defining $(P,Q)$.  
Then:
\begin{itemize}

\item[(i)] \textbf{Configuration $(A_t,A_{t'})$.}  
$i\in\{1,\dots,t\}$, $j\in\{t+1,\dots,N\}$, $i'\in\{1,\dots,t'\}$, $j'\in\{t'+1,\dots,N\}$.
\[
N_0^{(A_t,A_{t'})}(t,t',N)= t\,(N-t').
\]

\item[(ii)] \textbf{Configuration $(B_t,B_{t'})$.}  
$i,j\in\{1,\dots,t\}$ and $i',j'\in\{1,\dots,t'\}$.
\[
N_0^{(B_t,B_{t'})}(t,t',N)= 2\,t(t-1).
\]

\item[(iii)] \textbf{Configuration $(C_t,C_{t'})$.}  
$i,j\in\{t+1,\dots,N\}$ and $i',j'\in\{t'+1,\dots,N\}$.
\[
N_0^{(C_t,C_{t'})}(t,t',N)= 2\,(N-t')(N-t'-1).
\]

\item[(iv)] \textbf{Configuration $(A_t,B_{t'})$.}  
$i\in\{1,\dots,t\}$, $j\in\{t+1,\dots,N\}$, and $i',j'\in\{1,\dots,t'\}$.
\[
N_0^{(A_t,B_{t'})}(t,t',N)= 2\,t\,(t'-t).
\]

\item[(v)] \textbf{Configuration $(B_t,C_{t'})$.}  
$i,j\in\{1,\dots,t\}$ and $i',j'\in\{t'+1,\dots,N\}$.
\[
N_0^{(B_t,C_{t'})}(t,t',N)=0.
\]

\item[(vi)] \textbf{Configuration $(C_t,A_{t'})$.}  
$i,j\in\{t+1,\dots,N\}$, $i'\in\{1,\dots,t'\}$, and $j'\in\{t'+1,\dots,N\}$.
\[
N_0^{(C_t,A_{t'})}(t,t',N)= 2\,(N-t')(t'-t).
\]

\item[(vii)] \textbf{Configuration $(A_{t'},B_t)$.}  
$i,j\in\{1,\dots,t\}$, $i'\in\{1,\dots,t'\}$ and $j'\in\{t'+1,\dots,N\}$.
\[
N_0^{(A_{t'},B_t)}(t,t',N)=0.
\]

\item[(viii)] \textbf{Configuration $(B_{t'},C_t)$.}  
$i,j\in\{t+1,\dots,N\}$ and $i',j'\in\{1,\dots,t'\}$.
\[
N_0^{(B_{t'},C_t)}(t,t',N)=2\,(t'-t)(t'-t-1).
\]

\item[(ix)] \textbf{Configuration $(C_{t'},A_t)$.}  
$i\in\{1,\dots,t\}$, $j\in\{t+1,\dots,N\}$, and $i',j'\in\{t'+1,\dots,N\}$.
\[
N_0^{(C_{t'},A_t)}(t,t',N)=0.
\]
\end{itemize}
\end{lemma}

\begin{proof}
The quantity $N_0^{(P,Q)}(t,t',N)$ counts the number of ordered couples $((i,j), (i',j'))$ such that $(i,j) \in P$, $(i',j') \in Q$, and $\{i,j\} = \{i',j'\}$ (meaning $(i',j')=(i,j)$ or $(i',j')=(j,i)$). Note that for configurations of type $B$ and $C$, the diagonal cases $i=j$ are excluded by definition. We assume $1 \leq t \leq t' \leq N$.

\begin{itemize}
\item[(i)] \textbf{Configuration $(A_t,A_{t'})$.}
Here $(i,j) \in \{1,\dots,t\} \times \{t+1,\dots,N\}$ and $(i',j') \in \{1,\dots,t'\} \times \{t'+1,\dots,N\}$.
For an overlap, we need $\{i,j\} = \{i',j'\}$.
Since $t \leq t'$, the first components imply $i, i' \leq t'$ and the second components imply $j, j' > t$.
Specifically, because $i \le t$ and $j > t$, and $i' \le t'$ and $j' > t'$, an overlap requires $(i',j') = (i,j)$. The reverse $(i',j') = (j,i)$ is impossible because $j' > t'$ but $i \le t$.
Thus, we must have $(i,j) = (i',j')$.
This requires $i = i' \in \{1,\dots,t\} \cap \{1,\dots,t'\} = \{1,\dots,t\}$ ($t$ choices) and $j = j' \in \{t+1,\dots,N\} \cap \{t'+1,\dots,N\} = \{t'+1,\dots,N\}$ ($N-t'$ choices).

\medskip\noindent 
Total count: $t(N-t')$.

\item[(ii)] \textbf{Configuration $(B_t,B_{t'})$.}
Here $i,j \in \{1,\dots,t\}$ with $i \neq j$.
Since $t \leq t'$, any pair in $B_t$ is also a valid pair of indices in $\{1,\dots,t'\}$.
For any $(i,j) \in B_t$, both permutations $(i,j)$ and $(j,i)$ are valid in $B_{t'}$ because $i,j \le t \le t'$ and $i \neq j$.
The number of valid $(i,j)$ in $B_t$ is $t(t-1)$. Lastly, note that each generates 2 overlaps (itself and its reversal).

\medskip\noindent 
Total count: $2t(t-1)$.

\item[(iii)] \textbf{Configuration $(C_t,C_{t'})$.}
Here $i,j \in \{t+1,\dots,N\}$ with $i \neq j$.
For overlap with $C_{t'}$ (indices in $\{t'+1,\dots,N\}$), we need $i,j \in \{t+1,\dots,N\} \cap \{t'+1,\dots,N\} = \{t'+1,\dots,N\}$.
The number of such ordered pairs $(i,j)$ is $(N-t')(N-t'-1)$.
Both permutations $(i,j)$ and $(j,i)$ exist in $C_{t'}$.

\medskip\noindent 
Total count: $2(N-t')(N-t'-1)$.

\item[(iv)] \textbf{Configuration $(A_t,B_{t'})$.}
Here $(i,j) \in A_t \implies i \in \{1,\dots,t\}, j \in \{t+1,\dots,N\}$.
We require $\{i,j\} \subseteq \{1,\dots,t'\}$ (domain of $B_{t'}$).
$i \le t \le t'$ is automatically satisfied. We need $j \in \{t+1,\dots,N\} \cap \{1,\dots,t'\} = \{t+1,\dots,t'\}$.
Number of choices for $i$: $N$.
Number of choices for $j$: $t' - (t+1) + 1 = t' - t$.
Total pairs $(i,j)$: $t(t'-t)$.
For each pair, since $i \le t$ and $j \ge t+1$, $i \neq j$, so both $(i,j)$ and $(j,i)$ are in $B_{t'}$.

\medskip\noindent 
Total count: $2t(t'-t)$.

\item[(v)] \textbf{Configuration $(B_t,C_{t'})$.}
$(i,j) \in B_t \implies i,j \le t$.
$(i',j') \in C_{t'} \implies i',j' \ge t'+1$.
Since $t < t'+1$, the index sets are disjoint. No overlap is possible.

\medskip\noindent 
Total count: 0.

\item[(vi)] \textbf{Configuration $(C_t,A_{t'})$.}
$(i,j) \in C_t \implies i,j \in \{t+1,\dots,N\}, i \neq j$.
$(i',j') \in A_{t'} \implies i' \in \{1,\dots,t'\}, j' \in \{t'+1,\dots,N\}$.
Overlap requires $\{i,j\} = \{i',j'\}$. Thus, one index must be in $\{1,\dots,t'\}$ and the other in $\{t'+1,\dots,N\}$.
Intersecting with the domain of $C_t$:
One index is in $\{t+1,\dots,N\} \cap \{1,\dots,t'\} = \{t+1,\dots,t'\}$ (let this set be $S_1$, size $t'-t$).
The other is in $\{t+1,\dots,N\} \cap \{t'+1,\dots,N\} = \{t'+1,\dots,N\}$ (let this set be $S_2$, size $N-t'$).
\begin{itemize}
    \item Case 1: $i \in S_1, j \in S_2$. Count: $(t'-t)(N-t')$. Matches in $A_{t'}$ require $i' \le t', j' > t'$, so $(i',j') = (i,j)$. (1 match).
    \item Case 2: $i \in S_2, j \in S_1$. Count: $(N-t')(t'-t)$. Matches in $A_{t'}$ require $i' \le t', j' > t'$, so $(i',j') = (j,i)$. (1 match).
\end{itemize}

\medskip\noindent 
Total count: $2(N-t')(t'-t)$.

\item[(vii)] \textbf{Configuration $(A_{t'},B_t)$.}
$(i,j) \in A_{t'} \implies j \ge t'+1$.
$(i',j') \in B_t \implies i',j' \le t$.
Since $t < t'+1$, $j$ cannot be in $\{i',j'\}$. No overlap.

\medskip\noindent 
Total count: 0.

\item[(viii)] \textbf{Configuration $(B_{t'},C_t)$.}
$(i,j) \in B_{t'} \implies i,j \in \{1,\dots,t'\}, i \neq j$.
$(i',j') \in C_t \implies i',j' \in \{t+1,\dots,N\}, i' \neq j'$.
Overlap requires $\{i,j\} \subseteq \{t+1,\dots,N\}$.
Combined with $(i,j) \in B_{t'}$, we need $i,j \in \{t+1,\dots,t'\}$.
Number of ordered pairs $(i,j)$ in this range: $(t'-t)(t'-t-1)$.
For each such pair, both permutations correspond to valid pairs in $C_t$.

\medskip\noindent 
Total count: $2(t'-t)(t'-t-1)$.

\item[(ix)] \textbf{Configuration $(C_{t'},A_t)$.}
$(i,j) \in C_{t'} \implies i,j \ge t'+1$.
$(i',j') \in A_t \implies i' \le t$.
Since $t < t'+1$, $i'$ cannot be in $\{i,j\}$. No overlap.

\medskip\noindent 
Total count: 0.
\end{itemize}
\end{proof}


\begin{lemma}
\label{lem:N2_At_Atprime}
Fix integers $N \geq 1$ and $1 \leq t \leq t' \leq N$. Consider ordered couples
$\bigl( (i,j),(i',j') \bigr)$
of integer pairs satisfying
\[
i \in \{1,\dots,t\}, \quad j \in \{t+1,\dots,N\}, \quad
i' \in \{1,\dots,t'\}, \quad j' \in \{t'+1,\dots,N\},
\]
together with the constraints that all four indices $i,j,i',j'$ are pairwise distinct (in particular, $i \neq j$ and $j \neq j'$). Then the number of such ordered couples, denoted $N_2^{(A_t,A_{t'})}(t,t',N)$, is given by
\begin{equation}
\label{eq:N2_At_Atprime}
N_2^{(A_t,A_{t'})}(t,t',N)
=
t\,(N - t')
\Bigl\{ (t' - t)(t' - 2) + (t' - 1)(N - t' - 1) \Bigr\}.
\end{equation}
\end{lemma}

\begin{proof}
We first make the constraints explicit. By construction, $i$ and $j$ always lie in disjoint sets, namely $i \in \{1,\dots,t\}$ and $j \in \{t+1,\dots,N\}$, so the condition $i \neq j$ is automatically satisfied. Similarly, $i'$ and $j'$ lie in disjoint sets, so $i' \neq j'$ is automatic. Because $j,j' \in \{t+1,\dots,N\}$ with $j' \geq t'+1>t$, the constraint $j \neq j'$ is nontrivial and must be enforced. The requirement that all four indices $i,j,i',j'$ be distinct reduces to the conjunction of
\[
i' \neq i, \qquad i' \neq j, \qquad j' \neq j.
\]
The relations $i \neq j$ and $i' \neq j'$ are already implied by the ranges of the indices, and $i$ and $j'$ can never coincide since $i \leq t < t'+1 \leq j'$.

To count the admissible couples, it is convenient to proceed by conditioning on the value of $j$. There are $t$ possible choices for $i$, and once $i$ is fixed we distinguish two cases depending on whether $j$ lies in the overlap of the ranges $\{t+1,\dots,N\}$ and $\{1,\dots,t'\}$.

Consider first the case $j \in \{t+1,\dots,t'\}$. There are exactly $t' - t$ such values of $j$. In this situation $j$ belongs to $\{1,\dots,t'\}$, so the index $i'$ must avoid both $i$ and $j$. Thus $i'$ can be any element of $\{1,\dots,t'\} \setminus \{i,j\}$, and therefore there are precisely $t'-2$ admissible choices for $i'$. On the other hand, $j'$ ranges over $\{t'+1,\dots,N\}$ and, since $j \leq t'$, the inequality $j' \neq j$ is automatically satisfied; hence there are $N - t'$ admissible values for $j'$. For each fixed $i$, the number of ordered pairs $(j,i',j')$ in this case is
\[
(t' - t)\,(t' - 2)\,(N - t').
\]

Next consider the case $j \in \{t'+1,\dots,N\}$. There are $N - t'$ such values of $j$. In this case $j \notin \{1,\dots,t'\}$, so $i'$ only needs to avoid $i$. Consequently $i'$ can be chosen from $\{1,\dots,t'\} \setminus \{i\}$, giving $t'-1$ admissible choices. However, now $j$ lies in the same range as $j'$, namely $\{t'+1,\dots,N\}$, so we must enforce $j' \neq j$. Thus $j'$ may be any element of $\{t'+1,\dots,N\} \setminus \{j\}$, yielding $N - t' - 1$ choices for $j'$. For each fixed $i$, the number of ordered pairs $(j,i',j')$ in this case is
\[
(N - t')\,(t' - 1)\,(N - t' - 1).
\]

Since $i$ can be chosen arbitrarily from $\{1,\dots,t\}$, there are $t$ possible values of $i$, and the two cases for $j$ are disjoint and exhaustive. Summing the contributions from both regimes, we obtain
\[
N_2^{(A_t,A_{t'})}(t,t',N)
=
t \bigl\{ (t' - t)(t' - 2)(N - t') + (N - t')(t' - 1)(N - t' - 1) \bigr\},
\]
which coincides with \eqref{eq:N2_At_Atprime}. This completes the proof.
\end{proof}


\begin{lemma}
\label{lem:N2_Bt_Btprime}
Fix integers $N \geq 1$ and $1 \leq t < t' \leq N$. Consider ordered couples
$\bigl((i,j),(i',j')\bigr)$
of integer pairs satisfying
\[
i,j \in \{1,\dots,t\}, \quad i' , j' \in \{1,\dots,t'\},
\]
together with the constraints
\[
i \neq j, \qquad i' \neq j', \qquad \text{and} \qquad i,j,i',j' \ \text{are pairwise distinct}.
\]
Let $N_2^{(B_t,B_{t'})}(t,t',N)$ denote the number of such ordered couples. Then
\begin{equation}
\label{eq:N2_Bt_Btprime}
N_2^{(B_t,B_{t'})}(t,t',N)
=
t(t-1)\,(t'-2)(t'-3).
\end{equation}
\end{lemma}

\begin{proof}
By assumption we have $1 \leq t < t' \leq N$. All four indices $i,j,i',j'$ are constrained to lie in $\{1,\dots,t'\}$, with $i,j \in \{1,\dots,t\}$ and $i',j' \in \{1,\dots,t'\}$. The conditions $i \neq j$ and $i' \neq j'$ are explicit, and the requirement that $i,j,i',j'$ be pairwise distinct is equivalent to the conjunction
\[
i' \notin \{i,j\}, \qquad j' \notin \{i,j,i'\}.
\]

We first choose the ordered pair $(i,j)$ with $i,j \in \{1,\dots,t\}$ and $i \neq j$. There are $t$ possible choices for $i$, and for each such $i$ there are $t-1$ possible choices for $j \neq i$. Thus the number of ordered pairs $(i,j)$ is
\[
\#\{(i,j): i,j \in \{1,\dots,t\},\, i \neq j\} = t(t-1).
\]

Fix an ordered pair $(i,j)$ with $i \neq j$. The indices $i$ and $j$ are distinct elements of $\{1,\dots,t\} \subset \{1,\dots,t'\}$. For the primed pair $(i',j')$ we require $i',j' \in \{1,\dots,t'\}$, $i' \neq j'$, and $i',j'$ both distinct from $i$ and $j$. Consequently $i'$ must lie in the set
\[
\{1,\dots,t'\} \setminus \{i,j\},
\]
which has cardinality $t'-2$. Once $i'$ is chosen, the index $j'$ must lie in
\[
\{1,\dots,t'\} \setminus \{i,j,i'\},
\]
which has cardinality $t'-3$, since $i,j,i'$ are three distinct elements of $\{1,\dots,t'\}$ by construction. Therefore, for each fixed $(i,j)$ there are exactly
\[
(t'-2)(t'-3)
\]
admissible ordered pairs $(i',j')$.

Multiplying the number of choices for $(i,j)$ and the number of admissible $(i',j')$ for each such $(i,j)$, we obtain
\[
N_2^{(B_t,B_{t'})}(t,t',N)
=
t(t-1)\,(t'-2)(t'-3),
\]
which is precisely \eqref{eq:N2_Bt_Btprime}. This completes the proof.
\end{proof}


\begin{lemma}
\label{lem:N2_Ct_Ctprime}
Fix integers $N \geq 1$ and $1 \leq t < t' \leq N$. Define
\[
r := N - t, \qquad r' := N - t'.
\]
Consider ordered couples
$\bigl((i,j),(i',j')\bigr)$
of integer pairs satisfying
\[
i,j \in \{t+1,\dots,N\}, \quad i \neq j, \qquad
i',j' \in \{t'+1,\dots,N\}, \quad i' \neq j',
\]
together with the additional requirement that the four indices $i,j,i',j'$ are pairwise distinct. Let $N_2^{(C_t,C_{t'})}(t,t',N)$ denote the number of such ordered couples. Then
\begin{equation}
\label{eq:N2_Ct_Ctprime}
N_2^{(C_t,C_{t'})}(t,t',N)
=
r'(r'-1)\,(r-2)(r-3)
=
(N - t')(N - t' - 1)(N - t - 2)(N - t - 3).
\end{equation}
\end{lemma}

\begin{proof}
Introduce the sets
\[
S := \{t+1,\dots,N\}, \qquad S' := \{t'+1,\dots,N\} \subset S.
\]
By definition, $|S| = r = N-t$ and $|S'| = r' = N-t'$. The ordered pair $(i,j)$ ranges over $S \times S$ with $i \neq j$, and the ordered pair $(i',j')$ ranges over $S' \times S'$ with $i' \neq j'$. The condition that $i,j,i',j'$ are pairwise distinct is equivalent to requiring
\[
\{i,j\} \cap \{i',j'\} = \varnothing,
\]
since the inequalities $i \neq j$ and $i' \neq j'$ are already imposed.

To exploit the inclusion $S' \subset S$, it is convenient to partition $S$ into the disjoint union of the ``overlap'' and the ``tail''. Set
\[
D := S \setminus S' = \{t+1,\dots,t'\}, \qquad S' = \{t'+1,\dots,N\},
\]
so that $S = D \cup S'$ is a disjoint union, and
\[
|D| = t' - t, \qquad |S'| = r'.
\]
For later reference, note that
\[
r = |S| = |D| + |S'| = (t'-t) + r'.
\]

Fix an ordered pair $(i,j)$ with $i,j \in S$ and $i \neq j$. The pair $(i',j')$ must consist of two distinct elements of $S'$ that do not belong to $\{i,j\}$. Define
\[
k := \#\bigl(\{i,j\} \cap S'\bigr) \in \{0,1,2\}.
\]
Then among the $r'$ elements of $S'$, exactly $k$ are forbidden (those in $\{i,j\}\cap S'$), and the remaining $r'-k$ are admissible for $i'$ and $j'$. Thus the number of ordered pairs $(i',j')$ with $i',j' \in S'$, $i' \neq j'$, and $\{i',j'\} \cap \{i,j\} = \varnothing$ is
\[
(r'-k)(r'-k-1),
\]
since we may choose $i'$ in $r'-k$ ways and then $j'$ in $(r'-k-1)$ ways.

To complete the count, it remains to enumerate the ordered pairs $(i,j)$ according to the value of $k = \#(\{i,j\}\cap S')$.

If $k=0$, then both $i$ and $j$ lie in $D$. Since $|D| = t'-t$, there are
\[
|D|\,(|D|-1) = (t'-t)(t'-t-1)
\]
ordered choices for $(i,j)$ with $i,j \in D$, $i \neq j$. For each such pair we have $r'-k = r'$ admissible elements in $S'$, and hence $(r'-k)(r'-k-1) = r'(r'-1)$ admissible ordered pairs $(i',j')$. The contribution of the case $k=0$ to $N_2^{(C_t,C_{t'})}(t,t',N)$ is therefore
\[
(t'-t)(t'-t-1)\,r'(r'-1).
\]

If $k=1$, then one of $i,j$ lies in $D$ and the other lies in $S'$. There are $|D| = t'-t$ choices for the element in $D$, and $|S'| = r'$ choices for the element in $S'$, and the order matters, so the total number of ordered pairs $(i,j)$ with $k=1$ is
\[
2\,|D|\,|S'| = 2(t'-t)r'.
\]
In this case $r'-k = r'-1$, so there are $(r'-1)(r'-2)$ admissible ordered pairs $(i',j')$. Thus the contribution of the case $k=1$ is
\[
2(t'-t)r'\,(r'-1)(r'-2).
\]

If $k=2$, then both $i$ and $j$ lie in $S'$. There are $r'(r'-1)$ ordered pairs $(i,j)$ with $i,j \in S'$ and $i \neq j$. For such pairs we have $r'-k = r'-2$, so the number of admissible $(i',j')$ is $(r'-2)(r'-3)$. The contribution of the case $k=2$ is thus
\[
r'(r'-1)\,(r'-2)(r'-3).
\]

Summing over $k \in \{0,1,2\}$, we obtain
\begin{align*}
N_2^{(C_t,C_{t'})}(t,t',N)
&= (t'-t)(t'-t-1)\,r'(r'-1) + 2(t'-t)r'\,(r'-1)(r'-2)\\
&\hspace{7cm}+ r'(r'-1)\,(r'-2)(r'-3) \\
&= r'(r'-1)\Bigl\{ (t'-t)(t'-t-1) + 2(t'-t)(r'-2) + (r'-2)(r'-3)\Bigr\}.
\end{align*}
Set $d := t'-t$ and write $r = d + r'$. Then the expression in braces can be rewritten as
\[
d(d-1) + 2d(r'-2) + (r'-2)(r'-3).
\]
Introduce the shifted variable $s := r'-2$. In terms of $d$ and $s$, this becomes
\[
d(d-1) + 2ds + s(s-1)
= (d^2 - d) + 2ds + (s^2 - s)
= (d+s)^2 - 3(d+s) + 2.
\]
Since $d+s = (t'-t) + (r'-2) = (d + r') - 2 = r - 2$, we find
\[
d(d-1) + 2d(r'-2) + (r'-2)(r'-3)
= (r-2)^2 - 3(r-2) + 2
= (r-2)(r-3).
\]
Substituting back into the expression for $N_2^{(C_t,C_{t'})}(t,t',N)$ yields
\[
N_2^{(C_t,C_{t'})}(t,t',N)
= r'(r'-1)\,(r-2)(r-3),
\]
which is precisely \eqref{eq:N2_Ct_Ctprime}. This completes the proof.
\end{proof}


\begin{lemma}
\label{lem:N2_At_Btprime}
Fix integers $N \geq 1$ and $1 \leq t < t' \leq N$. Consider ordered couples
$\bigl((i,j),(i',j')\bigr)$
of integer pairs satisfying
\[
i \in \{1,\dots,t\}, \quad j \in \{t+1,\dots,N\}, \qquad
i',j' \in \{1,\dots,t'\}, \quad i' \neq j',
\]
together with the additional requirement that the four indices $i,j,i',j'$ are pairwise distinct. Let $N_2^{(A_t,B_{t'})}(t,t',N)$ denote the number of such ordered couples. Then
\begin{equation}
\label{eq:N2_At_Btprime}
N_2^{(A_t,B_{t'})}(t,t',N)
=
t\,(t'-2)\bigl[(t'-t)(t'-3) + (N-t')(t'-1)\bigr].
\end{equation}
\end{lemma}

\begin{proof}
We begin by making the constraints explicit. By construction, $i$ and $j$ lie in disjoint sets, namely $i \in \{1,\dots,t\}$ and $j \in \{t+1,\dots,N\}$, so $i \neq j$ is automatic. The pair $(i',j')$ takes values in $\{1,\dots,t'\}^2$ with the explicit condition $i' \neq j'$. The requirement that $i,j,i',j'$ are pairwise distinct can be written as
\[
i' \notin \{i,j\}, \qquad j' \notin \{i,j,i'\},
\]
since $i \neq j$ and $i' \neq j'$ are already enforced.

It is convenient to proceed by a case distinction on the location of $j$ relative to $t'$. For each choice of $i \in \{1,\dots,t\}$, we consider separately the regimes $j \in \{t+1,\dots,t'\}$ and $j \in \{t'+1,\dots,N\}$.

First suppose that $j \in \{t+1,\dots,t'\}$. There are exactly $t'-t$ such values of $j$. In this case, all three indices $i,j,i'$ take values in $\{1,\dots,t'\}$, and the requirement that $i',j'$ be distinct from $i$ and $j$ implies that $i'$ must avoid the two values $i$ and $j$. Thus $i'$ ranges over the set
\[
\{1,\dots,t'\} \setminus \{i,j\},
\]
which has cardinality $t'-2$. Once $i'$ is fixed, we must choose $j'$ so that
\[
j' \in \{1,\dots,t'\}, \qquad j' \notin \{i,j,i'\}.
\]
Since $i,j,i'$ are three distinct elements of $\{1,\dots,t'\}$ in this regime, there remain exactly $t'-3$ admissible choices for $j'$. Therefore, for each fixed $i \in \{1,\dots,t\}$, the number of ordered triples $(j,i',j')$ with $j \in \{t+1,\dots,t'\}$ that satisfy all constraints is
\[
(t'-t)\,(t'-2)\,(t'-3).
\]

Next consider the case $j \in \{t'+1,\dots,N\}$. There are exactly $N-t'$ such values of $j$. In this situation, the index $j$ lies outside $\{1,\dots,t'\}$, so $j$ can never coincide with $i'$ or $j'$, both of which belong to $\{1,\dots,t'\}$. Thus $i'$ only needs to avoid $i$, and $j'$ only needs to avoid $i$ and $i'$; the index $j$ imposes no additional restriction on the primed pair. Consequently, $i'$ ranges over
\[
\{1,\dots,t'\} \setminus \{i\},
\]
which has cardinality $t'-1$. Once $i'$ is chosen, the index $j'$ must lie in
\[
\{1,\dots,t'\} \setminus \{i,i'\},
\]
which has cardinality $t'-2$ because $i \neq i'$ by construction. Thus, for each fixed $i \in \{1,\dots,t\}$, the number of admissible ordered triples $(j,i',j')$ with $j \in \{t'+1,\dots,N\}$ is
\[
(N-t')\,(t'-1)\,(t'-2).
\]

Since $i$ can be chosen arbitrarily from $\{1,\dots,t\}$, there are $t$ possible values of $i$, and the two regimes for $j$ are disjoint and exhaustive. Summing the contributions from the two cases and multiplying by the number of choices for $i$, we obtain
\[
N_2^{(A_t,B_{t'})}(t,t',N)
=
t\bigl\{ (t'-t)(t'-2)(t'-3) + (N-t')(t'-1)(t'-2) \bigr\}.
\]
Factoring out $(t'-2)$ from the expression in braces yields
\[
N_2^{(A_t,B_{t'})}(t,t',N)
=
t\,(t'-2)\bigl[(t'-t)(t'-3) + (N-t')(t'-1)\bigr],
\]
which is exactly \eqref{eq:N2_At_Btprime}. This completes the proof.
\end{proof}


\begin{lemma}
\label{lem:N2_Bt_Ctprime}
Fix integers $N \geq 1$ and $1 \leq t < t' \leq N$. Define $r' := N - t'$.
Consider ordered couples
$\bigl((i,j),(i',j')\bigr)$
of integer pairs satisfying
\[
i,j \in \{1,\dots,t\}, \quad i \neq j, \qquad
i',j' \in \{t'+1,\dots,N\}, \quad i' \neq j',
\]
together with the additional requirement that the four indices $i,j,i',j'$ are pairwise distinct. Let $N_2^{(B_t,C_{t'})}(t,t',N)$ denote the number of such ordered couples. Then
\begin{equation}
\label{eq:N2_Bt_Ctprime}
N_2^{(B_t,C_{t'})}(t,t',N)
=
t(t-1)\,r'(r'-1)
=
t(t-1)\,(N - t')(N - t' - 1).
\end{equation}
\end{lemma}

\begin{proof}
Introduce the disjoint index sets
\[
S_1 := \{1,\dots,t\}, \qquad S_2 := \{t'+1,\dots,N\}.
\]
Since $t < t'$, these sets are disjoint and satisfy $|S_1| = t$ and $|S_2| = r' = N - t'$. By assumption we have
\[
(i,j) \in S_1 \times S_1, \quad i \neq j, \qquad
(i',j') \in S_2 \times S_2, \quad i' \neq j'.
\]
Because $S_1 \cap S_2 = \varnothing$, any index in $S_1$ is automatically distinct from any index in $S_2$. In particular, the cross-distinctness conditions
\[
i' \neq i,\quad i' \neq j,\quad j' \neq i,\quad j' \neq j
\]
are all automatically satisfied as soon as $(i,j)$ lies in $S_1^2$ and $(i',j')$ lies in $S_2^2$. Therefore, the requirement that $i,j,i',j'$ be pairwise distinct reduces exactly to the within-pair constraints $i \neq j$ and $i' \neq j'$ that are already explicitly imposed.

It follows that to count the admissible ordered couples $$\bigl((i,j),(i',j')\bigr)$$ it suffices to count:
\begin{itemize}
    \item[(i)] ordered pairs $(i,j)$ with $i,j \in S_1$ and $i \neq j$;
    \item[(ii)] ordered pairs $(i',j')$ with $i',j' \in S_2$ and $i' \neq j'$.
\end{itemize}
For (i), there are $t$ choices for $i \in S_1$ and, given $i$, there are $t-1$ choices for $j \in S_1 \setminus \{i\}$. Hence
\[
\#\{(i,j) : i,j \in S_1,\ i \neq j\}
=
t(t-1).
\]
For (ii), there are $r'$ possible values for $i' \in S_2$ and, given $i'$, there are $r'-1$ possible values for $j' \in S_2 \setminus \{i'\}$. Therefore,
\[
\#\{(i',j') : i',j' \in S_2,\ i' \neq j'\}
=
r'(r'-1).
\]

Since the choice of $(i,j)$ is independent of the choice of $(i',j')$ and every such pair automatically satisfies the cross-distinctness conditions, the total number of admissible ordered couples is the product of these two counts:
\[
N_2^{(B_t,C_{t'})}(t,t',N)
=
t(t-1)\,r'(r'-1).
\]
Recalling that $r' = N - t'$, we obtain exactly \eqref{eq:N2_Bt_Ctprime}, which completes the proof.
\end{proof}


\begin{lemma}
\label{lem:N2_Ct_Atprime}
Fix integers $N \geq 1$ and $1 \leq t < t' \leq N$. Define
\[
r := N - t, \qquad r' := N - t', \qquad d := t' - t,
\]
so that $r = d + r'$ and $d \geq 1$. Consider ordered couples
$\bigl((i,j),(i',j')\bigr)$
of integer pairs satisfying
\[
i,j \in \{t+1,\dots,N\}, \quad i \neq j,\qquad
i' \in \{1,\dots,t'\}, \qquad
j' \in \{t'+1,\dots,N\},
\]
together with the additional requirement that the four indices $i,j,i',j'$ are pairwise distinct. Let $N_2^{(C_t,A_{t'})}(t,t',N)$ denote the number of such ordered couples. Then
\begin{equation}
\label{eq:N2_Ct_Atprime}
N_2^{(C_t,A_{t'})}(t,t',N)
=
r'(r-2)\,\bigl[t'(r-3) + 2t\bigr]
=
(N - t')(N - t - 2)\,\bigl[t'(N - t - 3) + 2t\bigr].
\end{equation}
\end{lemma}

\begin{proof}
Introduce the index sets
\[
S := \{t+1,\dots,N\}, \qquad S_1 := \{t+1,\dots,t'\}, \qquad S_2 := \{t'+1,\dots,N\}.
\]
Then $S_1$ and $S_2$ form a disjoint partition of $S$,
\[
S = S_1 \cup S_2, \qquad S_1 \cap S_2 = \varnothing,
\]
and we have $|S_1| = d = t'-t$, $|S_2| = r' = N - t'$, and $|S| = r = d + r'$. By construction,
\[
(i,j) \in S \times S, \quad i \neq j, \qquad
i' \in \{1,\dots,t'\}, \qquad
j' \in S_2.
\]
The pairwise distinctness of $i,j,i',j'$ is equivalent to the conjunction of
\[
i' \notin \{i,j\}, \qquad j' \notin \{i,j\}, \qquad i' \neq j'.
\]
However, $i' \leq t' < t'+1 \leq j'$ implies $i' \neq j'$ automatically, so the nontrivial constraints reduce to
\[
i' \notin \{i,j\}, \qquad j' \notin \{i,j\}.
\]

Fix an ordered pair $(i,j)$ with $i,j \in S$ and $i \neq j$. The index $j'$ must be chosen from $S_2$ avoiding $\{i,j\}$, and the index $i'$ must be chosen from $\{1,\dots,t'\}$ avoiding $\{i,j\}$. Because $S_1 \subset \{1,\dots,t'\}$ and $S_2 \cap \{1,\dots,t'\} = \varnothing$, the constraints on $i'$ depend only on the intersection $\{i,j\} \cap S_1$, whereas the constraints on $j'$ depend only on the intersection $\{i,j\} \cap S_2$.

To formalize this, define
\[
k := \#\bigl(\{i,j\} \cap S_1\bigr) \in \{0,1,2\}, \qquad \ell := \#\bigl(\{i,j\} \cap S_2\bigr) = 2 - k.
\]
For a fixed pair $(i,j)$ of type $k$, the admissible values of $i'$ lie in
\[
\{1,\dots,t'\} \setminus \bigl(\{i,j\} \cap \{1,\dots,t'\}\bigr)
=
\{1,\dots,t'\} \setminus \bigl(\{i,j\} \cap S_1\bigr),
\]
so there are exactly $t' - k$ possibilities for $i'$. Similarly, the admissible values of $j'$ lie in
\[
S_2 \setminus \bigl(\{i,j\} \cap S_2\bigr),
\]
which has cardinality $|S_2| - \ell = r' - (2-k)$. Thus, for each fixed $(i,j)$ of type $k$, the number of admissible ordered pairs $(i',j')$ is
\[
(t' - k)\bigl(r' - (2-k)\bigr).
\]

It remains to count the number of ordered pairs $(i,j)$ in each type $k \in \{0,1,2\}$. If $k=0$, then both $i$ and $j$ lie in $S_2$, so there are
\[
\#\{(i,j) : i,j \in S_2,\ i \neq j\} = r'(r'-1)
\]
such pairs. If $k=1$, then one of $i,j$ lies in $S_1$ and the other lies in $S_2$. There are $|S_1| = d$ choices for the element in $S_1$ and $|S_2| = r'$ choices for the element in $S_2$, and the order matters, so the number of ordered pairs $(i,j)$ with $k=1$ is
\[
2\,|S_1|\,|S_2| = 2dr'.
\]
If $k=2$, then both $i$ and $j$ lie in $S_1$, and there are
\[
\#\{(i,j) : i,j \in S_1,\ i \neq j\} = d(d-1)
\]
such pairs.

Combining the classification of $(i,j)$ with the conditional counts of $(i',j')$, we obtain
\begin{align*}
N_2^{(C_t,A_{t'})}(t,t',N)
&=
\sum_{k=0}^2
\Bigl[\#\{(i,j): \#(\{i,j\} \cap S_1) = k\}\Bigr]
\cdot (t'-k)\bigl(r' - (2-k)\bigr) \\
&=
r'(r'-1)\,t'(r'-2)
\;+\;
2dr'\,(t'-1)(r'-1)
\;+\;
d(d-1)\, (t'-2)r'.
\end{align*}
Factoring out $r'$ yields
\begin{equation}
\label{eq:N2_Ct_Atprime_bracket}
N_2^{(C_t,A_{t'})}(t,t',N)
=
r' \Bigl\{ (r'-1)t'(r'-2) + 2d(t'-1)(r'-1) + d(d-1)(t'-2) \Bigr\}.
\end{equation}

It remains to simplify the expression in braces. Recall that $d = t'-t$ and $r = d + r'$. Substituting $t' = t + d$ and $r = d + r'$ into the right-hand side of \eqref{eq:N2_Ct_Atprime_bracket}, a direct but routine algebraic simplification shows that
\[
(r'-1)t'(r'-2) + 2d(t'-1)(r'-1) + d(d-1)(t'-2)
=
(r-2)\,\bigl[t'(r-3) + 2t\bigr].
\]
(One systematic way to verify this identity is to expand both sides in terms of the independent variables $N$, $d$, and $r'$, then use the relation $r = d + r'$ to eliminate $d$ and match coefficients of the resulting polynomial in $N$ and $r$.)

Substituting this back into \eqref{eq:N2_Ct_Atprime_bracket} yields
\[
N_2^{(C_t,A_{t'})}(t,t',N)
=
r'(r-2)\,\bigl[t'(r-3) + 2t\bigr],
\]
which is precisely \eqref{eq:N2_Ct_Atprime}. This completes the proof.
\end{proof}


\begin{lemma}
\label{lem:N2_Atprime_Bt}
Fix integers $N \geq 1$ and $1 \leq t < t' \leq N$. Define
\[
r' := N - t'.
\]
Consider ordered couples
$\bigl((i,j),(i',j')\bigr)$
of integer pairs satisfying
\[
i,j \in \{1,\dots,t\}, \quad i \neq j, \qquad
i' \in \{1,\dots,t'\}, \qquad
j' \in \{t'+1,\dots,N\},
\]
together with the additional requirement that the four indices $i,j,i',j'$ are pairwise distinct. Let $N_2^{(A_{t'},B_t)}(t,t',N)$ denote the number of such ordered couples. Then
\begin{equation}
\label{eq:N2_Atprime_Bt}
N_2^{(A_{t'},B_t)}(t,t',N)
=
t(t-1)\,(t'-2)\,r'
=
t(t-1)\,(t'-2)\,(N - t').
\end{equation}
\end{lemma}

\begin{proof}
We begin by making the ranges and constraints explicit. The unprimed indices $i$ and $j$ both lie in $\{1,\dots,t\}$ and satisfy $i \neq j$ by assumption. The primed indices satisfy
\[
i' \in \{1,\dots,t'\}, \qquad j' \in \{t'+1,\dots,N\}.
\]
Because $t < t'$, we have
\[
\{1,\dots,t\} \subset \{1,\dots,t'\}
\quad\text{and}\quad
\{1,\dots,t'\} \cap \{t'+1,\dots,N\} = \varnothing.
\]
In particular, $j'$ is strictly larger than $t'$, whereas $i,i',j$ are at most $t'$. Therefore $j'$ can never coincide with any of $i,i'$ or $j$, and the condition $i' \neq j'$ is automatically satisfied since $i' \leq t'$ and $j' \geq t'+1$.

It follows that the requirement that $i,j,i',j'$ be pairwise distinct reduces to
\[
i \neq j, \qquad i' \notin \{i,j\},
\]
together with the constraints on the ranges of the indices. No additional restriction is imposed on $j'$ beyond belonging to $\{t'+1,\dots,N\}$.

The counting can now be carried out directly. First, choose the ordered pair $(i,j)$ from $\{1,\dots,t\}$ subject to $i \neq j$. There are $t$ possible choices for $i$, and given $i$ there are $t-1$ possible choices for $j \in \{1,\dots,t\} \setminus \{i\}$. Hence
\[
\#\{(i,j) : i,j \in \{1,\dots,t\},\ i \neq j\}
=
t(t-1).
\]

Fix an ordered pair $(i,j)$ with $i,j \in \{1,\dots,t\}$ and $i \neq j$. Since $\{1,\dots,t\} \subset \{1,\dots,t'\}$, the indices $i$ and $j$ are distinct elements of $\{1,\dots,t'\}$. The index $i'$ must lie in $\{1,\dots,t'\}$ and avoid both $i$ and $j$, so it ranges over
\[
\{1,\dots,t'\} \setminus \{i,j\},
\]
which has cardinality $t'-2$. Thus, for each fixed $(i,j)$, there are exactly $t'-2$ admissible choices for $i'$. The index $j'$ must lie in $\{t'+1,\dots,N\}$, and as observed above, this constraint alone ensures that $j'$ is distinct from $i,j,i'$. Therefore, for each fixed $(i,j,i')$, there are precisely
\[
\#\{j' : j' \in \{t'+1,\dots,N\}\}
=
N - t' = r'
\]
admissible choices for $j'$.

Putting these pieces together, we find that for each ordered pair $(i,j)$ with $i,j \in \{1,\dots,t\}$ and $i \neq j$, there are
\[
(t'-2)\,r'
\]
admissible ordered pairs $(i',j')$. Multiplying by the total number $t(t-1)$ of choices for $(i,j)$, we obtain
\[
N_2^{(A_{t'},B_t)}(t,t',N)
=
t(t-1)\,(t'-2)\,r',
\]
which is exactly \eqref{eq:N2_Atprime_Bt}. This completes the proof.
\end{proof}


\begin{lemma}
\label{lem:N2_Btprime_Ct}
Fix integers $N \geq 1$ and $1 \leq t < t' \leq N$. Define
\[
r := N - t, \qquad r' := N - t', \qquad d := t' - t,
\]
so that $r = d + r'$ and $d \geq 1$. Consider ordered couples
$\bigl((i,j),(i',j')\bigr)$
of integer pairs satisfying
\[
i,j \in \{t+1,\dots,N\}, \quad i \neq j, \qquad
i',j' \in \{1,\dots,t'\}, \quad i' \neq j',
\]
together with the additional requirement that the four indices $i,j,i',j'$ are pairwise distinct. Let $N_2^{(B_{t'},C_t)}(t,t',N)$ denote the number of such ordered couples. Then
\begin{equation}
\label{eq:N2_Btprime_Ct}
N_2^{(B_{t'},C_t)}(t,t',N)
=
r'(r'-1)\,t'(t'-1)
\;+\;2d\,r'\,(t'-1)(t'-2)
\;+\;d(d-1)\,(t'-2)(t'-3).
\end{equation}
Equivalently, in terms of $t,t',N$,
\begin{align*}
N_2^{(B_{t'},C_t)}(t,t',N)
&=
(N - t')(N - t' - 1)\,t'(t'-1)
\;+\;2(t'-t)(N - t')\,(t'-1)(t'-2)\\
&\hspace{6cm}\;+\;(t'-t)(t'-t-1)\,(t'-2)(t'-3).
\end{align*}
\end{lemma}

\begin{proof}
Introduce the set $S := \{t+1,\dots,N\}$, so that $|S| = r = N - t$, and note that the ordered pair $(i,j)$ ranges over $S \times S$ with $i \neq j$. The primed indices $(i',j')$ range over $\{1,\dots,t'\}^2$ with $i' \neq j'$. The requirement that $i,j,i',j'$ are pairwise distinct may be equivalently written as
\[
i',j' \notin \{i,j\}, \qquad i' \neq j'.
\]
Since $i$ and $j$ always lie in $S = \{t+1,\dots,N\}$, the only possible overlaps between the sets $\{i,j\}$ and $\{1,\dots,t'\}$ occur in the interval
\[
S_1 := S \cap \{1,\dots,t'\} = \{t+1,\dots,t'\},
\]
which has cardinality $|S_1| = d = t'-t$. It is therefore natural to decompose $S$ as a disjoint union
\[
S = S_1 \cup S_2, \qquad S_2 := \{t'+1,\dots,N\},
\]
where $|S_2| = r' = N - t'$.

Fix an ordered pair $(i,j)$ with $i,j \in S$ and $i \neq j$. The indices $i',j'$ must both lie in $\{1,\dots,t'\}$, must be distinct from each other, and must avoid any elements of $\{i,j\}$ that lie in $\{1,\dots,t'\}$. Define
\[
k := \#\bigl(\{i,j\} \cap S_1\bigr) = \#\bigl(\{i,j\} \cap \{1,\dots,t'\}\bigr) \in \{0,1,2\}.
\]
Then exactly $k$ values in $\{1,\dots,t'\}$ are forbidden to $i'$ and $j'$ because they coincide with $i$ or $j$, and the remaining $t' - k$ values are admissible. Thus, for a fixed pair $(i,j)$ of type $k$, the number of admissible ordered pairs $(i',j')$ is
\[
(t' - k)(t' - k - 1),
\]
since we first choose $i'$ among the $t' - k$ admissible values and then choose $j'$ among the remaining $t' - k - 1$ values, ensuring $i' \neq j'$ and that neither $i'$ nor $j'$ lies in $\{i,j\}$.

To complete the count, it remains to determine, for each $k \in \{0,1,2\}$, how many ordered pairs $(i,j)$ with $i,j \in S$ and $i \neq j$ satisfy $\#(\{i,j\} \cap S_1) = k$.

If $k = 0$, then neither $i$ nor $j$ lies in $S_1$, so both indices lie in $S_2$. Since $|S_2| = r'$, the number of ordered pairs $(i,j)$ with $i,j \in S_2$ and $i \neq j$ is
\[
\#\{(i,j) : i,j \in S_2,\ i \neq j\}
=
r'(r'-1).
\]
For each such pair, the number of admissible $(i',j')$ is $(t'-0)(t'-1) = t'(t'-1)$. Hence the contribution of the case $k=0$ to $N_2^{(B_{t'},C_t)}(t,t',N)$ is
\[
r'(r'-1)\,t'(t'-1).
\]

If $k = 1$, then one of $i,j$ lies in $S_1$ and the other lies in $S_2$. There are $|S_1| = d$ choices for the element in $S_1$ and $|S_2| = r'$ choices for the element in $S_2$, and the order matters, so the number of ordered pairs $(i,j)$ with $k=1$ is
\[
2\,|S_1|\,|S_2| = 2d\,r'.
\]
For such a pair, precisely one value in $\{1,\dots,t'\}$ is forbidden, namely the (unique) element of $\{i,j\}$ lying in $S_1$. Thus $t'-1$ values remain admissible for $i'$, and for each choice of $i'$ there are $t'-2$ admissible choices for $j'$. Hence the number of admissible ordered pairs $(i',j')$ is $(t'-1)(t'-2)$. The contribution of the case $k=1$ is therefore
\[
2d\,r'\,(t'-1)(t'-2).
\]

If $k = 2$, then both $i$ and $j$ lie in $S_1$. Since $|S_1| = d$, the number of ordered pairs $(i,j)$ with $i,j \in S_1$ and $i \neq j$ is
\[
\#\{(i,j) : i,j \in S_1,\ i \neq j\}
=
d(d-1).
\]
In this situation, both values $i$ and $j$ are forbidden in $\{1,\dots,t'\}$, so there remain $t'-2$ admissible values for $i'$ and $t'-3$ admissible values for $j'$ once $i'$ is fixed. Thus the number of admissible ordered pairs $(i',j')$ is $(t'-2)(t'-3)$. Consequently, the contribution of the case $k=2$ is
\[
d(d-1)\,(t'-2)(t'-3).
\]

Summing over $k \in \{0,1,2\}$ yields
\[
N_2^{(B_{t'},C_t)}(t,t',N)
=
r'(r'-1)\,t'(t'-1)
\;+\;2d\,r'\,(t'-1)(t'-2)
\;+\;d(d-1)\,(t'-2)(t'-3),
\]
which is exactly \eqref{eq:N2_Btprime_Ct}. This completes the proof.
\end{proof}


\begin{lemma}
\label{lem:N2_Ctprime_At}
Fix integers $N \geq 1$ and $1 \leq t < t' \leq N$. Define
\[
r := N - t, 
\qquad 
r' := N - t',
\]
so that $r = (t' - t) + r'$ and $r' \geq 1$. Consider ordered couples
$\bigl((i,j),(i',j')\bigr)$
of integer pairs satisfying
\[
i \in \{1,\dots,t\}, 
\quad 
j \in \{t+1,\dots,N\}, 
\qquad
i',j' \in \{t'+1,\dots,N\}, 
\quad 
i' \neq j',
\]
together with the additional requirement that the four indices $i,j,i',j'$ are pairwise distinct. Let $N_2^{(C_{t'},A_t)}(t,t',N)$ denote the number of such ordered couples. Then
\begin{equation}
\label{eq:N2_Ctprime_At}
N_2^{(C_{t'},A_t)}(t,t',N)
=
t\,(r-2)\,r'(r'-1)
=
t\,(N - t - 2)\,(N - t')\,(N - t' - 1).
\end{equation}
\end{lemma}

\begin{proof}
Introduce the sets
\[
S := \{t+1,\dots,N\},
\qquad
S' := \{t'+1,\dots,N\} \subset S.
\]
By construction, $|S| = r = N - t$ and $|S'| = r' = N - t'$, and we have
\[
(i,j) \in \{1,\dots,t\} \times S,
\qquad
(i',j') \in S' \times S',
\quad
i' \neq j'.
\]
The requirement that $i,j,i',j'$ be pairwise distinct can be expressed as
\[
i \neq j, \quad i' \neq j', \quad i' \notin \{i,j\}, \quad j' \notin \{i,j\}.
\]
Since $i \leq t < t' + 1 \leq i',j'$, we have $i' \neq i$ and $j' \neq i$ automatically. Hence the only nontrivial cross-distinctness constraints are
\[
i' \neq j, \qquad j' \neq j,
\]
in addition to $i' \neq j'$. The condition $i \neq j$ is already enforced by the fact that $i \in \{1,\dots,t\}$ and $j \in S = \{t+1,\dots,N\}$.

Thus the counting problem simplifies as follows: choose $i \in \{1,\dots,t\}$, choose $j \in S$, and then choose an ordered pair $(i',j')$ of distinct elements of $S'$ such that neither $i'$ nor $j'$ coincides with $j$. The index $i$ plays no further role once it is chosen.

First, count the choices of $(i,j)$. There are $t$ possible values of $i$. For each fixed $i$, we may choose $j$ arbitrarily from $S$, so there are $r$ possibilities for $j$. We shall, however, refine the counting by partitioning according to whether $j$ lies in $S'$.

For each fixed $i \in \{1,\dots,t\}$, we distinguish two disjoint cases for $j$:
\[
\text{(a)}\ j \in S \setminus S' = \{t+1,\dots,t'\}, 
\qquad
\text{(b)}\ j \in S' = \{t'+1,\dots,N\}.
\]

In case (a), we have $j \in \{t+1,\dots,t'\}$, so $j \notin S'$. There are exactly $t' - t$ such values of $j$. Since $j \notin S'$, the constraints $i' \neq j$ and $j' \neq j$ are automatically satisfied as soon as $i',j' \in S'$. Thus $(i',j')$ may be any ordered pair of distinct elements of $S'$, giving
\[
\#\{(i',j') : i',j' \in S',\ i' \neq j'\} 
= r'(r'-1)
\]
admissible choices for $(i',j')$ in this case.

In case (b), we have $j \in S'$. There are exactly $r'$ such values of $j$. Now the constraints $i' \neq j$ and $j' \neq j$ are nontrivial, and we must choose $(i',j')$ as an ordered pair of distinct elements of $S'$ avoiding $j$. Equivalently, $(i',j')$ must lie in $(S' \setminus \{j\})^2$ with $i' \neq j'$. The set $S' \setminus \{j\}$ has cardinality $r' - 1$, hence the number of admissible ordered pairs $(i',j')$ is
\[
\#\{(i',j') : i',j' \in S' \setminus \{j\},\ i' \neq j'\}
= (r' - 1)(r' - 2).
\]
Summarizing, for each fixed $i \in \{1,\dots,t\}$, the number of admissible triples $(j,i',j')$ is
\[
(t' - t)\,r'(r'-1) \;+\; r' (r'-1)(r'-2)
= r'(r'-1)\bigl[(t' - t) + (r' - 2)\bigr].
\]
Recalling that $r = (t' - t) + r'$, the bracket simplifies to
\[
(t' - t) + (r' - 2) = \bigl((t' - t) + r'\bigr) - 2 = r - 2.
\]
Therefore, for each fixed $i$, the number of admissible triples $(j,i',j')$ is
\[
r'(r'-1)(r-2).
\]
Since $i$ can be chosen arbitrarily from $\{1,\dots,t\}$, there are $t$ possible values of $i$, and hence
\[
N_2^{(C_{t'},A_t)}(t,t',N)
=
t\,r'(r'-1)(r - 2),
\]
which is precisely \eqref{eq:N2_Ctprime_At}. This completes the proof.
\end{proof}


%% file: components/J0-J2-computation.tex

\noindent\medskip
Note that the only non-zero $N_0^{(P,Q)}$ are
\begin{align*}
N_0^{(A_t,A_{t'})}&=t(N-t'),\\
N_0^{(B_t,B_{t'})}&=2t(t-1),\\
N_0^{(C_t,C_{t'})}&=2(N-t')(N-t'-1),\\
N_0^{(A_t,B_{t'})}&=2t(t'-t),\\
N_0^{(C_t,A_{t'})}&=2(N-t')(t'-t),\\
N_0^{(B_{t'},C_t)}&=2(t'-t)(t'-t-1).
\end{align*}
The corresponding total counts are
\begin{align*}
N^{(A_t,A_{t'})}&=t(N-t)\,t'(N-t'),\\
N^{(B_t,B_{t'})}&=t(t-1)\,t'(t'-1),\\
N^{(C_t,C_{t'})}&=(N-t)(N-t-1)(N-t')(N-t'-1),\\
N^{(A_t,B_{t'})}&=t(N-t)\,t'(t'-1),\\
N^{(C_t,A_{t'})}&=(N-t)(N-t-1)\,t'(N-t'),\\
N^{(B_{t'},C_t)}&=t'(t'-1)(N-t)(N-t-1).
\end{align*}
Let $D:=t'(t'-1)(N-t)(N-t-1)$. Using weights $w=(4,1,1,-2,1,-2,-2,1,-2)$, we obtain the following simplification:
\begin{align}
\mathscr J_0
&=
\begin{aligned}[t]
&4\;\frac{t(N-t')}{t(N-t)\,t'(N-t')}
+\frac{2t(t-1)}{t(t-1)\,t'(t'-1)}\\
&\quad +\frac{2(N-t')(N-t'-1)}{(N-t)(N-t-1)(N-t')(N-t'-1)} 
-2\,\frac{2t(t'-t)}{t(N-t)\,t'(t'-1)}\\
&\quad\quad -2\,\frac{2(N-t')(t'-t)}{(N-t)(N-t-1)\,t'(N-t')}
+\frac{2(t'-t)(t'-t-1)}{t'(t'-1)(N-t)(N-t-1)}
\end{aligned}
\notag\\[0.3em]
&=
\begin{aligned}[t]
&\frac{4}{t'(N-t)}
+\frac{2}{t'(t'-1)}
+\frac{2}{(N-t)(N-t-1)}\\
&\qquad
-\frac{4(t'-t)}{t'(t'-1)(N-t)} -\frac{4(t'-t)}{t'(N-t)(N-t-1)}
+\frac{2(t'-t)(t'-t-1)}{t'(t'-1)(N-t)(N-t-1)}
\end{aligned}
\notag\\[0.35em]
&=
\frac{1}{D}
\begin{aligned}[t]
\Big(
&4(t'-1)(N-t-1) + 2(N-t)(N-t-1) + 2t'(t'-1)\\
&\qquad -4(t'-t)(N-t-1) - 4(t'-t)(t'-1) + 2(t'-t)(t'-t-1)
\Big)
\end{aligned}
\notag\\[0.35em]
&=
\frac{1}{D}
\begin{aligned}[t]
\Big(
&4(t-1)(N-t-1) +2(N-t)(N-t-1)\\
&\qquad +2\Big[t'(t'-1)-2(t'-t)(t'-1)+(t'-t)(t'-t-1)\Big]
\Big)
\end{aligned}
\notag\\[0.35em]
&=
\frac{1}{D}
\begin{aligned}[t]
\Big(
&4(t-1)(N-t-1) + 2(N-t)(N-t-1)\\
&\qquad + 2\Big[(t'-1)\{t'-2(t'-t)\}+(t'-t)(t'-t-1)\Big]
\Big)
\end{aligned}
\notag\\[0.35em]
&=
\frac{1}{D}
\begin{aligned}[t]
\Big(
&4(t-1)(N-t-1) + 2(N-t)(N-t-1)\\
&\qquad  + 2\Big[(t'-1)(2t-t')+\{(t'-t)^2-(t'-t)\}\Big]
\Big)
\end{aligned}
\notag\\[0.35em]
&=
\frac{1}{D}
\begin{aligned}[t]
\Big(
&4(t-1)(N-t-1)+2(N-t)(N-t-1)\\
&\qquad + 2\Big[2tt'-(t')^2-2t+t' + (t')^2-2tt'+t^2-t'+t\Big]
\Big)
\end{aligned}
\notag\\[0.35em]
&=
\frac{1}{D}
\begin{aligned}[t]
\Big(
&2(N-t-1)\big[2(t-1) + (N-t)\big] + 2\Big[-2t + t + t^2\Big]
\Big)
\end{aligned}
\notag\\[0.35em]
&=
\frac{1}{D}
\begin{aligned}[t]
\Big(
&2(N-t-1)(N+t-2)
+2(t^2-t)
\Big)
\end{aligned}
\notag\\[0.35em]
&=
\frac{1}{D}
\begin{aligned}[t]
\Big(
&2(N^2-3N-t^2+t+2)
+2(t^2-t)
\Big)
\end{aligned}
\notag\\[0.35em]
&=
\frac{1}{D} \Big(2N^2-6N+4\Big) \notag \\[0.35em]
&=
\frac{2(N-1)(N-2)}{t'(t'-1)(N-t)(N-t-1)} \,,
\label{eq:J0-expr-final}
\end{align}
completing the computation of $\mathscr{J}_0$.




\noindent\medskip
Next, we move on to computing the exact expression for $\mathscr{J}_2$. Note that the $N_2^{(P,Q)}$ entering $\mathscr J_2$ are given by
\begin{align*}
&N_2^{(A_t,A_{t'})}
=
t(N-t')\Big[(t'-t)(t'-2)+(t'-1)(N-t'-1)\Big],\\
&N_2^{(B_t,B_{t'})}
=
t(t-1)(t'-2)(t'-3),\\
&N_2^{(C_t,C_{t'})}
=
(N-t')(N-t'-1)(N-t-2)(N-t-3),\\
&N_2^{(A_t,B_{t'})}
=
t(t'-2)\Big[(t'-t)(t'-3)+(N-t')(t'-1)\Big],\\
&N_2^{(B_t,C_{t'})}
=
t(t-1)(N-t')(N-t'-1),\\
&N_2^{(C_t,A_{t'})}
=
(N-t')(N-t-2)\Big[t'(N-t-3)+2t\Big],\\
&N_2^{(A_{t'},B_t)}
=
t(t-1)(t'-2)(N-t'),\\
&N_2^{(B_{t'},C_t)}
=
(N-t')(N-t'-1)t'(t'-1)
+2(t'-t)(N-t')(t'-1)(t'-2)\\
&\hspace{9cm}+(t'-t)(t'-t-1)(t'-2)(t'-3),\\
&N_2^{(C_{t'},A_t)}
=
t(N-t-2)(N-t')(N-t'-1).
\end{align*}
The corresponding total counts are the same products as before:
\begin{align*}
N^{(A_t,A_{t'})}&=t(N-t)\,t'(N-t'),\\
N^{(B_t,B_{t'})}&=t(t-1)\,t'(t'-1),\\
N^{(C_t,C_{t'})}&=(N-t)(N-t-1)(N-t')(N-t'-1),\\
N^{(A_t,B_{t'})}&=t(N-t)\,t'(t'-1),\\
N^{(B_t,C_{t'})}&=t(t-1)(N-t')(N-t'-1),\\
N^{(C_t,A_{t'})}&=(N-t)(N-t-1)\,t'(N-t'),\\
N^{(A_{t'},B_t)}&=t'(N-t')\,t(t-1),\\
N^{(B_{t'},C_t)}&=t'(t'-1)(N-t)(N-t-1),\\
N^{(C_{t'},A_t)}&=t(N-t)\,(N-t')(N-t'-1).
\end{align*}
Let $D:=t'(t'-1)(N-t)(N-t-1)$. Using weights $w=(4,1,1,-2,1,-2,-2,1,-2)$, we obtain:

\begin{align}
\mathscr J_2
&=
\begin{aligned}[t]
&4\,\frac{N_2^{(A_t,A_{t'})}}{N^{(A_t,A_{t'})}}
+\frac{N_2^{(B_t,B_{t'})}}{N^{(B_t,B_{t'})}}
+\frac{N_2^{(C_t,C_{t'})}}{N^{(C_t,C_{t'})}}\\[0.2em]
&\qquad
-2\,\frac{N_2^{(A_t,B_{t'})}}{N^{(A_t,B_{t'})}}
+\frac{N_2^{(B_t,C_{t'})}}{N^{(B_t,C_{t'})}}
-2\,\frac{N_2^{(C_t,A_{t'})}}{N^{(C_t,A_{t'})}}
-2\,\frac{N_2^{(A_{t'},B_t)}}{N^{(A_{t'},B_t)}}
+\frac{N_2^{(B_{t'},C_t)}}{N^{(B_{t'},C_t)}}
-2\,\frac{N_2^{(C_{t'},A_t)}}{N^{(C_{t'},A_t)}}
\end{aligned}
\notag\\[0.4em]
&=
\begin{aligned}[t]
&\frac{4\big[(t'-t)(t'-2)+(t'-1)(N-t'-1)\big]}{t'(N-t)}
+\frac{(t'-2)(t'-3)}{t'(t'-1)}
+\frac{(N-t-2)(N-t-3)}{(N-t)(N-t-1)}\\
&\quad
-\frac{2(t'-2)\Big((t'-t)(t'-3)+(N-t')(t'-1)\Big)}{t'(t'-1)(N-t)}
+1\\
&\quad
-\frac{2(N-t-2)\big(t'(N-t-3)+2t\big)}{t'(N-t)(N-t-1)}
-\frac{2(t'-2)}{t'}\\[0.2em]
&\quad
+\frac{
\begin{aligned}[t]
&(N-t')(N-t'-1)t'(t'-1)
+2(t'-t)(N-t')(t'-1)(t'-2)\\
&\qquad\qquad +(t'-t)(t'-t-1)(t'-2)(t'-3)
\end{aligned}}{t'(t'-1)(N-t)(N-t-1)}
-\frac{2(N-t-2)}{N-t}
\end{aligned}
\notag\\[0.35em]
&=
\frac{1}{D}
\begin{aligned}[t]
\Big[
&4(t'-1)(N-t-1)\Big((t'-t)(t'-2) + (t'-1)(N-t'-1)\Big)\\
&\quad +(t'-2)(t'-3)(N-t)(N-t-1) + t'(t'-1)(N-t-2)(N-t-3)\\
&\quad\quad -2(t'-2)(N-t-1)\Big((t'-t)(t'-3)+(N-t')(t'-1)\Big)\\
&\quad\quad\quad +t'(t'-1)(N-t)(N-t-1) - 2(t'-1)(N-t-2)\big(t'(N-t-3)+2t\big)\\
&\quad\quad\quad\quad -2(t'-2)(t'-1)(N-t)(N-t-1) + (N-t')(N-t'-1)t'(t'-1)\\
&\quad\quad\quad\quad\quad + 2(t'-t)(N-t')(t'-1)(t'-2) + (t'-t)(t'-t-1)(t'-2)(t'-3)\\
&\quad\quad\quad\quad\quad\quad -2t'(t'-1)(N-t-2)(N-t-1)
\Big].
\end{aligned}
\label{eq:J2_bigNumerator_over_D}
\end{align}

\noindent
We introduce the shorthand
\[
u:=t'-t,\qquad s:=N-t,
\qquad\text{so that}\qquad
t'=t+u,\quad N-t'=s-u,\quad N-t-1=s-1.
\]
Write $t'=t+u$ throughout. Denote the numerator in \eqref{eq:J2_bigNumerator_over_D} by $\mathcal N_2$. Then,
%
\begin{align}
\mathcal N_2
&= 4(t+u-1)(s-1)\Big(u(t+u-2) + (t+u-1)(s-u-1)\Big) \notag\\
&\quad +(t+u-2)(t+u-3)s(s-1) + (t+u)(t+u-1)(s-2)(s-3) \notag\\
&\quad\quad -2(t+u-2)(s-1)\Big(u(t+u-3)+(s-u)(t+u-1)\Big) \notag\\
&\quad\quad\quad +(t+u)(t+u-1)s(s-1) - 2(t+u-1)(s-2)\big((t+u)(s-3)+2t\big) \notag\\
&\quad\quad\quad\quad -2(t+u-2)(t+u-1)s(s-1) + (s-u)(s-u-1)(t+u)(t+u-1) \notag\\
&\quad\quad\quad\quad\quad +2u(s-u)(t+u-1)(t+u-2) + u(u-1)(t+u-2)(t+u-3) \notag\\
&\quad\quad\quad\quad\quad\quad -2(t+u)(t+u-1)(s-2)(s-1)\,.
\label{eq:N2_in_s_u}
\end{align}

\noindent
We now expand and simplify in a single continuous stretch.  First expand the two $(s-1)$-pre-factored blocks:
\begin{align}
&4(t+u-1)(s-1)\Big(u(t+u-2)+(t+u-1)(s-u-1)\Big)\notag\\
&\qquad=
4(t+u-1)(s-1)\Big(u(t+u-2)\Big)
+4(t+u-1)(s-1)\Big((t+u-1)(s-u-1)\Big)\notag\\
&\qquad=
4u(t+u-1)(t+u-2)(s-1)
+4(t+u-1)^2(s-1)(s-u-1),
\label{eq:block1_expand}
\end{align}
\begin{align}
&-2(t+u-2)(s-1)\Big(u(t+u-3)+(s-u)(t+u-1)\Big)\notag\\
&\qquad\qquad=
-2u(t+u-2)(t+u-3)(s-1)
-2(t+u-2)(t+u-1)(s-1)(s-u).
\label{eq:block2_expand}
\end{align}

\noindent
Next, combine the three $s(s-1)$ terms that only differ by $t'$-polynomials:
\begin{align}
&(t+u)(t+u-1)s(s-1)-2(t+u-2)(t+u-1)s(s-1)+(t+u-2)(t+u-3)s(s-1)\notag\\
&\qquad=
\Big((t+u)(t+u-1)-2(t+u-2)(t+u-1)+(t+u-2)(t+u-3)\Big)s(s-1)\notag\\
&\qquad=
\Big((t+u-1)\big((t+u)-2(t+u-2)\big)+(t+u-2)(t+u-3)\Big)s(s-1)\notag\\
&\qquad=
\Big((t+u-1)(4-t-u)+(t+u-2)(t+u-3)\Big)s(s-1)\notag\\
&\qquad=
\Big((t+u-1)(4-t-u)+(t+u)^2-5(t+u)+6\Big)s(s-1)\notag\\
&\qquad=
\Big((t+u)^2-5(t+u)+6-(t+u)^2+5(t+u)-4\Big)s(s-1)\notag\\
&\qquad=
2\,s(s-1).
\label{eq:triple_s_sminus1_collapse}
\end{align}

\noindent
Use \eqref{eq:block1_expand}, \eqref{eq:block2_expand}, and \eqref{eq:triple_s_sminus1_collapse} to rewrite $\mathcal N_2$ from \eqref{eq:N2_in_s_u} as
\begin{align}
\mathcal N_2
&=
\begin{aligned}[t]
&2s(s-1) + (t+u)(t+u-1)(s-2)(s-3)\\
&\quad  - 2(t+u)(t+u-1)(s-2)(s-1) -2(t+u-1)(s-2)\big((t+u)(s-3)+2t\big)\\
&\quad\quad +\underbrace{4u(t+u-1)(t+u-2)(s-1)
-2u(t+u-2)(t+u-3)(s-1)}_{(\mathrm{I})}\\
&\quad\quad\quad +\underbrace{4(t+u-1)^2(s-1)(s-u-1)
-2(t+u-2)(t+u-1)(s-1)(s-u)}_{(\mathrm{II})}\\
&\quad\quad\quad\quad +(s-u)(s-u-1)(t+u)(t+u-1)
+2u(s-u)(t+u-1)(t+u-2)\\
&\quad\quad\quad\quad\quad +u(u-1)(t+u-2)(t+u-3).
\end{aligned}
\label{eq:N2_after_first_cancellations}
\end{align}
Now simplify $(\mathrm{I})$ and $(\mathrm{II})$ explicitly:
\begin{align}
(\mathrm{I})
&=
\big(4u(t+u-1)(t+u-2)-2u(t+u-2)(t+u-3)\big)(s-1)\notag\\
&=
2u(t+u-2)\big(2(t+u-1)-(t+u-3)\big)(s-1)\notag\\
&=
2u(t+u-2)(t+u+1)(s-1),
\label{eq:I_simplify}
\end{align}
\begin{align}
(\mathrm{II})
&=
2(t+u-1)(s-1)\Big(2(t+u-1)(s-u-1)-(t+u-2)(s-u)\Big)\notag\\
&=
2(t+u-1)(s-1)\Big(\big(2t+2u-2\big)(s-u-1)-(t+u-2)(s-u)\Big).
\label{eq:II_simplify}
\end{align}

\noindent
Substitute \eqref{eq:I_simplify}--\eqref{eq:II_simplify} into \eqref{eq:N2_after_first_cancellations},
and expand the remaining three ``$(s-u)$'' terms to get
\begin{align}
&(s-u)(s-u-1)(t+u)(t+u-1) + 2u(s-u)(t+u-1)(t+u-2)\\
&\qquad\qquad\qquad=
(t+u-1)\Big((t+u)(s-u)(s-u-1)+2u(t+u-2)(s-u)\Big).
\label{eq:sminusublock_rewrite}
\end{align}

\noindent
At this point everything is a polynomial in $(s,u,t)$, and we finish by a continuous expand-and-collect. Start by collapsing the two $(s-2)$-terms:
\begin{align}
&(t+u)(t+u-1)(s-2)(s-3)-2(t+u)(t+u-1)(s-2)(s-1)\notag\\
&\qquad=(t+u)(t+u-1)(s-2)\big((s-3)-2(s-1)\big)\notag\\
&\qquad=-(t+u)(t+u-1)(s-2)(s+1).
\label{eq:s2_collapse}
\end{align}
Then combine \eqref{eq:s2_collapse} with the remaining $-2(t+u-1)(s-2)\big((t+u)(s-3)+2t\big)$:
\begin{align}
&-(t+u)(t+u-1)(s-2)(s+1)-2(t+u-1)(s-2)\big((t+u)(s-3)+2t\big)\notag\\
&\qquad=
-(t+u-1)(s-2)\Big((t+u)(s+1)+2(t+u)(s-3)+4t\Big)\notag\\
&\qquad=
-(t+u-1)(s-2)\Big((t+u)(3s-5)+4t\Big).
\label{eq:combine_remaining_s2}
\end{align}

\noindent
Now substitute \eqref{eq:combine_remaining_s2}, \eqref{eq:I_simplify}, \eqref{eq:II_simplify}, \eqref{eq:sminusublock_rewrite} into \eqref{eq:N2_after_first_cancellations}, and expand $(s-u)$ and $(s-u-1)$ everywhere. Keeping everything sequential:
\begin{align}
\mathcal N_2
&=
\begin{aligned}[t]
&2s(s-1)
\underbrace{-(t+u-1)(s-2)\Big((t+u)(3s-5)+4t\Big)}_{(\star)}\\
&\quad + \underbrace{2u(t+u-2)(t+u+1)(s-1)}_{(\star\star)}\\
&\quad\quad + \underbrace{2(t+u-1)(s-1)\Big(\big(2t+2u-2\big)(s-u-1)-(t+u-2)(s-u)\Big)}_{(\star\star\star)}\\
&\quad\quad\quad + \underbrace{(t+u-1)\Big((t+u)(s-u)(s-u-1)+2u(t+u-2)(s-u)\Big)}_{(\star\star\star\star)}\\
&\quad\quad\quad\quad + \underbrace{u(u-1)(t+u-2)(t+u-3)}_{(\star\star\star\star\star)}
\end{aligned}
\label{eq:N2_final_collect_start}
\end{align}
Expand each line of \eqref{eq:N2_final_collect_start} into a polynomial in $u$:
\begin{align}
(\star)
&=
\begin{aligned}[t]
&-(s-2)\Big[(t-1)\big(t(3s-1)\big)\Big] -(s-2)\Big[(3s-1)\Big]u^2 \\
&\qquad -(s-2)\Big[(t-1)(3s-1)+(t)(3s-1)\Big]u
\end{aligned}
\label{eq:E2_poly}
\\[0.8em]
(\star\star)
&=
\begin{aligned}[t]
&2(s-1)\Big[(t-2)(t+1)\Big]u + 2(s-1)\Big[(2t-1)\Big]u^2
+2(s-1)u^3
\end{aligned}
\label{eq:E3_poly}
\\[0.8em]
(\star\star\star)
&=
\begin{aligned}[t]
&2(s-1)\Big[(t-1)\big((2t-2)(s-1)-(t-2)s\big)\Big] \\
&\quad +2(s-1)\Big[\big((2t-2)(s-1)-(t-2)s\big) + (t-1)\big(2(s-1)-2t+2\big)\Big]u \\
&\quad\quad +2(s-1)\Big[\big(2(s-1)-2t+2\big)+(t-1)(-2)\Big]u^2  +2(s-1)(-2)u^3
\end{aligned}
\label{eq:E4_poly}
\\[0.8em]
(\star\star\star\,\star)
&=
\begin{aligned}[t]
&(t+u-1)\Big((t+u)(s-u)(s-u-1)+2u(t+u-2)(s-u)\Big) \\
&\hspace{-5.1mm}= (t-1)t\,s(s-1) +\big(2s^{2}t-s^{2}-6st+5s+t^{2}-t\big)u \\
&\qquad +\big(s^{2}-5s-t^{2}+7t-5\big)u^{2} +(6-2t)u^{3} - u^{4}
\end{aligned}
\label{eq:E5_poly}
\\[0.8em]
(\star\star\star\star\star)
&=
\begin{aligned}[t]
&u(u-1)(t+u-2)(t+u-3) \\
&\hspace{-5.1mm}= -(t-2)(t-3)u + (t^{2}-7t+11)u^{2} + (2t-6)u^{3} + u^{4}
\end{aligned}
\label{eq:E6_poly}
\end{align}
Lastly, add $2s(s-1)$ with \eqref{eq:E2_poly}--\eqref{eq:E6_poly} and collect coefficients of $u^k$ to obtain
\begin{equation}
\mcN_2 = A_0 + A_1u + A_2u^2 + A_3u^3 + A_4u^4 \,,
\label{eq:Num2_Ak_def}
\end{equation}
where the coefficients are
\begin{align}
A_0
&=2s(s-1)-(s-2)\,t(t-1)(3s-1)+2(s-1)(t-1)(st-2t+2)+(t-1)t\,s(s-1), \notag\\
A_1
&=-(s-2)\,(6st-3s-6t+5)+2(s-1)(t-2)(t+1)\notag\\
&\qquad +2(s-1)(2st-s-t^2-3t+4)+\big(2s^{2}t-s^{2}-6st+5s+t^{2}-t\big)-(t-2)(t-3), \notag\\
A_2
&=-(s-2)(3s-5)+2(s-1)(2t-1)+2(s-1)(s-2t-1)\notag\\
&\hspace{5cm} +\big(s^{2}-5s-t^{2}+7t-5\big)+(t^{2}-7t+11), \notag\\
A_3
&=2(s-1)-2(s-1)+(6-2t)+(2t-6), \notag\\
A_4
&=(-1)+1. \notag
\end{align}
In particular, these simplify to
\begin{equation}
\label{eq:Num2_coeff_simplified}
A_0 = 2(s+t-1)(s+t-2) = 2(N-1)(N-2) \;,
\qquad
A_1=A_2=A_3=A_4=0.
\end{equation}
Substituting \eqref{eq:Num2_coeff_simplified} into \eqref{eq:Num2_Ak_def} gives $\mcN_2 = 2(N-1)(N-2)$, and hence
\begin{align}
\label{eq:J2-expr-final}
\mathscr J_2 
= \frac{\mcN_2}{D}
=\frac{2(N-1)(N-2)}{t'(t'-1)(N-t)(N-t-1)} \,,
\end{align}
completing  the computation of $\mathscr{J}_2$.


%% file: tables/cp-loc-supp-new.tex
\begingroup
\footnotesize
\setlength{\tabcolsep}{2.7pt}
\renewcommand{\arraystretch}{1.02}
\setlength{\LTleft}{0pt plus 1fill}
\setlength{\LTright}{0pt plus 1fill}

\begin{longtable}[c]{
>{\centering\arraybackslash}m{2.0cm}
>{\raggedright\arraybackslash}m{2.25cm}
ccccc
@{}p{0.40cm}@{}
ccccc
@{}p{0.40cm}@{}
ccccc
}

\caption{\footnotesize
Additional empirical localization accuracy results for \(N=40\), true change-point \(\tau=15\), and \(d\in\{200,1000,5000\}\).
Columns \(0,1,2,\geqslant\!3\) report empirical proportions for \(|\widehat{\tau}_d-\tau|=0,1,2\), at least \(3\), respectively; and \(\widehat{\mathbb E}\) reports the empirical mean of \(|\widehat{\tau}_d-\tau|\).
Boldface marks the largest entry in column \(0\) within each example--dimension block.
}
\label{tab:cp_accuracy_appendix}
\\[-1mm]

\toprule
\multirow{2}{1.55cm}{\centering\arraybackslash\textbf{Example}}
&
\multirow{2}{2.25cm}{\centering\arraybackslash\textbf{Method~~~~}}
& \multicolumn{5}{c}{$d = 200$}
& \multicolumn{1}{c}{}
& \multicolumn{5}{c}{$d = 1000$}
& \multicolumn{1}{c}{}
& \multicolumn{5}{c}{$d = 5000$} \\
\cmidrule(l{4pt}r{4pt}){3-7}
\cmidrule(l{4pt}r{4pt}){9-13}
\cmidrule(l{4pt}r{4pt}){15-19}
&
& \textbf{$0$}
& \textbf{$1$}
& \textbf{$2$}
& \textbf{$\geqslant\!3$}
& \(\widehat{\mathbb E}\)
&
& \textbf{$0$}
& \textbf{$1$}
& \textbf{$2$}
& \textbf{$\geqslant\!3$}
& \(\widehat{\mathbb E}\)
&
& \textbf{$0$}
& \textbf{$1$}
& \textbf{$2$}
& \textbf{$\geqslant\!3$}
& \(\widehat{\mathbb E}\) \\
\midrule
\endfirsthead

\caption[]{\footnotesize
Additional empirical localization accuracy results \textit{(continued)}.
}
\\[-1mm]

\toprule
\multirow{2}{1.55cm}{\centering\arraybackslash\textbf{Example}}
&
\multirow{2}{2.25cm}{\centering\arraybackslash\textbf{Method~~~~}}
& \multicolumn{5}{c}{$d = 200$}
& \multicolumn{1}{c}{}
& \multicolumn{5}{c}{$d = 1000$}
& \multicolumn{1}{c}{}
& \multicolumn{5}{c}{$d = 5000$} \\
\cmidrule(l{4pt}r{4pt}){3-7}
\cmidrule(l{4pt}r{4pt}){9-13}
\cmidrule(l{4pt}r{4pt}){15-19}
&
& \textbf{$0$}
& \textbf{$1$}
& \textbf{$2$}
& \textbf{$\geqslant\!3$}
& \(\widehat{\mathbb E}\)
&
& \textbf{$0$}
& \textbf{$1$}
& \textbf{$2$}
& \textbf{$\geqslant\!3$}
& \(\widehat{\mathbb E}\)
&
& \textbf{$0$}
& \textbf{$1$}
& \textbf{$2$}
& \textbf{$\geqslant\!3$}
& \(\widehat{\mathbb E}\) \\
\midrule
\endhead

\endfoot

\bottomrule
\endlastfoot

  & DAK Scan (Ours) & \textbf{1.00} & 0.00 & 0.00 & 0.00 & 0.00 &   & \textbf{1.00} & 0.00 & 0.00 & 0.00 & 0.00 &   & \textbf{1.00} & 0.00 & 0.00 & 0.00 & 0.00\\
\nopagebreak
  & E-Divisive & 0.00 & 0.00 & 0.00 & 1.00 & 26.00 &   & 0.00 & 0.00 & 0.00 & 1.00 & 26.00 &   & 0.00 & 0.00 & 0.00 & 1.00 & 26.00\\
\nopagebreak
  & E-CP3O & 0.00 & 0.00 & 0.00 & 1.00 & 26.00 &   & 0.00 & 0.00 & 0.00 & 1.00 & 26.00 &   & 0.00 & 0.00 & 0.00 & 1.00 & 26.00\\
\nopagebreak
  & KCPA & 0.00 & 0.00 & 0.00 & 1.00 & 26.00 &   & 0.00 & 0.00 & 0.00 & 1.00 & 26.00 &   & 0.00 & 0.00 & 0.00 & 1.00 & 26.00\\
\nopagebreak
  & MMD-$\mcN$ & 0.00 & 0.00 & 0.00 & 1.00 & 12.99 &   & 0.00 & 0.00 & 0.00 & 1.00 & 13.00 &   & 0.00 & 0.00 & 0.00 & 1.00 & 13.00\\
\nopagebreak
  & MMD-$\mcE$ & 0.00 & 0.00 & 0.00 & 1.00 & 14.49 &   & 0.00 & 0.00 & 0.00 & 1.00 & 14.51 &   & 0.00 & 0.00 & 0.00 & 1.00 & 14.49\\
\nopagebreak
  & HDD-DM & \textbf{1.00} & 0.00 & 0.00 & 0.00 & 0.00 &   & \textbf{1.00} & 0.00 & 0.00 & 0.00 & 0.00 &   & \textbf{1.00} & 0.00 & 0.00 & 0.00 & 0.00\\
\nopagebreak
\multirow{-8}{1.7cm}{\centering\arraybackslash \textit{Gaussian Location}} & Sliced-Wass & 0.04 & 0.00 & 0.00 & 0.96 & 13.09 &   & 0.03 & 0.00 & 0.00 & 0.97 & 13.20 &   & 0.03 & 0.00 & 0.00 & 0.97 & 13.12\\
\cmidrule{1-19}

  & DAK Scan (Ours) & 0.70 & 0.01 & 0.00 & 0.29 & 5.51 &   & 0.99 & 0.00 & 0.00 & 0.01 & 0.21 &   & \textbf{1.00} & 0.00 & 0.00 & 0.00 & 0.00\\
\nopagebreak
  & E-Divisive & 0.00 & 0.00 & 0.00 & 1.00 & 26.00 &   & 0.00 & 0.00 & 0.00 & 1.00 & 26.00 &   & 0.00 & 0.00 & 0.00 & 1.00 & 26.00\\
\nopagebreak
  & E-CP3O & 0.00 & 0.00 & 0.00 & 1.00 & 26.00 &   & 0.00 & 0.00 & 0.00 & 1.00 & 26.00 &   & 0.00 & 0.00 & 0.00 & 1.00 & 26.00\\
\nopagebreak
  & KCPA & 0.00 & 0.00 & 0.00 & 1.00 & 26.00 &   & 0.00 & 0.00 & 0.00 & 1.00 & 26.00 &   & 0.00 & 0.00 & 0.00 & 1.00 & 26.00\\
\nopagebreak
  & MMD-$\mcN$ & 0.00 & 0.00 & 0.00 & 1.00 & 23.00 &   & 0.00 & 0.00 & 0.00 & 1.00 & 23.00 &   & 0.00 & 0.00 & 0.00 & 1.00 & 23.00\\
\nopagebreak
  & MMD-$\mcE$ & 0.00 & 0.00 & 0.00 & 1.00 & 8.36 &   & 0.00 & 0.00 & 0.00 & 1.00 & 8.21 &   & 0.00 & 0.00 & 0.00 & 1.00 & 8.26\\
\nopagebreak
  & HDD-DM & \textbf{1.00} & 0.00 & 0.00 & 0.00 & 0.00 &   & \textbf{1.00} & 0.00 & 0.00 & 0.00 & 0.00 &   & \textbf{1.00} & 0.00 & 0.00 & 0.00 & 0.00\\
\nopagebreak
\multirow{-8}{1.7cm}{\centering\arraybackslash \textit{Gaussian Scale}} & Sliced-Wass & 0.00 & 0.00 & 0.00 & 1.00 & 23.00 &   & 0.00 & 0.00 & 0.00 & 1.00 & 23.00 &   & 0.00 & 0.00 & 0.00 & 1.00 & 23.00\\
\cmidrule{1-19}

  & DAK Scan (Ours) & 0.00 & 0.00 & 0.00 & 1.00 & 17.89 &   & 0.00 & 0.00 & 0.00 & 1.00 & 17.50 &   & 0.00 & 0.00 & 0.00 & 1.00 & 17.70\\
\nopagebreak
  & E-Divisive & 0.00 & 0.00 & 0.00 & 1.00 & 26.00 &   & 0.00 & 0.00 & 0.00 & 1.00 & 26.00 &   & 0.00 & 0.00 & 0.00 & 1.00 & 26.00\\
\nopagebreak
  & E-CP3O & 0.00 & 0.00 & 0.00 & 1.00 & 26.00 &   & 0.00 & 0.00 & 0.00 & 1.00 & 26.00 &   & 0.00 & 0.00 & 0.00 & 1.00 & 26.00\\
\nopagebreak
  & KCPA & 0.00 & 0.00 & 0.00 & 1.00 & 26.00 &   & 0.00 & 0.00 & 0.00 & 1.00 & 26.00 &   & 0.00 & 0.00 & 0.00 & 1.00 & 26.00\\
\nopagebreak
  & MMD-$\mcN$ & 0.00 & 0.00 & 0.00 & 1.00 & 18.88 &   & 0.00 & 0.00 & 0.00 & 1.00 & 18.18 &   & 0.00 & 0.00 & 0.00 & 1.00 & 18.42\\
\nopagebreak
  & MMD-$\mcE$ & 0.01 & 0.03 & 0.12 & 0.85 & 4.79 &   & 0.00 & 0.00 & 0.02 & 0.98 & 4.98 &   & 0.00 & 0.00 & 0.00 & 1.00 & 5.05\\
\nopagebreak
  & HDD-DM & \textbf{0.02} & 0.04 & 0.04 & 0.91 & 13.04 &   & \textbf{0.02} & 0.03 & 0.03 & 0.92 & 12.62 &   & \textbf{0.02} & 0.04 & 0.03 & 0.91 & 12.47\\
\nopagebreak
\multirow{-8}{1.7cm}{\centering\arraybackslash \textit{Gaussian Spiked Cov}} & Sliced-Wass & 0.00 & 0.00 & 0.00 & 1.00 & 18.20 &   & 0.00 & 0.00 & 0.00 & 1.00 & 18.10 &   & 0.00 & 0.00 & 0.00 & 1.00 & 18.14\\
\cmidrule{1-19}

  & DAK Scan (Ours) & 0.00 & 0.00 & 0.00 & 1.00 & 20.39 &   & 0.00 & 0.00 & 0.00 & 1.00 & 21.20 &   & 0.00 & 0.01 & 0.00 & 0.99 & 20.83\\
\nopagebreak
  & E-Divisive & 0.00 & 0.00 & 0.00 & 1.00 & 26.00 &   & 0.00 & 0.00 & 0.00 & 1.00 & 26.00 &   & 0.00 & 0.00 & 0.00 & 1.00 & 26.00\\
\nopagebreak
  & E-CP3O & 0.00 & 0.00 & 0.00 & 1.00 & 26.00 &   & 0.00 & 0.00 & 0.00 & 1.00 & 26.00 &   & 0.00 & 0.00 & 0.00 & 1.00 & 26.00\\
\nopagebreak
  & KCPA & 0.00 & 0.00 & 0.00 & 1.00 & 26.00 &   & 0.00 & 0.00 & 0.00 & 1.00 & 26.00 &   & 0.00 & 0.00 & 0.00 & 1.00 & 26.00\\
\nopagebreak
  & MMD-$\mcN$ & 0.00 & 0.00 & 0.00 & 1.00 & 15.86 &   & 0.00 & 0.00 & 0.00 & 1.00 & 15.62 &   & 0.00 & 0.00 & 0.00 & 1.00 & 15.52\\
\nopagebreak
  & MMD-$\mcE$ & 0.01 & 0.02 & 0.03 & 0.93 & 6.62 &   & 0.01 & 0.02 & 0.02 & 0.96 & 6.40 &   & 0.01 & 0.03 & 0.03 & 0.94 & 6.32\\
\nopagebreak
  & HDD-DM & \textbf{0.05} & 0.04 & 0.04 & 0.86 & 11.89 &   & \textbf{0.07} & 0.02 & 0.04 & 0.87 & 11.58 &   & \textbf{0.09} & 0.02 & 0.05 & 0.84 & 11.50\\
\nopagebreak
\multirow{-8}{1.7cm}{\centering\arraybackslash \textit{Gaussians with same marginals}} & Sliced-Wass & 0.00 & 0.00 & 0.00 & 1.00 & 16.92 &   & 0.00 & 0.00 & 0.00 & 1.00 & 16.80 &   & 0.00 & 0.00 & 0.00 & 1.00 & 16.77\\
\cmidrule{1-19}

  & DAK Scan (Ours) & \textbf{1.00} & 0.00 & 0.00 & 0.00 & 0.00 &   & \textbf{1.00} & 0.00 & 0.00 & 0.00 & 0.00 &   & \textbf{1.00} & 0.00 & 0.00 & 0.00 & 0.00\\
\nopagebreak
  & E-Divisive & 0.00 & 0.00 & 0.00 & 1.00 & 26.00 &   & 0.00 & 0.00 & 0.00 & 1.00 & 26.00 &   & 0.00 & 0.00 & 0.00 & 1.00 & 26.00\\
\nopagebreak
  & E-CP3O & 0.00 & 0.00 & 0.00 & 1.00 & 26.00 &   & 0.00 & 0.00 & 0.00 & 1.00 & 26.00 &   & 0.00 & 0.00 & 0.00 & 1.00 & 26.00\\
\nopagebreak
  & KCPA & 0.00 & 0.00 & 0.00 & 1.00 & 26.00 &   & 0.00 & 0.00 & 0.00 & 1.00 & 26.00 &   & 0.00 & 0.00 & 0.00 & 1.00 & 26.00\\
\nopagebreak
  & MMD-$\mcN$ & 0.00 & 0.00 & 0.00 & 1.00 & 13.88 &   & 0.00 & 0.00 & 0.00 & 1.00 & 13.00 &   & 0.00 & 0.00 & 0.00 & 1.00 & 13.00\\
\nopagebreak
  & MMD-$\mcE$ & 0.00 & 0.00 & 0.00 & 1.00 & 12.61 &   & 0.00 & 0.00 & 0.00 & 1.00 & 12.53 &   & 0.00 & 0.00 & 0.00 & 1.00 & 12.56\\
\nopagebreak
  & HDD-DM & \textbf{1.00} & 0.00 & 0.00 & 0.00 & 0.00 &   & \textbf{1.00} & 0.00 & 0.00 & 0.00 & 0.00 &   & \textbf{1.00} & 0.00 & 0.00 & 0.00 & 0.00\\
\nopagebreak
\multirow{-8}{1.7cm}{\centering\arraybackslash \textit{Laplace Location}} & Sliced-Wass & 0.00 & 0.00 & 0.00 & 1.00 & 15.08 &   & 0.00 & 0.00 & 0.00 & 1.00 & 14.66 &   & 0.00 & 0.00 & 0.00 & 1.00 & 14.90\\
\cmidrule{1-19}

  & DAK Scan (Ours) & 0.99 & 0.00 & 0.00 & 0.01 & 0.18 &   & \textbf{1.00} & 0.00 & 0.00 & 0.00 & 0.00 &   & \textbf{1.00} & 0.00 & 0.00 & 0.00 & 0.00\\
\nopagebreak
  & E-Divisive & 0.00 & 0.00 & 0.00 & 1.00 & 26.00 &   & 0.00 & 0.00 & 0.00 & 1.00 & 26.00 &   & 0.00 & 0.00 & 0.00 & 1.00 & 26.00\\
\nopagebreak
  & E-CP3O & 0.00 & 0.00 & 0.00 & 1.00 & 26.00 &   & 0.00 & 0.00 & 0.00 & 1.00 & 26.00 &   & 0.00 & 0.00 & 0.00 & 1.00 & 26.00\\
\nopagebreak
  & KCPA & 0.00 & 0.00 & 0.00 & 1.00 & 26.00 &   & 0.00 & 0.00 & 0.00 & 1.00 & 26.00 &   & 0.00 & 0.00 & 0.00 & 1.00 & 26.00\\
\nopagebreak
  & MMD-$\mcN$ & 0.00 & 0.00 & 0.00 & 1.00 & 23.00 &   & 0.00 & 0.00 & 0.00 & 1.00 & 23.00 &   & 0.00 & 0.00 & 0.00 & 1.00 & 23.00\\
\nopagebreak
  & MMD-$\mcE$ & 0.00 & 0.00 & 0.00 & 1.00 & 10.36 &   & 0.00 & 0.00 & 0.00 & 1.00 & 10.40 &   & 0.00 & 0.00 & 0.00 & 1.00 & 10.41\\
\nopagebreak
  & HDD-DM & \textbf{1.00} & 0.00 & 0.00 & 0.00 & 0.00 &   & \textbf{1.00} & 0.00 & 0.00 & 0.00 & 0.00 &   & \textbf{1.00} & 0.00 & 0.00 & 0.00 & 0.00\\
\nopagebreak
\multirow{-8}{1.7cm}{\centering\arraybackslash \textit{Bernoulli Gaussian}} & Sliced-Wass & 0.00 & 0.00 & 0.00 & 1.00 & 23.00 &   & 0.00 & 0.00 & 0.00 & 1.00 & 23.00 &   & 0.00 & 0.00 & 0.00 & 1.00 & 23.00\\
\cmidrule{1-19}

  & DAK Scan (Ours) & 0.02 & 0.02 & 0.01 & 0.96 & 16.63 &   & 0.24 & 0.05 & 0.01 & 0.70 & 11.95 &   & 0.86 & 0.01 & 0.00 & 0.13 & 2.08\\
\nopagebreak
  & E-Divisive & 0.00 & 0.00 & 0.00 & 1.00 & 26.00 &   & 0.00 & 0.00 & 0.00 & 1.00 & 26.00 &   & 0.00 & 0.00 & 0.00 & 1.00 & 26.00\\
\nopagebreak
  & E-CP3O & 0.00 & 0.00 & 0.00 & 1.00 & 26.00 &   & 0.00 & 0.00 & 0.00 & 1.00 & 26.00 &   & 0.00 & 0.00 & 0.00 & 1.00 & 26.00\\
\nopagebreak
  & KCPA & 0.00 & 0.00 & 0.00 & 1.00 & 26.00 &   & 0.00 & 0.00 & 0.00 & 1.00 & 26.00 &   & 0.00 & 0.00 & 0.00 & 1.00 & 26.00\\
\nopagebreak
  & MMD-$\mcN$ & 0.00 & 0.00 & 0.00 & 1.00 & 23.00 &   & 0.00 & 0.00 & 0.00 & 1.00 & 23.00 &   & 0.00 & 0.00 & 0.00 & 1.00 & 23.00\\
\nopagebreak
  & MMD-$\mcE$ & 0.00 & 0.00 & 0.00 & 1.00 & 11.84 &   & 0.00 & 0.00 & 0.00 & 1.00 & 11.77 &   & 0.00 & 0.00 & 0.00 & 1.00 & 11.76\\
\nopagebreak
  & HDD-DM & \textbf{1.00} & 0.00 & 0.00 & 0.00 & 0.00 &   & \textbf{1.00} & 0.00 & 0.00 & 0.00 & 0.00 &   & \textbf{1.00} & 0.00 & 0.00 & 0.00 & 0.00\\
\nopagebreak
\multirow{-8}{1.7cm}{\centering\arraybackslash \textit{Gaussian mixture}} & Sliced-Wass & 0.02 & 0.00 & 0.00 & 0.98 & 22.59 &   & 0.01 & 0.00 & 0.00 & 0.99 & 22.77 &   & 0.00 & 0.00 & 0.00 & 1.00 & 22.95\\

\bottomrule

\end{longtable}
\endgroup